\documentclass[reqno,12pt,letterpaper]{amsart}
\usepackage{amsmath,amssymb,amsthm,graphicx,mathrsfs,url}
\usepackage[usenames,dvipsnames]{color}
\usepackage[colorlinks=true,linkcolor=Black,citecolor=Black]{hyperref}
\usepackage{amsxtra}
\usepackage{enumitem}
\usepackage{epstopdf}
\usepackage{graphicx}
\usepackage[percent]{overpic}

\usepackage{lmodern}
\makeatletter
\def\Ddots{\mathinner{\mkern1mu\raise\p@
\vbox{\kern7\p@\hbox{.}}\mkern2mu
\raise4\p@\hbox{.}\mkern2mu\raise7\p@\hbox{.}\mkern1mu}}
\makeatother
\usepackage{graphicx}

\usepackage{array}
\usepackage[utf8]{inputenc}
\usepackage{tikz, tikz-3dplot}
\usepackage{pgfplots}
\def\?[#1]{\textbf{[#1]}\marginpar{\Large{\textbf{??}}}}

\let\varepsilon=\varepsilon 
\newcommand{\dd}{\, {\rm d}}

\newtheorem{theorem}{Theorem}
\newtheorem{prop}{Proposition}[section]

\newtheorem{lemma}[prop]{Lemma}
\newtheorem{corrolary}[prop]{Corollary}

\numberwithin{equation}{section}

\title[A review on the kinetic theory of oscillator chains]{A review on the kinetic theory \\ of oscillator chains}

\author{Pierre Germain} 
\address{Fakult\"at f\"ur Mathematik, Universit\"at Wien, Oskar-Morgenstern-Platz 1, 1090 Vienna, Austria}
\email{pierre.germain@univie.ac.at}

\author{Joonhyun La} 
\address{KIAS School of Mathematics, 85 Hoegi-ro, Dongdaemun-gu, Seoul 02455, Republic of Korea.}
\email{joonhyun@kias.re.kr}

\author{Angeliki Menegaki} 
\address{Department of Mathematics, Huxley building, South Kensington campus, Imperial College London, London SW7 2AZ, United Kingdom}
\email{a.menegaki@imperial.ac.uk}

\begin{document}
\maketitle

\begin{abstract}
We review the kinetic theory of one-dimensional nonlinear oscillator chains, of which the most famous example is the Fermi-Pasta-Ulam-Tsingou equation. We provide detailed, though not rigorous, accounts of the microscopic to mesoscopic, and mesoscopic to macroscopic limits: derivation of the kinetic wave equation and hydrodynamic limit. We also present the state of the art of the mathematical theory, including  proofs. We discuss the connection to two famous problems of Mathematical Physics: the Fermi-Pasta-Ulam-Tsingou paradox, and the derivation of Fourier's law. Finally, many open problems and possible directions for future research are proposed.
\end{abstract}

\setcounter{tocdepth}1
\tableofcontents

\section{Introduction}

\subsection{The aim of this review} We aim at presenting an overview of the kinetic theory of nonlinear oscillator chains, from the viewpoints of Physics and Mathematics. To be more precise and more modest, we will present the current state of the art, since this theory is far from complete, in particular lacking a rigorous mathematical formalization of many key aspects. 

The first difficulty one encounters in this subject is the number of names given to the same equation, namely the kinetic equation obtained in the mesoscopic limit from an oscillator chain. It is equivalently known as Kinetic Wave Equation, Phonon Boltzmann, Boltzmann-Peierls, or even kinetic Fermi-Pasta-Ulam-Tsingou equation, in the case of this particular nonlinear oscillator chain. In the present text, the name we adopt is \textit{Kinetic Wave Equation}.

The most important points of this review will be the following.

\begin{itemize}
\item We will describe in detail, though not rigorously, the singular limit from the microscopic (Hamiltonian) dynamics to the mesoscopic equation (kinetic wave equation), as well as the singular limit from the mesoscopic model to the macroscopic (hydrodynamic) equation. .

\medskip

\item Most of the discussion deals with general chains, but, for the sake of concreteness, we will specialize our results to the case of the Fermi-Pasta-Ulam-Tsingou and Discrete Nonlinear Klein-Gordon equations, which are the two most popular models.

\medskip

\item We will give an account of the mathematical theory of the kinetic wave equation, drawing from the recent mathematical literature, and providing essentially complete proofs.

\medskip

\item We will explain how the kinetic wave equation is related to two fundamental problems from Mathematical Physics: the Fermi-Pasta-Ulam-Tsingou paradox and the derivation of Fourier's law. Indeed, we will argue that the kinetic wave equation provides solutions to these two problems - at least if mathematical rigor is disregarded.
\medskip

\item Finally, we give many related questions and open problems for the most motivated readers!
\end{itemize}

The kinetic theory of nonlinear oscillator chains and its applications is a rapidly growing subject, to which great review articles were recently dedicated. We are thinking of Lukkarinen \cite{Luk2016} and Onorato-L'vov-DeMatteis-Chibaro \cite{OnoratoLvovDematteisChibaro}, on which we relied extensively.

\subsection{Phonons, oscillator chains, and wave turbulence} The vibration modes of atoms in a crystal are known as phonons; their Physics is very different from that of the electrons, which we disregard here. In normal conditions in a crystal, atoms cannot deviate much from their equilibrium positions, justifying a weakly nonlinear assumption. In a foundational article, Peierls \cite{Peierls} proposed the application of kinetic theory to the description of phonon dynamics; this is by now a classical part of the theory of phonons, see \cite{Callaway,Prigogine,Srivastava}. These references give general formulas without aiming at describing specific crystals; this is probably due to the very intricate formulas arising in realistic descriptions of three-dimensional crystals, which seem to defy an analytic treatment at the nonlinear level.

Spohn \cite{Spohn2005} proposed as a simplified model a higher-dimensional version of \eqref{NOC}, and developed its theory. Then he turned to one-dimensional models with Aoki and Lukkarinen \cite{AokiLukkSpohn} and Lukkarinen \cite{LukkarinenSpohn2008}, working on the nonlinear oscillator chain \eqref{NOC}. We followed the impetus of these articles by studying the associated kinetic problem from the viewpoint of nonlinear PDE analysis \cite{GermainLaMenegaki,EGLM25}. The aim of the present review is to give a comprehensive account of the state of the art of the kinetic theory associated to \eqref{NOC}.

Finally, the kinetic problem for phonons is part of the broader theory of weak turbulence, which we shall now present succinctly. Besides the kinetic theory of phonons, the work \cite{Peierls} is also laying the foundation for the theory of weak (or wave) turbulence. This theory aims at describing (weakly) nonlinear wave or dispersive equations in a turbulent regime through a kinetic equation. It was rediscovered independently by Hasselmann \cite{Hasselmann1962,Hasselmann1963} who was working on surface gravity waves. Finally, Zakharov and his school \cite{Zakharov} played a critical role in developing the theory; they showed the breadth of its applications and introduced key ideas such as the power law solutions known as Kolmogorov-Zakharov spectra. For modern accounts of wave turbulence theory, we refer to \cite{Nazarenko,Galtier}.

In its classical applications - see for instance the last few references - wave turbulence theory aims at describing the turbulent spectrum of nonlinear wave systems submitted to forcing and dissipation at difference scales, which causes the formation of an energy cascade. This is meaningful since these equations are set on Euclidean space, and the frequency parameter can be identified to the real line (disregarding the direction). This must be contrasted with equations such as \eqref{NOC} and its kinetic version \eqref{KWE}, for which the frequency space can be identified with the torus. As a consequence, frequency cascades and turbulent spectra are not (or less) relevant, but stationary solutions become the main characters: these are the Rayleigh-Jeans solutions, which will be defined in the following. They are universal attractors and will play a key role in the theory.

\subsection{From microscopic to macroscopic through mesoscopic} 

At the heart of the present review are two singular limits: from microscopic to mesoscopic, and from mescoscopic to macroscopic. Here, we summarize in a few paragraphs the different models involved.

\subsubsection{The microscopic (Hamiltonian) model}
Our starting point is the \textit{nonlinear oscillator chain}
\begin{equation}
\tag{NOC} \label{NOC}
\boxed{\ddot{q_j} = -  \partial_{q_j} F( (q_j) ), \qquad q_j \in \mathbb{R}, \quad j \in \mathbb{Z};}
\end{equation}
here and in the following, we denote $\dot f$ for the time derivative of $f$ and omit time-dependence in the notation; furthermore, the function $F((q_j))$ is assumed to be translation invariant in the $j$ variable. The two most famous examples are the \textit{Fermi-Pasta-Ulam-Tsingou equation}
\begin{equation} \tag{FPUT}
\boxed{\ddot{q_j} = V'(q_{j+1} - q_j) -  V'(q_{j} - q_{j-1}), \qquad V(x) = \frac 12 x^2 + \frac \alpha 3 x^3 + \frac \beta 4 x^4} \\
\end{equation}
(where $\alpha,\beta \in \mathbb{R}$ are parameters) and the \textit{Discrete Nonlinear Klein-Gordon equation}
\begin{equation} \tag{DNLKG}
\boxed{\ddot{q_j} = -\delta (2 q_{j} - q_{j+1} - q_{j-1}) - (1-2\delta) q_j - \lambda q_j^2 - \mu q_j^3},
\end{equation}
(where $\delta \in (0,\frac 12)$ and $\lambda,\mu \in \mathbb{R}$ are parameters)

Interpreting $q_j$ as the deviation of a particle from its rest position, such oscillator chains are naive and one-dimensional models for crystals.

\subsubsection{The mesoscopic (kinetic) model}
In the weakly turbulent regime characterized by weak nonlinearity, large volume and chaos (random phase approximation), this Hamiltonian system can be described by a kinetic equation. In this microscopic to mesoscopic limit, only some statistical quantities are retained; the resulting equation is easier to approach and statistical physics effects (such as the H-theorem) are transparent.
As we will argue, cubic nonlinearities are the most relevant, in which case the \textit{kinetic wave equation} takes the following form
\begin{equation}
\tag{KWE} \label{KWE}
\boxed{\partial_t f + \omega'(\xi) \partial_x f = \mathcal{C}(f), \qquad f = f(t,x,\xi) \geq 0, \qquad (x,\xi) \in \mathbb{T} \times \mathbb{T}.}
\end{equation}
Here, we denote $\mathbb{T} = \mathbb{R} / (2\pi \mathbb{Z})$, the \textit{dispersion relation} $\omega(\xi)$ is a function on $\mathbb{T}$ and the \textit{collision operator} is given by the formula
\begin{equation}
\label{formulacollision}
\boxed{
\mathcal{C}(f)(\xi) = \int K_{0,1,2,3} \left[ f_1 f_2 f_3 + f_0 f_2 f_3 - f_0 f_1 f_2 - f_0 f_1 f_3 \right] \delta(\Omega_{0,1,2,3}) \delta (\Sigma_{0,1,2,3}) \dd \xi_{1,2,3},}
\end{equation}
where we denoted
$$
\Omega_{0,1,2,3} = \omega_0 + \omega_1 - \omega_2 - \omega_3, \qquad \Sigma_{0,1,2,3} = \xi_0 + \xi_1 - \xi_2 - \xi_3
$$
and we used the customary notation 
$$
\xi_0 = \xi, \quad \omega_i = \omega(\xi_i), \quad K_{0,1,2,3} = K(\xi,\xi_1,\xi_2,\xi_3), \quad \mbox{etc...}
$$
and where the kernel $K$ enjoys the properties
\begin{equation}
\label{symmetryK}
K(\xi,\xi_1,\xi_2,\xi_3) \geq 0 
\qquad \mbox{and} \qquad
\begin{cases}
K(\xi,\xi_1,\xi_2,\xi_3) = K(\xi_2,\xi_3,\xi,\xi_1) \\
K(\xi,\xi_1,\xi_2,\xi_3) = K(\xi_1,\xi,\xi_2,\xi_3).
\end{cases}
\end{equation}

We considered above the simplest inhomogeneous setup, where $(x,\xi) \in \mathbb{T} \times \mathbb{T}$, but bounded domains in $x$ with appropriate boundary conditions are also of interest. Finally, dropping the $x$ variable yields the simpler homogeneous problem $\partial_t f = \mathcal{C}(f)$, which is already quite interesting.

\subsubsection{Macroscopic (fluid) equations} 
Consider the ansatz
\begin{equation}
\label{ansatzfg}
f(t,x,\xi) = \mathfrak{f}_{\beta,\gamma}(\xi)^2 \left[ 1 + \varepsilon g(\varepsilon^A t,\varepsilon^B x,\xi) \right],
\end{equation}
where $\mathfrak{f}_{\beta,\gamma}(\xi)$ is a stationary solution of \eqref{KWE} called Rayleigh-Jeans depending on two parameters $\beta,\gamma$, and $A,B>0$ are constants to be determined. Several cases have to be distinguished here: first, the problem is called \textit{degenerate} if the linearized operator $L$ of \eqref{KWE} around the Rayleigh-Jeans solution is not invertible; second, $(\beta,\gamma)$ in the above ansatz can be chosen to be constant, or to depend on the space and time variables.

\medskip

\noindent \underline{Non-degenerate case, non constant Rayleigh-Jeans.} Then the asymptotic dynamics of $(\beta,\gamma)$ is given by
$$
\boxed{A(\beta,\gamma) \partial_t \begin{pmatrix} \beta \\ \gamma \end{pmatrix} = \partial_x \left( D(\beta,\gamma) \partial_x \begin{pmatrix} \beta \\ \gamma \end{pmatrix} \right),}
$$
where $A$ and $D$ are positive symmetric matrices (see Section \ref{sec:Hydro_limit} for formulas), making the above a nonlinear diffusion equation. Restricting to the case $\gamma = 0$, the formula simplifies to
$$
\boxed{\partial_t \beta = c_0 \beta^2 \partial_x (\beta \partial_x \beta),}
$$
for a constant $c_0>0$.

\medskip
\noindent \underline{Non-degenerate case, constant Rayleigh-Jeans} In this case, $(\beta,\gamma)$ in \eqref{ansatzfg} are given a fixed value $(\beta_0,\gamma_0)$, and $g$ can be decomposed as
$$
g(t,x,\xi) = b(t,x) \omega(\xi) \mathfrak{f}_{\beta_0,\gamma_0} (\xi) + c(t,x) \mathfrak{f}_{\beta_0,\gamma_0} (\xi) + \{ \mbox{error} \}.
$$
The limiting behavior of $(b,c)$ can be given by one of the three following equations
\begin{itemize}
\item Linear diffusion equation: for symmetric non-negative $2 \times 2$ matrices $A$ and $D$,
$$
\boxed{A \partial_t \begin{pmatrix} b\\ c \end{pmatrix} + D \partial_x^2 \begin{pmatrix} b\\ c \end{pmatrix} = 0.}
$$
\item Conservation law: for an $\mathbb{R}^2$-valued function $P$ whose entries are cubic homogeneous polynomials in $v$,
$$
\boxed{A \partial_t \begin{pmatrix} b\\ c \end{pmatrix} + \partial_x P \begin{pmatrix} b\\ c \end{pmatrix} = 0.}
$$
\item Nonlinear diffusion equation: it combines the two above cases
$$
\boxed{A \partial_t \begin{pmatrix} b\\ c \end{pmatrix} + D \partial_x^2 \begin{pmatrix} b\\ c \end{pmatrix} + \partial_x P \begin{pmatrix} b\\ c \end{pmatrix}= 0.}
$$
\end{itemize}
(details and formulas for $A,D,P$ can be found in Section \ref{sec:Hydro_limit}).

\medskip

\noindent \underline{Degenerate case, constant Rayleigh-Jeans} We only mention here the case of perturbations of the RJ solution $\frac{1}{\omega(\xi)}$ in the \eqref{KWE} derived from (FPUT), and we further restrict to the linearized problem. Then the limiting equation is an anomalous diffusion
\begin{equation*}
\boxed{\partial_t b + c_0 (-\Delta)^{8/5} b = 0,}
\end{equation*}
where $c_0,c_1 >0$
(see Section \ref{sec:Hydro_limit_degenerate} for a more detailed discussion).

\noindent \subsection{Main results and organization of the article} The present review combines heuristic discussions and rigorous mathematical arguments, and we shall of course distinguish carefully between these two cases! We now want to summarize some of the most important results which can be found in this text.

\subsubsection{Heuristic results} \underline{Weakly nonlinear theory of \eqref{NOC}.} Section \ref{section_microscopic} is devoted to the \eqref{NOC} in the weakly nonlinear regime. We establish the formulation most adapted to the weakly nonlinear regime, and use it to discuss nonlinear resonances and their dynamical relevance.

\medskip

\noindent \underline{The weakly turbulent regime.} Section \ref{section_weaklyturbulent} is dedicated to the description of this limiting regime, which allows to transition from a microscopic to a mesoscopic description. We show how key objects (energy, entropy, and stationary solutions) converge from their microscopic to their mesoscopic avatars.

\medskip

\noindent \underline{Kinetic theory of general oscillator chains.} In Section \ref{Sec: Derivation_WKE}, we provide a full derivation of the kinetic wave equation \eqref{KWE} for nonlinear oscillator chains \eqref{NOC} for general potential functions $F$ in the absence of quadratic resonances. In particular, we provide detailed time scales and formulas, which were often implicit in the literature.

\medskip

\noindent \underline{Kinetic theory of (FPUT) and (DNLKG)} In Section \ref{section_examples}, we give complete formulas for the kinetic theory associated to the standard models of Fermi-Pasta-Ulam-Tsingou and Discrete-Nonlinear-Klein-Gordon equation. In the most involved cases of (FPUT$\alpha$) and quadratic (DNLKG), we were not able to find these derivations or formulas in print.

\medskip

\noindent \underline{First properties of \eqref{KWE}} Section \ref{section_first_look} delves into the structure of the homogeneous \eqref{KWE}. The first difficulty is to make sense of the collision integral, which is far from obvious, as we shall see. Then we discuss monotonic quantities and stationary solutions to understand the global dynamics.

\medskip

\noindent \underline{Hydrodynamic limit of \eqref{KWE} in the non-degenerate case} Section \ref{sec:Hydro_limit} is dedicated to the hydrodynamic limits of \eqref{KWE} when the linearized operator of \eqref{KWE} around RJ solution enjoys a spectral gap - which we call the nondegenerate case. We give a full description of the hydrodynamic limits, with fluid equations that seem to be new. 

\medskip

\noindent \underline{Hydrodynamic limit of \eqref{KWE} in the degenerate case} Section \ref{sec:Hydro_limit_degenerate} is dedicated to this question. Here, we give a sketch of the derivation of an anomalous diffusion equation due to \cite{mellet2015anomalous}.

\subsubsection{Rigorous results} \underline{Stationary solutions.} Section \ref{section_rigorous_stationary} focuses on the stationary solutions of \eqref{KWE}. Explicit stationary solutions are known: they form a 2-parameter family and are called Rayleigh-Jeans (RJ for short). It is expected that all stationary solutions are RJ, and this can be established in the case of (FPUT). We also discuss the minimization properties of RJ solutions. We present these results following \cite{LukkarinenSpohn2008,EGLM25}.

\medskip

\noindent \underline{Asymptotic stability of RJ solutions.} Section \ref{section_rigorous_FPU} develops the rigorous theory of the \eqref{KWE} associated to (FPUT). We show local well-posedness and give the main steps of the proof of asymptotic stability of RJ solutions, following \cite{GermainLaMenegaki}.

\subsubsection{Application to two Mathematical Physics riddles} Part of the interest of  \eqref{NOC} is that a large body of work in Mathematical Physics crystallized around the question of its out-of-equilibrium statistical physics. This is without a doubt related to the discrete nature of the equation, which makes computations much simpler than for field equations.

\medskip

\noindent \underline{The (FPUT) paradox} In Section \ref{sec:FPUT_paradox}, we recount the numerical experiments performed by Fermi-Pasta-Ulam-Tsingou in the 50's and their unexpected outcome: rather than converging to a statistical equilibrium, the solutions of the (FPUT) equation displayed periodic-in-time behavior. This observation aroused great interest amongst physicists and mathematicians; we will discuss the explanations which were proposed and the relevance of \eqref{KWE} in this context.

\medskip

\noindent \underline{Derivation of Fourier's law} Section \ref{section_Fourier} asks how to derive Fourier's law of heat conductivity from microscopic dynamics in a chain of oscillators. This is a famously difficult question from Mathematical Physics. It is natural to expect the \eqref{KWE} to play an intermediary role between the microscopic and macroscopic scales, just like the Boltzmann equation mediates between Newtonian particle dynamics and fluid mechanics. {Formal calculations show that the conductivity is determined by the decay rate of the semigroup of the linearized operator of (KWE) around
RJ solution. This gives indications of whether the system exhibits regular or anomalous energy transport.}   

\subsubsection{Perspectives} \underline{Related equations} are described in Section \ref{section_related}. Indeed, \eqref{NOC} can be seen as a prototype of Hamiltonian problems set on lattices, which invites many generalizations: one and higher-dimensional problems, classical and quantized systems, scalar and vector-valued equations. From a mathematically rigorous viewpoint, this is uncharted territory for the most part, but these models appear to be classical in the theory of phonons.

\medskip

\noindent \underline{A list of open mathematical problems} on the kinetic theory of \eqref{NOC} is given in Section \ref{section_open}. Indeed, most of the theory remains to be rigorously justified! Some problems are approachable, and some much more challenging!

\subsection{Acknowledgements} The authors are grateful to Thierry Bodineau, Emeric Bouin, Zaher Hani, Clément Mouhot, Stefano Olla, Miguel Onorato, Nathan Paratre and Herbert Spohn for comments on an earlier version of this review. 
This review is a contribution to the \emph{Festum Pi} 2025 proceedings. We are especially grateful to Cédric Villani for the invitation to contribute to this volume. PG was supported by a Wolfson fellowship from the Royal Society. AM acknowledges support by the Engineering and Physical Sciences Research Council with reference UKRI2025. JL acknowledges support by the Korea Institute for Advanced Study with reference MG094301.

\section{The microscopic model}
\label{section_microscopic}

In this section, we discuss the first properties of general nonlinear oscillator chain \eqref{NOC} without attempting to prove theorems. Taking advantage of translation invariance, the system is best understood after taking the Fourier transform, which gives in particular the dispersion relation of the linearized problem around zero. For the weakly nonlinear problem, resonances become the key to understand the dynamics; we define and discuss them. 

\subsection{The Hamiltonian and the equation}
We consider a Hamiltonian system of particles, labeled by their index $j \in \mathbb{Z}$ and described by their displacement from equilibrium $q_j \in \mathbb{R}$ and momentum $p_j \in \mathbb{R}$.
Denoting $\mathbf{p} = (p_j)_{j \in\mathbb{Z}}$ and $\mathbf{q} = (q_j)_{j \in \mathbb{Z}}$, we will consider Hamiltonians is of the form
$$
\mathscr{H}( \mathbf{p},\mathbf{q} ) = \frac{1}{2} \sum_j p_j^2 + F (\mathbf{q}).
$$
What do we want to assume on the potential function $F$? It should be smooth, such that $F(0) = F'(0) = 0$, and translation invariant: $F((q_j)_{j \in \mathbb{Z}}) = F((q_{j+1})_{j \in \mathbb{Z}})$. 

Denoting time differentiation with a dot ($\frac \dd {\dd t }f = \dot f$), the dynamics are deduced from the Hamiltonian through
$$
\begin{cases}
\dot{q_j} = p_ j \\
\dot{p_j} = - \partial_{q_j} F( \mathbf{q})
\end{cases}
\qquad \mbox{or} \qquad
\ddot{q_j} = -  \partial_{q_j} F( \mathbf{q}).
$$

A large body of research, both in Physics and Mathematics, is dedicated to the study of this equation in the regime which one might call deterministic, where the ideas and tools of Hamiltonian mechanics apply: solitons (traveling waves) \cite{FrieseckeWattis,Pankov}, breathers (time periodic, non propagating wave) \cite{Aubry,Bambusi1996,FlachWillis} and the continuous (long wave) limit, which yield nonlinear dispersive fields equations \cite{HongChulkwang,SchneiderWayne}. However, our focus in the present review will be the weakly turbulent regime, which will be introduced progressively.

We now look more closely at the structure of the potential function $F$; 
since it is smooth, we can expand it in powers of $q$, and the two first terms of the expansion vanish since $F(0) = F'(0) = 0$. Thus
$$
F (\mathbf{q}) = \sum_{n=2}^N F_n (\mathbf{q}) + O(|q|^{N+1}), \qquad F_n \;\mbox{homogeneous of order $n$}
$$
Translation invariance means that there exist (real) coefficients $(\alpha^{(n)}_{k_1,\dots,k_{n-1}})_{ k_1,\dots,k_{n-1} \in \mathbb{Z}}$ such that
$$
F_n( \mathbf{q}) = \frac 1 n \sum_{j \in \mathbb{Z}} \sum_{k_1, \dots, k_{n-1} \in \mathbb{Z}}  \alpha^{(n)}_{k_1, \dots, k_{n-1}} q_{j} q_{j-k_1} \dots q_{j-k_{n-1}}.
$$
These coefficients can be symmetrized without loss of generality to ensure that
\begin{equation}
\label{symmetry_alpha}
\begin{split}
& \alpha^{(2)}_k = \alpha^{(2)}_{-k} \qquad \forall k \\
& \alpha^{(3)}_{k_1,k_2} = \alpha^{(3)}_{k_2,k_1} = \alpha^{(3)}_{-k_1,k_2-k_1} \qquad \forall k_1,k_2 \\
& \alpha^{(4)}_{k_1,k_2,k_3} = \alpha^{(4)}_{k_2,k_1,k_3} = \alpha^{(4)}_{k_3,k_1,k_2} = \alpha^{(4)}_{-k_1,k_2-k_1,k_3-k_1} \qquad \forall k_1,k_2,k_3 \;\;\;\mbox{etc...}
\end{split}
\end{equation}
To explain this symmetrization procedure, we consider quadratic terms first, which can be written
$$
\sum_{j,k} \alpha_k^{(2)} q_j q_{j-k} = \alpha_0^{(2)} \sum_j q_j^2 + \sum_{j_1 > j_2} \left[ \alpha^{(2)}_{j_1-j_2} + \alpha^{(2)}_{j_2-j_1} \right] q_{j_1} q_{j_2}.
$$
Therefore we can assume that $\alpha^{(2)}_{j_1-j_2} = \alpha^{(2)}_{j_2 - j_1}$ for all $j_1,j_2$, which means $\alpha^{(2)}_k = \alpha^{(2)}_{-k}$. Turning to cubic terms, they can be written
\begin{align*}
& \sum_{j,k_1,k_2} \alpha^{(3)}_{k_1,k_2} q_{j} q_{j-k_1} q_{j-k_2} \\
& \qquad = \sum_{j_1 > j_2 > j_3} q_{j_1} q_{j_2} q_{j_3} \left[ \alpha^{(3)}_{j_1-j_2,j_1-j_3} +  \alpha^{(3)}_{j_1-j_3,j_1-j_2} + \alpha^{(3)}_{j_2-j_1,j_2-j_3} \right. \\
& \qquad \qquad \qquad \qquad \qquad \left. + \alpha^{(3)}_{j_2-j_3,j_2-j_1} + \alpha^{(3)}_{j_3-j_1,j_3-j_2} + \alpha^{(3)}_{j_3-j_2,j_3-j_1}\right] + \dots
\end{align*}
(where the $\dots$ stand for terms for which $j_1,j_2,j_3$ are not distinct, which are omitted for simplicity). Upon symmetrization, we can thus assume that, for any $j_1,j_2,j_3$,
$$
\alpha^{(3)}_{j_1-j_2,j_1-j_3} = \alpha^{(3)}_{j_1-j_3,j_1-j_2} = \alpha^{(3)}_{j_2-j_1,j_2-j_3} = \alpha^{(3)}_{j_2-j_3,j_2-j_1} = \alpha^{(3)}_{j_3-j_1,j_3-j_2} = \alpha^{(3)}_{j_3-j_2,j_3-j_1},
$$
which implies the condition in \eqref{symmetry_alpha} for cubic coefficients. The discussion for quartic terms is similar and will be omitted.

\medskip

Writing down the first terms in the expansion of $F$, the evolution problem becomes
$$
\ddot{q_j}  = - \sum_k \alpha^{(2)}_{k} q_{j-k} - \sum_{k_1,k_2} \alpha^{(3)}_{k_1,k_2} q_{j-k_1} q_{j-k_2} + \sum_{k_1,k_2,k_3} \alpha^{(4)}_{k_1,k_2,k_3} q_{j-k_1} q_{j-k_2} q_{j-k_3} + \dots$$

What are the conserved quantities through this evolution?
\begin{itemize}
\item For a general choice of $F$, the only conserved quantity is the Hamiltonian itself.
\item If $F$ is invariant by uniform translation of all particles, $F((q_j)_{j \in \mathbb{Z}}) = F((q_j+a)_{j \in \mathbb{Z}})$ for all $a \in \mathbb{R}$, the system is called \textit{unpinned}. In that case, Noether's theorem gives a further conservation law, namely the total momentum 
$$\mathscr{P}(\mathbf{p}) = \sum_j p_j.$$
\item Finally, we will discuss later the case of the Toda lattice (given by a specific choice of $F$) which is an integrable system and as such enjoys an infinite number of conservation laws.
\end{itemize}

\subsection{The dispersion relation}
We adopt the following normalization for the Fourier transform $\widehat{u}(\xi)$, $\xi \in \mathbb{T} = \mathbb{R} / (2\pi \mathbb{Z})$ of a function on the lattice $(u_k)$, $k\in \mathbb{Z}$:
$$
\widehat{u}(\xi) = \sum_{k \in \mathbb{Z}} u_k e^{-ik\xi} \quad \Longleftrightarrow \quad  u_k = \frac 1 {2\pi} \int_{\mathbb{T}} \widehat{u}(\xi) e^{i k \xi} \dd \xi.
$$
With this normalization, we record the formulas
\begin{equation}
\label{formulafourier}
\widehat{a * b}(\xi) = \widehat{a}(\xi)  \widehat{n}(\xi), \qquad  
\widehat{a \cdot b}(\xi) = \frac{1}{2\pi} \int \widehat{a}(\eta) \widehat{b}(\xi-\eta) \dd \eta, \qquad \widehat{q_{k\pm 1}}(\xi) = e^{\pm i \xi} \widehat{q}(\xi).
\end{equation}

The Fourier transform will play a central role in the following since it diagonalizes the linear part of the equation (which is in turn a consequence of the translation invariance assumption made initially):
$$
\ddot{q_j}  = - \sum_k \alpha_k q_{j-k} \quad \Longleftrightarrow \quad \ddot{\widehat{q}}(\xi) = -\widehat{\alpha}(\xi) \widehat{q}(\xi)
$$
(here and in the following, we are simply denoting $\alpha$ instead of $\alpha^{(2)}$ for simplicity). Since $\alpha$ is real-valued and even, so is its Fourier transform. The linear stability condition for the above ODE is
$$
\widehat{\alpha}(\xi) \geq 0, \qquad \xi \in \mathbb{T},
$$
and we will henceforth assume that it is satisfied. This leads to the definition of the \textit{dispersion relation}
$$
{\omega}(\xi) = \sqrt{ \widehat{\alpha}(\xi)}
$$
(notice that $\omega$ is even). We will now review some classical examples of dispersion relations. 

\medskip \noindent
\underline{Born-von Karman chain.} In this standard model, we think of particles as  connected by linear (Hookean) springs. In other words, the quadratic part of $F$ is a positive linear combination of $\sum_j |q_j-q_{j+k}|^2$, or equivalently  $\alpha_k \leq 0$ for $k\neq 0$ and $\alpha_0 = - 2 \sum_{k \neq 0} \alpha_k$, or equivalently $\widehat{\alpha}$ can be represented as series in $\sin \left(\frac{k x}{2}\right)^2$ with positive coefficients.

The most basic case is when only neighboring particles interact through the potential energy, which implies (up to a constant factor) $\alpha_k = - \delta_{1,k} + 2 \delta_{0,k} - \delta_{-1,k}$ and
$$
\omega(\xi) = 2 \left| \sin \left( \frac{\xi}{2} \right) \right|.
$$

\medskip \noindent
\underline{Nearest neighbor interaction.} Here, we assume that particles only interact with themselves or with their direct neighbors through the potential energy. In other words $\alpha_j = 0$ if $|j| \geq 2$. Up to a multiplicative constant, the stability condition gives that
$$
\omega(\xi) = \sqrt{1 - 2\delta \cos \xi}, \qquad 0 \leq \delta \leq \frac 12.
$$
If $\delta = 0$, the above degenerates to $1$, while if $\delta = \frac{1}{2}$,
$$
\omega(\xi) = \sqrt 2 \left| \sin \left( \frac \xi 2 \right) \right|,
$$
which already appeared above. To understand the physical meaning of $\delta$, it is helpful to write the quadratic part of the potential (which gives the linear part of the equation)
$$
F_{\operatorname{quad}}(\mathbf{q}) = \sum_{j \in \mathbb{Z}} \frac{q_j^2}{2} - \delta q_j q_{j+1} = \sum_{j \in \mathbb{Z}} \left[ \frac{1}{2} - \delta \right] q_j^2 + \frac{\delta}{2} ( q_j - q_{j+1})^2.
$$
Here, the term with the coefficient $\frac{1}{2} - \delta$ corresponds to a pinning onsite potential while the term with the coefficient $\delta$ is of Born-von Karman type.

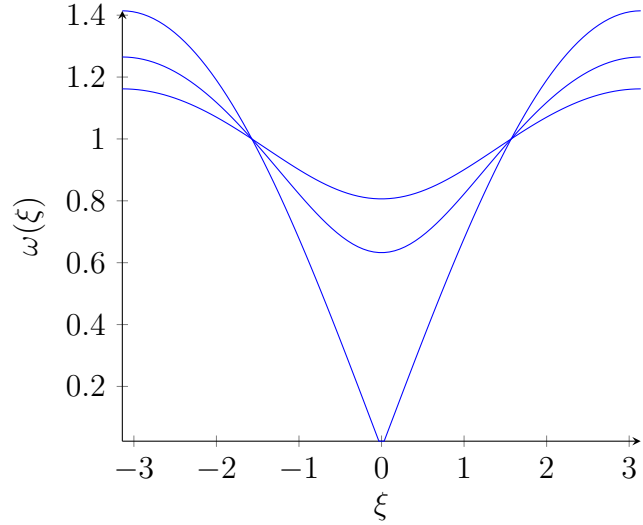
\begin{figure}
    \begin{tikzpicture}
\begin{axis}[
    axis lines = left,
    xlabel = \(\xi\),
    ylabel = {\(\omega(\xi)\)},
]

\addplot [
    domain=-3.14:3.14, 
    samples=100, 
    color=blue,
]
{sqrt(1-0.35*cos(57.3*x))};

\addplot [
    domain=-3.14:3.14, 
    samples=100, 
    color=blue,
    ]
{sqrt(1-0.6*cos(57.3*x))};

\addplot [
    domain=-3.14:3.14, 
    samples=100, 
    color=blue,
    ]
{sqrt(1-cos(57.3*x))};
\end{axis}
\end{tikzpicture}
\caption{$\omega(\xi) = \sqrt{1-2 \delta \cos \xi}$, the nearest-neighbor dispersion relation, for $\delta = \frac{1}{10} , \frac{3}{10} , \frac 12$}
\end{figure}

\medskip \noindent
\underline{Next-to-nearest neighbor interaction} This corresponds to requiring that $\alpha_j = 0$ for $|j| \geq 3$. Then
$$
\omega(\xi) = \sqrt{A + B \cos (\xi) + C \cos(2\xi)},
$$
where the constants $A,B,C$ within the square root are such that this quantity is non-negative. This gives a more general class than nearest-neighbor interaction, used in \cite{Spohn2005,LukkarinenSpohn2008} to illustrate more general behaviors (e.g. quadratic resonances which, as we shall see, are absent in the nearest neighbor case). Moreover, it is used to give a more realistic description of some crystals, see Chapter 6 in \cite{Maugin}.

\medskip \noindent
\underline{Smooth dispersion relations} are equivalent to rapidly decaying coupling  $\alpha_j$; thus any smooth, positive and even function for $\omega$ is associated to an interaction coefficient $(\alpha_j)$.

\medskip \noindent
\underline{Long range interactions} are relevant in a number of physical situations \cite{AvilaPereiraTeixeira,IubiniDiCintioLepriLiviCasetti,MNDKR}. They
typically correspond to polynomially decaying coupling: 
$$
\alpha_j \sim |j|^{-\gamma} \qquad \mbox{as $j \to \infty$, with $\gamma>0$}.
$$
For the dispersion relation, this corresponds to a power singularity at $0$:
$$
\omega(\xi) - \omega(0) \sim |\xi|^{\frac{\gamma -1}{2}}.
$$

\subsection{Weakly nonlinear form of the equation} 

When a weak nonlinearity is added to the linear equation, it is most natural to adopt a set of coordinates in which the dynamics would be trivial (constant) in the linear case, and therefore slow in the weakly nonlinear framework. This will be our aim in this section.

For simplicity in the notations, we will henceforth neglect nonlinear terms of order $\geq 4$. This is because equations of interest are of order $2$ or $3$; and because higher order terms can be analyzed by following the exact same pattern. 

\medskip

\noindent \underline{Step 1: viewing the nonlinear problem in the Fourier variable}. Our equation becomes
\begin{equation}
\begin{split}
\label{2d_order_eq_q}
\ddot{\widehat{q}}(\xi) & = - \omega(\xi)^2 \widehat{q}(\xi) -  \frac{1}{2\pi} \int \widehat{\alpha^{(3)}}(\xi_1,\xi_2) \widehat{q}(\xi_1) \widehat{q}(\xi_2) \delta(\xi - \xi_1 - \xi_2) \dd \xi_1 \dd \xi_2 \\
& \qquad \qquad - \frac{1}{(2\pi)^2} \int \widehat{\alpha^{(4)}}(\xi_1,\xi_2,\xi_3) \widehat{q}(\xi_1) \widehat{q}(\xi_2)  \widehat{q}(\xi_3) \delta(\xi - \xi_1 - \xi_2 - \xi_3) \dd \xi_1 \dd \xi_2 \dd \xi_3 + \dots
\end{split}
\end{equation}
We pause to record the symmetry properties of $\widehat{\alpha}(\xi_1,\xi_2)$ and $\widehat{\alpha}(\xi_1,\xi_2,\xi_3)$: they satisfy 
\begin{equation}
\label{symmetry_alpha_hat}
\begin{split}
& \widehat{\alpha^{(3)}}(\xi_1,\xi_2) = \overline{\widehat{\alpha^{(3)}}(-\xi_1,-\xi_2)} = \widehat{\alpha^{(3)}}(\xi_2,\xi_1) = \widehat{\alpha^{(3)}}(\xi_1,-\xi_1-\xi_2),\\
& \widehat{\alpha^{(4)}}(\xi_1,\xi_2,\xi_3) =  \overline{\widehat{\alpha^{(4)}}(-\xi_1,-\xi_2,-\xi_3)} = \widehat{\alpha^{(4)}}(\xi_1,\xi_3,\xi_2) = \widehat{\alpha^{(4)}}(\xi_3,\xi_2,\xi_1) \\
& \qquad \qquad \quad \;\;\;= \widehat{\alpha^{(4)}}(\xi_1,\xi_2,-\xi_1-\xi_2-\xi_3).
\end{split}
\end{equation}
All equalities but the last one in each line follows immediately by reality of $\alpha$ and symmetry of $\alpha$ by permutation of its indices. To check the last equality, we change the sommation variables and use \eqref{symmetry_alpha} as follows
\begin{align*}
& \widehat{\alpha^{(3)}}(\xi_1,-\xi_1-\xi_2) = \sum_{k_1,k_2} \alpha^{(3)}_{k_1,k_2} e^{-ik_1 \xi_1 + i k_2(\xi_1+\xi_2)} = \sum_{m,n} \alpha^{(3)}_{-m,n-m} e^{-im \xi_1 - i n \xi_2} = \widehat{\alpha^{(3)}}(\xi_1,\xi_2) \\
& \widehat{\alpha^{(4)}}(\xi_1,\xi_2,-\xi_1-\xi_2-\xi_3)  = \sum_{k_1,k_2,k_3} \alpha^{(4)}_{k_1,k_2,k_3} e^{-ik_1 \xi_1 - i k_2 \xi_2 + ik_3(\xi_1 + \xi_2 +\xi_3)} \\
& \qquad \qquad= \sum_{k_1,k_2} \alpha^{(4)}_{n-\ell,m-\ell,-\ell} e^{-im \xi_1 - i n \xi_2 -i \ell \xi_3}  = \widehat{\alpha^{(4)}}(\xi_1,\xi_2,\xi_3).
\end{align*}

\medskip

\noindent \underline{Step 2: first order system.} To rewrite \eqref{2d_order_eq_q} as a first order system, we set
$$
b(\xi) = \frac{1}{\sqrt{ \omega(\xi)}} \left[ i  \widehat{p}(\xi) + \omega(\xi) \widehat{q}(\xi) \right].
$$
By reality of $q$ and $p$, the unknown $q$ can be reconstructed through 
\begin{equation}\label{eq:q hat reconstr}
    \widehat{q}(\xi) = \frac{1}{2\sqrt{\omega(\xi)}} \left[ b(\xi) + \overline{b(-\xi)} \right]
\end{equation}
and we have furthermore
$$
\ddot{\widehat{q}}(\xi) + \omega(\xi)^2 \widehat{q}(\xi) =\sqrt{\omega(\xi)} (-i \partial_t + \omega(\xi)) b(\xi).
$$
Therefore, the evolution problem on $b$ becomes
\begin{align*}
& (-i \partial_t + \omega(\xi)) b(\xi) \\
&= -\frac{1}{8\pi} \int \frac{\widehat{\alpha^{(3)}}(\xi_1,\xi_2) }{\sqrt{\omega(\xi) \omega(\xi_1) \omega(\xi_2)}} \left[ b(\xi_1) + \overline{b(-\xi_1)}\right] \left[ b(\xi_2) + \overline{b(-\xi_2)}\right] \delta(\xi -  \xi_1 -  \xi_2 ) \dd \xi_{1,2} + \dots
\end{align*}
which can be rearranged as
\begin{equation}
\label{eqb}\boxed{
\begin{aligned}
 - i \partial_t b(\xi)  + \omega(\xi) b(\xi) & = \sum_{\sigma_i = \pm 1}  \int  Q (\sigma_1 \xi_1, \sigma_2 \xi_2) b^{\sigma_1}(\xi_1)  b^{\sigma_2}(\xi_2) \delta(\xi - \sigma_1 \xi_1 - \sigma_2 \xi_2 ) \dd \xi_{1,2} \\
& \qquad + \sum_{\sigma_i = \pm 1}  \int  C(\sigma_1 \xi_1,\sigma_2 \xi_2,\sigma_3 \xi_3) b^{\sigma_1}(\xi_1)  b^{\sigma_2}(\xi_2) b^{\sigma_3}(\xi_3) \\
& \qquad\qquad \qquad\qquad \qquad\qquad \delta(\xi - \sigma_1 \xi_1 - \sigma_2 \xi_2 - \sigma_3 \xi_3) \dd \xi_{1,2,3}
\end{aligned}}
\end{equation}
where we used the shorthands
$$
\dd \xi_{1,2,3} = \dd \xi_1 \dd \xi_2 \dd \xi_3 \qquad \mbox{and} \qquad b^\sigma(\xi) = \begin{cases} y  & \mbox{if $\sigma = + 1$} \\  \overline{y}  & \mbox{if $\sigma = - 1$} \end{cases}.
$$
and the notations $Q$ and $C$ for the quadratic and cubic symbols given by
\begin{equation}
\label{harlebievre}
\begin{split}
& Q(\xi_1, \xi_2) = \frac{-1}{8 \pi} \frac{ \widehat{\alpha^{(3)}}( \xi_1, \xi_2)}{\sqrt {\omega(\xi_1+\xi_2) \omega(\xi_1) \omega(\xi_2) }} \\
& C( \xi_1, \xi_2, \xi_3) = \frac{-1}{32 \pi^2} \frac{ \widehat{\alpha^{(4)}}( \xi_1,\xi_2,  \xi_3)}{\sqrt {\omega(\xi_1+\xi_2+\xi_3) \omega(\xi_1) \omega(\xi_2) \omega(\xi_3)}}.
\end{split}
\end{equation}
These symbols inherit the symmetry properties \eqref{symmetry_alpha_hat} of $\widehat{\alpha^{(3)}}$ and $\widehat{\alpha^{(4)}}$:
\begin{equation}
\label{symmetry_QC}
\begin{split}
& Q(\xi_1,\xi_2) = \overline{Q(-\xi_1,-\xi_2)} = Q(\xi_2,\xi_1) = Q(\xi_1,-\xi_1-\xi_2),\\
& C(\xi_1,\xi_2,\xi_3) = \overline{C(-\xi_1,-\xi_2,-\xi_3)} = C(\xi_1,\xi_3,\xi_2) = C(\xi_3,\xi_2,\xi_1) \\
& \qquad \qquad \quad = C(\xi_1,\xi_2,-\xi_1-\xi_2-\xi_3).
\end{split}
\end{equation}

\medskip

\noindent \underline{Step 3: Filtering by the linear flow.} There remains to filter $b(\xi)$ through the linear flow by letting
$$
a(\xi) = e^{it \omega(\xi)} b(\xi),
$$
the equation becomes
\begin{equation*}
\boxed{
\begin{aligned}
 - i \partial_t a(\xi) & = \sum_{\sigma_i = \pm 1} \int Q(\sigma_1 \xi_1,\sigma_2 \xi_2) e^{it \Omega_{0,1,2}^{\sigma_1,\sigma_2}} a^{\sigma_1}(\xi_1)  a^{\sigma_2}(\xi_2) \delta(\xi - \sigma_1 \xi_1 - \sigma_2 \xi_2 )  \dd \xi_{1,2} \\
& \qquad + \sum_{\sigma_i = \pm 1} \int C(\sigma_1 \xi_1,\sigma_2 \xi_2,\sigma_3 \xi_3) e^{it \Omega_{0,1,2,3}^{\sigma_1,\sigma_2,\sigma_3}} a^{\sigma_1}(\xi_1)  a^{\sigma_2}(\xi_2) a^{\sigma_3}(\xi_3) \\
& \qquad \qquad \qquad \qquad \qquad \qquad \qquad \qquad \delta(\xi - \sigma_1 \xi_1 - \sigma_2 \xi_2 - \sigma_3 \xi_3)  \dd \xi_{1,2,3}
\end{aligned}}
\end{equation*}
where we define the \textit{resonant moduli}
\begin{align*}
& \Omega^{\sigma_1,\sigma_2}_{0,1,2} = \omega(\xi) - \sigma_1 \omega(\xi_1) - \sigma_2 \omega(\xi_2), \\
& \Omega^{\sigma_1,\sigma_2,\sigma_3}_{0,1,2,3} = \omega(\xi) - \sigma_1 \omega(\xi_1) - \sigma_2 \omega(\xi_2) - \sigma_3 \omega(\xi_3).
\end{align*}

The equation on the unknown $a$ above is known as the \textit{interaction representation} or \textit{Zakharov representation}; it appears for instance in \cite{ZakharovLvovFalkovich}. This formulation is very natural in the case of a weak nonlinearity, since $a$ would be constant if the nonlinearity was zero, and thus slowly varying if the nonlinearity is weak. Indeed, it is so natural that it already appeared in the papers of Peierls \cite{Peierls} and Hasselmann \cite{Hasselmann1962}. 

\subsection{Nonlinear resonances}
\label{section_resonances}

Nonlinear resonance is the key concept to understand how the weak nonlinearity interacts with the linear part of the equation. This section is devoted to the definition of the resonant set, its characterization in the case of nearest-neighbor interactions, and the exploration of its dynamical consequences.

\subsubsection{Definition of the resonant set}
"Resonances" may refer to a variety of mathematical concepts. We are concerned here with nonlinear resonances, as they appear in the theory of Hamiltonian systems, or dynamical systems in general. In the weakly nonlinear regime which is our focus in this review, nonlinear resonances correspond to these nonlinear interactions between frequencies which are determinant for large time dynamics.

In the integral equation for $a$ above, we can view $e^{it \Omega_{0,1,2}^{\sigma_1,\sigma_2}}$ as a rapidly oscillatory factor, the other terms varying slowly in time. These rapid oscillations lead to cancellations (often called 'averaging') in the nonlinear term... unless $\Omega^{\sigma_1,\sigma_2}_{0,1,2}$ or $\Omega^{\sigma_1,\sigma_2,\sigma_3}_{0,1,2,3}$ (in the quadratic and cubic terms, respectively) vanish! The zero set of the resonance modulus is known as the \textit{resonant set}
\begin{align*}
& \mathscr{R}^{\sigma_1,\sigma_2}(\omega) = \{ (\xi,\xi_1,\xi_2) \in \mathbb{T}^3 \;\mbox{s.t.}\; \xi = \sigma_1 \xi_1 + \sigma_2 \xi_2 \;\mbox{and} \; \Omega^{\sigma_1,\sigma_2}_{0,1,2}=0 \} \\
& \mathscr{R}^{\sigma_1,\sigma_2,\sigma_3}(\omega) = \{ (\xi,\xi_1,\xi_2,\xi_3) \in \mathbb{T}^4 \;\mbox{s.t.}\; \xi = \sigma_1 \xi_1 + \sigma_2 \xi_2 + \sigma_3 \xi_3 \;\mbox{and} \; \Omega^{\sigma_1,\sigma_2,\sigma_3}_{0,1,2,3}=0 \}.
\end{align*}

\subsubsection{Resonant sets for nearest-neighbor interactions}
As we saw, the resonant sets gather interactions which do not undergo averaging and thus dominate the long-term Hamiltonian dynamics\footnote{The reader might have noticed that the resonant modulus typically vanishes on a zero-measure set, and that it actually suffices that $\Omega \ll \frac 1t$ for averaging to be operant. We gloss over these more technical considerations but they will actually be crucial for the derivation of the kinetic wave equation.}. As we shall see, the resonant sets are also at the heart of the kinetic wave equation. Since nearest-neighbor interactions are the most common model, we now turn our attention to their resonant sets. Recall that nearest-neighbor interactions correspond to the dispersion relation
$$
\omega(\xi) = \sqrt{1 - 2 \delta \cos(\xi)}, \qquad 0 \leq \delta \leq \frac 12.
$$

\medskip

\noindent \underline{Quadratic resonances.} There holds if $0 < \delta < \frac 12$
$$
\mathscr{R}^{\sigma_1,\sigma_2}(\omega) = \emptyset \qquad \mbox{for any $\sigma_1,\sigma_2=\pm$},
$$
while if $\delta = \frac 12$, $\omega(\xi) = \sqrt{2} \left| \sin \left( \frac \xi 2 \right) \right|$
\begin{align*}
& \mathscr{R}^{--}(\omega) = \{(0,0,0) \} \\
& \mathscr{R}^{+,+}(\omega) = \{ (\xi,\xi,0), \; \xi \in \mathbb{T} \} \cup \{ (\xi,0,\xi), \; \xi \in \mathbb{T} \} \\
& \mathscr{R}^{+,-}(\omega) = \{ (\xi,\xi,0), \; \xi \in \mathbb{T} \} \cup \{ (0,\xi,\xi), \; \xi \in \mathbb{T} \}.
\end{align*}
An analytic proof of the above description of quadratic resonant sets can be found in \cite{Lukkarinen2016}.

\medskip

\noindent \underline{Cubic resonances.} If $0 < \delta < \frac 1 2$
$$
\mathscr{R}^{-,-,-}(\omega) = \mathcal{R}^{+,+,+}(\omega) = \emptyset \qquad \mbox{if $0<\delta<\frac{1}{2}$}
$$
If $\delta = \frac{1}{2}$,
\begin{align*}
& \mathscr{R}^{-,-,-}(\omega) = \{ (0,0,0,0) \}\\
&  \mathscr{R}^{+,+,+}(\omega) = \{ (\xi,\xi,0,0),\; \xi \in \mathbb{T} \} \cup \{ (\xi,0,\xi,0),\; \xi \in \mathbb{T} \} \cup \{ (\xi,0,0,\xi),\; \xi \in \mathbb{T} \}
\end{align*}
Once again, an analytic proof of the descriptions of $\mathscr{R}^{-,-,-}$ and $\mathscr{R}^{+,+,+}$ given above can be found in \cite{Luk2016}.
For reasons which will become clear later on, the resonant sets $\mathscr{R}^{-,-,-}$ and $\mathscr{R}^{+,+,+}$ will not matter in the dynamics at the kinetic level, the important case being $\mathscr{R}^{-,+,+}$, to which we now turn. 

It will be natural to split $\mathscr{R}^{-,+,+}$ into \textit{trivial} and \textit{effective} resonances
\begin{align*}
\mathscr{R}^{-,+,+}(\omega) = \mathscr{R}^{-,+,+}_{\operatorname{triv}}(\omega) \cup \mathscr{R}^{-,+,+}_{\operatorname{eff}}(\omega),
\end{align*}
where
\begin{equation}
\label{ReffRtriv}
\begin{split}
& \mathscr{R}^{-,+,+}_{\operatorname{triv}}(\omega) = \{ (\xi_0,\xi_1,\xi_0,\xi_1),\,(\xi_0,\xi_1) \in \mathbb{T}^2 \}\, \cup \, \{ (\xi_0,\xi_1,\xi_1,\xi_0),\,(\xi_0,\xi_1) \in \mathbb{T}^2 \}\\
& \mathscr{R}^{-,+,+}_{\operatorname{eff}}(\omega) = \mathscr{R}^{-,+,+}(\omega) \setminus \mathscr{R}^{-,+,+}_{\operatorname{triv}}(\omega).
\end{split}
\end{equation}
Trivial resonances deserve this qualification for two reasons: first, they do not have any impact on the kinetic dynamics, as we shall see, and second, they are independent of the dispersion relation $\omega$. Effective resonances are the exact opposite: they are the key to the kinetic dynamics, and they depend strongly on the dispersion relation $\omega$.

In the case of nearest-neighbor interactions, effective resonances can be parameterized through a function $h(\xi_0,\xi_2)$ as
$$
\mathscr{R}^{-,+,+}_{\operatorname{eff}}(\omega) = \{ (\xi_0,h(\xi_0,\xi_2),\xi_2,\xi_0 + h(\xi_0,\xi_2) - \xi_2 ), \,(\xi_0,\xi_2) \in \mathbb{T}^2 \}.
$$
If $\delta = \frac{1}{2}$, an explicit formula for the function $h$ was derived in \cite{LukkarinenSpohn2008}
\begin{equation}
\label{formula_h}
h(\xi_0,\xi_2) = \frac{\xi_2 - \xi_0}{2} + 2 \arcsin \left(\tan \left| \frac{\xi_2-\xi_0}{4} \right| \cos \left( \frac{\xi_0+\xi_2} 4 \right) \right)   
\end{equation}
(to be precise: this formula is valid if $\xi_0$ and $\xi_2$ are chosen to be the representatives in $[0,2\pi)$ of $\xi_0,\xi_2 \in \mathbb{T}$).

\begin{figure}
  \centering

  \begin{minipage}{.4\linewidth}
    \begin{overpic}[width=\linewidth]{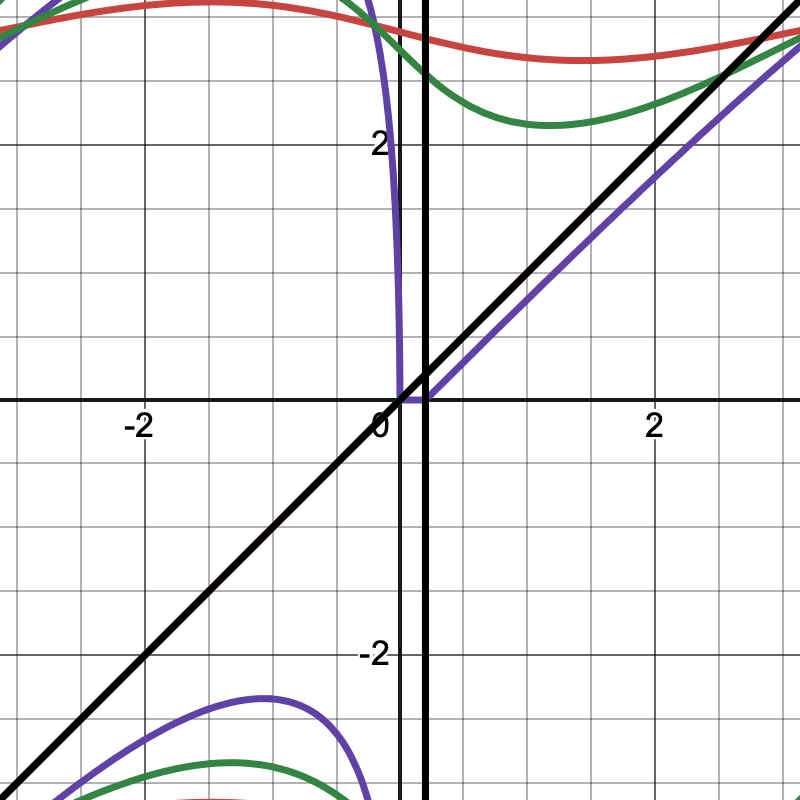} 
      \put(-10,95){(\textit{a})}
    \end{overpic}
  \end{minipage}\hspace{0,05\textwidth}
  \begin{minipage}{0.40\linewidth}
    \begin{overpic}[width=\linewidth]{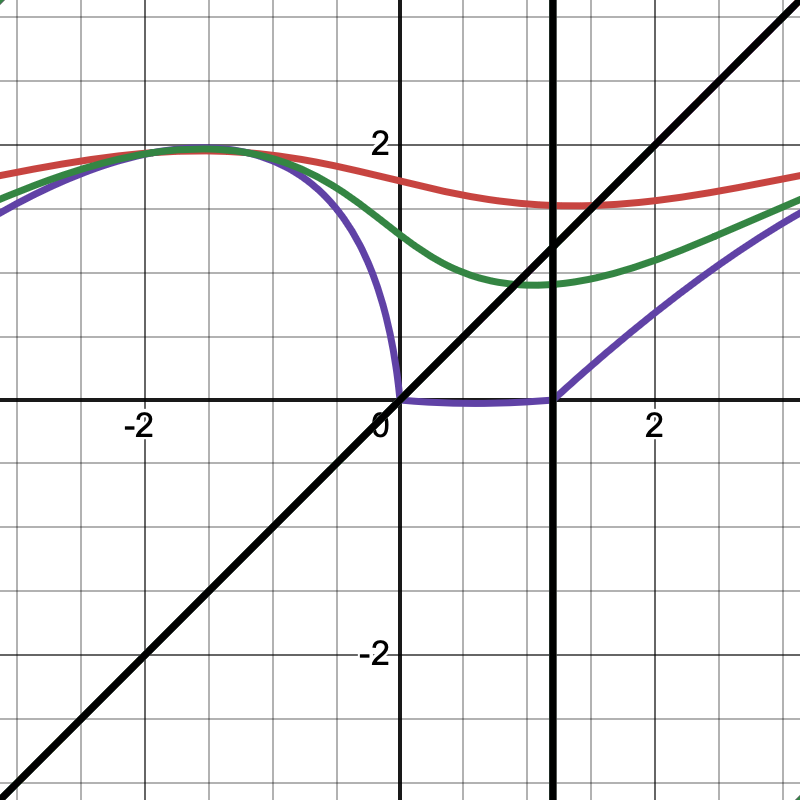} 
      \put(-10,95){(\textit{b})}
    \end{overpic}
  \end{minipage}

\vspace{0.05 \textwidth}
  \begin{minipage}{0.40\linewidth}
    \begin{overpic}[width=\linewidth]{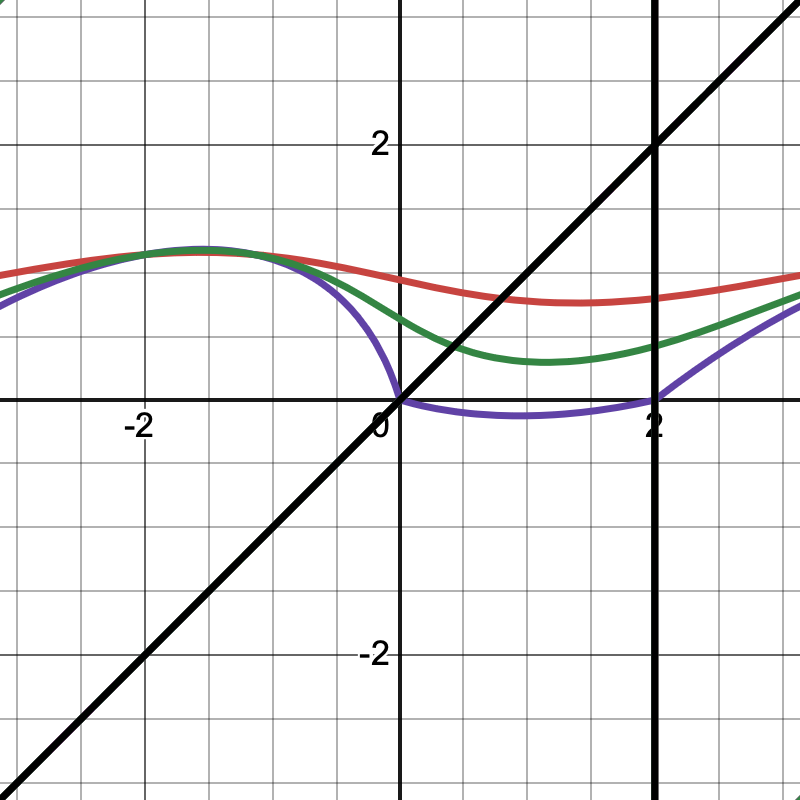} 
      \put(-10,95){(\textit{c})}
    \end{overpic}
  \end{minipage}\hspace{0.05\textwidth}
  \begin{minipage}{0.40\linewidth}
    \begin{overpic}[width=\linewidth]{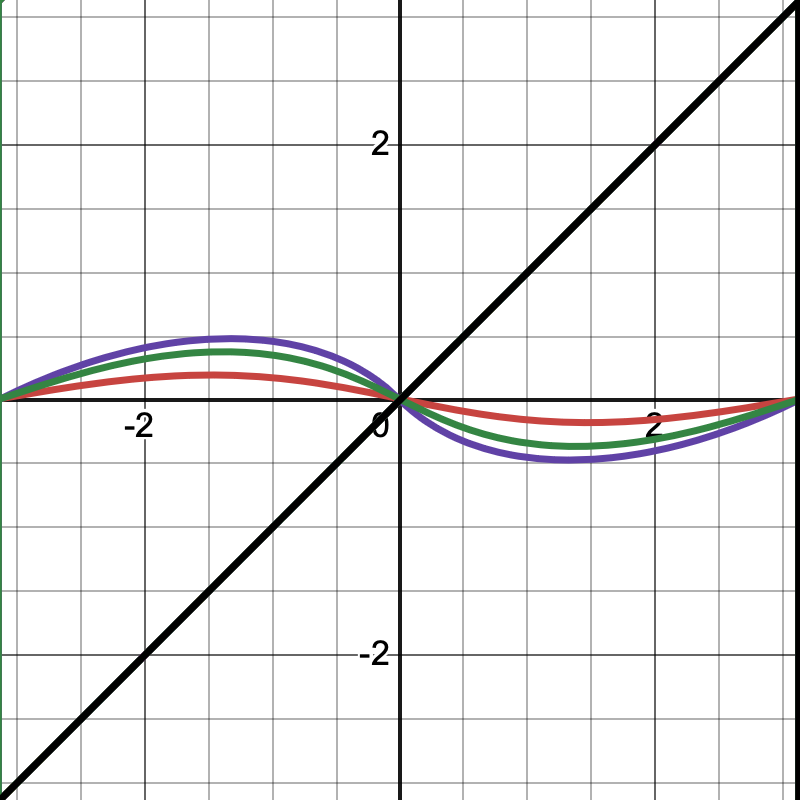} 
      \put(-10,95){(\textit{d})}
    \end{overpic}
  \end{minipage}

\caption{The resonant set $\mathscr{R}^{-,+,+}(\omega)$ for fixed $\xi_0$. The figures $(a)$, $(b)$, $(c)$, $(d)$ represent the resonant set for $\xi_0 = 0.2,1.2,2,\pi$, with the horizontal coordinate being $\xi_2$, the vertical $\xi_1$, and as always $\xi_3 = \xi_0+\xi_1 -\xi_2$. The black lines are the trivial resonances $\xi_1=\xi_2$ and $\xi_2=\xi_0$, while the colored curves are the effective resonances for nearest-neighbor interactions $\omega(\xi) = \sqrt{1-2\delta \cos \xi}$: the red, green and violet curves correspond to the values $\delta = 0.2, 0.4, 0.5$ respectively.}
\end{figure}

\subsubsection{Normal form} We saw that quadratic resonances do not exist in cases of interest. Assuming that the nonlinear term is quadratic, the fact that the resonant sets are empty allows to transform the equation into a cubic one, a procedure often called \textit{normal form}. Our starting point is the general quadratic equation seen earlier
$$
a(t,\xi) = a_0(\xi) + i \sum_{\sigma_i} \int_0^t \int Q(\sigma_1 \xi_1,\sigma_2 \xi_2) e^{it \Omega_{0,1,2}^{\sigma_1,\sigma_2}} a^{\sigma_1}(\xi_1) a^{\sigma_2}(\xi_2) \delta(\xi - \sigma_1 \xi_1 - \sigma_2 \xi_2) \dd \xi_{1,2}.
$$
Integrating by parts gives
\begin{align*}
a(t,\xi) & = a_0(\xi) - 2 \sum_{\sigma_i} \int_0^t \int \frac{Q(\sigma_1 \xi_1,\sigma_2 \xi_2)}{\Omega_{0,1,2}^{\sigma_1,\sigma_2}} e^{it \Omega_{0,1,2}^{\sigma_1,\sigma_2}} a^{\sigma_1}(\xi_1) \partial_t a^{\sigma_2}(\xi_2) \delta(\xi - \sigma_1 \xi_1 - \sigma_2 \xi_2) \dd \xi_{1,2} \\ 
& \qquad \qquad  + \{ \mbox{boundary term} \}.
\end{align*}
We now argue that the boundary term can be ignored: indeed, the expectation is that the dominant effects will result from cubic resonances. Therefore, we drop the boundary term and replace $\partial_t a$ by the quadratic expression in $a$; this gives the equation
\begin{align*}
& a(t,\xi) = a_0(\xi) - 2i \sum_{\sigma_i} \sigma_2 \int_0^t \int \frac{Q(\sigma_1 \xi_1,\sigma_2 \xi_2)Q^{\sigma_2}(\sigma_3 \xi_3,\sigma_4 \xi_4)}{\Omega_{0,1,2}^{\sigma_1,\sigma_2}} e^{it \Omega_{0,1,3,4}^{\sigma_1,\sigma_2 \sigma_3,\sigma_2 \sigma_4}} \times \\
& \qquad \qquad \quad \quad  a^{\sigma_1}(\xi_1) a^{\sigma_2 \sigma_3}(\xi_3) a^{\sigma_2 \sigma_4}(\xi_4) \delta (\xi - \sigma_1 \xi_1 - \sigma_2 \xi_2) \delta(\xi_2 - \sigma_3 \xi_3 - \sigma_4 \xi_4)  \dd \xi_{1,2,3,4}.
\end{align*}
The integral expression can be written in a more friendly way by eliminating $\xi_2$,  and relabeling $\tau$, $\sigma_2$, $\sigma_3$, $\xi_2$, $\xi_3$ instead of $\sigma_2$, $\sigma_2 \sigma_3$, $\sigma_2 \sigma_4$, $\xi_3$, $\xi_4$ respectively. Using furthermore that $Q(\tau \xi,\tau \eta)^\tau = Q(\xi,\eta)$ by \eqref{symmetry_QC}, the above becomes
\begin{align*}
&a(t,\xi) = a_0(\xi) + i \sum_{\sigma_1,\sigma_2,\sigma_3} \int_0^t \int D^{\sigma_1,\sigma_2,\sigma_3}(\xi_1,\xi_2,\xi_3) e^{it \Omega_{0,1,2,3}^{\sigma_1,\sigma_2, \sigma_3}} a^{\sigma_1}(\xi_1) a^{\sigma_2}(\xi_2) a^{\sigma_3}(\xi_3) \\
& \qquad \qquad\qquad \qquad\qquad \qquad\qquad \qquad \qquad \qquad \qquad \delta(\xi - \sigma_1 \xi_1 - \sigma_2 \xi_2 - \sigma_3 \xi_3)  \dd \xi_{1,2,3}.
\end{align*}
where
\begin{equation}
\label{formulaD}
D^{\sigma_1,\sigma_2,\sigma_3}(\xi_1,\xi_2,\xi_3) = -2 \sum_{\tau} \tau \frac{Q(\sigma_1 \xi_1, \sigma_2 \xi_2 + \sigma_3 \xi_3) Q(\sigma_2 \xi_2,\sigma_3 \xi_3)}{\omega(\sigma_1 \xi_1 + \sigma_2 \xi_2 +\sigma_3 \xi_3) - \sigma_1 \omega(\xi_1) - \tau \omega(\sigma_2 \xi_2 + \sigma_3 \xi_3)}.
\end{equation}
Finally, we can symmetrize the variables $2,3,4$ and obtain
\begin{align*}
&a(t,\xi) = a_0(\xi) + i \sum_{\sigma_1,\sigma_2,\sigma_3} \int_0^t \int C^{\sigma_1,\sigma_2,\sigma_3}( \xi_1, \xi_2, \xi_3) e^{it \Omega_{0,1,2,3}^{\sigma_1,\sigma_2, \sigma_3}} a^{\sigma_1}(\xi_1) a^{\sigma_2}(\xi_2) a^{\sigma_3}(\xi_3) \\
& \qquad \qquad\qquad \qquad\qquad \qquad\qquad \qquad \qquad \qquad \qquad  \delta(\xi - \sigma_1 \xi_1 - \sigma_2 \xi_2 - \sigma_3 \xi_3)  \dd \xi_{1,2,3}.
\end{align*}
after setting
\begin{equation}
\label{defC}
\begin{split}
& C^{\sigma_1,\sigma_2,\sigma_3}( \xi_1, \xi_2, \xi_3) \\
& \qquad \qquad = \frac 13 \left[D^{\sigma_1,\sigma_2,\sigma_3}(\xi_1,\xi_2,\xi_3) + D^{\sigma_2,\sigma_1,\sigma_3}(\xi_2,\xi_1,\xi_3) + D^{\sigma_3,\sigma_2,\sigma_1}(\xi_3,\xi_2,\xi_1) \right].
\end{split}
\end{equation}
We recognize a cubic nonlinear term, with a new cubic symbol $C$. It now depends on $\sigma_1,\sigma_2,\sigma_3$ in a more involved way than in \eqref{harlebievre}.

\section{The weakly turbulent regime}

\label{section_weaklyturbulent}

Our aim in this section is to present the regime in which the microscopic to mesoscopic limit can be performed, namely the \textit{weakly turbulent regime}.  In the mesoscopic description of the system, only some statistical information is retained from the original microscopic (Hamiltonian) system.

As we shall see, the weakly turbulent regime is characterized by weak nonlinearity, large volume and random phase approximation. In this regime, we will also examine the limits of the energy, of the Gibbs entropy and of Gibbs measures, which yield objects which will be important in the following developments.

\subsection{Definition of the weakly turbulent regime} 
The \textit{weakly turbulent regime} requires three conditions: weak nonlinearity, large volume, and random phase approximation. Though it seems to be the simpler of the three conditions, we will discuss the weakly nonlinear condition last, for a reason that will become clear.

Our starting point is the general form of the Hamiltonian system in the case of a cubic nonlinearity - perhaps after a normal form in the absence of quadratic resonances. We recall it here and introduce some convenient notation
$$
- i \partial_t a(\xi) = \sum_{\sigma_i = \pm 1} \int C_{1,2,3}^\sigma e^{it \Omega_{0,1,2,3}^{\mathbf{\sigma}}} a^{\sigma_1}_1 a^{\sigma_2}_2 a^{\sigma_3}_3  \delta_{0,1,2,3}^{\mathbf{\sigma}}  \dd \xi_{1,2,3}.
$$
where we denote as usual $\dd \xi_{1,2,3} = \dd \xi_1 \dd \xi_2 \dd \xi_3$ and furthermore
\begin{align*}
& a_i = a(t,\xi_i)\\
& \delta^{\mathbf{\sigma}}_{0,1,2,3} = \delta(\xi_0 - \sigma_1 \xi_1 - \sigma_2 \xi_2 - \sigma_3 \xi_3) \\
& \Omega^{\mathbf{\sigma}}_{0,1,2,3} = \omega_0 - \sigma_1 \omega_1 - \sigma_2 \omega_2 - \sigma_3 \omega_3.
\end{align*}
Finally, two cases need to be distinguished for the cubic symbol $C$: if $Q=0$, then
$$
C_{1,2,3}^\sigma = C(\sigma_1 \xi_1, \sigma_2 \xi_2, \sigma_3\xi_3)
$$
with $C$ given in \eqref{harlebievre}; and if $C=0$ and quadratic resonances are absent, then $C_{1,2,3}^\sigma$ is defined in \eqref{defC} (a combination of these two cases is of course also possible).

\medskip

\noindent \underline{Infinite volume limit.} To give a meaning to this limit, we shall set the problem on the compact domain $\{1,\dots,N\}$ and let the size of the domain go to infinity. It is convenient to think of the solutions as being periodic solutions: we will assume that $q_{j+N} = q_j$ for all $j$.

The Fourier transform can be defined similarly to the above, but the Fourier space becomes discrete: while $j$ ranges over $\{ 0,1,\dots,N-1 \}$, the dual variable $\underline{\xi}$ ranges over $\{ 0, \frac{2\pi}{N} , \dots , \frac{(N-1)2 \pi}{N} \}$; note that we are denoting $\underline{\xi}$ instead of $\xi$ in order to emphasize that the Fourier variable is now discrete. The formulas are as follows
$$
\widehat{q}(\underline{\xi}) = \sum_k q_k e^{-ik \underline{\xi}}, \qquad q_k = \frac{1}{N} \sum_{\underline{\xi}} \widehat{q}(\xi) e^{i k \underline{\xi}}. 
$$
With this normalization, we have
$$
\widehat{q_{j+1}}(\underline{\xi}) = e^{i \underline{\xi}} \widehat{q}(\underline{\xi}), \qquad \widehat{a \cdot b}(\underline \xi) = \frac{1}{N} \sum_{ \underline{\xi_1}, \underline{\xi_2}}\widehat{a} (\underline{\xi_1}) \widehat{b} (\underline{\xi_2}) \delta (\underline{\xi} - \underline{\xi_1} - \underline{\xi_2}).  
$$
Here, the $\delta$ notation mimicks the notation we saw for a continuous Fourier variable, but it should be understood as a Kronecker instead of a Dirac $\delta$. In other words,
$$
\sum_{ \underline{\xi_1}, \underline{\xi_2}} \widehat{a} (\underline{\xi_1}) \widehat{b} (\underline{\xi_2}) \delta (\underline{\xi} - \underline{\xi_1} - \underline{\xi_2})= \sum_{\substack{ \underline{\xi_1},\underline{\xi_2} \in \frac{2\pi}{N}\mathbb{Z} \operatorname{mod} 1 \\ \underline{\xi_1} + \underline{\xi_2} = \underline{\xi}}} \widehat{a} (\underline{\xi_1}) \widehat{b} (\underline{\xi_2}).
$$
In the limit $N \to \infty$, the discrete variable $\underline{\xi}$ becomes the continuous variable $\xi$ and sums turn into integrals: we will often use the formula
$$
\frac{1}{N} \sum_{\underline{\xi}} \varphi(\underline{\xi}) \overset{N \to \infty}{\longrightarrow} \frac{1}{2\pi} \int \varphi( \xi) \dd \xi, \qquad \varphi \;\;\mbox{smooth}.
$$

Finally, with $N$-periodicity in place, the equation becomes in Fourier
\begin{equation}
\label{eq_epsilon_N0}
- i \partial_t a(\underline \xi) = \frac{(2\pi)^2}{N^2} \sum_{\sigma_i = \pm 1} \sum_{ \underline{\xi_1}, \underline{\xi_2}, \underline{\xi_3}} C_{1,2,3}^\sigma e^{it \Omega_{0,1,2,3}^{\mathbf{\sigma}}} a^{\sigma_1}_1 a^{\sigma_2}_2 a^{\sigma_3}_3  \delta_{0,1,2,3}^{\mathbf{\sigma}}
\end{equation}
(with obvious notations).

\medskip

\noindent \underline{Random phase approximation.} We now consider an ensemble of solutions and consider the phases in the polar decomposition of the Fourier modes. Strictly speaking, the random phase approximation is requiring that the phases, viewed as random variables, have uniform laws on $[0,2\pi]$ and are independent. 

We will adopt the stronger assumption that the Fourier modes are given by independent Gaussians. Not only does this simplify some details, but it also corresponds to the expected behavior in the weakly turbulent regime. To put it into a formula, we are assuming that, for each $t$, the law of $a(t)$ is given by
$$
a(t,\underline{\xi}) = \sqrt{f (t,\underline{\xi})} \mathcal{G}_{\underline{\xi}}, \qquad \mbox{$(\mathcal{G}_{\underline{\xi}})_{\underline{\xi}}$ independent normalized complex Gaussians}.
$$
This should be taken with a grain of salt! This ansatz cannot be literally true since correlations between Fourier modes will build up through the nonlinearity. Proving that this ansatz indeed holds up to a small error is a very hard problem known as \textit{propagation of chaos}.

At the mesoscopic level, the system will be simply described by $f (t,\underline{\xi})$, which obviously carries much less information than the full microscopic description.

In accordance with the random phase approximation, we choose as data
\begin{equation}
\label{data_Gaussian}
a(t=0,\underline{\xi}) = A(\underline{\xi}) = \sqrt{F (\underline{\xi})} \mathcal{G}_{\underline{\xi}}, \qquad \mbox{$(\mathcal{G}_{\underline{\xi}})_{\underline{\xi}}$ independent normalized complex Gaussians}.
\end{equation}

\medskip

\noindent \underline{Weakly nonlinear limit.} Weak nonlinearity can be achieved in two equivalent ways: either by choosing small data, or by adding a small factor in front the nonlinearity. Physically, we are trying to ensure that linear effects are dominant over nonlinear effects. We consider the initial form of the equation (in physical space), which can be written schematically
$$
\ddot q = L q + C(q,q,q),
$$
where $L$ stands for the linear operator and $C$ for the cubic non-linearity. 

We now want to compare the sizes of the linear and nonlinear terms on the right-hand side. Recall that $q_k = \frac{1}{N} \sum \widehat{q}(\underline{\xi}) e^{ik \underline{\xi}}$, and that we think of $\widehat{q}(\underline{\xi})$ as independent Gaussian variables of variance $\sim 1$. By the central limit theorem (square root cancellation), we expect $q_k$ to have size $\sim \frac{1}{\sqrt{N}}$. 

This means that the linear term has size $\sim q \sim \frac{1}{\sqrt{N}}$ while the nonlinear term has size $\sim q^3 \sim N^{-\frac 32}$. Thus we are already in the weakly nonlinear regime! In order to be more precise, we can ask what is the optimal factor that can multiply the nonlinear term, still ensuring that the equation is weakly nonlinear. Comparing the sizes of $q$ and  $q^3$, we see that a factor $N\varepsilon^2$, with $\varepsilon \ll 1$, is the most we can do; the schematic equation turns into
$$
\ddot q = L q + N \varepsilon^2 C(q,q,q).
$$
and \eqref{eq_epsilon_N0} becomes
\begin{equation}
\label{eq_epsilon_N}
- i \partial_t a(\underline \xi) = \frac{(2\pi)^2 \varepsilon^2}{N} \sum_{\sigma_i = \pm 1} \sum_{ \underline{\xi_1}, \underline{\xi_2}, \underline{\xi_3}} C_{1,2,3}^\sigma e^{it \Omega_{0,1,2,3}^{\mathbf{\sigma}}} a^{\sigma_1}_1 a^{\sigma_2}_2 a^{\sigma_3}_3  \delta_{0,1,2,3}^{\mathbf{\sigma}}.
\end{equation}

\medskip

\noindent \underline{Summary: the weakly turbulent regime} Over the last three paragraphs, we identified the weakly turbulent regime. It can be summarized through the initial value problem
$$
\boxed{
\begin{cases} \displaystyle - i \partial_t a(\underline \xi) = \frac{(2\pi)^2 \varepsilon^2}{N} \sum_{\sigma_i = \pm 1} \sum_{ \underline{\xi_1}, \underline{\xi_2}, \underline{\xi_3}} C_{1,2,3}^\sigma e^{it \Omega_{0,1,2,3}^{\mathbf{\sigma}}} a^{\sigma_1}_1 a^{\sigma_2}_2 a^{\sigma_3}_3  \delta_{0,1,2,3}^{\mathbf{\sigma}}, \quad \underline{\xi} \in \frac{2\pi}N \mathbb{Z} \mod 1 \\
a(t=0,\underline{\xi}) = A(\underline{\xi}) = \sqrt{F (\underline{\xi})} \mathcal{G}_{\underline{\xi}}, \qquad \mbox{$(\mathcal{G}_{\underline{\xi}})_{\underline{\xi}}$ i.i.d complex Gaussians},
\end{cases}
}
$$
where $\varepsilon \to 0$ and $N \to \infty$.

It is instructive to write the Hamiltonian from which this evolution equation derives; we shall only do so in the case $Q=0$ for simplicity. We saw above that the right scaling is obtained by multiplying the cubic term by $N\varepsilon^2$; at the level of the Hamiltonian, this corresponds to the quartic term, thus
$$
\mathscr{H}(\mathbf{p},\mathbf{q}) = \frac{1}2 \sum_j |p_j|^2 + \frac{1}{2} \sum_{j,k} \alpha^{(2)}_k q_j q_{j-k} + \frac{N \varepsilon^2}{4} \sum_{j,k_1,k_2,k_3} \alpha^{(4)}_{k_1,k_2,k_3} q_j q_{j-k_1} q_{j-k_2} q_{j-k_4}, 
$$
where $j,k_i \in \mathbb{Z} \mod N$.

\subsection{From microscopic to mesoscopic conserved quantities}
We saw the formula for the Hamiltonian in the weakly turbulent regime, which we split  into quadratic and quartic part:
$$
\mathscr{H}(\mathbf{p},\mathbf{q}) = \underbrace{\frac{1}2 \sum_j |p_j|^2 + \frac{1}{2} \sum_{j,k} \alpha^{(2)}_k q_j q_{j-k}}_{\displaystyle \mathscr{H_{\operatorname{quad}}}(\mathbf{p},\mathbf{q})} + \underbrace{\frac{N \varepsilon^2}{4} \sum_{j,k_1,k_2,k_3} \alpha^{(4)}_{k_1,k_2,k_3} q_j q_{j-k_1} q_{j-k_2} q_{j-k_3}}_{\displaystyle \mathscr{H}_{\operatorname{quart}}(\mathbf{q})}.
$$
The quadratic part is best understood after taking the Fourier transform and turning to the variable $b(\underline{\xi}) = \frac{1}{\sqrt{\omega(\underline{\xi})}} \left[i \widehat{p}(\underline{ \xi}) + \omega(\underline{\xi}) \widehat{q}(\underline{\xi})\right]$:
$$
\mathscr{H}_{\operatorname{quad}}(\mathbf{p},\mathbf{q})  = \frac{1}{2N} \sum_{\underline{\xi}} \left[  | \widehat{p}(\underline{\xi}) |^2 + \omega(\underline{\xi}) | \widehat{q}(\xi)|^2 \right]  = \frac{1}{2N} \sum_{\underline{\xi}} \omega(\underline{\xi}) | b (\underline{\xi})|^2 .
$$
Taking the expectation, we find
$$
\mathbb{E}\, \mathscr{H}_{\operatorname{quad}}(\mathbf{p},\mathbf{q}) = \frac{1}{2N} \sum_{\underline{\xi}} \omega(\underline{\xi}) f(\underline{\xi}) \overset{N\to \infty}{\longrightarrow} \frac{1}{4\pi} \int \omega(\xi) f(\xi) \dd \xi.
$$
Turning to the quartic part of the energy, the sum over $k$ contributes $O(1)$ since we assume that $\alpha$ decays rapidly, the sum over $j$ contributes $O(N)$ since there are $N$ terms, and $q_j$ has typical size $\sim \frac{1}{\sqrt N}$. Therefore, we can bound
$$
\left| \mathscr{H}_{\operatorname{quart}}(\mathbf{q}) \right| \lesssim N \varepsilon^2 \cdot N \cdot  \left( \frac{1}{\sqrt{N}} \right)^4 = \varepsilon^2 \longrightarrow 0.
$$
Overall, we find that
$$
\mathscr{H}(\mathbf{p},\mathbf{q}) \overset{\substack{N \to \infty \\ \varepsilon \to 0}}{\longrightarrow}  \frac{1}{4\pi} \int \omega(\xi) f(\xi) \dd \xi.
$$
The quantity of the right-hand side is the mesoscopic version of the energy, and we will see that it is indeed conserved through the evolution by the kinetic wave equation.

\smallskip

In the unpinned case, we saw that another conserved quantity emerges, namely the total momentum
$$
\mathscr{P}(\mathbf{p}) = \sum_j p_j.
$$
In the weakly turbulent regime, this quantity is to leading order a Gaussian random variable. In particular, it does not have a limit, even after rescaling, and thus it does not yield a conserved quantity at the mesoscopic level.

\subsection{From microscopic to mesoscopic entropy} 
\label{micro_to_meso_entropy} As explained above, we consider the equation set on $\{ 1,\dots , N \}$ with periodic boundary conditions, in which case the system is finite-dimensional, and we denote $\mathbf{p} = (p_j)_{1\leq j \leq N}$, $\mathbf{q} = (q_j)_{1\leq j \leq N}$ as well as $\dd \mathbf{p} = \prod \dd p_j$ and $\dd \mathbf{q} = \prod \dd q_j$. 

Letting $f(\mathbf{p},\mathbf{q})$ denote the density of probability of finding the system in the state $(\mathbf{p},\mathbf{q})$, the Gibbs entropy with respect to the Liouville measure is given by
$$
S = - \iint f(\mathbf{p},\mathbf{q}) \log f (\mathbf{p},\mathbf{q}) \dd \mathbf{p} \dd \mathbf{q}.
$$
The question we want to answer is: what does this entropy becomes in the weakly turbulent regime which was described above?

We will now change the integration variable to $b(\underline{\xi}) = \frac{1}{\sqrt{\omega(\underline{\xi})}} \left[i \widehat{p}(\underline{ \xi}) + \omega(\underline{\xi}) \widehat{q}(\underline{\xi})\right]$, define
$$
\dd \mathbf{b} = \dd \operatorname{Re} \mathbf{b} \dd \operatorname{Im} \mathbf{b} = \prod_{\underline{\xi}} \dd \operatorname{Re} b(\underline{\xi}) \dd \operatorname{Im} b(\underline{\xi})
$$
and use that $\dd \mathbf{b} = \dd \mathbf{p} \dd \mathbf{q}$. Then the formula for the entropy becomes
$$
S = - \frac 1 N \iint f(\mathbf{b}) \log f (\mathbf{b}) \dd \mathbf{b},
$$
where we abused notations by writing $f(\mathbf{b}) = f(\mathbf{p},\mathbf{q})$.

In the weakly turbulent limit, the random variables $b(\underline{\xi})$ are asymptotically independent complex Gaussians with variance $f(\underline{\xi})$. We use two facts: on the one hand, the entropy of a complex Gaussian with variance $\sigma^2$ is $1+\log 2\pi + \log \sigma^2$, and on the other hand, the entropies of independent variables add up. Therefore, we find
$$
S \sim \sum_{\underline{\xi}} \left[ 1 + \log 2\pi + 2 \log f(\underline{\xi}) \right] \overset{N \to \infty}\sim N \left[1+\log 2\pi  + \frac{1}{\pi} \int \log f(\xi) \dd \xi \right].
$$
Disregarding additive and multiplicative constants, we obtain $\int \log f$ as the continuous limit of the entropy. And indeed, we will see in Section \ref{section_monotonic} that the quantity $\int \log f$ is increasing for the kinetic wave equation, as predicted by the second law of thermodynamics.

\subsection{From Gibbs measure to Rayleigh-Jeans solutions} 
\label{Gibbs_to_RJ} The Gibbs measure 
$$
\frac{1}{Z_{N,T}} e^{-\frac{1}{T} \mathscr{H}(\mathbf{p},\mathbf{q})} \dd \mathbf{p} \dd \mathbf{q}
$$
($Z_T$: normalizing constant) is an invariant measure which plays a key role for the Hamiltonian dynamics. We would like to understand what it becomes in the weakly turbulent regime.

Under the scaling we adopted, we saw that
$$
\mathscr{H}(\mathbf{p},\mathbf{q}) \sim \mathscr{H}_{\operatorname{quad}}(\mathbf{p},\mathbf{q}) = \frac{1}{2N} \sum_{\underline{\xi}} \omega(\underline{\xi}) | b (\underline{\xi})|^2 .
$$
Since $\dd \mathbf{p} \dd \mathbf{q} = \dd \mathbf{b}$, the Gibbs measure becomes in the limit we consider
\begin{align*}
\frac{1}{Z_{N,T}} e^{-\frac{1}{T} \mathscr{H}(\mathbf{p},\mathbf{q})} \dd \mathbf{p} \dd \mathbf{q} & \sim \frac{1}{Z_{N,T}} e^{-\frac 1 T \mathscr{H}_{\operatorname{quad}}(\mathbf{p},\mathbf{q})} \dd \mathbf{p} \dd \mathbf{q} \\
& \sim \frac{1}{Z_{N,T}} e^{-\frac 1 {2TN} \sum_{\underline{\xi}} \omega(\underline{\xi}) | b (\underline{\xi})|^2}  \dd \mathbf{b} = \prod_{\underline{\xi}} \frac{\omega(\underline{\xi})}{2\pi TN} e^{-\frac 1 {2TN} \omega(\underline{\xi}) | b (\underline{\xi})|^2} \dd b(\underline{\xi})
\end{align*}

In other words, the Gibbs measure is now a product measure over all the Fourier modes, and the $b(\underline{\xi})$ are independent complex Gaussians with variance $\frac{TN}{\omega(\underline{\xi})}$:
$$
f(\underline{\xi}) = \mathbb{E} |b(\underline{\xi})|^2 = \frac{TN}{\omega(\underline{\xi})}.
$$
Such functions of $\xi$ are known as Rayleigh-Jeans states and we will see in Section \ref{subsection_stationary} that they are indeed stationary solutions of the kinetic wave equation.

\section{Derivation of the kinetic wave equation}

\label{Sec: Derivation_WKE}

This section is dedicated to the (formal) derivation of the kinetic wave equation \eqref{KWE} which arises from the nonlinear oscillator chain \eqref{NOC} in the weakly turbulent regime. For simplicity, we will only consider the space-homogeneous setting, and thus derive the space-homogeneous version of \eqref{KWE}. We refer to \cite{Spohn2005} for a treatment of the inhomogeneous case.

A rigorous derivation of the kinetic wave equation is still missing, and even partial results are very demanding technically. Nevertheless, we will describe the state of the art in the mathematical treatment of the kinetic limit.

\subsection{From \eqref{NOC} to \eqref{KWE}} \label{section_from_NOC_to_KWE} We will summarize here the conditions under which the derivation of the kinetic wave equation is - heuristically - valid. These conditions are needed in the actual derivation, which will occupy Section \ref{section_heuristic_proof}.

Our starting point is the nonlinear oscillator chain \eqref{NOC}. In general the nonlinearity contains terms of order two and higher, but we will assume that terms of order two can be eliminated.

\begin{itemize}
\item[(C1)] \underline{Elimination of terms of order $2$.} Either the nonlinearity of equation \eqref{NOC} does not have terms of order two, or there are no quadratic resonances, in which case terms of order two are eliminated by a normal form transform.
\end{itemize}

Under this assumption, the nonlinearity is of order $3$. In the homogeneous expansion, terms of order $4$ and higher can be dropped since they will not matter on the time-scale of interest for us. Finally, the equation can be written on the unknown $a$, the nonlinearity scaled as explained above, and the data chosen according to the random phase approximation. This gives the initial data problem (set on $\underline{\xi} \in \{\frac{2\pi}{N},\dots, 1\}$)
\begin{equation}
\label{IDP}
\begin{cases}
& \displaystyle - i \partial_t a(\underline \xi) = \frac{(2\pi)^2 \varepsilon^2}{N} \sum_{\sigma_i = \pm 1} \sum_{ \underline{\xi_1}, \underline{\xi_2}, \underline{\xi_3}} C_{1,2,3}^\sigma e^{it \Omega_{0,1,2,3}^{\mathbf{\sigma}}} a^{\sigma_1}_1 a^{\sigma_2}_2 a^{\sigma_3}_3  \delta_{0,1,2,3}^{\mathbf{\sigma}}   \\
& a(t=0,\underline{\xi}) = A(\underline{\xi}) = \sqrt{F(\underline{\xi})} \mathcal{G}_{\underline{\xi}}, \qquad (\mathcal{G}_{\underline{\xi}}) \; \mbox{i.i.d. Gaussians.}
\end{cases}
\end{equation}
Furthermore, we make the following additional assumptions
\begin{itemize}
\item[(C2)] \underline{Cubic resonances of type $-++$ only} We assume that other resonances (up to permutation of the indices) do not occur: this means that $\mathscr{R}^{+++}(\omega)=\emptyset$.

\medskip
\item[(C3)] \underline{Scaling the parameters $\varepsilon$ and $N$.} We assume that $N^{-\frac 12} \ll \varepsilon \ll 1$.

\medskip

\item[(C4)] \underline{Nonlinear cancellation.} We assume that $C_{1,2,3}=0$ if $\omega_0=0$.

\medskip

\item[(C5)] \underline{Reality of the symbols.} We assume that $C$ and $Q^2$ are real-valued.

\end{itemize}

We now set
$$
T_{\operatorname{kin}} = 36 \pi^3 \varepsilon^{-4}.
$$
Under assumptions (C1)-(C5), we have
$$
\mathbb{E} \left| a \left( \frac{t}{T_{\operatorname{kin}}},\underline{\xi} \right) \right|^2 \overset{\varepsilon \to 0}{\longrightarrow} f(t,\xi), \qquad t \lesssim T_{\operatorname{kin}}
$$
where $f$ solves the kinetic wave equation \eqref{KWE} with data $F$:
$$
\begin{cases}
& \partial_t f = \mathcal{C}(f) \\
& f(t=0,\xi) = F(\xi).
\end{cases}
$$

We have not stated yet how the collision kernel can be computed from the microscopic model. Two cases have to be distinguished
\begin{itemize} 
\item If \eqref{NOC} is cubic (in other words, the symbol $Q$ vanishes in the formulation \eqref{eqb}), then
$$
\boxed{
K(\xi,\xi_1,\xi_2,\xi_3) = |C(-\xi_1,\xi_2,\xi_3)|^2, \qquad \xi = - \xi_1 + \xi_2 + \xi_3.}
$$
This will be proved in Section \ref{caseQ0}.

\medskip
\item If \eqref{NOC} is quadratic ($C=0$ in \eqref{eqb}) and furthermore there are no quadratic resonances, then the system can be made cubic through a normal form transform and
$$
\boxed{K(\xi,\xi_1,\xi_2,\xi_3) = \frac 19 |D^{-,+,+}(\xi_1,\xi_2,\xi_3)+D^{+,-,+}(\xi_2,\xi_1,\xi_3) + D^{+,-,+}(\xi_3,\xi_1,\xi_2)|^2}
$$
with
$$
\boxed{
D^{\sigma_1,\sigma_2,\sigma_3}(\xi_1,\xi_2,\xi_3) = -2 \sum_{\tau} \tau \frac{Q(\sigma_1 \xi_1,\sigma_2 \xi_2 + \sigma_3 \xi_3) Q(\sigma_2 \xi_2,\sigma_3 \xi_3)}{\omega(\sigma_1 \xi_1 + \sigma_2 \xi_2 +\sigma_3 \xi_3) - \sigma_1 \omega(\xi_1) - \tau \omega(\sigma_2 \xi_2 + \sigma_3 \xi_3)}.}
$$
This will be proved in Section \ref{caseC0}.
\end{itemize}

In both cases, $K$ enjoys the natural symmetry \eqref{symmetryK}.

\subsection{Treating \eqref{NOC} in full generality}
The result which was stated in the previous section holds if conditions (C1) to (C5) are satisfied, and we will not aim here for a more general statement. This is because these conditions hold on models of interest, and to avoid further developments which would be similar in spirit and technical. However, we indicate briefly how the correct kinetic equation can be derived if these conditions are not satisfied. 
\begin{itemize}
\item If (C1) does not hold, then the correct model is a quadratic kinetic equation instead of a cubic one, and it can be computed by the same means.

\medskip

\item If (C2) is violated, then another term has to be added to the cubic collision kernel, which corresponds to resonances of the type $+++$ or $--+$. Its formula can be read off from the computations below.

\medskip

\item If (C3) does not hold, there seems to be little hope that a kinetic equation involving quasi-resonances can be derived. For $\varepsilon \gg 1$, there is no reason why resonances should play any role. For $\varepsilon \ll N^{-\frac 12}$, exact resonances become key but they depend on the exact model under consideration, and no general theory is expected.

\medskip

\item If (C4) does not hold, an additional term should be added to the kinetic equation; however it seems that this condition is always satisfied for systems which are natural from a physical viewpoint (e.g. pinned systems, for which the dispersion relation does vanish at the origin).

\medskip

\item If (C5) is not satisfied, there is no doubt that a kinetic equation still holds, but its derivation would require further developments from which we refrain.
\end{itemize}

\subsection{Heuristic proof of the kinetic limit: the case $Q=0$} \label{caseQ0}
\label{section_heuristic_proof}
Considering the nonlinear oscillator chain in the weakly turbulent regime which was defined in the previous section, we will now present the derivation of the wave kinetic equation in the case $Q=0$. The cubic symbol is given by
$$
C^\sigma(\xi_1,\xi_2,\xi_3) = C(\sigma_1 \xi_1, \sigma_2 \xi_2, \sigma_3 \xi_3)=C^\sigma_{123}.
$$

The essential ideas behind this derivation are well-known and can be found in the classical references \cite{Peierls,Hasselmann1962,ZakharovLvovFalkovich}, albeit on different examples and with varying degrees of rigor. In the particular case of \eqref{NOC}, we refer to \cite{OnoratoLvovDematteisChibaro,Lukkarinen2016}. Nevertheless, we hope that the version of the derivation which follows will be useful to the reader, in that it is self-contained, keeps track of all constants, and specifies the regime of validity of \eqref{KWE}.

\noindent
\underline{Step 1: weakly nonlinear expansion}. We consider the initial data problem \eqref{IDP}, and we expand $a(t,\xi)$
in powers of $\varepsilon$. To order zero, $a(t,\xi) = A(\xi)$. There are no terms of order one; to order two, we find
\begin{align*}
a(t,\xi) & = A(\xi) + i \frac{(2\pi)^2 \varepsilon^2}{N} \sum_{\sigma_i} \sum_{\underline{\xi}_{1,2,3}}  \int_0^t  C_{1,2,3}^\sigma e^{is \Omega_{0,1,2,3}^{\mathbf{\sigma}}} A^{\sigma_1}_1 A^{\sigma_2}_2 A^{\sigma_3}_3  \delta_{0,1,2,3}^{\mathbf{\sigma}} \dd s + O(\varepsilon^5) \\
& = A(\xi) + \frac{(2\pi)^4 \varepsilon^2}{N} \sum_{\sigma_i} \sum_{\underline{\xi}_{1,2,3}} \frac{C_{1,2,3}^\sigma}{\Omega_{0,1,2,3}^{\mathbf{\sigma}}} \left[ e^{it \Omega_{0,1,2,3}^{\mathbf{\sigma}} }- 1 \right] A^{\sigma_1}_1 A^{\sigma_2}_2 A^{\sigma_3}_3  \delta_{0,1,2,3}^{\mathbf{\sigma}}  + O(\varepsilon^5) .
\end{align*}
We can then compute the terms of order five, which are (by symmetry in the entries of the cubic term)
\begin{align*}
& 3i \frac{(2\pi)^4 \varepsilon^4}{N^2} \sum_{\sigma_i} \int_0^t \sum_{\underline{\xi}_{1,2,3,4,5,6}}  C_{1,2,3}^\sigma e^{is \Omega_{0,1,2,3}^{\mathbf{\sigma}}} \frac{C_{4,5,6}^\sigma}{\Omega_{3,4,5,6}^{\mathbf{\sigma}}} \left[ e^{is \sigma_3 \Omega_{3,4,5,6}^{\mathbf{\sigma}} }- 1 \right] \times \\
& \qquad \qquad\qquad\qquad \qquad\qquad\qquad\qquad A^{\sigma_1}_1 A^{\sigma_2}_2 A^{\sigma_3 \sigma_4}_4 A^{\sigma_3 \sigma_5}_5 A^{\sigma_3 \sigma_6}_6 \delta_{0,1,2,3}^{\mathbf{\sigma}} \delta_{3,4,5,6}^\sigma    \\
& = 3 \frac{(2\pi)^4 \varepsilon^4}{N^2} \sum_{\sigma_i} \sum_{\underline{\xi}_{1,2,3,4,5,6}} C_{1,2,3}^\sigma \frac{C_{4,5,6}^\sigma}{\Omega_{3,4,5,6}^{\mathbf{\sigma}}}
\left[ \frac{e^{it \Omega_{0,1,2,4,5,6}^\sigma} - 1}{ \Omega_{0,1,2,4,5,6}^\sigma } -  \frac{e^{it \Omega_{0,1,2,3}^{\mathbf{\sigma}}} - 1}{\Omega_{0,1,2,3}^{\mathbf{\sigma}}} \right] \times \\
& \qquad\qquad\qquad\qquad \qquad\qquad\qquad\qquad A^{\sigma_1}_1 A^{\sigma_2}_2 A^{\sigma_3 \sigma_4}_4 A^{\sigma_3 \sigma_5}_5 A^{\sigma_3 \sigma_6}_6 \delta_{0,1,2,3}^{\mathbf{\sigma}} \delta_{3,4,5,6}^\sigma  
\end{align*}
where we denoted
$$
 \Omega_{0,1,2,4,5,6}^\sigma = \omega_0 - \sigma_1 \omega_1 - \sigma_2 \omega_2 - \sigma_3 \sigma_4 \omega_4 - \sigma_3 \sigma_5 \omega_5 - \sigma_3 \sigma_6 \omega_6. 
$$
Overall, we found the expansion
\begin{align*}
a(t,\underline{\xi}) & = A(\xi) + \frac{(2\pi)^4 \varepsilon^2}{N}  \sum_{\sigma_i}  \sum_{\underline{\xi}_{1,2,3} } \frac{C_{1,2,3}^\sigma}{\Omega_{0,1,2,3}^{\mathbf{\sigma}}} \left[ e^{it \Omega_{0,1,2,3}^{\mathbf{\sigma}} }- 1 \right] A^{\sigma_1}_1 A^{\sigma_2}_2 A^{\sigma_3}_3  \delta_{0,1,2,3}^{\mathbf{\sigma}}  \\
& \qquad + \frac{3 (2\pi)^4 \varepsilon^4}{N^2} \sum_{\sigma_i} \sum_{\underline{\xi}_{1,2,3,4,5,6}} C_{1,2,3}^\sigma \frac{C_{4,5,6}^\sigma}{\Omega_{3,4,5,6}^{\mathbf{\sigma}}}
\left[ \frac{e^{it \Omega_{0,1,2,4,5,6}^\sigma} - 1}{ \Omega_{0,1,2,4,5,6}^\sigma } -  \frac{e^{it \Omega_{0,1,2,3}^{\mathbf{\sigma}}} - 1}{\Omega_{0,1,2,3}^{\mathbf{\sigma}}} \right] \times \\
& \qquad \qquad \qquad 
A^{\sigma_1}_1 A^{\sigma_2}_2 A^{\sigma_3 \sigma_4}_4 A^{\sigma_3 \sigma_5}_5 A^{\sigma_3 \sigma_6}_6 \delta_{0,1,2,3}^{\mathbf{\sigma}} \delta_{3,4,5,6}^\sigma  \\
& \qquad + O(\varepsilon^6) \\
& = \{ \mbox{linear} \} + \{ \mbox{cubic} \} + \{ \mbox{quintic} \} + \{ \mbox{higher-order terms} \}
\end{align*}
We are using the convention that, if $a=b=0$, then
\begin{equation}
\label{canardsiffleur1}
\frac{e^{ita}-1}{a} = it, \qquad \frac{1}{b}\left( \frac{e^{it(a+b)}-1}{a+b} - \frac{e^{ita}-1}{a}\right) = - \frac{1}{2} t^2.
\end{equation}

\medskip
\noindent
\underline{Step 2: statistical cancellations through pairing} 
The next step is to compute the expectation (with respect to the random data) of $|a|^2$ with the help of Isserlis' formula
\begin{align*}
\mathbb{E}(\mathcal{G}_{\underline{\xi}_1} \dots \mathcal{G}_{\underline{\xi}_{n_+}} \overline{\mathcal{G}_{\underline{\eta}_1}} \dots \overline{\mathcal{G}_{\underline{\eta}_{n_-}}}) = 
\begin{cases} 
0 & \mbox{if $n_+ \neq n_-$} \\
\# \; \mbox{pairings} & \mbox{if $n_+ = n_-=n$} \},
\end{cases}
\end{align*}
where a pairing is a map $\varphi : \{ 1,\dots,n \} \to \{ 1,\dots,n \}$ such that  $\underline{\xi}_{k} = \underline{\eta}_{\varphi(k)}$. 

Writing $a$ as a sum of a linear, trilinear, and five-linear term, expand $\mathbb{E} |a(\xi)|^2$ and examine the various contributions, examining the pairings between linear, cubic and quinitic terms.
\begin{itemize}
\item Linear-linear interactions obviously yield $|A(\xi)|^2$, which is order $0$ in $\varepsilon$.
\item Linear-cubic interactions would yield a term of order $\varepsilon^2$, but it vanishes. Indeed, it equals
$$
2 \frac{(2\pi)^2 \varepsilon^2}{N} \operatorname{Re} \sum_{\sigma_i}  \sum_{\underline{\xi}_{1,2,3} } \frac{C_{1,2,3}^\sigma}{\Omega_{0,1,2,3}^{\mathbf{\sigma}}} \left[ e^{it \Omega_{0,1,2,3}^{\mathbf{\sigma}} }- 1 \right] \mathbb{E} \left[\overline{A_0} A^{\sigma_1}_1 A^{\sigma_2}_2 A^{\sigma_3}_3  \right]  \delta_{0,1,2,3}^{\mathbf{\sigma}}.
$$
Applying Isserlis formula, we consider possible pairings between the indices $0,1,2,3$. By symmetry, it suffices to consider the case where $0$ is paired to $1$ and $2$ to $3$. In that case, $\Omega^\sigma_{0,1,2,3} =0$ and therefore the above term vanishes since it equals
\begin{equation}
\label{canardsiffleur2}
\operatorname{Re} \left[ \frac{e^{it\cdot0 }- 1 }{0}\right] = 0.
\end{equation}
\item Cubic-cubic interactions give the following term of order $\varepsilon^4$ 
\begin{align*}
&\frac{(2\pi)^4 \varepsilon^4}{N^2}  \sum_{\sigma_i}  \sum_{\underline{\xi}_{1,2,3,4,5,6}} \frac{C_{1,2,3}^\sigma}{\Omega_{0,1,2,3}^{\mathbf{\sigma}}} \left[ e^{it \Omega_{0,1,2,3}^{\mathbf{\sigma}} }- 1 \right]\frac{C_{4,5,6}^\sigma}{\Omega^\sigma_{0,4,5,6}} \overline{ \left[ e^{it \Omega_{0,4,5,6}^\sigma }- 1 \right]} \times \\
& \qquad\qquad\qquad\qquad\qquad\qquad\qquad \mathbb{E} \left[ A^{\sigma_1}_1 A^{\sigma_2}_2 A^{\sigma_3}_3 \overline{A^{\sigma_4}_4 A^{\sigma_5}_5 A^{\sigma_6}_6} \right] \delta_{0,1,2,3}^{\mathbf{\sigma}} \delta_{0,4,5,6}^\sigma
\end{align*}
There remains to examine all pairing possibilities in the expectation term. If $\{ 1,2,3 \}$ is paired with $\{ 4,5,6 \}$, then up to permutations $\underline{\xi_i} = \underline{\xi_{3+i}}$ and $\Omega_{0,1,2,3}^{\mathbf{\sigma}} = \Omega^\sigma_{0,4,5,6}$. There are 6 such permutations, which gives a factor 6 and thus the contribution
\begin{equation}
\label{harlebievre1}
\frac{6 (2\pi)^4 \varepsilon^4}{N^2} \sum_{\sigma_1,\sigma_2,\sigma_3} \sum_{\underline{\xi}_{1,2,3}} |C_{1,2,3}^\sigma|^2 \left| \frac{e^{it \Omega_{0,1,2,3}^{\mathbf{\sigma}} }- 1}{\Omega_{0,1,2,3}^{\mathbf{\sigma}}} \right|^2 f_1 f_2 f_3 \delta_{0,1,2,3}^{\mathbf{\sigma}}.
\end{equation}
If there is a pairing within $\{1,2,3\}$, then up to permutation ($9$ possibilities), we can assume that $1$ is paired with $2$, $4$ is paired with $5$, and $3$ with $6$; this implies $\xi_3 = \sigma_3 \xi_0$ and $\Omega_{0,1,2,3}^{\mathbf{\sigma}} = \Omega_{3,4,5,6}^{\mathbf{\sigma}} = \omega_0(1-\sigma_3)$. Separating the cases $\sigma_3=1$ and $\sigma_3 = -1$ gives the contribution
\begin{align}
& \frac{9 (2\pi)^4 \varepsilon^4}{N^2} t^2 \sum_{\sigma_1,\sigma_4} \sum_{\underline{\xi}_{1,4} } C_{1,1,0}^{\sigma_1,-\sigma_1,+} C_{4,4,0}^{\sigma_4,-\sigma_4,+} f_0 f_1 f_4 \label{harlebievre2} \\
& \qquad \qquad + \frac{9 (2\pi)^4 \varepsilon^4}{N^2} \sum_{\sigma_1,\sigma_4} \sum_{\underline{\xi}_{1,4} } C_{1,1,0}^{\sigma_1,-\sigma_1,-} C_{4,4,0}^{\sigma_4,-\sigma_4,-} \left| \frac{e^{2it \omega_0 }- 1}{2\omega_0} \right|^2 f_0 f_1 f_4 \label{harlebievre3}
\end{align}

\item Quintic-linear interactions give the term of order $\varepsilon^4$
\begin{align*}
& \frac{6 (2\pi)^4 \varepsilon^4}{N^2} \operatorname{Re} \sum_{\sigma_i} \sum_{\underline{\xi}_{1,2,3,4,5,6}} C_{1,2,3}^\sigma \frac{C^\sigma_{4,5,6}}{\Omega_{3,4,5,6}^{\mathbf{\sigma}}}
\left[ \frac{e^{it \Omega_{0,1,2,4,5,6}^\sigma} - 1}{ \Omega_{0,1,2,4,5,6}^\sigma } - \frac{e^{it \Omega_{0,1,2,3}^{\mathbf{\sigma}}} - 1}{\Omega_{0,1,2,3}^{\mathbf{\sigma}}} \right] \times \\
& \qquad\qquad\qquad\qquad\qquad\qquad\qquad \mathbb{E} \left[\overline{A_0} A^{\sigma_1}_1 A^{\sigma_2}_2 A^{\sigma_3 \sigma_4}_4 A^{\sigma_3 \sigma_5}_5 A^{\sigma_3 \sigma_6}_6 \right] \delta_{0,1,2,3}^{\mathbf{\sigma}} \delta_{3,4,5,6}^\sigma 
\end{align*}
Once again, we examine the possible pairings in the expectation term. If $\{ 0,1,2 \}$ is paired with $\{ 4,5,6 \}$, then we can assume up to permutation of the indices that $0$ is mapped to $4$, $1$ to $5$ and $2$ to $6$ - the number of permutations being $6$. We first notice that $\Omega_{0,1,2,4,5,6}^\sigma = 0$ and as a consequence, the corresponding term cancels because of the $\operatorname{Re}$ prefactor. Second, $\sigma_3 \sigma_4 =1$, $\sigma_1 =- \sigma_3 \sigma_5$ and $\sigma_2 = - \sigma_3 \sigma_6$, which implies $\Omega_{0,1,2,3} = -\sigma_3 \Omega_{3,4,5,6}$. Third, there holds $C^\sigma_{1,2,3} = C^\sigma_{4,5,6}$; indeed we can use the symmetries \eqref{symmetry_QC} of $C$ to compute that
\begin{align*}
C^\sigma_{1,2,3} & = C(\sigma_1 \xi_1,\sigma_2 \xi_2, \sigma_3 \xi_3) =  C(\sigma_1 \xi_1,\sigma_2 \xi_2, -\xi_0) = C(-\sigma_3 \sigma_5 \xi_5,-\sigma_3 \sigma_6 \xi_6,-\sigma_3 \sigma_4 \xi_4) \\
& = C( \sigma_5 \xi_5,\sigma_6 \xi_6, \sigma_4 \xi_4) = C^\sigma_{4,5,6}.
\end{align*}
Finally, we use the identity 
\begin{equation}
\label{canardsiffleur3}
\operatorname{Re}(e^{i \theta} - 1) = -\frac 12 |e^{i\theta}-1|^2
\end{equation}
to obtain
\begin{equation}
\label{harlebievre4}
- \frac{18 (2\pi)^4 \varepsilon^4}{N^2} \sum_{\sigma_1,\sigma_2,\sigma_3} \sigma_3 \sum_{\underline{\xi}_{1,2,3}} |C_{1,2,3}^\sigma|^2 \left| \frac{e^{it \Omega_{0,1,2,3}^{\mathbf{\sigma}}} - 1}{\Omega_{0,1,2,3}^{\mathbf{\sigma}}} \right|^2 f_0 f_1 f_2 \delta_{0,1,2,3}^\sigma
\end{equation}
If $0$ is paired with $1$ or $2$, then we can assume up to permutation that $0$ is paired to $1$, $4$ to $5$ and $2$ to $6$ (6 permutations). This implies that $\Omega_{0,1,2,4,5,6}^\sigma = 0$ and furthermore $\xi_2 = -\sigma_2 \sigma_3 \xi_3$ hence $-\sigma_3 \Omega_{0,1,2,3} = \Omega_{3,4,5,6}^{\mathbf{\sigma}} = \omega_2 (1 + \sigma_2 \sigma_3)$. We now distinguish between the cases $\sigma_3=-\sigma_2$ and $\sigma_3 = \sigma_2$, with the help of \eqref{canardsiffleur1}, \eqref{canardsiffleur2} and \eqref{canardsiffleur3}, leading to
\begin{align}
& -\frac{18 (2\pi)^4 \varepsilon^4}{N^2} t^2 \sum_{\sigma_3,\sigma_4} \sigma_3 \sum_{\underline{\xi}_{2,4}} C^{+,-\sigma_3,\sigma_3}_{0,2,2} C_{4,4,2}^{\sigma_4,-\sigma_4,+} f_0 f_2 f_4 \label{harlebievre5} \\
& \qquad -\frac{18 (2\pi)^4 \varepsilon^4}{N^2} \sum_{\sigma_3,\sigma_4} \sigma_3 \sum_{\underline{\xi}_{2,4}} C^{+,\sigma_3,\sigma_3}_{0,2,-2} C_{4,4,2}^{\sigma_4,-\sigma_4,-}  \left| \frac{e^{2it \omega_3}-1}{2\omega_3}\right|^2 f_0 f_2 f_4 \label{harlebievre6}
\end{align}
If $1$ and $2$ are paired together, then up to permutation of the indices ($3$ such permutations), we can assume that $0$ is paired with $4$ and $5$ with $6$. Then $\Omega_{012456}^\sigma=0$, and we also have $\xi_3 = \sigma_3 \xi_0$, hence $\Omega_{0,1,2,3}^\sigma = \Omega^\sigma_{3,4,5,6} = \omega_0(1-\sigma_3)$. Distinguishing between the cases $\sigma_3=1$ and $\sigma_3 = -1$, this gives
\begin{align}
&- \frac{9 (2\pi)^4 \varepsilon^4}{N^2} t^2 \sum_{\sigma_1,\sigma_5} \sum_{\underline{\xi}_{1,5}} C_{1,1,0}^{\sigma_1,-\sigma_1,+} C_{0,5,5}^{+,\sigma_5,-\sigma_5} f_0 f_1 f_5 \label{harlebievre7} \\
& \qquad + \frac{9 (2\pi)^4 \varepsilon^4}{N^2} \sum_{\sigma_1,\sigma_5} \sum_{\underline{\xi}_{1,5}} C_{1,1,-0}^{\sigma_1,-\sigma_1,-} C_{0,5,5}^{-,\sigma_5,-\sigma_5} \left| \frac{e^{2it \omega_0}-1}{2 \omega_0} \right|^2 f_0 f_1 f_5 \label{harlebievre8}
\end{align}
\item Cubic-Quintic and higher-order interactions are of order equal to or higher than $\varepsilon^6$, and will be disregarded.
\end{itemize}
We now record a number of cancellations between the various terms that were computed. First, \eqref{harlebievre5} and \eqref{harlebievre6} vanish because of the sum over $\sigma_3$. Second, \eqref{harlebievre2} and \eqref{harlebievre7} add up to zero. Third, \eqref{harlebievre3} and \eqref{harlebievre8} become large if $\omega_0$ vanishes; but this is canceled by the coefficient $C$, which we assume to vanish there (assumption (C4)). There only remains \eqref{harlebievre1} and \eqref{harlebievre4}, namely
\begin{equation}
\label{colvert}
\begin{split}
& \frac{6 (2\pi)^4 \varepsilon^4}{N^2} \sum_{\sigma_1,\sigma_2,\sigma_3} \sum_{\underline{\xi}_{1,2,3}} |C_{1,2,3}^\sigma|^2 \left| \frac{e^{it \Omega_{0,1,2,3}^{\mathbf{\sigma}} }- 1}{\Omega_{0,1,2,3}^{\mathbf{\sigma}}} \right|^2 f_1 f_2 f_3 \delta_{0,1,2,3}^{\mathbf{\sigma}} \\
& \qquad \qquad - \frac{18 (2\pi)^4 \varepsilon^4}{N^2} \sum_{\sigma_1,\sigma_2,\sigma_3} \sigma_3 \sum_{\underline{\xi}_{1,2,3}} |C_{1,2,3}^\sigma|^2 \left| \frac{e^{it \Omega_{0,1,2,3}^{\mathbf{\sigma}}} - 1}{\Omega_{0,1,2,3}^{\mathbf{\sigma}}} \right|^2 f_0 f_1 f_2 \delta_{0,1,2,3}^\sigma
\end{split}
\end{equation}

\medskip
\noindent
\underline{Step 3: large volume limit.} We want to take the limit $N \to \infty$ in the above expression. Formally, this turns the discrete variable $\underline{\xi}$ into the continuous variable $\xi$, and costs a factor $\frac{N^2}{(2\pi)^2}$. Denoting furthermore 
$$
K_{0,1,2,3}^\sigma = |C_{1,2,3}^\sigma|^2 \qquad \mbox{(with $\xi_0 = \sigma_1 \xi_1 + \sigma_2 \xi_2 + \sigma_3 \xi_3$)},
$$
we find that \eqref{colvert} is equivalent to
\begin{equation}
\label{fuligulemorillon0}
\begin{split}
& 6 (2\pi)^2 \varepsilon^4 \sum_{\sigma_1,\sigma_2,\sigma_3}  K_{0,1,2,3}^\sigma \left| \frac{e^{it \Omega_{0,1,2,3}^{\mathbf{\sigma}} }- 1}{\Omega_{0,1,2,3}^{\mathbf{\sigma}}} \right|^2 f_1 f_2 f_3 \delta_{0,1,2,3}^{\mathbf{\sigma}} \\
& \qquad \qquad - 18 (2\pi)^2 \varepsilon^4 \sum_{\sigma_1,\sigma_2,\sigma_3} \sigma_3  K_{0,1,2,3}^\sigma \left| \frac{e^{it \Omega_{0,1,2,3}^{\mathbf{\sigma}}} - 1}{\Omega_{0,1,2,3}^{\mathbf{\sigma}}} \right|^2 f_0 f_1 f_2 \delta_{0,1,2,3}^\sigma
\end{split}
\end{equation}

When is this limit justified? The greatest difficulty is posed by the rapidly varying factor $\left| \frac{e^{it \Omega_{0,1,2,3}^{\mathbf{\sigma}} }- 1}{\Omega_{0,1,2,3}^{\mathbf{\sigma}}} \right|^2$. We can think of this expression as a function of $\xi_1$ and $\xi_2$ which varies locally on a scale $\sim \frac 1 t$ in the direction of the gradient. Thus its level sets have area $\sim \frac 1t$, and they should be well-sampled by the points of the two-dimensional lattice $(\underline{\xi_1},\underline{\xi_2})$. Since there are $\sim N^2$ lattice points, we need, assuming that these points are equidistributed,
\begin{equation}
\label{equidistribution}
t \ll N^2.
\end{equation}

\medskip
\noindent
\underline{Step 4: large time limit.} In the sense of distributions,
$$
\left| \frac{\sin(t \Omega)}{\Omega} \right|^2 \sim \pi t \delta(\Omega) \qquad \mbox{as $t \to \infty$},
$$
so that, as $t\to \infty$, \eqref{fuligulemorillon0} becomes equivalent to
\begin{equation} 
\label{sittelle0}
\begin{split}
& 12 \pi^3 t \varepsilon^4 \left[ \sum_{\sigma_1,\sigma_2,\sigma_3} \int K_{0,1,2,3}^\sigma f_1 f_2 f_3 \delta(\Omega_{0,1,2,3}^{\mathbf{\sigma}}) \delta_{0,1,2,3}^{\mathbf{\sigma}} \dd \underline{\xi}_{1,2,3} \right. \\
& \left. \qquad \qquad - 3 \sum_{\sigma_1,\sigma_2,\sigma_3} \sigma_3 \int K_{0,1,2,3}^\sigma f_0 f_1 f_2 \delta(\Omega_{0,1,2,3}^{\mathbf{\sigma}})  \delta_{0,1,2,3}^\sigma \dd \underline{\xi}_{1,2,3} \right].
\end{split}
\end{equation}

Finally, we are assuming that $\Omega_{0,1,2,3}^\sigma$ only vanishes if $\sigma_1 + \sigma_2 + \sigma_3 = 1$, so that we need to restrict the summation in \eqref{sittelle0} to these configurations. For the first term of \eqref{sittelle0}, it suffices to consider the term $(-,+,+)$ by symmetry, and to multiply the corresponding contribution by $3$ to account for the permutations. Considering the second term in $\eqref{sittelle0}$, the choice $\sigma_3= -$ implies $\sigma_1 = \sigma_2 = +$ and comes with a positive sign, while the choice $\sigma_3 = +$ gives the possibilities $(+,-)$ and $(-,+)$ for $(\sigma_1,\sigma_2)$ and comes with a negative sign. Overall, we find after relabeling that \eqref{sittelle0} equals
$$
36 \pi^3 t \varepsilon^4 \int K_{0,1,2,3}^{-,+,+} \left[f_1 f_2 f_3 + f_0 f_2 f_3 - f_0 f_1 f_2 - f_0 f_2 f_3 \right] \delta(\Omega_{0,1,2,3}) \delta(\Omega_{0,1,2,3}) \dd \xi_{1,2,3}.
$$
We recognize the collision operator, and it is now natural to define the kinetic time scale
$$
T_{\operatorname{kin}} = \varepsilon^{-4}.
$$
Since we are requiring $t \to \infty$, we need $T_{\operatorname{kin}} \to \infty$, which is consistent with the weakly nonlinear limit $\varepsilon \to 0$. Finally, we need to compare $T_{\operatorname{kin}}$ with the equidistribution condition \eqref{equidistribution}, which results into
$$
\varepsilon \gg N^{-\frac{1}{2}}.
$$

\medskip

\noindent
\underline{Symmetry of $K$}. We simply write $K(\xi,\xi_1,\xi_2,\xi_3)$ instead of $K^{-,+,+}(\xi,\xi_1,\xi_2,\xi_3)$, in other words
$$
K(\xi,\xi_1,\xi_2,\xi_3) = |C(-\xi_1,\xi_2,\xi_3)|^2 \qquad \mbox{with $\xi+\xi_1= \xi_2 + \xi_3$}.
$$
We want to check that $K$ enjoys the symmetry identities \eqref{symmetryK}. Since $C$ is invariant by permutation of its entries, it is clear that
$$
K(\xi,\xi_1,\xi_2,\xi_3) = K(\xi,\xi_1,\xi_3,\xi_2).
$$
To prove invariance by exchanging the first two with the last two variables, we use the symmetry properties of $C$ recalled in \eqref{symmetry_QC} and the identity $\xi+\xi_1 = \xi_2 + \xi_3$ as follows:
\begin{align*}
K(\xi,\xi_1,\xi_2,\xi_3) & = |C(-\xi_1,\xi_2,\xi_3)|^2 
= |C(-\xi_1,\xi_3,\xi_1-\xi_2-\xi_3)|^2 \\
& = |C(-\xi_1,\xi_3,-\xi)|^2 
= |C(-\xi_3,\xi,\xi_1)|^2 \\
& = K(\xi_2,\xi_3,\xi,\xi_1).
\end{align*}

\subsection{Heuristic derivation of the kinetic limit: the case $C=0$} \label{caseC0}
We are now working with the formula 
\begin{align*}
& C^{\sigma_1,\sigma_2,\sigma_3}(\xi_1, \xi_2, \xi_3) = \frac 13 \left[D^{\sigma_1,\sigma_2,\sigma_3}(\xi_1,\xi_2,\xi_3) + D^{\sigma_2,\sigma_1,\sigma_3}(\xi_2,\xi_1,\xi_3) + D^{\sigma_3,\sigma_2,\sigma_1}(\xi_3,\xi_2,\xi_1) \right]
\end{align*}
with
$$
D^{\sigma_1,\sigma_2,\sigma_3}(\xi_1,\xi_2,\xi_3) = -2 \sum_{\tau} \tau \frac{Q(\sigma_1 \xi_1,\sigma_2 \xi_2 + \sigma_3 \xi_3) Q(\sigma_2 \xi_2,\sigma_3 \xi_3)}{\omega(\sigma_1 \xi_1 + \sigma_2 \xi_2 +\sigma_3 \xi_3) - \sigma_1 \omega(\xi_1) - \tau \omega(\sigma_2 \xi_2 + \sigma_3 \xi_3)}.
$$
(assuming $Q$ real-valued for simplicity). Furthermore, we abbreviate
$$
C_{1,2,3}^\sigma = C^{\sigma_1,\sigma_2,\sigma_3}_{1,2,3}
$$
and we will use the identities
\begin{equation}
\label{symmetryC}
\begin{split}
& C_{1,2,3}^{\sigma_1, \sigma_2,\sigma_3} = C_{2,1,3}^{\sigma_2, \sigma_1,\sigma_3}
= C_{1,3,2}^{\sigma_1, \sigma_3,\sigma_2} \\
& C_{0,1,2}^{\sigma_3,-\sigma_3 \sigma_1,-\sigma_3 \sigma_2} = C_{1,2,3}^{\sigma_1, \sigma_2,\sigma_3}
\end{split}
\end{equation}
provided $\xi = \sigma_1 \xi_1 + \sigma_2 \xi_2 + \sigma_3 \xi_3$. The first line above is obvious, and the second requires a lengthy but elementary computation. We will only show that $D_{0,1,2}^{\sigma_3,-\sigma_3 \sigma_1,-\sigma_3 \sigma_2} = D_{3,2,1}^{\sigma_3, \sigma_2,\sigma_1}$, since the full proof follows a similar pattern. By the formula for $D$, the symmetries of Q in \eqref{symmetry_QC}, the reality of $Q^2$, and finally recalling that $\xi_0 = \sigma_1 \xi_1 + \sigma_2 \xi_2 + \sigma_3 \xi_3$,
\begin{align*}
D_{0,1,2}^{\sigma_3,-\sigma_3 \sigma_1,-\sigma_3 \sigma_2} & = -2 \sum_{\tau} \tau \frac{Q(\sigma_3 \xi_0,-\sigma_3 \sigma_1 \xi_1 - \sigma_3 \sigma_2 \xi_2) Q(-\sigma_3\sigma_1 \xi_1,-\sigma_3 \sigma_2 \xi_2)}{\omega(\xi_3) - \sigma_3 \omega(\xi_0) - \tau \omega(\sigma_1 \xi_1 + \sigma_2 \xi_2)} \\
& = -2 \sum_{\tau} \frac{Q(-\xi_0,\sigma_1 \xi_1 + \sigma_2 \xi_2) Q(\sigma_1 \xi_1, \sigma_2 \xi_2)}{\tau(\omega(\xi_3) - \sigma_3 \omega(\xi_0)) -  \omega(\sigma_1 \xi_1 + \sigma_2 \xi_2)} \\
& = -2 \sum_{\tau} \frac{Q(\sigma_1 \xi_1 + \sigma_2 \xi_2,\xi_0 - \sigma_1 \xi_1 - \sigma_2 \xi_2) Q(\sigma_1 \xi_1, \sigma_2 \xi_2)}{\tau(\omega(\xi_0) - \sigma_3 \omega(\xi_3)) -  \omega(\sigma_1 \xi_1 + \sigma_2 \xi_2)} \\
& = -2 \sum_{\tau} \frac{Q(\sigma_3 \xi_3,\sigma_1 \xi_1 + \sigma_2 \xi_2) Q(\sigma_1 \xi_1, \sigma_2 \xi_2)}{\tau(\omega(\xi_0) - \sigma_3 \omega(\xi_3)) -  \omega(\sigma_1 \xi_1 + \sigma_2 \xi_2)} = D_{3,2,1}^{\sigma_3, \sigma_2,\sigma_1}.
\end{align*}
After this preliminary discussion, we can now turn to the derivation itself! It follows the same pattern as in the case $Q=0$, but there are some subtle differences in the formulas. For this reason, we will omit some of the steps but write down all relevant identities.

\medskip

\noindent
\underline{Step 1: weakly nonlinear expansion}. 
We find the expansion
\begin{align*}
a(t,\underline{\xi}) & = A(\xi) + \frac{\varepsilon^2}{N}  \sum_{\sigma_i}  \sum_{\underline{\xi}_{1,2,3} } \frac{C^{\sigma}_{1,2,3}}{\Omega_{0,1,2,3}^{\mathbf{\sigma}}} \left[ e^{it \Omega_{0,1,2,3}^{\mathbf{\sigma}} }- 1 \right] A^{\sigma_1}_1 A^{\sigma_2}_2 A^{\sigma_3}_3  \delta_{0,1,2,3}^{\mathbf{\sigma}}  \\
& \qquad + \frac{3\varepsilon^4}{N^2} \sum_{\sigma_i} \sum{\underline{\xi}_{1,2,3,4,5,6}} C_{1,2,3}^\sigma \frac{C_{4,5,6}^{\sigma}}{\Omega_{3,4,5,6}^{\mathbf{\sigma}}}
\left[ \frac{e^{it \Omega_{0,1,2,4,5,6}^\sigma} - 1}{ \Omega_{0,1,2,4,5,6}^\sigma } -  \frac{e^{it \Omega_{0,1,2,3}^{\mathbf{\sigma}}} - 1}{\Omega_{0,1,2,3}^{\mathbf{\sigma}}} \right] \times \\
& \qquad \qquad \qquad 
A^{\sigma_1}_1 A^{\sigma_2}_2 A^{\sigma_3 \sigma_4}_4 A^{\sigma_3 \sigma_5}_5 A^{\sigma_3 \sigma_6}_6 \delta_{0,1,2,3}^{\mathbf{\sigma}} \delta_{3,4,5,6}^\sigma  \\
& \qquad + O(\varepsilon^7) \\
& = \{ \mbox{linear} \} + \{ \mbox{cubic} \} + \{ \mbox{quintic} \} + \{ \mbox{higher-order terms} \},
\end{align*}
still using the convention that, if $a=b=0$, then
\begin{equation}
\label{canardsiffleur1bis}
\frac{e^{ita}-1}{a} = it, \qquad \frac{1}{b}\left( \frac{e^{it(a+b)}-1}{a+b} - \frac{e^{ita}-1}{a}\right) = - \frac{1}{2} t^2.
\end{equation}

\medskip
\noindent
\underline{Step 2: statistical cancellations through pairing} 
We expand $\mathbb{E} |a(\xi)|^2$ and examine the various contributions of the pairings between linear, cubic and quinitic terms.
\begin{itemize}
\item Linear-linear interactions obviously yield $|A(\xi)|^2$, which is order $0$ in $\varepsilon$.
\item Linear-cubic interactions would yield a term of order $\varepsilon^2$, but it vanishes. Indeed, it equals
$$
2 \frac{\varepsilon^2}{N} \operatorname{Re} \sum_{\sigma_i}  \sum_{\underline{\xi}_{1,2,3} } \frac{C_{1,2,3}^\sigma}{\Omega_{0,1,2,3}^{\mathbf{\sigma}}} \left[ e^{it \Omega_{0,1,2,3}^{\mathbf{\sigma}} }- 1 \right] \mathbb{E} \left[\overline{A_0} A^{\sigma_1}_1 A^{\sigma_2}_2 A^{\sigma_3}_3  \right]  \delta_{0,1,2,3}^{\mathbf{\sigma}}.
$$
Applying Isserlis formula, we consider possible pairings between the indices $0,1,2,3$. By symmetry, it suffices to consider the case where $0$ is paired to $1$ and $2$ to $3$. In that case, $\Omega^\Omega_{0,1,2,3} =0$ and therefore the above term vanishes.
\item Cubic-cubic interactions give the following term of order $\varepsilon^4$
\begin{align*}
&\frac{(2\pi)^4 \varepsilon^4}{N^2}  \sum_{\sigma_i}  \sum_{\underline{\xi}_{1,2,3,4,5,6}} \frac{C_{1,2,3}^\sigma}{\Omega_{0,1,2,3}^{\mathbf{\sigma}}} \left[ e^{it \Omega_{0,1,2,3}^{\mathbf{\sigma}} }- 1 \right]\frac{C_{4,5,6}^\sigma}{\Omega^\sigma_{0,4,5,6}} \overline{ \left[ e^{it \Omega_{0,4,5,6}^\sigma }- 1 \right]} \times \\
& \qquad\qquad\qquad\qquad\qquad\qquad\qquad \mathbb{E} \left[ A^{\sigma_1}_1 A^{\sigma_2}_2 A^{\sigma_3}_3 \overline{A^{\sigma_4}_4 A^{\sigma_5}_5 A^{\sigma_6}_6} \right] \delta_{0,1,2,3}^{\mathbf{\sigma}} \delta_{0,4,5,6}^\sigma
\end{align*}
There remains to examine all pairing possibilities in the expectation term. If $\{ 1,2,3 \}$ is paired with $\{ 4,5,6 \}$, then up to permutations $\underline{\xi_i} = \underline{\xi_{3+i}}$, $\sigma_i = \sigma_{i+3}$ and $\Omega_{0,1,2,3}^{\mathbf{\sigma}} = \Omega^\sigma_{0,4,5,6}$. There are 6 such permutations, which gives a factor 6 and thus the contribution
\begin{equation}
\label{harlebievre1bis}
\frac{6 (2\pi)^4 \varepsilon^4}{N^2} \sum_{\sigma_1,\sigma_2,\sigma_3} \sum_{\underline{\xi}_{1,2,3}} |C_{1,2,3}^\sigma|^2 \left| \frac{e^{it \Omega_{0,1,2,3}^{\mathbf{\sigma}} }- 1}{\Omega_{0,1,2,3}^{\mathbf{\sigma}}} \right|^2 f_1 f_2 f_3 \delta_{0,1,2,3}^{\mathbf{\sigma}}.
\end{equation}
If there is a pairing within $\{1,2,3\}$, then up to permutation ($9$ possibilities), we can assume that $1$ is paired with $2$, $4$ is paired with $5$, and $3$ with $6$. This implies $\sigma_1 = -\sigma_2$, $\sigma_3 = \sigma_6$, $\sigma_4 = -\sigma_5$,  $\xi_3 = \sigma_3 \xi_0$ and $\Omega_{0,1,2,3}^{\mathbf{\sigma}} = \Omega_{3,4,5,6}^{\mathbf{\sigma}} = \omega_0(1-\sigma_3)$. Separating the cases $\sigma_3=1$ and $\sigma_3 = -1$ gives the contribution
\begin{align}
& \frac{9 (2\pi)^4 \varepsilon^4}{N^2} t^2 \sum_{\sigma_1,\sigma_4} \sum_{\underline{\xi}_{1,4} } C_{1,1,0}^{\sigma_1,-\sigma_1,+} C_{4,4,0}^{\sigma_4,-\sigma_4,+} f_0 f_1 f_4 \label{harlebievre2bis} \\
& \qquad \qquad + \frac{9 (2\pi)^4 \varepsilon^4}{N^2} \sum_{\sigma_1,\sigma_4} \sum_{\underline{\xi}_{1,4} } C_{1,1,-0}^{\sigma_1,-\sigma_1,-} C_{4,4,-0}^{\sigma_4,-\sigma_4,-} \left| \frac{e^{2it \omega_0 }- 1}{2\omega_0} \right|^2 f_0 f_1 f_4 \label{harlebievre3bis}
\end{align}
(where we denote $C_{1,1,-0} = C(\xi_0,\xi_1,\xi_1,-\xi_0)$).

\item Quintic-linear interactions give the term of order $(2\pi)^4 \varepsilon^4$
\begin{align*}
& \frac{6(2\pi)^4 \varepsilon^4}{N^2} \operatorname{Re} \sum_{\sigma_i} \sum_{\underline{\xi}_{1,2,3,4,5,6}} C_{1,2,3}^\sigma \frac{C_{4,5,6}^\sigma}{\Omega_{3,4,5,6}^{\mathbf{\sigma}}}
\left[ \frac{e^{it \Omega_{0,1,2,4,5,6}^\sigma} - 1}{ \Omega_{0,1,2,4,5,6}^\sigma } - \frac{e^{it \Omega_{0,1,2,3}^{\mathbf{\sigma}}} - 1}{\Omega_{0,1,2,3}^{\mathbf{\sigma}}} \right] \times \\
& \qquad\qquad\qquad\qquad\qquad\qquad\qquad \mathbb{E} \left[\overline{A_0} A^{\sigma_1}_1 A^{\sigma_2}_2 A^{\sigma_3 \sigma_4}_4 A^{\sigma_3 \sigma_5}_5 A^{\sigma_3 \sigma_6}_6 \right] \delta_{0,1,2,3}^{\mathbf{\sigma}} \delta_{3,4,5,6}^\sigma 
\end{align*}
Once again, we examine the possible pairings in the expectation term. If $\{ 0,1,2 \}$ is paired with $\{ 4,5,6 \}$, then we can assume up to permutation of the indices that $0$ is mapped to $4$, $1$ to $3$ and $2$ to $6$ - the number of permutations being $6$. We first notice that $\Omega_{0,1,2,4,5,6}^\sigma = 0$ and as a consequence, the corresponding term cancels because of the $\operatorname{Re}$ prefactor. Second, $\sigma_3 \sigma_4 =1$, $\sigma_1 = -\sigma_3 \sigma_5$ and $\sigma_2 = -\sigma_3 \sigma_6$, which implies $\Omega_{0,1,2,3} = -\sigma_3 \Omega_{3,4,5,6}$ and, by \eqref{symmetryC}, that $C^{\sigma}_{1,2,3} = C^{\sigma}_{4,5,6}$. Finally, we use the identity \eqref{canardsiffleur3} to obtain
\begin{equation}
\label{harlebievre4bis}
- \frac{18 (2\pi)^4 \varepsilon^4}{N^2} \sum_{\sigma_1,\sigma_2,\sigma_3} \sigma_3 \sum_{\underline{\xi}_{1,2,3}} |C_{1,2,3}^{\sigma}|^2 \left| \frac{e^{it \Omega_{0,1,2,3}^{\mathbf{\sigma}}} - 1}{\Omega_{0,1,2,3}^{\mathbf{\sigma}}} \right|^2 f_0 f_1 f_2 \delta_{0,1,2,3}^\sigma
\end{equation}
If $0$ is paired with $1$ or $2$, then we can assume up to permutation that $0$ is paired to $1$, $4$ to $5$ and $2$ to $6$ (6 permutations), and as a result $\sigma_1 =+$, $\sigma_2 = -\sigma_3 \sigma_6$ and $\sigma_4 = -\sigma_5$. This implies that $\Omega_{0,1,2,4,5,6}^\sigma = 0$ and furthermore $\xi_2 = -\sigma_2 \sigma_3 \xi_3$ hence $-\sigma_3 \Omega_{0,1,2,3} = \Omega_{3,4,5,6}^{\mathbf{\sigma}} = \omega_2 (1+ \sigma_2 \sigma_3)$. We now distinguish between the cases $\sigma_3= - \sigma_2$ and $\sigma_3 = \sigma_2$, leading to
\begin{align}
& -\frac{18 (2\pi)^4 \varepsilon^4}{N^2} t^2 \sum_{\sigma_3,\sigma_4} \sigma_3 \sum_{\underline{\xi}_{2,4}} C_{0,2,2}^{+,-\sigma_3,\sigma_3} C_{4,4,2}^{\sigma_4,-\sigma_4,+} f_0 f_2 f_4 \label{harlebievre5bis} \\
& \qquad -\frac{18 (2\pi)^4 \varepsilon^4}{N^2} \sum_{\sigma_3,\sigma_4} \sigma_3 \sum_{\underline{\xi}_{2,4}} C_{0,2,-2}^{+,\sigma_3,\sigma_3} C_{4,4,2}^{\sigma_4,-\sigma_4,-} \left| \frac{e^{2it \omega_3}-1}{2\omega_3}\right|^2 f_0 f_2 f_4 \label{harlebievre6bis}
\end{align}
If $1$ and $2$ are paired together, then up to permutation of the indices ($3$ such permutations), we can assume that $0$ is paired with $4$ and $5$ with $6$. Then $\Omega_{012456}^\sigma=0$, and we also have $\xi_3 = \sigma_3 \xi_0$, hence $\Omega_{0,1,2,3}^\sigma = \Omega^\sigma_{3456} = \omega_0(1-\sigma_3)$. Distinguishing between the cases $\sigma_3=1$ and $\sigma_3 = -1$, this gives
\begin{align}
&- \frac{9 (2\pi)^4 \varepsilon^4}{N^2} t^2 \sum_{\sigma_1,\sigma_5} \sum_{\underline{\xi}_{1,5}} C_{1,1,0}^{\sigma_1,-\sigma_1,+} C_{0,0,5,5}^{+,\sigma_5,-\sigma_5} f_0 f_1 f_5 \label{harlebievre7bis} \\
& \qquad + \frac{9 (2\pi)^4 \varepsilon^4}{N^2} \sum_{\sigma_1,\sigma_5} \sum_{\underline{\xi}_{1,5}} C_{1,1,-0}^{\sigma_1,-\sigma_1,-} C_{-0,0,5,5}^{-,\sigma_5,-\sigma_5} \left| \frac{e^{2it \omega_0}-1}{2 \omega_0} \right|^2 f_0 f_1 f_5 \label{harlebievre8bis}
\end{align}
\item Cubic-Quintic and higher-order interactions are of order equal to or higher than $\varepsilon^6$, and will be disregarded.
\end{itemize}
We now record a number of cancellations between the various terms that were computed. First, \eqref{harlebievre5bis} and \eqref{harlebievre6bis} vanish because of the sum over $\sigma_3$. Second, \eqref{harlebievre2bis} and \eqref{harlebievre7bis} add up to zero. Third, \eqref{harlebievre3bis} and \eqref{harlebievre8bis} become large if $\omega_0$; but this is canceled by the coefficient $C$, which we assume to vanish there.

There only remains \eqref{harlebievre1bis} and \eqref{harlebievre4bis}, namely
\begin{equation}
\label{colvert}
\begin{split}
& \frac{6 (2\pi)^4 \varepsilon^4}{N^2} \sum_{\sigma_1,\sigma_2,\sigma_3} \sum_{\underline{\xi}_{1,2,3}} |C_{1,2,3}^\sigma|^2 \left| \frac{e^{it \Omega_{0,1,2,3}^{\mathbf{\sigma}} }- 1}{\Omega_{0,1,2,3}^{\mathbf{\sigma}}} \right|^2 f_1 f_2 f_3 \delta_{0,1,2,3}^{\mathbf{\sigma}} \\
& \qquad \qquad - \frac{18 (2\pi)^4 \varepsilon^4}{N^2} \sum_{\sigma_1,\sigma_2,\sigma_3} \sigma_3 \sum_{\underline{\xi}_{1,2,3}} |C_{1,2,3}^\sigma|^2 \left| \frac{e^{it \Omega_{0,1,2,3}^{\mathbf{\sigma}}} - 1}{\Omega_{0,1,2,3}^{\mathbf{\sigma}}} \right|^2 f_0 f_1 f_2 \delta_{0,1,2,3}^\sigma
\end{split}
\end{equation}

\medskip
\noindent
\underline{Steps 3 and 4: large volume limit and large time limit.} 
Setting
$$
K_{0,1,2,3} = |C^{-,+,+}_{1,2,3}|^2,
$$
we find as $N,t\to \infty$ that \eqref{colvert} becomes equivalent to
$$
36 \pi^3 t \varepsilon^4 \int K_{0,1,2,3} \left[f_1 f_2 f_3 + f_0 f_2 f_3 - f_0 f_1 f_2 - f_0 f_2 f_3 \right] \delta(\Omega_{0,1,2,3}) \delta(\Omega_{0,1,2,3}) \dd \xi_{1,2,3},
$$
in which we recognize the collision operator.

\medskip

\noindent \underline{Symmetry of $K$} By the above definition,
$$
K_{0,1,2,3} = \frac 19 \left| D^{-,+,+}_{1,2,3} + D^{+,-,+}_{2,1,3} + D^{+,-,+}_{3,1,2} \right|^2.
$$
As a result, the symmetry $K_{0,1,2,3} = K_{0,1,3,2}$ is obvious. Proving that, on the resonant manifold $\Omega=\Sigma=0$,
$$
K_{0,1,2,3} = K_{2,3,0,1}
$$
is a lengthy and tedious computation. Rather than giving it in full, we will show that
$$
D_{1,2,3}^{-,+,+} = D^{-,+,+}_{3,0,1};
$$
the terms $D^{+,-,+}_{2,1,3}$ and $D^{+,-,+}_{3,1,2}$ can be treated similarly. First, on the resonant manifold there holds
$$
\xi+\xi_1 = \xi_2 + \xi_3 \quad \mbox{and} \quad \omega(\xi) + \omega(\xi_1) = \omega(\xi_2) + \omega(\xi_3).
$$
Second, by the symmetries \eqref{symmetry_QC} of $Q$, we have
$$
Q(-\xi_1,\xi+\xi_1)Q(\xi_2,\xi_3) = Q(-\xi_1,-\xi) Q(\xi_3,-\xi_2-\xi_3) = Q(\xi_1,\xi) Q(-\xi_3,\xi_2+\xi_3).
$$
Using these identities, we have
\begin{align*}
D_{1,2,3}^{-,+,+} & =\frac{Q(-\xi_1,\xi_2+\xi_3)Q(\xi_2,\xi_3)}{\omega(\xi) + \omega(\xi_1) - \tau \omega(\xi+\xi_1)} = \frac{Q(-\xi_3,\xi_1+\xi_2)Q(\xi_1,\xi)}{\omega(\xi_2) + \omega(\xi_3) - \tau \omega(\xi_2+\xi_3)} = D_{3,1,2}^{-,+,+}.
\end{align*}

\subsection{Rigorous results on the derivation of the kinetic equation} The derivation of \eqref{KWE} from \eqref{NOC} is one instance of the more general question of deriving mesoscopic kinetic equations from microscopic Hamiltonian systems. 

Over the last decade, there was remarkable progress in the case where the Hamiltonian system is the Nonlinear Schr\"odinger equation (NLS), and the kinetic model the associated kinetic wave equation. Just like we saw above in the framework of \eqref{NOC}, there is a natural kinetic time scale $T_{kin}$ over which the kinetic equation operates. In first attempts \cite{BuckmasterGermainHaniShatah,DengHani0,CollotGermain}, it was possible to prove the convergence from (NLS) to the kinetic wave equation on times $t \in [0,T_{kin}^{1-\delta}]$, for any $\delta >0$. The breakthrough came from Deng and Hani \cite{DengHani, DengHani2}, who reached successively a small multiple $cT_{kin}$ ($c$ small constant), and then arbitrary multiples $C T_{kin}$ ($C$ fixed constant) of the kinetic time scale.

Does the success we described for (NLS) have a counterpart for the system which occupies us here, namely \eqref{NOC}? It turns out that one-dimensional problems are particularly challenging. The first results in dimension one \cite{Vassilev} were obtained for the MMT model. Very recently, \cite{BoyangWu} improved by  \cite{VassilevWu} were able to prove convergence in the case of the (FPUT) over a time scale $T_{kin}^{\frac 23}$. We have not presented the (FPUT) model yet; this will be done in Section \ref{section_examples}. For now, suffice it to say it is the most standard example of a nonlinear oscillator chain; as to the reduced form, it restricts the system to the interactions which are most important dynamically.

All the results we have mentioned so far deal with the homogeneous problem. In the inhomogeneous case, the only rigorous result is \cite{AmpatzoglouCollotGermain}, which deals with a quadratic nonlinear wave equation and is valid up to $T_{kin}^{1-\delta}$, with $\delta>0$ arbitrary.

\section{Examples of nonlinear oscillatory chains and their kinetic wave equations} \label{section_examples}
Up to this point, we worked with very general nonlinear oscillator systems, only imposing translation invariance and nonlinear stability. To further our understanding of these systems, it is instructive to consider particular cases of physical, mathematical, and historical significance. To these we now turn.

\subsection{Unpinned systems} In such systems, the potential $F$ is invariant if all particles are uniformly shifted: $F((q_j)_j) = F((q_j + r)_j)$, with $r \in \mathbb{R}$, hence the name; as a consequence, $\omega(0) = 0$. We can further restrict this class of systems by asking that 
$$
F(\mathbf{q}) = \sum_j V(q_{j+1} - q_j)
$$
(nearest neighbor interactions). If $V''(0) \neq 0$, the linearized system is not trivial, and, up to a multiplicative constant, it is given by the discrete Laplacian
$$
\ddot{q_j} = \Delta_d q_j, \qquad \Delta_d q_j = -2 q_j + q_{j-1} + q_{j+1},
$$
which corresponds to the dispersion relation
$$
\boxed{\omega(\xi) =  2 \left| \sin \left( \frac \xi 2 \right) \right|.}
$$
The full equation is given by
$$
\ddot{q_j} = V'(q_{j+1} - q_j) - V'(q_j - q_{j-1}).
$$

\medskip
\noindent
\underline{FPUT$\alpha$ model}. The class of (FPUT) systems goes back to the article (or rather, Los Alamos report) by Fermi, Pasta and Ulam \cite{FermiPastaUlam} in 1950, which had a lasting influence on the field of nonlinear dynamics. To their initials was added that of Tsingou, whose contribution has been recently highlighted. We shall come back in Section \ref{sec:FPUT_paradox} to the unexpected behavior observed numerically by these authors and the theory behind it, but for now we will just present the model itself and the resulting kinetic equation.

The FPUT$\alpha$ model is the unpinned system with nearest neighbor interaction and
$$
V(x) = \frac 12 x^2 + \frac \alpha 3 x^3
$$
(with $\alpha \in \mathbb{R}$). The equation becomes in physical space
\begin{equation}
\label{FPUTalpha} \tag{FPUT$\alpha$}
\ddot{q_j} = \Delta_d q_j + \alpha (q_{j+1}-q_{j-1})\Delta_d q_j.
\end{equation}
To translate this formula to Fourier space, a small computation is needed. It is  simpler if the nonlinearity is written $(q_{j+1} - q_j)^2 - (q_{j} -q_{j-1})^2$. By the formulas \eqref{formulafourier}, we compute successively
\begin{align*}
& \mathcal{F} (q_{j \pm 1} - q_j)^2 = \frac{1}{2\pi} \int (e^{\pm i\xi_1} - 1) \widehat{q}(\xi_1 )  (e^{\pm i\xi_2} - 1) \widehat{q}(\xi_2 ) \delta(\xi - \xi_1 - \xi_2) \dd \xi_1 \dd \xi_2 \\
& \mathcal{F} [(q_{j +1} - q_j)^2 - (q_{j -1} - q_j)^2] = \frac{1}{2\pi} \int 
\left[ (e^{i\xi_1} - 1)  (e^{ i\xi_2} - 1) - (e^{-i\xi_1} - 1)(e^{- i\xi_2} - 1) \right]\\
& \qquad \qquad \qquad \qquad \qquad\qquad \qquad  \qquad \qquad \qquad \widehat{q}(\xi_1 ) \widehat{q}(\xi_2 )  \delta(\xi - \xi_1 - \xi_2) \dd \xi_1 \dd \xi_2
\end{align*}
If $\xi_1,\xi_2 \in \mathbb{R}$, we can write with the help of Euler's formula
\begin{align*}
& (e^{i\xi_1} - 1)  (e^{ i\xi_2} - 1) - (e^{-i\xi_1} - 1)(e^{- i\xi_2} - 1)\\
& \qquad = e^{i \frac{\xi_1 + \xi_2 }{2}} (e^{i\frac{\xi_1}{2}} - e^{-i\frac{\xi_1}{2}})(e^{i\frac{\xi_2}{2}} - e^{-i\frac{\xi_2}{2}}) - e^{-i \frac{\xi_1 + \xi_2 }{2}} (e^{-i\frac{\xi_1}{2}} - e^{+i\frac{\xi_1}{2}})(e^{-i\frac{\xi_2}{2}} - e^{+i\frac{\xi_2}{2}}) \\
& \qquad = -8 i \sin \left( \frac{\xi_1 + \xi_2}{2} \right) \sin \left( \frac{\xi_1}{2} \right) \sin \left( \frac{\xi_2}{2} \right).
\end{align*}
At this point, a word of caution is needed: the expressions $\sin \left( \frac{\xi_1}{2} \right)$, $\sin \left( \frac{\xi_2}{2} \right)$ or $\sin \left( \frac{\xi_1+\xi_2}{2} \right)$ are not $2\pi$ periodic, hence they do not make sense if $\xi_1,\xi_2 \in \mathbb{T}$ - \textit{but} their product $\sin \left( \frac{\xi_1 + \xi_2}{2} \right) \sin \left( \frac{\xi_1}{2} \right) \sin \left( \frac{\xi_2}{2} \right)$ does! Indeed, as we saw above, given $\Xi_1,\Xi_2 \in \mathbb{T}$, this product is independent of the choice of $\xi_1,\xi_2 \in \mathbb{R}$ such that $\xi_i = \Xi_i \mod 2\pi$. Therefore, we write $\sin \left( \frac{\xi_1 + \xi_2}{2} \right) \sin \left( \frac{\xi_1}{2} \right) \sin \left( \frac{\xi_2}{2} \right)$ with this understanding.

Overall, the equation \eqref{FPUTalpha} becomes in Fourier space
\begin{align*}
\ddot{\widehat{q}}(\xi) & = - 4\sin \left( \frac \xi 2 \right)^2 \widehat{q}(\xi)  \\
& \qquad\qquad - \frac{4 \alpha i}{\pi} \int \sin \left( \frac {\xi_1+\xi_2} 2 \right) \sin \left( \frac {\xi_1} 2 \right) \sin \left( \frac{\xi_2} 2 \right) \widehat{q}(\xi_1) \widehat{q}(\xi_2) \delta(\xi - \xi_1 - \xi_2) \dd \xi_1 \dd \xi_2.
\end{align*}
This means that $\widehat{\alpha^{(3)}}(\xi_1,\xi_2) = 8 \alpha i \sin \left( \frac {\xi_1+\xi_2} 2 \right) \sin \left( \frac {\xi_1} 2 \right) \sin \left( \frac{\xi_2} 2 \right)$ and thus the associated symbol (see equation \eqref{eqb}) for the unknown $b$ is
$$
\displaystyle
\boxed{Q(\xi_1,\xi_2) = - \frac{i \alpha}{2\sqrt 2} \frac{\sin \left( \frac {\xi_1+\xi_2} 2 \right) \sin \left( \frac {\xi_1} 2 \right) \sin \left( \frac{\xi_2} 2 \right)}{\sqrt{\left|\sin \left( \frac {\xi_1+\xi_2} 2 \right) \sin \left( \frac {\xi_1} 2 \right) \sin \left( \frac{\xi_2} 2 \right) \right|}}}
$$
The next step is to perform a normal form transformation. We learned from Section \ref{section_resonances} that the only quadratic resonances are the trivial ones for the dispersion relation $\left| \sin \left( \frac \xi 2 \right) \right|$. These might nevertheless lead to singularities, but we will see that this is not the case.
From formula \eqref{formulaD}, we have (denoting $\xi = \sigma_1 \xi_1 + \sigma_2 \xi_2 + \sigma_3 \xi_3$)

\begin{align*}
D^{\sigma_1,\sigma_2,\sigma_3}& (\xi_1,\xi_2,\xi_3) = \frac{\alpha^2}{4 \pi^2} \sum_{\tau}  \frac{\tau}{\omega(\xi) - \sigma_1 \omega(\xi_1) - \tau \omega(\xi - \sigma_1 \xi_1)} \times \\
& \;\;\;\;\;\;\;\;\;\;\;\;\;\;\;\; \frac{\sin \left( \frac {\xi} 2 \right) \sin \left( \frac {\sigma_1 \xi_1} 2 \right) \sin \left( \frac{\xi-\sigma_1\xi_1} 2 \right)}{\sqrt{\left|\sin \left( \frac {\xi} 2 \right) \sin \left( \frac {\xi_1} 2 \right) \sin \left( \frac{\xi-\sigma_1 \xi_1} 2 \right) \right|}} \frac{\sin \left( \frac {\xi-\sigma_1 \xi_1} 2 \right) \sin \left( \frac { \sigma_2 \xi_2} 2 \right) \sin \left( \frac{\sigma_3 \xi_3} 2 \right)}{\sqrt{\left|\sin \left( \frac {\xi-\sigma_1 \xi_1} 2 \right) \sin \left( \frac {\xi_2} 2 \right) \sin \left( \frac{\xi_3} 2 \right) \right|}} \\
& = \frac {\alpha^2} {4 \pi^2} \frac{\sin \left( \frac {\xi} 2 \right) \sin \left( \frac {\sigma_1 \xi_1} 2 \right)\sin \left( \frac { \sigma_2 \xi_2} 2 \right) \sin \left( \frac{ \sigma_3 \xi_3} 2 \right) }{\sqrt{\left|\sin \left( \frac {\xi} 2 \right) \sin \left( \frac {\xi_1} 2  \right) \sin \left( \frac {\xi_2} 2 \right) \sin \left( \frac{\xi_3} 2 \right) \right|}} \frac{\omega(\xi-\sigma_1 \xi_1)^2}{\left[ \omega(\xi) - \sigma_1 \omega(\xi_1)\right]^2 - \omega(\xi-\sigma_1 \xi_1)^2}.
\end{align*}
By the result in Section \ref{Sec: Derivation_WKE}, the kernel of the kinetic equation is given by
$$
K_{0,1,2,3} = \frac{1}{9} | D^{-,+,+}_{1,2,3} + D^{+,-,+}_{2,1,3} + D^{+,-,+}_{3,1,2}|^2.
$$
We claim that, if $\xi + \xi_1 = \xi_2 + \xi_3$ and $\omega_0 + \omega_1 = \omega_2 + \omega_3$,
$$
\boxed{|K_{0,1,2,3}|^2 = \frac{\alpha^4}{144 \pi^4} \left|\sin \left( \frac {\xi} 2 \right) \sin \left( \frac { \xi_1} 2 \right)\sin \left( \frac { \xi_2} 2 \right) \sin \left( \frac{ \xi_3} 2 \right) \right|.}
$$
The computation is very delicate and will be postponed till Section \ref{section_computation}
\medskip

\noindent
\underline{FPUT$\beta$ model}. This is the unpinned system with nearest neighbor interaction and
$$
V(x) = \frac 12 x^2 + \frac \beta 4 x^4
$$
The equation becomes in physical space
\begin{equation}
\label{FPUTbeta} \tag{FPUT$\beta$}
\ddot{q_j} = \Delta_d q_j + \beta \left[ (q_{j+1} - q_j)^3 - (q_{j} - q_{j-1})^3 \right].
\end{equation}
To obtain the formula in Fourier, we write
\begin{align*}
& \mathcal{F} [(q_{j +1} - q_j)^3 - (q_{j} - q_{j-1})^3] \\
& \qquad \qquad \qquad \qquad = \frac{1}{4\pi^2} \int 
\left[ (e^{i\xi_1} - 1)  (e^{ i\xi_2} - 1) (e^{ i\xi_3} - 1) + (e^{-i\xi_1} - 1)(e^{- i\xi_2} - 1) (e^{-i\xi_3} - 1) \right]\\
& \qquad \qquad \qquad \qquad \qquad\qquad \qquad  \qquad \qquad \qquad \widehat{q}(\xi_1 ) \widehat{q}(\xi_2 )\widehat{q}(\xi_3 )  \delta(\xi - \xi_1 - \xi_2 -\xi_3) \dd \xi_1 \dd \xi_2 \dd \xi_3
\end{align*}
If $\xi_1,\xi_2,\xi_3 \in \mathbb{R}$, we can then write
\begin{align*}
& (e^{i\xi_1} - 1)  (e^{ i\xi_2} - 1) (e^{i\xi_3} - 1)  + (e^{-i\xi_1} - 1)(e^{- i\xi_2} - 1) (e^{-i\xi_3} - 1)\\
& \qquad = e^{i \frac{\xi_1 + \xi_2 +\xi_3}{2}} (e^{i\frac{\xi_1}{2}} - e^{-i\frac{\xi_1}{2}})(e^{i\frac{\xi_2}{2}} - e^{-i\frac{\xi_2}{2}})(e^{i\frac{\xi_3}{2}} - e^{-i\frac{\xi_3}{2}}) \\
& \qquad \qquad \qquad \qquad + e^{-i \frac{\xi_1 + \xi_2 +\xi_3}{2}} (e^{-i\frac{\xi_1}{2}} - e^{+i\frac{\xi_1}{2}})(e^{-i\frac{\xi_2}{2}} - e^{+i\frac{\xi_2}{2}})(e^{-i\frac{\xi_3}{2}} - e^{+i\frac{\xi_3}{2}}) \\
& \qquad = 16 \sin \left( \frac{\xi_1 + \xi_2 +\xi_3}{2} \right) \sin \left( \frac{\xi_1}{2} \right) \sin \left( \frac{\xi_2}{2} \right) \sin \left( \frac{\xi_2}{2} \right).
\end{align*}
With the same word of caution as for \eqref{FPUTalpha}, we obtain the expression in Fourier space
\begin{align*}
& \ddot{\widehat{q}}(\xi) = - 4 \sin \left( \frac \xi 2 \right)^2 \widehat{q}(\xi) + \frac{4 \beta}{\pi^2} \int \sin \left( \frac {\xi_1+\xi_2+\xi_3} 2 \right) \sin \left( \frac {\xi_1} 2 \right) \sin \left( \frac {\xi_2} 2 \right) \ \sin \left( \frac {\xi_3} 2 \right)  \\
& \qquad \qquad \qquad \qquad \qquad \qquad \qquad \qquad \qquad \widehat{q}(\xi_1)  \widehat{q}(\xi_2) \widehat{q}(\xi_3) \delta(\xi - \xi_1 - \xi_2 - \xi_3) \dd \xi_1 \dd \xi_2 \dd \xi_3.
\end{align*}
For the unknown $b$, the symbol becomes
\begin{equation}
\label{symbolC}
\boxed{C(\xi_1,\xi_2,\xi_3) = \frac \beta {8 \pi^2} \frac{\sin \left( \frac {\xi_1+\xi_2+\xi_3} 2 \right) \sin \left( \frac {\xi_1} 2 \right) \sin \left( \frac {\xi_2} 2 \right) \ \sin \left( \frac {\xi_3} 2 \right) }{\sqrt{ \left| \sin \left( \frac {\xi_1+\xi_2+\xi_3} 2 \right) \sin \left( \frac {\xi_1} 2 \right) \sin \left( \frac{\xi_2} 2 \right)  \sin \left( \frac {\xi_3} 2 \right)\right|}}.}
\end{equation}
By the derivation result in Section \ref{Sec: Derivation_WKE}, the associated \eqref{KWE} is characterized by the kernel
$$
\boxed{K(\xi_0,\xi_1,\xi_2,\xi_3) = \frac{\beta^2}{64\pi^4} \left| \sin \left( \frac {\xi} 2 \right) \sin \left( \frac {\xi_1} 2 \right) \sin \left( \frac{\xi_2} 2 \right)  \sin \left( \frac {\xi_3} 2 \right)\right|.}
$$

\medskip

\noindent
\underline{FPUT $\alpha + \beta$ model}. It combines both previous models with the choice
$$
V(x) = \frac 12 x^2 + \frac \alpha 3 x^3 + \frac \beta 4 x^4.
$$
The cubic symbol resulting from the quartic term in $V$ has been computed in \eqref{symbolC}. Turning to the cubic term in $V$, it yields a quadratic nonlinearity, which gives, after normal form and restricting to the $-,+,+$ signature, the symbol
\begin{align*}
&\frac{1}3\left[D^{-,+,+}(\xi_1,\xi_2,\xi_3) + D^{+,-,+}(\xi_2,\xi_1,\xi_3) + D^{+,+,-}(\xi_3,\xi_2,\xi_1) \right] \\
& \qquad\qquad\qquad\qquad\qquad = -\frac{\alpha^2}{12 \pi^2}  \frac{\sin \left( \frac {\xi_1+\xi_2+\xi_3} 2 \right) \sin \left( \frac {\xi_1} 2 \right) \sin \left( \frac {\xi_2} 2 \right) \ \sin \left( \frac {\xi_3} 2 \right) }{\sqrt{ \left| \sin \left( \frac {\xi_1+\xi_2+\xi_3} 2 \right) \sin \left( \frac {\xi_1} 2 \right) \sin \left( \frac{\xi_2} 2 \right)  \sin \left( \frac {\xi_3} 2 \right)\right|}}
\end{align*}
by the computation in Section \ref{section_computation}.

This leads to the kinetic kernel
\begin{equation}
\label{formulaK}
\boxed{K(\xi_0,\xi_1,\xi_2,\xi_3) = \frac{1}{16 \pi^4} \left[ -\frac{\alpha^2}{3} + \frac{\beta}{2} \right]^2 \left| \sin \left( \frac {\xi} 2 \right) \sin \left( \frac {\xi_1} 2 \right) \sin \left( \frac{\xi_2} 2 \right)  \sin \left( \frac {\xi_3} 2 \right)\right|.}
\end{equation}

\medskip

\noindent
\underline{Toda lattice.} It is the equation corresponding to the potential
$$
V(x) = e^{-x} + x - 1,
$$
and it was introduced by M. Toda in 1967 \cite{Toda}; it has the remarkable property of being completely integrable, see for instance the recent textbook presentation in \cite{DeiftDubachTomeiTrogdon}. As a result, it should be expected that all resonant interactions cancel and that the kinetic description in the weakly turbulent regime is trivial. 

To check this prediction, we expand $V$ as
$$
V(x) = \frac{1}{6} x^3 + \frac{1}{24} x^4 + \dots,
$$
which means that $V$ agrees to order $4$ with the $\alpha+\beta$ (FPUT) system, with $\alpha = -\frac 12$ and $\beta = \frac 1 6$. With these values of $\alpha$ and $\beta$, the formula \eqref{formulaK} gives $K=0$, confirming the intuition that integrability precludes turbulent behavior.

\subsection{Pinned systems} \label{subsec:pinnedsystems} In such systems, the potential function is not invariant by a uniform shift of the positions of the particles. The most natural choice would be $F(\textbf{q}) = \sum_n V(q_n)$, but it does not allow any interaction between particles, resulting in trivial dynamics! For this reason, we will let
$$
F(\textbf{q}) = \sum_{k,\ell} \alpha_k q_{\ell} q_{k+\ell} + \sum_n V(q_n).
$$
Further restricting $\alpha$ to nearest neighbor interactions, $\alpha_k = 0$ if $|k| \geq 2$, we obtain up to a multiplicative constant
$$
\boxed{\omega(\xi) = \sqrt{1 - 2 \delta \cos \xi}, \qquad 0 \leq \delta \leq \frac 12.}
$$
For $\delta = \frac{1}{2}$, the system is unpinned at the linear level, $\omega(\xi) = \sqrt 2 \left| \sin \left( \frac \xi 2 \right) \right|$, and we are brought back to the previous paragraph. If $\delta < \frac 12$, then $\omega'(0) = 0$ but $\omega (0),\omega''(0) \neq 0$. This brings to mind the dispersion relation of the Klein-Gordon equation, earning the system its name Discrete Nonlinear Klein-Gordon, DNLKG for short.

\medskip

\noindent 
\underline{Quadratic DNLKG} This corresponds to the choice
$$
F(\mathbf{q}) = \sum_j  \frac 12 \left[\delta (q_{j+1}-q_j)^2 + (1-2\delta) q_j^2 \right] + \frac{1}{4} \sum_j q_j^4.
$$
The evolution problem is
$$
\ddot q_j = (\delta \Delta_d - (1-2\delta) ) q_j - q_j^3
$$
or in the Fourier variable
$$
\ddot{\widehat{q}}(\xi) = - (1 - 2 \delta \cos \xi) \widehat{q}(\xi) - \frac{1}{2\pi} \int \widehat{q}(\xi_1) \widehat{q}(\xi_2) \delta(\xi - \xi_1 - \xi_2)  \dd \xi_1 \dd \xi_2.
$$
Changing to the unknown function $b$, the symbol is
$$
\boxed{Q(\xi,\xi_1,\xi_2) = \frac{-1}{8\pi} \frac{1}{\sqrt{\omega(\xi) \omega(\xi_1) \omega(\xi_2)}}.}
$$
We saw that there are no resonances if $\delta<\frac 12$, in which case we can perform a normal form transformation. The formula for $D$ given in \eqref{formulaD} becomes
\begin{align*}
D^{-,+,+}(\xi_1,\xi_2,\xi_3) & = - \frac 2 {(8 \pi)^2} \frac{1}{\sqrt{\omega(\xi)\omega(\xi_1)\omega(\xi_2) \omega(\xi_3)}} \frac{1}{\omega(\xi+\xi_1)} \sum_\tau \tau \frac{1}{\omega(\xi)+\omega(\xi_1)-\tau \omega(\xi + \xi_1)} \\
& = \frac{1}{16 \pi^2} \frac{1}{\sqrt{\omega(\xi)\omega(\xi_1)\omega(\xi_2) \omega(\xi_3)}} \frac{1}{\omega(\xi+\xi_1)^2-[\omega(\xi) + \omega(\xi_1)]^2} 
\end{align*}
provided $\xi + \xi_1 = \xi_2 + \xi_3$.

We can then follow the derivation in Section \ref{Sec: Derivation_WKE} to obtain
\begin{equation*}
\boxed{
\begin{aligned}
K(\xi,\xi_1,\xi_2,\xi_3) & = \frac{1}{9\cdot 16^2 \pi^4} \frac{1}{\omega(\xi) \omega(\xi_1) \omega(\xi_2) \omega(\xi_3)} \left[ \frac{1}{\omega(\xi+\xi_1)^2-[\omega(\xi) + \omega(\xi_1)]^2} \right. \\
& \left. + \frac{1}{\omega(\xi-\xi_2)^2-[\omega(\xi) - \omega(\xi_2)]^2} + \frac{1}{\omega(\xi-\xi_3)^2-[\omega(\xi) - \omega(\xi_3)]^2} \right]^2
\end{aligned}
}
\end{equation*}

\medskip
\noindent 
\underline{Cubic DNLKG} (a.k.a. discrete $\varphi^4$ model). It corresponds to the choice
$$
F(\mathbf{q}) = \sum_j  \frac 12 \left[2\delta (q_{j+1}-q_j)^2 + (1-2\delta) q_j^2 \right] + \frac{1}{3} \sum_j q_j^3.
$$
The evolution problem is
$$
\ddot q_j = (2\delta \Delta_d - (1-2\delta) ) q_j - q_j^2
$$
or in the Fourier variable
$$
\ddot{\widehat{q}}(\xi) = - (1 - 2 \delta \cos \xi) \widehat{q}(\xi) + \frac{1}{4\pi^2} \int \widehat{q}(\xi_1) \widehat{q}(\xi_2) \widehat{q}(\xi_3) \delta(\xi - \xi_1 - \xi_2 - \xi_3)  \dd \xi_1 \dd \xi_2 \dd \xi_3.
$$
Changing to the unknown function $b$, the symbol is
$$
\boxed{C(\xi,\xi_1,\xi_2) = \frac{-1}{32\pi^2} \frac{1}{\sqrt{\omega(\xi_1+\xi_2+\xi_3) \omega(\xi_1) \omega(\xi_2) \omega(\xi_3)}}.}
$$
The associated \eqref{KWE} is characterized by the kernel
$$
\boxed{K(\xi_0,\xi_1,\xi_2,\xi_3) = \frac{1}{(32 \pi^2)^2} \frac{1}{\omega(\xi) \omega(\xi_2) \omega(\xi_3) \omega(\xi_4)}.}
$$

\medskip
\noindent \underline{Frenkel-Kontorova model.} It is given by the Hamiltonian
$$
F(\mathbf{q}) = \sum_j  \frac 12 \left[2\delta (q_{j+1}-q_j)^2 + (1-2\delta) q_j^2 \right] +  \sum_j \cos(q_j),
$$
so that it can be thought of as a discrete version of the Sine-Gordon equation. It is a fundamental model in solid state physics, see \cite{BraunKivshar0,BraunKivshar1} for an introduction to its theory. As far as its kinetic theory is concerned, this model is indistinguishable from the Cubic DNLKG (upon adjusting parameters) since the nonlinearities of these models agree to third order.

\subsection{The computation of the kinetic kernel for \eqref{FPUTalpha}}
\label{section_computation}

We aim at proving that
\begin{align*}
& \frac{\omega(\xi+ \xi_1)^2}{\left[ \omega(\xi) + \omega(\xi_1)\right]^2 - \omega(\xi + \xi_1)^2} \\
& \qquad \qquad + \frac{\omega(\xi - \xi_2)^2}{\left[ \omega(\xi) - \omega(\xi_2)\right]^2 - \omega(\xi - \xi_2)^2} + \frac{\omega(\xi- \xi_3)^2}{\left[ \omega(\xi) - \omega(\xi_3)\right]^2 - \omega(\xi- \xi_3)^2} = -1
\end{align*}
on the resonant manifold
$$
\Sigma = \xi + \xi_1 - \xi_2 - \xi_3 = 0, \qquad\qquad \Omega = \omega_0 + \omega_1 -\omega_2 - \omega_3 = 0.
$$
In order to prove this identity, we use the trigonometric identities 
\begin{align*} 
&(\sin(x)+ \sin(y))^2 - \sin^2\left(x+y\right) = 4\sin(x)\sin(y) \sin^2\left( \frac{x+y}{2}\right),\\  &\sin^2\left(x+y\right) = 4 \sin^2 \left( \frac{x+y}{2}\right) \cos^2 \left( \frac{x+y}{2}\right),\\
&(\sin(x)- \sin(y))^2 - \sin^2\left(x-y\right) = -4\sin(x)\sin(y) \sin^2\left( \frac{x-y}{2}\right).
\end{align*} 
to rewrite the main term as 
\begin{equation} \label{eq: main term_FPUalpha kernel}
\begin{split}
    &\frac{\cos^2\left( \frac{\xi_0+\xi_1}{4} \right) }{\sin\left( \frac{\xi_0}{2}\right) \sin\left( \frac{\xi_1}{2}\right)} 1_{s_0s_1\geq 0} - 
    \frac{\cos^2\left( \frac{\xi_0-\xi_2}{4} \right) }{\sin\left( \frac{\xi_0}{2}\right) \sin\left( \frac{\xi_2}{2}\right)} 1_{s_0s_2\geq 0} - 
    \frac{\cos^2\left( \frac{\xi_0-\xi_3}{4} \right) }{\sin\left( \frac{\xi_0}{2}\right) \sin\left( \frac{\xi_3}{2}\right)}1_{s_0s_3\geq 0} \\
    & + \frac{\sin^2\left( \frac{\xi_0+\xi_1}{4} \right) }{|\sin\left( \frac{\xi_0}{2}\right) \sin\left( \frac{\xi_1}{2}\right)|} 1_{s_0s_1< 0}  - 
    \frac{\sin^2\left( \frac{\xi_0-\xi_2}{4} \right) }{|\sin\left( \frac{\xi_0}{2}\right) \sin\left( \frac{\xi_2}{2}\right)|} 1_{s_0s_2< 0} - 
    \frac{\sin^2\left( \frac{\xi_0-\xi_3}{4} \right) }{|\sin\left( \frac{\xi_0}{2}\right) \sin\left( \frac{\xi_3}{2}\right)|}1_{s_0s_3 < 0}
    \end{split}
\end{equation}

We use the identities 
\begin{align*} 
&\sin\left( \frac{x}{2}  \right) = \frac{2 \tan\left(\frac{x}{4} \right)}{1+\tan^2 \left(\frac{x}{4} \right) } \  \text{ and }\ \cos^2 \left( \frac{x\pm y}{4}  \right) = \frac{\left[ 1\mp \tan \left( \frac{x}{4} \right) \tan \left( \frac{y}{4} \right)\right]^2 }{(1+ \tan^2 \left(\frac{x}{4} \right) )(1+ \tan^2 \left(\frac{y}{4} \right))}, \\ 
& \hspace{2cm} \sin^2 \left( \frac{x\pm y}{4}  \right) = \frac{\left[ \tan \left( \frac{x}{4} \right) \pm \tan \left( \frac{y}{4} \right) \right]^2 }{(1+ \tan^2 \left(\frac{x}{4} \right) )(1+ \tan^2 \left(\frac{y}{4} \right))}
\end{align*}
in particular we have that $s_x = \operatorname{sign}(\sin(x/2) ) =\operatorname{sign}(\tan(x/4) ))$. 
With the help of these identities we may rewrite our main term in terms of $t_j$, with obvious notation $t_j = \tan(\xi_j/4)$. 

Also the resonant constraints become: \begin{equation}\label{eq:reson contrant for t_i}
\begin{split} &\frac{|t_0|}{1+t_0^2} + \frac{|t_1|}{1+t_1^2} - \frac{|t_2|}{1+t_2^2} = \frac{|t_3|}{1+t_3^2}, \quad \text{ when } t_3 \quad \text{ is }\\
& \mbox{either:}\; \tan \left(\frac{\xi_0+\xi_1-\xi_2}{4} \right) = \frac{t_0+t_1-t_2+t_0t_1t_2}{1-t_0t_1+t_0t_2+t_1t_2} =: \frac{a}{b}\\
& \mbox{or:} \; \tan \left(  \frac{\xi_0+\xi_1-\xi_2}{4} \pm \pi/2 \right) = -\frac{1}{\tan(\frac{\xi_0+\xi_1-\xi_2}{4} )} = -\frac{b}{a}.
\end{split}
\end{equation}
To arrive to the above relations we used the addition formula for tangent: $\tan(x+y) = \frac{\tan(x) + \tan(y) }{1-\tan(x) \tan(y)}$ applied first to $x=\xi_0/4, y =(\xi_1-\xi_2 )/4$, and then that $\tan(x-y) = \frac{\tan(x) - \tan(y) }{1+\tan(x) \tan(y)}$. 

We want to show on the resonant manifold that
the main term \eqref{eq: main term_FPUalpha kernel} equals to $-1$. 
Depending on the signs we need to split into the different cases. We will assume that $t_0> 0$ and $t_i \neq 0$. \\
\underline{Case 1: $t_0t_1 \geq 0,  t_0t_2\geq 0, t_0t_3\geq 0$}:  In this case \eqref{eq: main term_FPUalpha kernel}$=-1$ means:  
$$ \frac{(1-t_0t_1)^2}{4t_0t_1} - \frac{(1+t_0t_2)^2}{4t_0t_2}  - \frac{(1+t_0t_3)^2}{4t_0t_3} = -1. $$
 This is the case if and only if
\begin{align} \label{eq: aim_FPUalpha}
t_2t_3-t_1t_3-t_1t_2 = (2-t_0t_1+t_0t_2+t_0t_3)\prod_{\ell=0}^3 t_\ell 
\end{align} 
We insert $t_3 = -b/a$ into the energy constraint \eqref{eq:reson contrant for t_i}: 
$$ \frac{t_0}{1+t_0^2} + \frac{t_1}{1+t_1^2} - \frac{t_2}{1+t_2^2} = -\frac{ab}{a^2+b^2}  = -\frac{ab}{(1+t_0^2) (1+t_1^2)(1+t_2^2)}.$$
After expanding everything, the above is equivalent to
$$ F_{reson}^1(t_0,t_1,t_2) := t_0+t_1-t_2 + t_0t_1^2t_2^2 + t_1t_0^2t_2^2 - t_2t_1^2t_0^2 + 2t_0t_1t_2 =0. $$
Given this, we aim to verify \eqref{eq: aim_FPUalpha}. Compute that \eqref{eq: aim_FPUalpha} factorises as $$\frac{ (t_1-t_2)(t_0t_2+1)(t_0t_1-1) F_{reson}^1(t_0,t_1,t_2)   }{a^2} =0 $$
which is true due to the resonant constraint! 
When $t_3 = a/b$, then we compute the energy constraint and find 
$$ \frac{t_0}{1+t_0^2} + \frac{t_1}{1+t_1^2} - \frac{t_2}{1+t_2^2} = \frac{ab}{a^2+b^2}  = \frac{ab}{(1+t_0^2) (1+t_1^2)(1+t_2^2)}$$
which rewrites as $$\frac{ (t_2-t_1)(t_0+t_1)(t_2-t_0)}{(1+t_0^2) (1+t_1^2)(1+t_2^2)} =0 $$ and so either $t_2=t_1$ or $t_2=t_0$ or $t_0=-t_1$. These all correspond to the trivial resonances.
\\
\underline{Case 2: $t_0t_1\leq 0, t_0t_2\geq 0, t_0t_3\geq 0$}:  In this case \eqref{eq: main term_FPUalpha kernel}$=-1$ means:  
\begin{align} \label{eq:MT_case2}
- \frac{(t_0+t_1)^2}{4t_0t_1} - \frac{(1+t_0t_2)^2}{4t_0t_2}  - \frac{(1+t_0t_3)^2}{4t_0t_3} = -1, 
\end{align} 
and the resonant constraint \eqref{eq:reson contrant for t_i} becomes
$$ \frac{t_0}{1+t_0^2} - \frac{t_1}{1+t_1^2} - \frac{t_2}{1+t_2^2} = \frac{t_3}{1+t_3^2} $$ which for $t_3=a/b$, it is equivalent to 
$$ F_{reson}^2(t_0,t_1,t_2) := t_0^2t_1t_2^2 + t_0^2t_2-t_0t_1^2+2t_0t_1t_2 - t_0t_2^2+t_1^2t_2+t_1 =0.$$
We then factorise as in Case 1 the main term and we find: 
$$ (\text{LHS of } \eqref{eq:MT_case2}) + 1 = \frac{(t_0+t_1)(1+t_0t_2)(1+t_1t_2) F_{reson}^2 (t_0,t_1,t_2) }{4t_0t_1t_2 ba}, $$
which is indeed $0$. 
For $t_3=-b/a$, the numerator of the energy constraint factorises as  
$-2 (1+t_0^2)(1+t_1^2)(1+t_2^2)(t_0-t_2) (t_0t_1-1)(t_1t_2+1)=0 $ and so the resonant condition reduces to either $t_0=t_2$, or $t_0t_1=1$ or $t_2t_1=-1$, which all again correspond to the trivial resonances (for example  $t_0t_1=1$ means that $\xi_0 + \xi_1 =0\ (\operatorname{mod} 2\pi)$, and $t_1t_2=-1$ means $\xi_1 - \xi_2 =0  \ (\operatorname{mod} 2\pi)$).
\\
\underline{Case 3: $t_0t_1\geq 0,\ t_0t_2 \leq 0, t_0t_3\leq 0$}:
In this case \eqref{eq: main term_FPUalpha kernel}$=-1$ means:  
\begin{align} \label{eq:MT_case3}
\frac{(1-t_0t_1)^2}{4t_0t_1} + \frac{(t_0-t_2)^2}{4t_0t_2}  + \frac{(t_0-t_3)^2}{4t_0t_3} = -1, 
\end{align} 
and the resonant constraint \eqref{eq:reson contrant for t_i}:
$$ \frac{t_0}{1+t_0^2} + \frac{t_1}{1+t_1^2} + \frac{t_2}{1+t_2^2} = -\frac{t_3}{1+t_3^2} $$ which for $t_3=a/b$, factorising the numerator and taking it equal to $0$ we get  
$$2(t_0+t_1)(t_0t_2+1)(t_1t_2+1)=0$$
where all the possible cases here correspond to trivial resonances.  
While for $t_3=-b/a$ we get:
$$ F_{reson}^3(t_0,t_1,t_2) := 
t_0^2t_1^2t_2 + t_0^2t_1 + t_0t_1^2 - 2t_0t_1t_2 + t_0t_2^2 + t_1t_2^2 + t_2 = 0.$$
Then looking at our main term: 
$$ (\text{LHS of } \eqref{eq:MT_case3}) + 1 = \frac{-(t_0-t_2)(1+t_1t_2)(t_1t_0-1) F_{reson}^3 (t_0,t_1,t_2) }{4t_0t_1t_2ba}, $$
which is indeed $0$. 
\\
\underline{Case 4: $t_0t_1\geq 0,\ t_0t_2 \geq 0, t_0t_3\leq 0$ or $t_0t_1\geq 0,\ t_0t_2 \leq 0, t_0t_3\geq 0$ }: We are only looking at the first subcase (the other is symmetric by exchanging $t_2$, $t_3$).  
Here \eqref{eq: main term_FPUalpha kernel}$=-1$ means:  
\begin{align} \label{eq:MT_case4}
\frac{(1-t_0t_1)^2}{4t_0t_1} - \frac{(1+t_0t_2)^2}{4t_0t_2} + \frac{(t_0-t_3)^2}{4t_0t_3} = -1, 
\end{align} 
and the resonant constraint 
\eqref{eq:reson contrant for t_i}:
$$ \frac{t_0}{1+t_0^2} + \frac{t_1}{1+t_1^2} - \frac{t_2}{1+t_2^2} = -\frac{t_3}{1+t_3^2} $$ which for $t_3=a/b$, it is equivalent to 
$$ F_{reson}^4(t_0,t_1,t_2) := t_0^2t_1^2t_2 - t_0^2t_1t_2^2 - t_0t_1^2t_2^2 - 2t_0t_1t_2 - t_0 - t_1 + t_2 = 0.$$
Then looking at our main term: 
$$ (\text{LHS of } \eqref{eq:MT_case4}) + 1 = \frac{-(t_0-t_2)(1+t_1t_2)(t_1t_0-1) F_{reson}^4 (t_0,t_1,t_2) }{4t_0t_1t_2ba}, $$
which is indeed $0$. 
\\
Finally for $t_3=-b/a$ reduces to: $(t_0 + t_1)(t_0 - t_2)(t_1 - t_2) =0$, where all possibilities correspond to trivial resonances. 
\\
\underline{Case 5: $t_0t_1\leq 0,\ t_0t_2\leq 0, t_0t_3\leq 0$}:
Here \eqref{eq: main term_FPUalpha kernel}$=-1$ means:  
\begin{align} \label{eq:MT_case5}
-\frac{(t_0+t_1)^2}{4t_0t_1} + \frac{(t_0-t_2)^2}{4t_0t_2} + \frac{(t_0-t_3)^2}{4t_0t_3} = -1, 
\end{align} 
and the resonant constraint 
\eqref{eq:reson contrant for t_i}:
$$ \frac{t_0}{1+t_0^2} - \frac{t_1}{1+t_1^2} + \frac{t_2}{1+t_2^2} = -\frac{t_3}{1+t_3^2}.$$
For $t_3=a/b$, the energy constraint factorises exactly as in Case 6 with $t_3=-b/a$. For $t_3=-b/a$ it factorises as in Case 6 with $t_3=a/b$.  Indeed term $\frac{t_3}{1+t_3^2}$ changes only sign between $t_3=-b/a$ and $t_3=a/ab$ and Case 6 differs from Case 5 only in $t_3$-sign. The case that does not fall into trivial resonances is now $t_3=a/b$.  We look at our main term: 
$$ (\text{LHS of } \eqref{eq:MT_case5}) + 1 = \frac{-(t_0 + t_1)(t_0 - t_2)(t_1-t_2)F_{reson}^6(t_0,t_1,t_2)}{4t_0t_1t_2 a b } $$
which is $0$ and $F_{reson}^6$ is defined in Case 6 below. 
\\
\underline{Case 6: $t_0t_1\leq 0,t_0t_2\leq 0, t_0t_3\geq 0$ or $t_0t_1\leq 0,t_0t_2\geq 0, t_0t_3\leq 0$ }: We are only looking at the first subcase (the other is symmetric).
Here \eqref{eq: main term_FPUalpha kernel}$=-1$ means:  
\begin{align} \label{eq:MT_case6}
-\frac{(t_0+t_1)^2}{4t_0t_1} + \frac{(t_0-t_2)^2}{4t_0t_2} - \frac{(1+t_0t_3)^2}{4t_0t_3} = -1, 
\end{align} 
and the resonant constraint 
\eqref{eq:reson contrant for t_i}:
$$ \frac{t_0}{1+t_0^2} - \frac{t_1}{1+t_1^2} + \frac{t_2}{1+t_2^2} = \frac{t_3}{1+t_3^2} $$ which for $t_3=a/b$, it is equivalent to 
$ (t_1-t_2)(t_0t_2+1)(t_0t_1-1)=0$ where again all possibilities fall into the trivial resonances. 
Finally for $t_3=-b/a$ reduces to:
$$ F_{reson}^6(t_0,t_1,t_2) := t_0^2t_1 - t_0^2t_2 - t_0t_1^2t_2^2 - 2t_0t_1t_2 - t_0 - t_1^2t_2 + t_1t_2^2 =0. $$
Then looking at our main term: 
$$ (\text{LHS of } \eqref{eq:MT_case6}) + 1 = \frac{-(t_0+t_1)(t_0-t_2)(t_1-t_2)F_{reson}^6(t_0,t_1,t_2)}{4t_0t_1t_2ba} $$
which vanishes either due to the trivial resonances or to due to the energy constraint. 
\\

\section{A first look at the kinetic wave equation} \label{section_first_look}
In this section, we shall present the kinetic wave equation and its main properties. The aim of this section is to present the structure of the equation while remaining at a formal or semi-rigorous level; hard analysis and theorems will be postponed till the next section.

Recall that the kinetic wave equation is given by
\begin{equation}
\tag{KWE} \label{KWE}
\partial_t f + \omega'(\xi) \partial_x f = \mathcal{C}(f), \qquad f = f(t,x,\xi), \qquad (x,\xi) \in \mathbb{T} \times \mathbb{T}.
\end{equation}
with the collision operator
$$
\mathcal{C}(f)(\xi) = \int K_{0,1,2,3} \left[f_1 f_2 f_3 + f_0 f_2 f_3 - f_0 f_1 f_2 - f_0 f_1 f_3 \right] \delta(\Omega_{0,1,2,3}) \delta (\Sigma_{0,1,2,3}) \dd \xi_{1,2,3},
$$
where we denote

\begin{align*}
& \Sigma_{0,1,2,3} = \xi_0 + \xi - \xi_2 - \xi_3, \qquad \Omega_{0,1,2,3} = \omega_0 + \omega_1 - \omega_2 - \omega_3,\\
& f_i = f(t,x,\xi_i), \qquad \dd \xi_{123} = \dd \xi_1 \dd \xi_2 \dd \xi_3.
\end{align*}
The homogeneous problem corresponds to the particular case where $f$ does not depend on the $x$ variable, but only on the frequency $\xi$.

\subsection{Making sense of the collision operator} 
\label{section_makingsense} This is obviously the first thing to check: does the integral defining the collision operator make sense? It is not obvious since this definition involves the product of two Dirac $\delta$ functions, which can fail to make sense if the supports of these $\delta$ functions are tangent.

\subsubsection{A more explicit formula}
Let us consider for simplicity the case where the resonant manifold $\mathscr{R}_{\operatorname{eff}}$ can be parameterized with the help of a function $h$ as
$$
\mathscr{R}_{\operatorname{eff}} = \{ (\xi_0, \, h(\xi_0,\xi_2), \,\xi_2,\,\xi_0+ h(\xi_0,\xi_2)), \, (\xi_0,\xi_2) \in \mathbb{T}^2 \}
$$
In that case, we claim that the collision operator can be written as 
\begin{equation}\label{collisionformula_h}
\mathcal{C}(f) (\xi) = \int K_{0,1,2,3} \left[f_1 f_2 f_3 + f_0 f_2 f_3 - f_0 f_1 f_2 - f_0 f_1 f_3 \right] \frac{1}{|\omega'_1 - \omega_3'|} \dd \xi_2,
\end{equation}
where it is understood in the above expression that
\begin{equation}
\xi_1 = h(\xi_0,\xi_2), \qquad \xi_3 = \xi_0 + h(\xi_0,\xi_2) - \xi_2
\end{equation}
and we used the notation
$$
\omega'_i = \omega'(\xi_i).
$$
Recall that the set $\mathscr{R}^{+,--}$ where $\Omega_{0,1,2,3}$ and $\Omega_{0,1,2,3} =0$ vanish can be written as the union of trivial $\mathscr{R}_{\operatorname{eff}}$ and effective $\mathscr{R}_{\operatorname{eff}}$ resonances. Therefore, two things are needed to justify \eqref{collisionformula_h}: first, that the contribution of trivial resonances is zero, and second, that the contribution of effective resonances gives this formula.

\begin{itemize}
\item \underline{Vanishing of trivial resonances.} Let us consider for instance the trivial resonances given by $\xi_2 = \xi_0$, $\xi_1 = \xi_3$. First of all, we eliminate $\delta(\Omega_{0,1,2,3})$ by setting $\xi_3 = \xi_0+\xi_1 - \xi_2$. Then we set $\xi_2 = \xi_0 +y$ so that
$$
\Omega_{0,1,2,3} = \omega(\xi_0) + \omega(\xi_1) - \omega(\xi_0 + y) - \omega(\xi_1 -y) = G_{\xi_0,\xi_1}(y).
$$
The contribution of the trivial resonances to the collision operator is
$$
\mathcal{C}_{\operatorname{triv}}[f](\xi) = \int K_{0,1,2,3}[f_1 f_2 f_3 + f_0 f_2 f_3 - f_0 f_1 f_2 - f_0 f_1 f_3] \delta(G_{\xi_0,\xi_1}(y)) \dd y \dd \xi_2.
$$
Since $G_{\xi_0,\xi_1}(0) =0$ and $G_{\xi_0,\xi_1}'(0) = \omega'(\xi_1) - \omega'(\xi_0)$, this becomes
$$
\int K_{0,1,2,3}[f_1 f_2 f_3 + f_0 f_2 f_3 - f_0 f_1 f_2 - f_0 f_1 f_3] \frac{1}{|G_{\xi_0,\xi_1}'(0)|} \delta(y) \dd y \dd \xi_2.
$$
Finally, since $f_1 f_2 f_3 + f_0 f_2 f_3 - f_0 f_1 f_2 - f_0 f_1 f_3$ vanishes for $y=0$, the above is zero! (we glossed over the possible zeros of $G_{\xi_0,\xi_1}'(0)$).

\item \underline{Parameterization of effective resonances.} Here too, we eliminate $\xi_3$ and view $\Omega_{0,1,2,3}$ as a function of $\xi_1$:
$$
\Omega_{0,1,2,3} = \omega(\xi) + \omega(\xi_1) - \omega(\xi_2) - \omega(\xi + \xi_1 - \xi_2 - \xi_3) = F_{\xi_0,\xi_2}(\xi_1).
$$
With this definition, we have that 
$$
F_{\xi_0,\xi_2}(h(\xi_0,\xi_2)) = 0 \quad \mbox{and} \quad \frac{\dd}{\dd \xi_1} F_{\xi_0,\xi_2}(h(\xi_0,\xi_2)) = \omega'(\xi_1) - \omega'(\xi_3).
$$
Therefore,
$$
\delta(\Omega_{0,1,2,3}) \dd \xi_1 = \frac{1}{|\omega'(\xi_1) - \omega'(\xi_3)|} \delta(\xi_1 - h(\xi_0,\xi_2)) \dd \xi_1.
$$
\end{itemize}

\subsubsection{Convergence of the integral} The most obvious way to define the integral \eqref{collisionformula_h} occurs when it is absolutely converging, or in other words when the integrand is integrable as a function of $\xi_2$ (this corresponds in the terminology of the Boltzmann equation to the case "with cut-off" \cite{Villani}). 
We will distinguish two cases.

\medskip 

\noindent \underline{The $FPUT$ case.} Then it is possible to use the explicit formula \eqref{formula_h} for $h$ to show that $\omega_1'-\omega_3'$ only vanishes when one of the frequencies is zero. This singularity can then be absorbed in the kernel $K$ to show that the integral is well-defined. This was first observed in \cite{LukkarinenSpohn2008}, and we shall come back to that point in Section \ref{section_lwp}.

\medskip
\noindent \underline{The general case.} We will argue that $\omega_1'-\omega_3'$ does vanish in the general case, and as a consequence the integrand in \eqref{collisionformula_h} is not integrable.

Indeed, viewing the problem in the coordinates $(\xi_0,\xi_1,\xi_2)$, the resonant manifold is a smooth $2$-dimensional manifold, and the function $\omega'(\xi_1) - \omega'(\xi_3)$ also vanishes on a smooth $2$-dimensional manifold. Such structures do generically intersect! The validity of this heuristic argument can be checked numerically in the case of nearest-neighbor interaction with $\delta \in (0,\frac 12)$.

This means that generically, the collision operator is of "non cutoff" type, to use the Boltzmann terminology. Note that such a situation was succesfully dealt with for the one-dimensional MMT model \cite{GermainLaZhang}.

\subsection{Monotonic quantities} 
\label{section_monotonic} The classical symmetrization procedure gives that
\begin{align*}
\frac{\dd}{\dd t} \int \varphi(\xi) f(t,x,\xi) \dd x \dd \xi
& = \int \varphi(\xi) \mathcal{C}(f)(\xi) \dd x \dd \xi \\
= \int K_{0,1,2,3}& f_1 f_2 f_3 \left[ \varphi_0 + \varphi_1 - \varphi_2 - \varphi_3 \right] \delta(\Omega_{0,1,2,3}) \delta (\Sigma_{0,1,2,3}) \dd \xi_{0,1,2,3} \dd x.
\end{align*}
This gives immediately the conservation of the \textit{mass} $\mathcal{M}$ and \textit{energy} $\mathcal{E}$
$$
\mathcal{M}[f] = \int f(x,\xi) \dd x \dd \xi \qquad \mbox{and} \qquad \mathcal{E}[f] = \int \omega(\xi) f(x,\xi) \dd x \dd \xi.
$$
While the conservation of the energy $\mathcal{E}$ can be seen as a consequence of the conservation of the Hamiltonian at the microscopic level, the mass $\mathcal{M}$ is not tied to a microscopically conserved quantity. As a consequence, the conservation of mass is not expected, even approximately, for the microscopic system beyond kinetic time scales; we shall come back to this point.

Considering now functions of $f$, we obtain after symmetrizing twice
\begin{equation}
\label{formulaF}
\begin{split}
& \frac{\dd}{\dd t} \int F(f(t,x,\xi)) \dd x \dd \xi =  \int F'(f(t,x,\xi)) \mathcal{C}(f) \dd x \dd \xi \\
& \quad = \int K_{0,1,2,3} f_0 f_1 f_2 f_3 \left[ \frac{1}{f_0} + \frac{1}{f_1} - \frac{1}{f_2} - \frac{1}{f_3} \right] \times \\
& \qquad  \qquad \qquad \qquad \left[ F'(f_0) + F'(f_1) - F'(f_2) - F'(f_3)  \right] \delta(\Omega_{0,1,2,3}) \delta (\Sigma_{0,1,2,3}) \dd \xi_{0,1,2,3} \dd x.
\end{split}
\end{equation}
Choosing $F(f) = \log f$ gives the \textit{entropy}
$$
S[f] = \int \log f(x,\xi) \dd x \dd \xi.
$$
The $H$-theorem states that the entropy is non-decreasing
$$
\frac{\dd}{\dd t} S[f(t)] \geq 0, \qquad S[f] = \int \log f(x,\xi) \dd x \dd \xi;
$$
it follows from the formula for the entropy dissipation given by \eqref{formulaF} specialized to the case $F(f) = \log f$
\begin{equation}
\label{entropydissipation}
\frac{\dd}{\dd t} S[f(t)]
= \int K_{0,1,2,3} f_0 f_1 f_2 f_3 \left| \frac{1}{f_0} + \frac{1}{f_1} - \frac{1}{f_2} - \frac{1}{f_3} \right|^2 \delta(\Omega_{0,1,2,3}) \delta (\Sigma_{0,1,2,3}) \dd \xi_{1,2,3} \dd x \geq 0.
\end{equation}

The form of the entropy would seem surprising to the reader used to the Boltzmann equation if we had not seen the explanation in Section \ref{micro_to_meso_entropy}. 

\subsection{Rayleigh-Jeans (RJ) states} 
\label{subsection_stationary}
For any $\beta,\gamma \in \mathbb{R}$ such that $\beta + \gamma \omega(\xi) \geq 0$ for any $\xi$, the Rayleigh-Jeans (for short RJ) solution
$$
\mathfrak{f}_{\beta,\gamma}(\xi) = \frac{1}{\beta \omega(\xi) + \gamma}
$$
is a stationary solution of \eqref{KWE} since (denoting simply $\mathfrak{f} = \mathfrak{f}_{\beta,\gamma}$)
$$
\frac{1}{\mathfrak{f}_0} + \frac{1}{\mathfrak{f}_1} - \frac{1}{\mathfrak{f}_2} - \frac{1}{\mathfrak{f}_3} = 0 \qquad \mbox{if \;\;$\Omega_{0,1,2,3}= \Sigma_{0,1,2,3} = 0$}.
$$
The condition $\mathfrak{f}_{\beta,\gamma}(\xi) \geq 0$ for all $\xi$ gives the allowed range for $\beta$ and $\gamma$, namely
\begin{equation}
\label{rangebetagamma}
\mbox{either} \;\; \begin{cases} \beta \geq 0 \\ \gamma \geq - \beta m \end{cases} \quad \mbox{or}  \;\; \begin{cases} \beta \leq 0 \\ \gamma \geq - \beta M,  \end{cases} \qquad m = \min_{\mathbb{T}} \omega, \quad M= \max_{\mathbb{T}} \omega.
\end{equation}

We saw in Section \ref{Gibbs_to_RJ} that these stationary solutions can be understood as the limit of the Gibbs invariant measure if $\gamma =0$. This restriction comes from the fact which was already noted above that the conservation of mass $\mathcal{M}$ cannot be related to a microscopic conservation law, as opposed to the energy $\mathcal{E}$.

Given the conservation of mass and energy and the decay of entropy, it is natural to ask whether these solutions are the minimizers of the entropy for given mass and energy. We will prove in Section \ref{section_rigorous_stationary} that
$$
\min_{\substack{\mathcal{M}[f] = \mathcal{M}_0 \\ \mathcal{E}[f] = \mathcal{E}_0}} S[f] = S[\mathfrak{f}_{\beta,\gamma}] \quad \mbox{for $\beta,\gamma$ such that $\mathcal{M}[\mathfrak{f}_{\beta,\gamma}] = \mathcal{M}_0$ and $\mathcal{E}[\mathfrak{f}_{\beta,\gamma}] = \mathcal{E}_0$},
$$
at least for smooth $\omega$. Because of this minimizing property, RJ states are expected to be stable; see Section \ref{section_rigorous_FPU} for a rigorous proof of this fact in the case of the (FPUT) model. 

Finally, denoting
$$
\mathcal{E}(\beta,\gamma) = \mathcal{E}[\mathfrak{f}_{\beta,\gamma}], \quad \mathcal{M}(\beta,\gamma) = \mathcal{M}[\mathfrak{f}_{\beta,\gamma}], \quad S(\beta,\gamma) = S[\mathfrak{f}_{\beta,\gamma}],
$$
we claim that
\begin{equation}
\label{SEMbetagamma}
\left. \frac{\partial S}{\partial \mathcal{E}} \right|_{\mathcal{M} = \operatorname{cst}} = \beta, \quad \left. \frac{\partial S}{\partial \mathcal{M}} \right|_{\mathcal{E} = \operatorname{cst}} = \gamma.
\end{equation}
In classical thermodynamics, the \textit{temperature} $T$ and the \textit{chemical potential} $\mu$ are defined through the relations $\frac{1}{T} =  \frac{\partial S}{\partial \mathcal{E}}$ and $\frac{\mu}{T} = -\frac{\partial S}{\partial \mathcal{M}}$. Comparing with the above, we see that
$$
\beta = \frac{1}{T}, \qquad \gamma = - \frac{\mu}{T}.
$$
Comparing with the allowed range for $\beta$ and $\gamma$ in \eqref{rangebetagamma}, we see that the temperature as well as the chemical potential may take positive and negative values, as observed in \cite{ODPPBR}. It is surprising, but it does not seem that states with negative temperatures behave differently (for instance, as far as their dynamical stability goes) from states with positive temperature.

There remains to prove \eqref{SEMbetagamma}. In the following, derivative with respect to $\beta$ or $\gamma$ is always understood with $\gamma$ or $\beta$ fixed respectively; and the same goes for $\mathcal{E}$ and $\mathcal{M}$. Differentiating the entropy $S$ with respect to $\beta$ and $\gamma$ gives
\begin{equation}
\label{pouledeau1}
\begin{split}
& \frac{\partial S}{\partial \beta} = \frac{\partial}{\partial \beta} \int \log \frac{1}{\beta \omega + \gamma} \dd \xi = - \int \frac{\omega}{\beta \omega + \gamma} \dd \xi = - \mathcal{E}\\
& \frac{\partial S}{\partial \gamma} = \frac{\partial}{\partial \gamma} \int \log \frac{1}{\beta \omega + \gamma} \dd \xi = - \int \frac{1}{\beta \omega + \gamma} \dd \xi = - \mathcal{M}.
\end{split}
\end{equation}
We may also differentiate the identity $\beta \mathcal{E} + \gamma \mathcal{M} = 2\pi$ with respect to $\mathcal{E}$ and $\mathcal{M}$ to obtain
\begin{equation}
\label{pouledeau2}
\frac{\partial \beta}{\partial \mathcal E} \mathcal E + \beta + \frac{\partial \gamma}{\partial \mathcal E} \mathcal M = 0 \quad \mbox{and} \quad \frac{\partial \beta}{\partial \mathcal M} \mathcal E + \frac{\partial \gamma}{\partial \mathcal M} \mathcal M + \gamma = 0.
\end{equation}
Finally, we differentiate $S$ with respect to $\mathcal E$ and $\beta$ and use successively \eqref{pouledeau1} and \eqref{pouledeau2}
\begin{align*}
& \frac{\partial S}{\partial \mathcal E} =  \frac{\partial S}{\partial \beta}  \frac{\partial \beta}{\partial \mathcal E} + \frac{\partial S}{\partial \gamma}  \frac{\partial \gamma}{\partial \mathcal E} 
= - \mathcal{E} \frac{\partial \beta}{\partial \mathcal E} - \mathcal{M} \frac{\partial \gamma}{\partial \mathcal E} = \beta \\
& \frac{\partial S}{\partial \mathcal M} =  \frac{\partial S}{\partial \beta}  \frac{\partial \beta}{\partial \mathcal M} + \frac{\partial S}{\partial \gamma}  \frac{\partial \gamma}{\partial \mathcal M} 
= - \mathcal{E} \frac{\partial \beta}{\partial \mathcal M} - \mathcal{M} \frac{\partial \gamma}{\partial \mathcal M} = \gamma. 
\end{align*}

\subsection{Linearization around Rayleigh-Jeans solutions} The RJ equilibria are important as attractors of the flow of \eqref{KWE} and as the basic states around which the hydrodynamic limit - which will be discussed in Section \ref{sec:Hydro_limit} - can be built. For these reasons, it is crucial to understand the local behavior around RJ states which is given to leading order by the linearization of \eqref{KWE}.

\subsubsection{The formula for the linearized operator}
It is helpful to adopt the following ansatz
$$
f = \mathfrak{f} (1 +g),
$$
which leads to the following equation for $g$
\begin{align*}
& \partial_t (\mathfrak{f} g) = 
\int K_{0,1,2,3} \left[ \mathfrak{f}_1  \mathfrak{f}_2 \mathfrak{f}_3 (1+g_1)(1+g_2)(1+g_3) + \mathfrak{f}_0  \mathfrak{f}_2 \mathfrak{f}_3 (1+g_0)(1+g_2)(1+g_3) \right. \\ 
& \qquad \qquad \qquad \qquad - \left. \mathfrak{f}_0  \mathfrak{f}_1 \mathfrak{f}_2 (1+g_0)(1+g_1)(1+g_2) - \mathfrak{f}_0  \mathfrak{f}_1 \mathfrak{f}_3 (1+g_0)(1+g_1)(1+g_3) \right] \\
& \qquad \qquad \qquad \qquad \qquad\qquad \qquad \qquad \qquad\qquad \qquad \qquad \qquad \delta(\Sigma_{0,1,2,3}) \delta(\Omega_{0,1,2,3}) \dd \xi_{1,2,3}.
\end{align*}
Linearizing around $g=0$ gives
\begin{align*}
\partial_t g & = \frac{1}{\mathfrak{f}} \int K_{0,1,2,3} \left[ \mathfrak{f}_1  \mathfrak{f}_2 \mathfrak{f}_3 (g_1 + g_2 + g_3) +  \mathfrak{f}_0  \mathfrak{f}_2 \mathfrak{f}_3 (g_0 + g_2 + g_3) \right. \\
& \qquad \qquad \qquad \left. - \mathfrak{f}_0  \mathfrak{f}_1 \mathfrak{f}_2 (g_0 + g_1 + g_2) - \mathfrak{f}_0  \mathfrak{f}_1 \mathfrak{f}_3 (g_0 + g_1 + g_3)\right]  \delta(\Sigma_{0,1,2,3}) \delta(\Omega_{0,1,2,3}) \dd \xi_{1,2,3} \\
& = \frac{1}{\mathfrak{f}} \int K_{0,1,2,3} \mathfrak{f}_0 \mathfrak{f}_1  \mathfrak{f}_2 \mathfrak{f}_3
\left[ g_0 \left( \frac{1}{\mathfrak{f}_1} - \frac{1}{\mathfrak{f}_2} - \frac{1}{\mathfrak{f}_3} \right) + g_1 \left( \frac{1}{\mathfrak{f}_0} - \frac{1}{\mathfrak{f}_2} - \frac{1}{\mathfrak{f}_3} \right) \right. \\
& \qquad \qquad  \qquad \left.+g_2 \left( \frac{1}{\mathfrak{f}_0} + \frac{1}{\mathfrak{f}_1} - \frac{1}{\mathfrak{f}_3} \right) + g_3 \left( \frac{1}{\mathfrak{f}_0} + \frac{1}{\mathfrak{f}_1} - \frac{1}{\mathfrak{f}_2} \right) \right]  \delta(\Sigma_{0,1,2,3}) \delta(\Omega_{0,1,2,3}) \dd \xi_{1,2,3} \\
& = - \frac{1}{\mathfrak{f}} \int K_{0,1,2,3} \mathfrak{f}_0 \mathfrak{f}_1  \mathfrak{f}_2 \mathfrak{f}_3 \left[ \frac{g_0}{\mathfrak{f}_0} + \frac{g_1}{\mathfrak{f}_1} - \frac{g_2}{\mathfrak{f}_2} - \frac{g_3}{\mathfrak{f}_3} \right]  \delta(\Sigma_{0,1,2,3}) \delta(\Omega_{0,1,2,3}) \dd \xi_{1,2,3},
\end{align*}
where we used in the last line that $\frac{1}{\mathfrak{f}_0} + \frac{1}{\mathfrak{f}_1} - \frac{1}{\mathfrak{f}_2} - \frac{1}{\mathfrak{f}_3}$ vanishes on the resonant manifold $\Sigma = \Omega = 0$. We found the following formula for the linearized operator:
$$
\boxed{L = - \frac{1}{\mathfrak{f}} \int K_{0,1,2,3} \mathfrak{f}_0 \mathfrak{f}_1  \mathfrak{f}_2 \mathfrak{f}_3 \left[ \frac{g_0}{\mathfrak{f}_0} + \frac{g_1}{\mathfrak{f}_1} - \frac{g_2}{\mathfrak{f}_2} - \frac{g_3}{\mathfrak{f}_3} \right]  \delta(\Sigma_{0,1,2,3}) \delta(\Omega_{0,1,2,3}) \dd \xi_{1,2,3}}.
$$

\subsubsection{First properties of the linearized operator} The two most important properties of the linearized operator are as follows.
\begin{itemize}
\item \underline{Kernel} It follows immediately from the formula above and the definition of the resonant manifold that
$$
\operatorname{Span}(\mathfrak{f},\mathfrak{f}\omega) \subset \operatorname{Ker} L.
$$
It is expected in general that equality holds in this equation, at lest generically in $\omega$, but it is by no means obvious! We shall come back to that problem.

\medskip
\item \underline{Non-positivity} Taking the inner product of $L$ with $g$ and symmetrizing gives the identity
\begin{align*}
\langle L g,g \rangle = - \int K_{0,1,2,3} \mathfrak{f}_0 \mathfrak{f}_1  \mathfrak{f}_2 \mathfrak{f}_3 \left[ \frac{g_0}{\mathfrak{f}_0} + \frac{g_1}{\mathfrak{f}_1} - \frac{g_2}{\mathfrak{f}_2} - \frac{g_3}{\mathfrak{f}_3} \right]^2  \delta(\Sigma_{0,1,2,3}) \delta(\Omega_{0,1,2,3}) \dd \xi_{0,1,2,3} \leq 0.
\end{align*}
\end{itemize}

Besides these two properties, the structure of $L$ depends heavily on the dispersion relation $\omega$ and the kernel $K$. Consider such questions as: is $L$ bounded on $L^2$ or other Lebesgue or Sobolev spaces? Is $L^{-1}$ bounded on these same spaces? These questions are crucial to understand the local and global behavior of \eqref{KWE}, as well as its kinetic limit. The answers are only known in the case of FPU, and will be discussed in the following; in general, this is a difficult problem and the global picture is still eluding us.

\subsection{Narrowly localized functions}
\subsubsection{Smooth functions with small compact support}
For $a \in \mathbb{T}$, we define the set
$$
S_a = \{ (\eta,a,a,a), \, \eta \in \mathbb{T} \} \,\cup\, \{ (a,\eta,a,a), \, \eta \in \mathbb{T} \} \,\cup\, \{ (a,a,\eta,a), \, \eta \in \mathbb{T} \} \,\cup\, \{ (a,a,a,\eta), \, \eta \in \mathbb{T} \}.
$$
and its $r$-neighborhood $S_{a,r}$.

From looking at pictures of the resonant set, or by elementary arguments, we see that
$$
S_a \cap \mathcal{R}_{\operatorname{eff}} = \{ a,a,a,a \} \cap \mathscr{R}_{\operatorname{eff}} = \emptyset
$$
for all but exceptional values of $a \in \mathbb{T}$. By continuity, there holds for $r>0$ sufficiently small
$$
S_{a,r} \cap \mathscr{R}_{\operatorname{eff}} = \emptyset.
$$
Assuming that $f$ is supported on $B(a,r)$, this implies that
$$
f_1 f_2 f_3 + f_0 f_2 f_3 - f_0 f_1 f_2 - f_0 f_1 f_3 = 0 \qquad \mbox{on $\mathscr{R}_{\operatorname{eff}}$}
$$
which implies in turn that
$$
\mathcal{C}(f) = 0.
$$
In other words, $f$ is a stationary solution! We just showed that any function $f$ is a stationary solution as long as it is localized in a sufficiently small ball around $a \in \mathbb{T}$, a few exceptional values of $a$ being excluded.

Thus, the set of stationary solutions of \eqref{KWE} is extremely rich! How can this be compatible with the relaxation towards a statistical equilibrium which is expected from a kinetic equation with an $H$-theorem? The point is that the stationary solutions we found are zero on most of $\mathbb{T}$; this makes the collision operator identically zero, and prevents relaxation. But we do expect relaxation to occur as soon as $\inf_{\mathbb{T}} f >0$.

\subsubsection{Dirac $\delta$ solutions?} Since $\delta(\xi-a)$, for $a \in \mathbb{T}$, can be envisioned as limiting cases of localized solutions, it seems reasonable to define them as stationary solutions. 

As we shall see in Section \ref{section_rigorous_stationary}, Dirac deltas also occur in generalized maximizers of the emtropy for singular dispersion relations, and it seems natural to accept these maximizers as stationary solutions.

However, it is unclear at this stage whether these Dirac delta solutions have any stability and can be given an analytic meaning, beyond declaring them to be solutions by fiat.

\subsection{Boundary conditions}
Besides the torus $\mathbb{T}$, another natural domain where \eqref{KWE} can be set is the line. Going further, it is natural to try and set \eqref{KWE} on the interval $[0,1]$, and the crucial question becomes: how should one choose the boundary conditions?
By analogy with the Boltzmann equation (see for instance \cite{Villani}), the following possibilities are natural.

\medskip

\noindent \underline{Specular reflection.} This means that phonons are simply bouncing at both ends, $x=0,1$
\begin{align*}
& f(t,x=0,\xi) = f(t,x=0,-\xi) \quad \forall t,\xi \\
& f(t,x=1,\xi) = f(t,x=1,-\xi)\quad \forall t,\xi.
\end{align*}
This ensures the conservation of the mass and energy as well as the decay of the entropy.

\medskip
\noindent \underline{Absorption and emission.} By this we mean the following: at both ends of the interval, incoming phonons are absorbed while outgoing phonons are emitted at a fixed rate
\begin{align*}
& f(t,x=0,\xi) = \frac{T_0}{{\omega(\xi)}} \quad \mbox{if $\omega'(\xi)>0$} \\
& f(t,x=1,\xi) = \frac{T_1}{{\omega(\xi)}} \quad \mbox{if $\omega'(\xi) <0$},
\end{align*}
with $T_0,T_1 > 0$. Physically, the system is connected to heat baths at both ends of the interval, with temperatures $T_0$ and $T_1$ respectively. With this boundary condition, the conservation of mass and energy as well as the decay of the entropy are lost, as should be expected for an open system. Let us finally remark that the emission profile we chose corresponds to RJ equilibria with energy equidistribution. We saw that these RJ solutions are naturally singled out, but other choices would certainly be possible.

\medskip
\noindent \underline{Convex combination.} As the name suggests, this means that
\begin{align*}
& f(t,x=0,\xi) = \theta(\xi) f(t,x=0,-\xi) + (1-\theta(\xi)) \frac{T_0}{\sqrt{\omega(\xi)}} \quad \mbox{if $\omega'(\xi)>0$} \\
& f(t,x=1,\xi) = \theta(\xi) f(t,x=0,-\xi) + (1-\theta(\xi)) \frac{T_1}{\sqrt{\omega(\xi)}} \quad \mbox{if $\omega'(\xi) <0$},
\end{align*}
where $T_0,T_1 > 0$ and $0 < \theta(\xi) <1$

\medskip

These choices of boundary conditions are natural in view of the corresponding theory of the Boltzmann equation. The most convincing justification of these boundary conditions would be the derivation from microscopic dynamics, but it is a very delicate question. First because of the very technical mathematical and physical considerations involved; and second since there is no canonical way of setting boundary conditions for the Hamiltonian system (or of coupling it to a heat bath) so that the answer might not be universal, but rather depend on the specifics of the microscopic boundary condition.

There is however one case where a derivation from microscopic dynamics could be carried out: for the linear system on the line coupled to a heat bath at a point \cite{KORS20}. The findings of that paper correspond to the case 'convex combination' above, giving this possibility further credibility.

\section{Rigorous results on stationary solutions}

\label{section_rigorous_stationary}

This section exposes the state of the art of the mathematical theory of stationary solutions of \eqref{KWE}. Following \cite{EGLM25}, we completely characterize the minimizers of the entropy for fixed mass and energy: these stationary solutions are expected to act as global attractors. Next, we turn to the question of collisional invariants, which are intimately related to stationary solutions. These collisional invariants are only known in the case of FPUT; we present the proof due to \cite{LukkarinenSpohn2008}.

\subsection{Entropy maximizers} We rely here on \cite{EGLM25}.
We saw in Section \ref{section_first_look} that the mass $\mathcal{M}$ and energy $\mathcal{E}$ are conserved while the entropy $S$ is increasing. It is therefore natural 
to look for the maximizers of the entropy for fixed mass and entropy, which are expected to act as attractors of the flow.
\underline{We will henceforth assume} the dispersion relation $\omega$ to be continuous and taking its maximal and minimal value on a set of measure zero - both assumptions could be easily relaxed, but they allow for cleaner statements and proofs.

\subsubsection{RJ states with prescribed mass and energy}
Since $\omega \geq 0$ and by the homogeneity of the problem, we will furthermore assume that 
\begin{equation}
\label{a1}
\min_{\mathbb{T}} \omega = a \geq 0, \qquad \max_{\mathbb{T}} \omega = 1.
\end{equation}
It follows from (\ref{a1}) that for any non negative bounded measure $f$ on  $\mathbb{T}$, 
\begin{equation}
\label{admissible}
a\mathcal M(f)\le \mathcal E(f)\leq \mathcal M(f)
\end{equation}
and this is then a necessary condition on any pair $\mathcal M, \mathcal E$ to be the mass and energy of some non negative, bounded measure $\lambda$.

The Euler-Lagrange equation for the constrained maximization of the entropy with fixed mass and energy leads to the RJ equilibria $\mathfrak{f}_{\beta,\gamma}$ which we already encountered
\begin{equation}
\label{RJBE}
\mathfrak{f}_{\beta,\gamma}(\xi) 
= \frac{1}{\beta \omega(\xi) + \gamma}.
\end{equation}
We have $\mathfrak{f}_{\beta,\gamma}(\xi) \geq 0$ if and only if
\begin{equation}
\label{rangemunu}
\mbox{either} \;\; \begin{cases} \gamma \geq 0 \\ \beta \geq - \gamma \end{cases} \qquad \mbox{or}  \;\; \begin{cases} \gamma \leq 0 \\ \beta \geq - \frac \gamma a \end{cases}.
\end{equation}

THe first question we want to answer is to characterize the couples $(\mathcal{M},\mathcal{E})$ such that there exists a RJ state with this mass and energy. A necessary condition is given by \eqref{admissible} above. But it is not sufficient, as we shall now see: let
$$
c_1 = a + 2\pi \left( \int \frac{\dd \xi}{\omega - a} \right)^{-1} > a, \qquad c_2 =  1 - 2\pi \left( \int \frac{\dd \xi}{1- \omega} \right)^{-1} < 1.
$$
Observe that $c_1 = a$ and $c_2 =1$ provided $\omega$ is smooth; but for instance long-tail interactions can result in singular $\omega$ for which $c_1>a$.
To know whether there exists a RJ equilibrium with given mass and energy, it suffices to compare the quotient of the mass and energy to the thresholds $c_1$ and $c_2$. This is the content of the following theorem.
\begin{theorem} \label{Theorem1}
Let $\mathcal{M}_0, \mathcal{E}_0 \in (0,\infty)$ be given. The existence of RJ equilibria $\mathfrak{f}_{\beta,\gamma}$ satisfying 
\begin{equation}
\label{M0E0}
\mathcal{M}(\mathfrak{f}_{\beta,\gamma}) = \mathcal{M}_0, \qquad \mathcal{E}(\mathfrak{f}_{\beta,\gamma}) = \mathcal{E}_0, 
\end{equation}
is characterized as follows:
\begin{itemize}
\item[(i)] If $ \displaystyle
c_1 < \frac{\mathcal{E}_0}{\mathcal{M}_0} < c_2$, then for any $\mathcal{M}_0, \mathcal{E}_0 > 0$, there exists a unique RJ equilibrium satisfying \eqref{M0E0}.
\item[(ii)] If $\displaystyle \frac{\mathcal{E}_0}{\mathcal{M}_0} = c_2 <1$, there exists a unique RJ equilibrium satisfying \eqref{M0E0}. It is furthermore such that $\beta = -\gamma$, $\gamma > 0$.
\item[(iii)] If $\displaystyle \frac{\mathcal{E}_0}{\mathcal{M}_0} = c_1 >a$, there exists a unique RJ equilibrium satisfying \eqref{M0E0}. It is furthermore such that $\gamma = -\beta a$, $\beta > 0$.
\item[(iv)] If $\displaystyle \frac{\mathcal{E}_0}{\mathcal{M}_0} > c_2$ or $\displaystyle \frac{\mathcal{E}_0}{\mathcal{M}_0} < c_1$, there does not exist a RJ equilibrium satisfying \eqref{M0E0}.
\end{itemize}
\end{theorem}

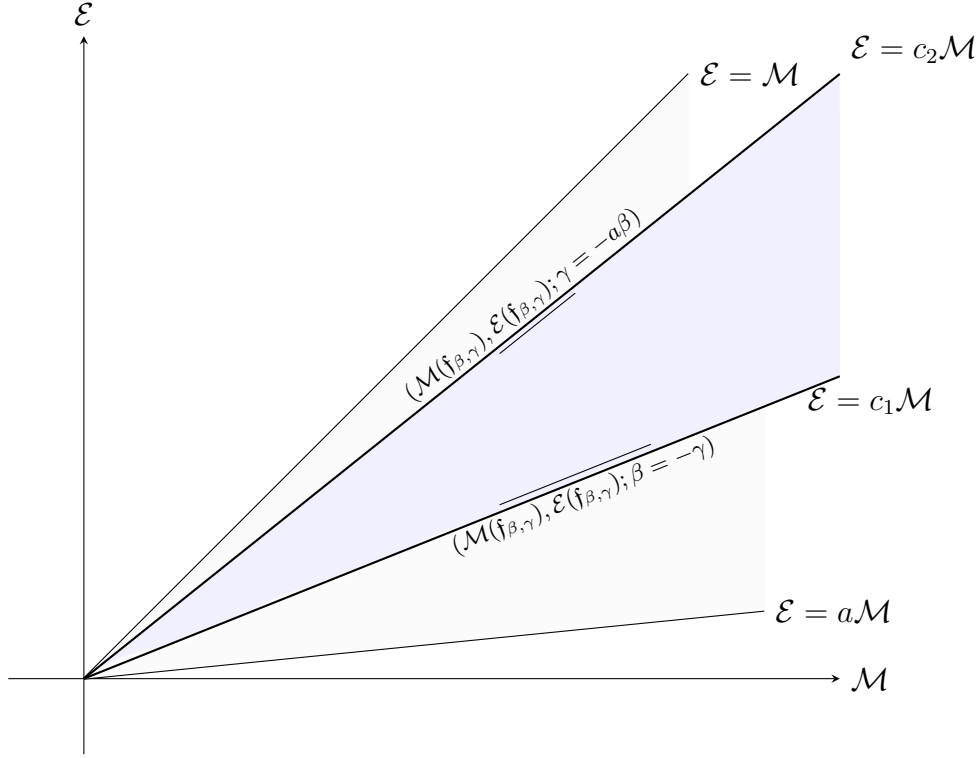
\begin{figure}
\label{figureEM}
\begin{tikzpicture}[>=stealth,scale=1.0]

\def\a{0.1}   
\def\b{0.8}   
\def\alph{0.4} 

\draw[->] (-1,0) -- (10,0) node[right] {$\mathcal{M}$};
\draw[->] (0,-1) -- (0,8.5) node[above] {$\mathcal{E}$};

\draw[thick,black] (0,0) -- (8,8) node[anchor=west] {$\mathcal{E} = \mathcal{M}$};
\draw[thick,black] (0,0) -- (9,{9*\a}) node[anchor=west] {$\mathcal{E} = a\mathcal{M}$};

\begin{scope}
    \clip (0,0) -- plot[smooth,domain=0:10,samples=200] (\x,{(\b*\x)-0.1*sin(deg(0*\x))/sqrt(1+\x)})
          -- plot[smooth,domain=10:0,samples=200] (\x,{(\alph*\x)+0.1*sin(deg(0*\x))/sqrt(1+\x)}) -- cycle;
    \fill[blue!6] (0,0) rectangle (10,8.5);
\end{scope}
\begin{scope}
    \clip (0,0) -- plot[smooth,domain=0:10,samples=200] (\x,{(\b*\x)-0.1*sin(deg(0*\x))/sqrt(1+\x)})
          -- plot[smooth,domain=10:0,samples=200] (\x,{(\x)}) -- cycle;
  \fill[gray!3, dashed] (0,0) rectangle (8,8.5);
\end{scope}
\begin{scope}
    \clip (0,0) -- plot[smooth,domain=0:10,samples=200] (\x,{(\alph*\x)+0.1*sin(deg(3*\x))/sqrt(1+\x)})
         -- plot[smooth,domain=10:0,samples=200] (\x,{(\a*\x)}) -- cycle;
    \fill[gray!3, dashed] (0,0) rectangle (9,8.5);
\end{scope}

\draw[thick,black,smooth,domain=0:10,samples=200] 
    plot (\x,{(\b*\x)-0.1*sin(deg(0*\x))/sqrt(1+\x)}) node[anchor=south west] {{\textbf{$\mathcal{E} = c_2 \mathcal{M}$}}};
\draw[thick,black,smooth,domain=0:10,samples=200] 
    plot (\x,{(\alph*\x)+0.1*sin(deg(0*\x))/sqrt(1+\x)}) node[anchor=north] {\qquad$\mathcal{E} = c_1 \mathcal{M}$};

\draw[] (5.5,{5.5*\alph+0.1}) -- (7.5,{7.5*\alph+0.1}) 
node[midway,below,sloped,rotate=0] 
    {\scriptsize $(\mathcal{M}(\mathfrak{f}_{\beta,\gamma}), \mathcal{E}(\mathfrak{f}_{\beta,\gamma}); \beta = -\gamma)$};
    
\draw[] (5.5,{5.5*\b-0.1}) -- (6.5,{6.5*\b-0.1}) node[midway,above,sloped,rotate=0] 
        {\scriptsize $(\mathcal{M}(\mathfrak{f}_{\beta,\gamma}), \mathcal{E}(\mathfrak{f}_{\beta,\gamma}); \gamma = -a \beta)$};
\end{tikzpicture}
\caption{For $(\mathcal{M}_0, \mathcal{E}_0)$ in the blue region, the entropy maximizer is a regular RJ equilibrium $\mathfrak{f}_{\beta,\gamma}$. In the remaining light gray regions, an additional singular measure is needed to maximize the entropy.}
 \label{fig:Entropy_max_picture}
\end{figure}

\begin{proof} 
We will denote $\mathcal{M}(\beta,\gamma) = \mathcal{M}(\mathfrak{f}_{\beta,\gamma})$. It follows from the identity
\begin{equation}
\label{identitymunu}
\beta \mathcal{E}(\beta,\gamma) + \gamma \mathcal{M}(\beta,\gamma) = 2\pi
\end{equation}
that
$$
\mathcal{E}(\beta,\gamma) = \frac{1}{\beta}(2\pi - \gamma \mathcal{M}(\beta,\gamma)).
$$
Parametrizing $\beta$ and $\gamma$ as
$$
\beta = \rho \cos \varphi, \qquad \gamma = \rho \sin \varphi,
$$
there holds
\begin{align*} 
&\mathcal{M}(\beta,\gamma)) = \frac{1}{\rho} \mathcal{M}(\cos \varphi,\sin \varphi) = \frac{1}{\rho} \mathcal{M}(\varphi),
\\ 
&\mathcal{E}(\beta,\gamma)) = \frac{1}{\rho} \mathcal{E}(\cos \varphi,\sin \varphi) = \frac{1}{\rho \cos \varphi}(2\pi - \sin \varphi \mathcal{M}(\varphi)).
\end{align*} 
One can then eliminate $\rho$ and obtain that the equation \eqref{M0E0} is equivalent to
$$
\frac{\mathcal{E}_0}{\mathcal{M}_0} = \frac{2\pi}{\cos \varphi \mathcal{M} (\varphi)} - \tan \varphi,
$$
or in other words
$$
\frac{\mathcal{E}_0}{\mathcal{M}_0} = F(\tan \varphi) \quad \mbox{with} \quad F(x) = 2\pi \left( \int \frac{\dd \xi}{\omega + x} \right)^{-1} - x.
$$
The allowed range for $\beta$ and $\gamma$ in \eqref{rangemunu} translates into the restriction that 
$$
\varphi \in (\varphi^*,\frac{3\pi}{4}) \qquad \mbox{with} \;\;\varphi^* = - \arctan (a).
$$
which means that 
$$
\tan \varphi \in (-\infty,-1) \cup (-a,\infty).
$$
There remains to compute the image of that set by $F$! 
First, we notice that $F$ is strictly increasing since
$$
F'(x) = \frac{\frac{1}{2\pi} \int \frac{\dd \xi}{(\omega+x)^2}}{\left( \frac{1}{2\pi} \int \frac{\dd \xi}{\omega+x} \right)^2} - 1 > 0
$$
by Jensen's inequality.

Thus, there remains to compute the value of $F$ at the points $-\infty, -1, -a, \infty$.
\begin{itemize}
\item $\displaystyle F(-a+)= a + 2\pi \left( \int \frac{1}{\omega -a} \dd \xi \right)^{-1}$. If $\int \frac{1}{\omega -a} < \infty$, the limiting RJ equilibrium (corresponding to $\varphi \to \varphi^*$) has finite mass. It is given by $\varphi = \varphi^*$, or in other words $\gamma = -a \beta$, $\beta > 0$. 
\item $\displaystyle F(-1-) = 1 - 2\pi \left( \int \frac{1}{1 - \omega} \dd \xi \right)^{-1}$. If $\int \frac{1}{1- \omega} < \infty$, the limiting RJ equilibrium (corresponding to $\varphi \to \frac {3\pi} 4$) has finite mass. It is given by $\varphi = \frac{3\pi}4$ or in other words $\beta = -\gamma$, $\gamma>0$.
\item $\displaystyle F(\infty) = \frac{I}{2\pi}$, with $\displaystyle I = \int \omega \dd \xi$. Indeed, we can expand
$$
\frac{1}{\omega + x} = \frac{1}{x} - \frac{\omega}{x^2} + O_{\infty}(\frac{1}{x^3}) \quad \mbox{so that} \quad \int \frac{\dd \xi}{\omega + x} \sim \frac{2\pi}{x} - \frac{I}{x^2} + O_\infty \left( \frac 1 {x^3} \right).
$$
We could have found this expression directly since the limit $\tan \varphi \to \infty$ corresponds to the limit $\varphi \to \frac \pi 2 -$ or equivalently $\beta \to 0+$, in which case $\frac{\mathcal{E}(\mathfrak{f}_{0,\gamma})}{\mathcal{M}(\mathfrak{f}_{0,\gamma})} = \frac{I}{2\pi}$.
\item $\displaystyle F(-\infty) = \frac{I}{2\pi}$ since the expansion above remains valid as $x \to -\infty$. Here again, the limit $\tan \varphi \to - \infty$ corresponds to the limit $\varphi \to\frac \pi 2 +$ or equivalently $\beta \to 0-$, in which case $\frac{\mathcal{E}(\mathfrak{f}_{0,\gamma})}{\mathcal{M}(\mathfrak{f}_{0,\gamma})} = \frac{I}{2\pi}$.
\end{itemize}
Therefore, we find that the image of $(-\infty,-1) \cup (-a,\infty)$ by $F$ is $(c_1,c_2)$, since $F$ is increasing and satisfies $F(\pm \infty)=I/2\pi \in (c_1,c_2).$ This was the desired result. Note that if the integrals $\int \frac{1}{\omega - a} d p  =\int \frac{1}{1-\omega} dp = \infty$, then, $c_1=a, c_2=1$ in which case the image of $F$ is $(a,1)$.    
\end{proof}

\subsubsection{Solving the maximization problem}

As a consequence of Theorem \eqref{Theorem1}, for certain choices of $\omega $, there are values of $\mathcal M$ and $\mathcal E$, satisfying (\ref{admissible}) for which there is no entropy maximizers of the form (\ref{RJBE}).
 As it was observed first in \cite{Einstein} for ideal Bose gases, the solution to this difficulty is to relax the problem at hand by allowing general non negative measures as possible maximizers (cf. for example  \cite{EscMiscValle} for a detailed description). This leads to maximizers with singular parts, similar to the  Bose-Einstein distributions in presence of a condensate, for the Nordheim equation.

In order to set up our minimization problem rigorously for general measures, we consider a general nonnegative measure $\lambda$. By the Radon-Nikodym theorem, $\lambda$ can be uniquely decomposed into an absolutely continuous part with respect to the Lebesgue measure and a singular part:
\begin{equation}
\label{RadonNikodym}
\dd \lambda = f \dd \xi + \dd \lambda_{\operatorname{sing}}
\end{equation}
where $f,\lambda_{\operatorname{sing}} \geq 0$, $f \in L^1(\mathbb{T})$, and  $\lambda_{\operatorname{sing}}$ is singular with respect to the Lebesgue measure. 

Following \cite{DemengelTemam,EscMiscValle}, the definitions of the mass, energy and entropy can be extended to general measures as follows
\begin{align*}
 \displaystyle \mathcal{M}(\lambda) = \int_{\mathbb{T}} \dd \lambda, \qquad 
\mathcal{E}(\lambda) = \int_{\mathbb{T}} \omega \dd \lambda, \qquad 
 \displaystyle  S(\lambda) = \int_{\mathbb{T}} \ln f(\xi) \dd \xi.
\end{align*}
The definitions for the mass and energy are natural; as for the entropy of a general measure, we define it to be equal to the absolutely continuous part of the measure. In heuristic terms, the logarithmic growth of the entropy functional cancels the contribution of the singular part of the measure. 
More precisely, let $a \in \mathbb{T}$ and $(\varphi_n)_{n\in \mathbb{N}}$ be a sequence of nonnegative functions so that $\varphi_n \rightharpoonup \delta_a$ in the sense of distributions. Then $H_{cl}(f + \varphi_n) \to H_{cl}(f)$. 

\begin{theorem} [Maximizers of the entropy]
\label{maintheorem}
Let $\mathcal{M}_0, \mathcal{E}_0 \in (0,\infty)$ be given mass and energy. The maximizers of the entropy $S$ over positive measures subject to the constraints $\mathcal{M}(f) = \mathcal{M}_0$ and $\mathcal{E}(f) = \mathcal{E}_0$ are characterized as follows:
\begin{itemize}
\item[(i)] If $\displaystyle c_1 < \frac{\mathcal{E}_0}{\mathcal{M}_0} < c_2$, the unique maximizer is the unique RJ equilibrium with this mass and energy.
\item[(ii)] If $\displaystyle c_2 \leq \frac{\mathcal{E}_0}{\mathcal{M}_0} < 1$, maximizers are of the type $\mathfrak{f}_{\beta,\gamma} + \lambda_{\operatorname{sing}}$, where $(\beta,\gamma)$ are characterized by 
$$
\begin{cases}
\mathcal{M}(\mathfrak{f}_{\beta,\gamma}) = \frac{\mathcal{M}_0 - \mathcal{E}_0}{1-c_2} \\ \mathcal{E}(\mathfrak{f}_{\beta,\gamma}) = c_2 \frac{\mathcal{M}_0 - \mathcal{E}_0}{1-c_2},
\end{cases}
$$
and $\lambda_{\operatorname{sing}}$ has mass $\frac{1}{1-\beta}(\mathcal{E}_0 - \beta \mathcal{M}_0)$ and is supported on the set where $\omega$ is maximal.

\item[(iii)] If $\displaystyle a < \frac{\mathcal{E}_0}{\mathcal{M}_0} \leq c_1 $, maximizers are of the type $\mathfrak{f}_{\beta,\gamma} + \lambda_{\operatorname{sing}}$ with 
$$
\begin{cases}
\mathcal{M}(\mathfrak{f}_{\beta,\gamma}) = \frac{\mathcal{M}_0 - \mathcal{E}_0}{1-c_1} \\ \mathcal{E}(\mathfrak{f}_{\beta,\gamma}) = c_1 \frac{\mathcal{M}_0 - \mathcal{E}_0}{1-c_1}.
\end{cases}
$$ Also $\lambda_{\operatorname{sing}}$ has mass and energy equal to $\frac{1}{1-c_1}(\mathcal{E}_0 - c_1 \mathcal{M}_0)$ and is supported on the set where $\omega$ is maximal.
\end{itemize}
\end{theorem}

\begin{proof} We start with a measure $\lambda$ satisfying
$
\mathcal{M}(\lambda) = \mathcal{M}_0,$ and $\mathcal{E}(\lambda) = \mathcal{E}_0$, 
which we decompose as in \eqref{RadonNikodym}. Applying the inequality $\ln x \leq x -1$ to $\frac{f}{\mathfrak{f}_{\beta,\gamma}}$ and integrating gives
$$ 
\int \ln f \dd \xi - \int \ln \mathfrak{f}_{\beta,\gamma} \dd \xi \leq \int f (\beta \omega + \gamma) \dd \xi - 2\pi.
$$
Starting from this inequality and using in addition that $\beta \omega + \gamma \geq 0$ and the identity \eqref{identitymunu}, we obtain
\begin{align*}
S(\lambda) = S(f) & \leq S(\mathfrak{f}_{\beta,\gamma}) + \int f (\beta \omega + \gamma) \dd \xi - 2\pi \\
& \leq S(\mathfrak{f}_{\beta,\gamma}) + \int (\beta \omega + \gamma)(f \dd \xi + \dd \lambda_{\operatorname{sing}}) - \beta \mathcal{E}(\mathfrak{f}_{\beta,\gamma}) - \gamma \mathcal{M}(\mathfrak{f}_{\beta,\gamma}).
\end{align*}
By the mass and energy constraints on $\lambda$, this means that
\begin{equation}
\label{inegalitedebase}
S(\lambda) \leq S(\mathfrak{f}_{\beta,\gamma}) + \beta (\mathcal{E}_0-\mathcal{E}(\mathfrak{f}_{\beta,\gamma})) + \gamma (\mathcal{M}_0 -\mathcal{M}(\mathfrak{f}_{\beta,\gamma})).
\end{equation}

\medskip
\noindent $(i)$ 
Choosing $\beta$ and $\gamma$ in the above equation such that $\mathcal{E}(\mathfrak{f}_{\beta,\gamma}) = \mathcal{E}_0$ and $\mathcal{M}(\mathfrak{f}_{\beta,\gamma}) = \mathcal{M}_0$ (which is possible by Theorem \ref{Theorem1} since $c_1 < \frac{\mathcal{E}_0}{\mathcal{M}_0} < c_2$) gives the desired result.

\medskip
\noindent $(ii)$ We choose $\beta$ and $\gamma$ such that
$$
\begin{cases}
\mathcal{E}_0 - \mathcal{E}(\mathfrak{f}_{\beta,\gamma}) = \mathcal{M}_0 - \mathcal{M}(\mathfrak{f}_{\beta,\gamma}) \\
\mathcal{E}(\mathfrak{f}_{\beta,\gamma}) = c_2 \mathcal{M}(\mathfrak{f}_{\beta,\gamma})
\end{cases}
\quad 
\mbox{or equivalently}
\quad
\begin{cases}
\mathcal{M}(\mathfrak{f}_{\beta,\gamma}) = \frac{\mathcal{M}_0 - \mathcal{E}_0}{1-c_2} \\ \mathcal{E}(\mathfrak{f}_{\beta,\gamma}) = c_2 \frac{\mathcal{M}_0 - \mathcal{E}_0}{1-c_2}.
\end{cases}
$$
Such $\beta$ and $\gamma$ exist by Theorem \ref{Theorem1}, which gives furthermore that $\beta = -\gamma$, with $\gamma>0$. With this choice of $\beta$ and $\gamma$, inequality \eqref{inegalitedebase} becomes $S(\lambda) \leq S(\mathfrak{f}_{\beta,\gamma})$.

Examining the equality case in the derivation of inequality \eqref{inegalitedebase}, we see that necessarily $f = \mathfrak{f}_{\beta,\gamma}$. This implies that $\lambda_{\operatorname{sing}}$ is such that
$$
\begin{cases}
\mathcal{M}(\lambda_{\operatorname{sing}}) = \mathcal{M}_0 - \mathcal{M}(\mathfrak{f}_{\beta,\gamma}) = \mathcal{M}_0 - \frac{\mathcal{M}_0 - \mathcal{E}_0}{1-c_2} = \frac{1}{1-c_2}(\mathcal{E}_0 - c_2 \mathcal{M}_0) \\
\mathcal{E}(\lambda_{\operatorname{sing}}) = \mathcal{E}_0 - \mathcal{E}(\mathfrak{f}_{\beta,\gamma}) = \mathcal{E}_0 -  c_2 \frac{\mathcal{M}_0 - \mathcal{E}_0}{1-c_2} = \frac{1}{1-c_2}(\mathcal{E}_0 - c_2 \mathcal{M}_0).
\end{cases}
$$
Since the mass and energy of $\lambda_{\operatorname{sing}}$ are equal, it is necessarily supported on the set where $\omega$ takes the value $1$. Conversely, such a measure $\mu$ achieves the maximum value for the entropy, namely $S(\mathfrak{f}_{\beta,\gamma})$.
\end{proof}

\subsection{Classification of collisional invariants} \label{sec:collisional invariants}
We saw in Section \ref{subsection_stationary} that stationary solutions include RJ solutions and "most" functions concentrated on small balls. Do other examples exist? The entropy dissipation functional \eqref{entropydissipation} provides the key to this question, provided it makes sense. We need to assume that the entropy is well-defined, or that $\log f$ is integrable - this excludes functions whose support is distinct from $\mathbb{T}$. Then it is clear that the dissipation functional is zero if and only if
$$
\frac{1}{f_0} + \frac{1}{f_1} - \frac{1}{f_2} - \frac{1}{f_3} = 0  \qquad \mbox{on $\mathscr{R}_{\operatorname{eff}}$},
$$
or, after setting $g =\frac 1f$,
$$
g_0 + g_1 - g_2 -g_3 = 0 \qquad \mbox{on $\mathscr{R}_{\operatorname{eff}}$}.
$$
In the case where $\mathscr{R}_{\operatorname{eff}}$ can be parameterized as $\xi_1 = h(\xi_0,\xi_2)$ - which is the case for nearest-neighbor interactions, see \eqref{ReffRtriv} - the above equation becomes
\begin{equation}
\label{collisional_invariants_eq}
g(\xi_0) + g(h(\xi_0,\xi_2)) - g(\xi_2) - g(\xi_0 + h(\xi_0,\xi_2) - \xi_2) = 0 \qquad \forall \xi_0, \xi_2 \in \mathbb{T}.
\end{equation}
Obvious solutions are $1$, $\omega$, and their linear combinations, which explains why solutions of \eqref{collisional_invariants_eq} are also known as \textit{collisional invariants}. It is expected that these are the only solutions of \eqref{collisional_invariants_eq} and this could indeed be proed in the case $\omega(\xi) = \left| \sin \left(\frac{\xi}2 \right) \right|$ (which corresponds to the unpinned nearest-neighbors interaction).

\begin{theorem}[Lukkarinen-Spohn \cite{LukkarinenSpohn2008}]\label{theo:FPU collisional_invar}
If $\omega(\xi) = \left| \sin \left(\frac{\xi}2 \right) \right|$, all solutions of \eqref{collisional_invariants_eq} in $L^1(\mathbb{T})$ can be written
$$
g = c_1 + c_2 \omega, \qquad \mbox{where $c_1,c_2 \geq 0$.}
$$
\end{theorem}

It is conjectured in \cite{Lukkarinen2016,AokiLukkSpohn} that the above statement also holds for $\delta \in (0,\frac{1}{2})$, and it is tempting to think that it is indeed true for generic dispersion relations $\omega$. This is a central question for the further development of the theory, but it seems to be quite difficult

Note that the corresponding problem was solved in higher dimension in \cite{Spohn2006}.

\begin{proof}[Sketch of proof of Theorem \ref{theo:FPU collisional_invar}] In this case where $\delta=1/2$, we have the explicit parametrisation given by the formula of $h$ in \eqref{formula_h}. \\ 
\underline{Step 1}: We are assuming that $f \in \mathcal{C}(\mathbb{T}) \cap \mathcal{C}^1(0,2\pi)$, and that the limits $\lim_{\xi \downarrow 0} f'(\xi) = a \in \mathbb{R}$ and $\lim_{\xi \uparrow 2\pi} f'(\xi) = b  \in \mathbb{R}$, exist.
We are going to relax these assumptions later.

We  are going to exploit this jump of the derivative. For this, we consider the limit as $\xi_0 \uparrow 2 \pi$ and $\xi_2 $ is arbitrary point. In that limit explicit computations yield that as $\xi_0 \uparrow 2\pi$, then $h \searrow \xi_2-2\pi$ and thus $\xi_0-\xi_2+h \searrow 0 $. Moreover 
$\lim_{\xi_0 \uparrow 2\pi} (\partial_{\xi_2}h)(x,z)  = 1$,  $\lim_{\xi_0 \uparrow 2\pi} (\partial_{\xi_0} h)(\xi_0,\xi_2)  =-1-\tan^2 ((2\pi - \xi_2)/4)$.
With these in hand we write 
\begin{align*}
    - f'(\xi_2) \lim_{\xi_0 \uparrow 2\pi} (\partial_{\xi_0} h)(\xi_0,\xi_2) &=
    \lim_{\xi_0 \uparrow 2\pi} \left[ \frac{f(h+2\pi)-f(\xi_2)}{2\pi -\xi_0} \right] \\ & = 
    \lim_{\xi_0 \uparrow 2\pi} 
    \left[ \frac{f(2\pi)-f(\xi_0)}{2\pi -\xi_0} +  \frac{f(\xi_0-\xi_2+h)-f(0)}{2\pi -\xi_0}
    \right] \\ & = b - a \lim_{\xi_0 \uparrow 2\pi}(\partial_{\xi_0} h +1)(\xi_0,\xi_2)
\end{align*}
Rearranging and using the formula for the limit of $\partial_{\xi_0} h$ we get that $$f'(\xi_2) =  \frac{b-a}{1+\tan^2 ((2\pi - \xi_2)/4)} + a. $$
Now miraculously 
$\frac{1}{1+\tan^2 ((2\pi - \xi_2)/4) } = \frac{1}{2} - \omega'(\xi_2)$ and so for $\xi_2 \in (0,2\pi)$, 
$$f'(\xi_2) = \frac{a+b}{2} - (b-a) \omega'(\xi_2).$$
 We now proceed to justify the assumptions made on the regularity of $f$. 
\\


\noindent
\underline{Step 2}: Here we are assuming that $f \in \mathcal{C}^2(0,2\pi)$. We are going to explain how to pick $f(0)$ so that $f$ is continuous, and $\lim_{\xi_0 \downarrow 0} f'(\xi_0) = a \in \mathbb{R}$,  $\lim_{\xi_0 \uparrow 2\pi} f'(\xi_0) = b  \in \mathbb{R}$: 
First observe that if $f$ is antisymmetric (i.e. odd),
the situation is simple since $h(\xi_0,2\pi -\xi_0) =\pi -\xi_0$ and for all $\xi_0 \in (0,\pi)$ and $\xi_2=2\pi -\xi_0$: 
$$ f(\xi_0) + f(\pi - \xi_0) - f(2\pi-\xi_0)- f(\xi_0 + \pi)=0 \ \Rightarrow \ f(\xi_0) = f(\pi+\xi_0) \text{ and } f(\pi)=0.$$
Thus in this case, setting $f(2\pi n)=0$, $n \in \mathbb{Z}$ extends $f$ to $C^1$ at $\xi_0=0$ as the right limit coincides the limit at $\pi$.

If $f$ is symmetric,  we have $f'(\xi_2) = -f'(2\pi -\xi_2)$. We differentiate \eqref{collisional_invariants_eq} first w.r.t. to $\xi_0$ and then w.r.t. $\xi_2$. Adding them up yields: 
$$(1-\partial_{\xi_2} h)f'(\xi_0) - (1+\partial_{\xi_0} h)f'(\xi_2) +(\partial_{\xi_0} h+\partial_{\xi_2} h)f'(h)=0.$$
For $0<\xi_2<\xi_0<2\pi$, and for fixed $\xi_2$, 
we divide by $(1-\partial_{\xi_2} h)$ and rearrange to get 
\begin{align} \label{eq:existlim1}
    f'(\xi_0) = f'(\xi_{2}) + (\partial_{\xi_0} h+\partial_{\xi_2} h)\frac{f'(\xi_2) - f'(h)}{(1-\partial_{\xi_2} h)}.
\end{align}
We plan to take the limit $\xi_0 \to 2\pi$, 
and explicit calculations yield 
$$ \lim_{\xi_0 \uparrow 2\pi}\partial_{\xi_0} h = -1 - \tan^2 ((2\pi - \xi_2)/4), \qquad \lim_{\xi_0 \uparrow 2\pi}\partial_{\xi_2} h = 1.$$
In particular as $\xi_0 \uparrow 2\pi$, $\partial_{\xi_2} h - 1 $ vanishes at a linear rate $|\xi_0-2\pi|$. To be sure that the limit makes sense, we further fix $\xi_2=\pi$, (where $f'(\pi)=0$) and inserting in \eqref{eq:existlim1}: $$ \lim_{\xi_0 \uparrow 2\pi}f'(\xi_0) =\lim_{\xi_0 \uparrow 2\pi} \frac{f'(\pi) - f'(h(\xi_0,\pi))}{1-\partial_{\xi_2} h} = \lim_{\xi_0 \uparrow 2\pi} f''(h)\partial_{\xi_0} h  = -2 f''(\pi) =:b \in \mathbb{R}.$$
By symmetry of $f$ then we also conclude the existence of the limit at the other end $$a:= \lim_{\xi_0 \downarrow 0} f'(\xi_0) = - \lim_{\xi_0 \uparrow 2\pi} f'(\xi_0)  = -b.$$
Finally, since now $f'$ is bounded on $[0,2\pi]$, $\lim_{\xi_0 \to 0}f(\xi_0)$ exists and if we set $f(2\pi n)=\lim_{\xi_0 \to 0}f(\xi_0)$, we are done. 
\\
\noindent
\underline{Step 3}: Find a continuous representative $f$ from $\psi \in L^1(\mathbb{T})$ satisfying the initial identity. We build $f_{1,\varepsilon}$ continuous as $$f_{1,\varepsilon}(\xi_2) : =\varepsilon^{-1}\int_{2\pi - \varepsilon}^{2\pi}  [\psi(x) + \psi(h(\xi_0,\xi_2)) - \psi(\xi_0-\xi_2+h(\xi_0,\xi_2)) ]d\xi_0, $$ for $\xi_2 \in [\varepsilon, \pi + \varepsilon]$ and $\varepsilon>0$ sufficiently small. We are in the case $\xi_0>\xi_2$ so the explicit branch of $h$ is valid,  and from the main identity: $f_{1,\varepsilon} = \psi(\xi_2)$ a.e. in $[\varepsilon, \pi + \varepsilon]$. We then change the variables in each term, leading to 
\begin{equation}\label{eq:f_1epsilon}
\begin{split}
    f_{1,\varepsilon}(\xi_2) = \varepsilon^{-1}\int_{2\pi - \varepsilon}^{2\pi}  \psi(\xi_0)\dd \xi_0  & -
\varepsilon^{-1} \int_{\xi_2-2\pi}^{h(2\pi - \varepsilon,\xi_2)} \psi(y) \frac{d y}{(\partial_{\xi_0} h)(x_1(y,\xi_2),\xi_2)} 
 \\ & \qquad+ \varepsilon^{-1} \int_0^{2\pi - \varepsilon-\xi_2-h(2\pi - \varepsilon,\xi_2)} \psi(y) \frac{dy}{(1+\partial_{\xi_0}h)((x_2(y,\xi_2),\xi_2))}.
 \end{split}
\end{equation}
A direct computation shows that the Jacobians in the last two integrals are uniformly bounded for such $\xi_0,\xi_2, \varepsilon$, and depend continuously on $z$. Thus we may conclude continuity of the map $\xi_2 \mapsto f_{1,\varepsilon}(\xi_2)$ by dominated convergence. \\
Using the parity transformation $[P\psi] (\xi_0) = \psi(2\pi -\xi_0 ) = \sigma \psi$, $\sigma \in \{\pm1\}$ (since $\psi$ is either symmetric or antisymmetric), we can extend $f_{1,\varepsilon}$ to a function $f_\varepsilon$ defined on a larger interval by reflection. This extension is still continuous and still satisfies $f_\varepsilon = \psi$ a.e. on its domain.\\
Finally, $\varepsilon$ was arbitrary so for each $z \in (0,2\pi)$ and letting $\varepsilon < \min(\xi_2, 2\pi -\xi_2)$, we can extend the definitions and find $C^0$ representative $f =f_\varepsilon$ and a.e. equal to $\psi$.
Moreover, since in the collision identity we have continuous compositions of $f$ (at least when $\xi_0 \neq 0, \xi_2 \neq 0, \xi_0 \neq \xi_2$) the a.e. identity holds in a pointwise sense. 
\\
\noindent
\underline{Step 4}: Bootstrap the regularity for $f$ to $C^2((0,2\pi))$. We have that \eqref{eq:f_1epsilon} holds after replacing $\psi$ by $f$ and $f$ is $C^0$. So $f_{1,\varepsilon}$ is $C^1((0,2\pi))$. Repeat to improve it to $C^2$, etc.

Finally let us note that we have assumed $\psi$ is either symmetric or antisymmetric and this can be justified since by linearity of the collision identity, we may decompose $\psi = \psi_s +\psi_a$ and treat each case separately.
\end{proof}

Next, ideally we would like to have an analogous statement for the $1$-dim general $\delta \in (0,1/2)$ case. This question has remained open, as noted also in \cite[below (4.10)]{AokiLukkSpohn}, \cite{Lukkarinen2016},
and appears to be quite challenging. Below we give a short proof for the finite-dimensionality.

\begin{theorem}[General $\delta \in (0,1/2)$ case]
Assume $\omega(\xi, \delta) = \sqrt{1-2\delta \cos(\xi)}$. The space of $C^1$ solutions to \eqref{collisional_invariants_eq} is finite-dimensional. 
\end{theorem}
\begin{proof}
    We differentiate both in $\xi_0$ and in $\xi_2$ variable, we get the identity 
$$f'(\bar{h})(\partial_{\xi_0}h+ \partial_{\xi_2}h ) = (\partial_{\xi_0}h)f'(\xi_2) + (\partial_{\xi_2}h)f'(\xi_0).$$
We also know $\omega'$ satisfies this from the resonant condition,  and so for any linear combination of $f'$ and $\omega'$, say $g'$ we have this identity $$g'(\bar{h})(\partial_{\xi_0}h+ \partial_{\xi_2}h ) = (\partial_{\xi_0}h)g'(\xi_2) + (\partial_{\xi_2}h)g'(\xi_0), \qquad \bar{h} = \xi_0 - \xi_2 + h(\xi_0, \xi_2).$$ 
We aim to show $g'\equiv 0$.
\\
We project onto the $\xi_2$ variable, with the normalised-in-$\xi_2$ weight $$w(\xi_0,\xi_2,\delta):= \frac{|\partial_{\xi_2}h|}{\int_{\mathbb{T}} |\partial_{\xi_2}h|\dd \xi_2}\ \text{ and } \ \left(\frac{w}{\partial_{\xi_2}h}\right) (\xi_0,\xi_2,\delta)\ \text { to be bounded} $$ 
$$ \int_{\mathbb{T}} \left[ \frac{\partial_{\xi_0}h + \partial_{\xi_2}h}{\partial_{\xi_2}h} g'(\bar{h}) - \frac{\partial_{\xi_0}h}{\partial_{\xi_2}h} g'(\xi_2) \right] w(\xi_0,\xi_2,\delta)\dd \xi_2 =  g'(\xi_0). $$
In other words, our problem rewrites in the form 
\begin{equation} \label{eq:Compact-Id}
    \begin{split}
        &\hspace{5cm}(K(\delta) -  \operatorname{Id})g'=0,\  \text{ where }\\
        &K = K(\delta) : L^2 \to L^2,\quad \   g' \mapsto \int_{\mathbb{T}} \left[ \frac{\partial_{\xi_0}h + \partial_{\xi_2}h}{\partial_{\xi_2}h} g'(\bar{h}) - \frac{\partial_{\xi_0}h}{\partial_{\xi_2}h} g'(\xi_2) \right] w(\xi_0,\xi_2,\delta)\dd \xi_2.
    \end{split}
\end{equation}
\noindent
\underline{Now for $\delta<1/2$:} We can show that $K(\delta)$ is a compact operator: Indeed explicit computations show that  $\partial_{\xi_2} h < 1$, we may change the variable in the first integral:   
$$ K(\delta) g' = \int_{\mathbb{T}}- \frac{(\partial_{\xi_0}h + \partial_{\xi_2}h)}{\partial_{\xi_2} h-1} (\xi_0, \xi_2(\xi_0,y)) w (\xi_0, \xi_2(\xi_0,y)) g'(y) d y  - \int_{\mathbb{T}} (\partial_{\xi_0}h)(\xi_0, \xi_2) g'(\xi_2) w(\xi_0,\xi_2, \delta)\dd \xi_2
$$ and identify the integral kernels of the two pieces of $K$ to be  
\begin{equation}
    \begin{split}
&k_1(\xi_0,\xi_2(\xi_0,y),\delta) = - \frac{(\partial_{\xi_0}h + \partial_{\xi_2}h)}{\partial_{\xi_2} h-1} (\xi_0, \xi_2(\xi_0,y)) w (\xi_0,\xi_2(\xi_0,y)) \text{ and } \\  &k_2(\xi_0,\xi_2)= (\partial_{\xi_0}h)(\xi_0,\xi_2) w(\xi_0,\xi_2, \delta).
    \end{split}
\end{equation}
They are smooth, thus the operator is Hilbert-Schmidt, and so compact since each piece is. That shows that the space of solutions to \eqref{eq:Compact-Id} is finite dimensional, since we have a Fredholm type of operator.
\end{proof}

However note that this is not the case for $\delta=0$, where the space of solutions is in fact infinite-dimensional. To see that, we can use the explicit formula for the parametrization $h$ for a general delta, see \cite[Eq. (132)]{Luk2016}. Even though this formula is much more complicated (and thus way less handy) than the one in \eqref{formula_h}, by expanding for small $\delta$, we see that 
$$ h(\xi_0, \xi_2;\delta) = \pi - \xi_0 - \delta(\sin(\xi_0) +\sin(\xi_2)) + \mathcal{O}(\delta^2).$$ In other words, $h$ is independent of $\xi_2$ when $\delta=0$.  
Then the problem of finding all $L^1$ functions $f$ solving   \eqref{collisional_invariants_eq}: 
$$f(\xi_0) + f(h(\xi_0,\xi_2)) - f(\xi_2) - f(\xi_0 + h(\xi_0,\xi_2) - \xi_2) = 0 \qquad \forall \xi_0, \xi_2 \in \mathbb{T},$$ becomes 
$$f(\xi_0) + f(\pi - \xi_0) - f(\xi_2) - f(\pi - \xi_2) = 0, \forall \xi_0, \xi_2 $$
which means $f(\xi_0) + f(\pi - \xi_0)=\text{constant}$ for all $\xi_0 \in \mathbb{T}$. This is an infinite dimensional space. 

So the problem around $\delta$ close to $0$
 and $\delta$ close to $1/2$ is actually very different, and one most likely needs to perturb around the $\delta=1/2$ case in order to determine the collisional invariants.

\section{Rigorous results on the dynamics of the kinetic (FPUT) equation}

\label{section_rigorous_FPU}

In this section, we focus on the kinetic wave equation \eqref{KWE} associated to the (FPUT) model, for short the kinetic (FPUT) equation. In the homogeneous case, which will mostly occupy us, it takes the form
\begin{equation}
\tag{KFPUT} \label{KFPUT}
\partial_t f = \mathcal{C}(f), \qquad f = f(t,\xi) \geq 0, \qquad (x,\xi) \in \mathbb{R} \times \mathbb{T},
\end{equation}
where the collision operator is given by the formula
\begin{align*}
& \mathcal{C}(f)(\xi) = \int \omega_0 \omega_1 \omega_2 \omega_3 \left[ f_1 f_2 f_3 + f_0 f_2 f_3 - f_0 f_1 f_2 - f_0 f_1 f_3 \right] \delta(\Omega_{0,1,2,3}) \delta (\Sigma_{0,1,2,3}) \dd \xi_{1,2,3} \\
& \omega(\xi) = \left|\sin\left( \frac \xi 2 \right) \right|, \quad \Omega_{0,1,2,3} = \omega_0 + \omega_1 - \omega_2 - \omega_3, \qquad \Sigma_{0,1,2,3} = \xi_0 + \xi_1 - \xi_2 - \xi_3.
\end{align*}
Following \cite{LukkarinenSpohn2008,GermainLaMenegaki}, we will summarize here the existing mathematical theory for this equation, centered around the stability of RJ solutions. We will also say a word of the inhomogeneous case following \cite{Xiang}.

\subsection{Main results and organization of the section}
\textit{Section \ref{section_lwp}} establishes the basic local well-posedness result in weighted $L^\infty$ spaces.

\begin{theorem} \label{thmlwp} The collision operator is bounded in $\omega^\alpha L^\infty(\mathbb{T})$ for $\alpha > - \frac{5}{4}$. As a consequence, the equation is locally well-posed in $\mathcal{C} (\omega^\alpha L^\infty)$: for any data $\Omega$ in $\omega^\alpha L^\infty$, there exists a unique solution in $\mathcal{C}([0,T],\omega^\alpha L^\infty)$ with $T\sim \|\omega^{-\alpha} \Omega \|^{-2}_{L^\infty}$.
\end{theorem}

In \textit{Section \ref{sectionfirstproperties}}, we turn to the linearized problem around RJ equilibria. The traditional tools of kinetic theory yield the following weighted $L^2$-type estimate for the perturbation.

\begin{theorem}
We denote by $L$ the linearized operator around a RJ solution with $\beta,\gamma >0$. There exists a function $a(\xi)$ satisfying  $|a(\xi)| \sim |\sin \left( \frac \xi 2 \right)|^{\frac 53}$ as $\xi \to 0$,  such that  
for all initial data $g_0 \in L^2$ with $$ 
\int_0^{2\pi} a(p) g_0(p) \phi(p) \dd p = 0 \ \text{ for all } \phi \in \operatorname{Ker} L,$$ 
it holds
$$
\int_0^\infty \int_0^{2\pi} a(\xi) |e^{tL} g_0(\xi)|^2 \dd \xi \dd t \lesssim \| g_0 \|_{L^2}^2.
$$
\end{theorem}

This estimate is insufficient for nonlinear purposes: the decay in time is weak, and we cannot hope to close nonlinear estimates since the equation is ill-posed in $L^2$ type spaces. Furthermore, the lack of a spectral gap leads to the degenerate weight $a(\xi)$. \textit{Section \ref{sectionpointwise}} addresses these shortcomings by proving pointwise decay at a polynomial rate.

\begin{theorem} \label{theo:decay linear semigr FPUT}
Assume that $\gamma>0$ and that $g_0 \in L^\infty$, with a zero projection (in $L^2$) on $\operatorname{Ker} L$. Then for any $\mu ,\nu \in [\frac 16, \frac 12]$ and any $\delta>0$
$$
\| \omega^\mu e^{tL} g \|_{L^\infty} \lesssim_\delta \langle t \rangle^{- \frac{3}{5}(\mu + \nu)+ \delta} \| \omega^{-\nu} g_0 \|_{L^\infty}.
$$
\end{theorem}

This is achieved by understanding how the edges of the frequency domain, where dissipation degenerates, interact with the bulk of the domain. At a more technical level, we resort to an iterative scheme to gain decay increasingly.

Finally, \textit{Section \ref{sectionnonlinear}} deals with the fully nonlinear problem, showing global existence of solutions around non-singular RJ. This is accomplished through a careful definition of an appropriate norm that allows us to control the nonlinearity, relying crucially on the structure of the equation.

\begin{theorem} There exists $\varepsilon_0$ such that the following holds. If
$$
f(t=0) = \mathfrak{f}_{\beta,\gamma} [ 1 + g_0]
$$
where $\int \mathfrak{f}_{\beta,\gamma} g_0 \dd \xi = \int \omega  \mathfrak{f}_{\beta,\gamma} g_0 \dd \xi = 0$ and $\| \omega^{-\frac 12} g_0 \|_{L^\infty} = \varepsilon < \varepsilon_0$,
then there exists a global solution which can be written
$$
f(t) = \mathfrak{f}_{\beta,\gamma} [ 1 + g] \qquad \mbox{where} \quad \| \omega^{\frac{1}{2}} g(t,\xi) \|_{L^\infty_\xi} \lesssim \varepsilon \langle t \rangle^{-\frac 12} \quad \mbox{for all $t \geq 0$}.
$$
\end{theorem}

Finally, we say a word of the inhomogeneous case, where little is known, even though an important step was accomplished in \cite{Xiang}. There, it was proved that the equation is locally well-posed in the inhomogeneous setting, and that the time of existence can be improved for localized data thanks to the dispersive effect, which implies decay for the solution. Indeed, the time of existence in the theorem below is $\varepsilon^{-4}$, to be compared with $\varepsilon^{-2}$ in Theorem \eqref{thmlwp} (the agreement with the kinetic time scale $T_{\operatorname{kin}} = \varepsilon^{-4}$ being purely coincidental!).

\begin{theorem} We consider the Cauchy problem
$$
\begin{cases}
& \partial_t f + \omega'(\xi) \partial_x f = \mathcal{C}(f) \\
& f(t=0,x,\xi) = F(x,\xi)
\end{cases}
$$
Let $\varepsilon$ denote
$$
\varepsilon = \left\| \omega^{\frac 12} F \right\|_{L^\infty(\mathbb{R} \times \mathbb{T})} + \left\| F \right\|_{L^1_{x,\xi}(\mathbb{R} \times \mathbb{T})} + \left\| \omega^{-\frac 12} F \right\|_{L^2(\mathbb{R},L^\infty ( \mathbb{T}))} + \left\|  F \right\|_{L^1(\mathbb{R},L^2 ( \mathbb{T}))}.
$$
If $\varepsilon$ is sufficiently small, then there exists a unique solution $f(t,x,\xi)$ on $[0,C\varepsilon^{-4}]$ (with $C$ constant) which furthermore satisfies
$$ 
\sup_{t\in [0,T]} \left\| \omega^{\frac 12} f(t) \right\|_{L^\infty(\mathbb{R} \times \mathbb{T})} + \left\| f(t) \right\|_{L^1_{x,\xi}(\mathbb{R} \times \mathbb{T})} + (1+t)^{\frac 12} \left\| \omega^{\frac 12} f \right\|_{L^2(\mathbb{R},L^1 ( \mathbb{T}))} \lesssim \varepsilon.
$$

\end{theorem}

\subsection{Local well-posedness for \eqref{KFPUT}} 
\label{section_lwp}
In the case $\omega(\xi) = \left| \sin \left( \frac \xi 2 \right) \right|$ which occupies us here, we learned in Section \ref{section_resonances} that $\mathscr{R}_{\operatorname{eff}}$ can be parameterized as 
\begin{equation} \label{eq: param h_FPUT}
\begin{split}
    & \mathscr{R}_{\operatorname{eff}}(\omega) = \{ (\xi_0,h(\xi_0,\xi_2),\xi_2,\xi_0 + h(\xi_0,\xi_2) - \xi_2 ), \,(\xi_0,\xi_2) \in \mathbb{T}^2 \} \\
& h(\xi_0,\xi_2) = \frac{\xi_2 - \xi_0}{2} + 2 \arcsin \left(\tan \left| \frac{\xi_2-\xi_0}{4} \right| \cos \left( \frac{\xi_0+\xi_2} 4 \right) \right).
\end{split}
\end{equation}
A semi-explicit formula for the collision operator was given in\eqref{collisionformula_h}; in the case of \eqref{KFPUT}, the computation in \cite{LukkarinenSpohn2008} gives a completely explicit formula: 
\begin{equation}
\label{formulacollision}
\mathcal{C}(f)(p_0)= \int_0^{2\pi}
\frac{ \omega_0 \omega_1 \omega_2 \omega_3
 }{\sqrt{F^+_{0,2}}}  \left[ f_1 f_2 f_3 + f_0 f_2 f_3 - f_0 f_1 f_2 - f_0 f_1 f_3 \right] \dd \xi_2
\end{equation}
where it is understood that $\xi_3 = \xi_0 + \xi_1 - \xi_2$ and $\xi_1 = h(\xi_0,\xi_2)$ and furthermore
$$
F^+(\xi_0,\xi_2) = \left[ \cos \left( \frac{\xi_0}{2}\right) +\cos \left( \frac{\xi_2}{2}\right) \right]^2 +4\sin\left( \frac{\xi_0}{2}\right)\sin\left( \frac{\xi_2}{2}\right).
$$

In order to prove the local well posedness of \eqref{KFPUT} in $\omega^\alpha L^\infty$ in Theorem \eqref{thmlwp}, it suffices to prove boundedness of $\mathcal{C}$ on $\omega^\alpha L^\infty$. Indeed, this turns \eqref{KFPUT} into a Banach space-valued ODE whose local solution can be constructed through the Banach fixed-point theorem.

\begin{lemma}
The collision operator $\mathcal{C}$ is bounded on $\omega^\alpha L^\infty$ if $\alpha > -\frac{5}{4}$.
\end{lemma}

\begin{proof} We first show this for $\alpha \geq 1$, by rather crude estimates. 

The denominator in the integrand in \eqref{formulacollision} is bounded by 
\begin{equation*}
\sqrt{ F_+(\xi_0,\xi_2)} = \sqrt{\left[ \cos \left( \frac{\xi_0}{2}\right) +\cos \left( \frac{\xi_2}{2}\right) \right]^2 +4\sin\left( \frac{\xi_0}{2}\right)\sin\left( \frac{\xi_2}{2}\right)} \geq 2 \sqrt{\omega_0 \omega_2}
\end{equation*}
implying that 
\begin{equation} 
\begin{split} 
|\mathcal{C}[f](t,\xi_0)| &\leq \int_0^{2\pi} \sqrt{\omega_0 \omega_2} \omega_1 \omega_3 \left| \prod_{\ell=0}^3 f_\ell \left(\frac{1}{f}+\frac{1}{f_1}-\frac{1}{f_2}-\frac{1}{f_3}\right) \right|
d \xi_2.
\end{split}
\end{equation} 

We now resort to the inequalities
\begin{align*}
& \sqrt{\omega_0 \omega_2} \omega_1 \omega_3  \omega_0^{-\alpha} \prod_{\ell=0}^3 \omega_\ell^{\alpha} \left({\omega_0^{-\alpha}}+ {\omega_1^{-\alpha}}+{\omega_2^{-\alpha}}+{\omega_3^{-\alpha}}\right)
\lesssim
\begin{cases}
\omega_2^{\frac{1}{2}+\alpha} & \mbox{if $-1 \leq \alpha \leq 0$} \\
1 & \mbox{if $\alpha \geq 0$}.
\end{cases}
\end{align*}
For the first inequality, when $\alpha \in [-1,0]$, expanding the terms in the bracket we see that $\omega_2^{\frac{1}{2}+\alpha}$ is the worst term and all the rest of the factors are bounded by $1$.

The second inequality, when $\alpha>0$, is a consequence of $\omega_1 \omega_2 \omega_3 \lesssim \omega_0 $ (which can be proved by an elementary computation, see \cite{GermainLaMenegaki}, which in this case implies  $(\omega_1 \omega_2 \omega_3)^\alpha \lesssim \omega_0^\alpha$. Indeed by expanding the terms in the bracket, the whole expression becomes $\omega_0^{\frac{1}{2}}\omega_2^{\frac{1}{2}+\alpha} \omega_1^{1+\alpha}\omega_3^{1+\alpha} \big( {\omega_0^{-\alpha}}+ {\omega_1^{-\alpha}}+{\omega_2^{-\alpha}}+{\omega_3^{-\alpha}}\big)$. 
The worst term then is the first one, which is bounded by $\omega_0^{\frac{1}{2}} \omega_1 \omega_2^{\frac{1}{2}}\omega_3 \lesssim 1$. 
All the others are bounded, as there is no negative exponent. As a consequence,
$$
\| \omega^{-\alpha} \mathcal{C}[f] \|_{ L^\infty} \lesssim \|\omega^{-\alpha} f \|_{L^\infty}^3 \int_0^{2\pi} \max(1,\omega_2^{\frac{1}{2} + \alpha}) \dd p_2 \lesssim \|\omega^{-\alpha} f \|_{L^\infty}^3.
$$
since this integral is finite whenever $\alpha > -\frac{3}{2}$, which is always the case since  $\alpha \geq -1$.

In fact, using the explicit parametrisation $h$ of the resonant manifold, \eqref{eq: param h_FPUT}, 
we can sharpen the threshold for $\alpha$. Namely we will show that 
$$
Q_\alpha (\xi_0,\xi_2) := \sqrt{\omega_0 \omega_2} \omega_1 \omega_3  \prod_{\ell=1}^3 \omega_\ell^{\alpha} \left({\omega_0^{-\alpha}}+ {\omega_1^{-\alpha}}+{\omega_2^{-\alpha}}+{\omega_3^{-\alpha}}\right)
\lesssim \max(1,\omega_2^{\frac{3}{2} + 2\alpha}).$$

We focus on the region where both $\xi_0, \xi_2$ are close to $0$. In particular, due to symmetry, we treat only the case $0<\xi_2<\xi_0 \ll 1$. We Taylor expand $h$, around $(0,0)$ from which we obtain the asymptotics 
\begin{align*}
h(\xi_0,\xi_2) &= -\frac{1}{16} \xi_0\xi_2 (\xi_0-\xi_2) + \mathcal{O}\Big((\xi_0-\xi_2)^3(\xi_0+\xi_2)^2\Big), \\  
\xi_3  &= (\xi_0-\xi_2) - \frac{1}{16} \xi_0\xi_2 (\xi_0-\xi_2) + \mathcal{O}\Big((\xi_0-\xi_2)^3(\xi_0+\xi_2)^2\Big).
\end{align*}
 In this region then, up to leading order,  $\omega_0 \sim \xi_0, \omega_2\sim \xi_2$, $\omega_1\sim \xi_0\xi_2 (\xi_0-\xi_2)$ and $\omega_3 \sim \xi_0-\xi_2$.
 We turn then to our quantity $Q_\alpha$ which is the sum of four terms and we estimate each term separately (for $\alpha<0$): 
 \begin{align*}
     &\omega_0^{\frac{1}{2}-\alpha } \omega_2^{\frac{1}{2}+\alpha} \omega_1^{1+\alpha} \omega_3^{1+\alpha} \sim \xi_0^{\frac{3}{2}}\xi_2^{\frac{3}{2} + 2\alpha} (\xi_0-\xi_2)^{2+2\alpha} 
     \\
     &\omega_0^{\frac{1}{2}} \omega_2^{\frac{1}{2}+\alpha} \omega_3^{1+\alpha}\omega_1 \sim
     \xi_0^{\frac{3}{2}}\xi_2^{\frac{3}{2} + \alpha} (\xi_0-\xi_2)^{2+\alpha} 
     \\
     & \omega_0^{\frac{1}{2}} \omega_2^{\frac{1}{2}} \omega_3^{1+\alpha}\omega_1^{1+\alpha}\sim 
     \xi_0^{\frac{3}{2}+ \alpha}\xi_2^{\frac{3}{2} + \alpha} (\xi_0-\xi_2)^{2+2\alpha} \\
     & \omega_0^{\frac{1}{2}} \omega_2^{\frac{1}{2}+ \alpha} \omega_3 \omega_1^{1+\alpha}\sim 
     \xi_0^{\frac{3}{2}+ \alpha}\xi_2^{\frac{3}{2} + 2 \alpha} (\xi_0-\xi_2)^{2+\alpha}. 
 \end{align*}
 So indeed if $\alpha < 0$, the worst exponent in $Q_\alpha$ is $\xi_2^{\frac{3}{2} + 2 \alpha}$.   While when $\alpha \geq 0$, the quantity is upper bounded by $1$, as written above. 

 When $\xi_0,\xi_2$ are far from the edges, or when one of them is, the same (or simpler) bounds hold.

 Finally,
$$
\| \omega^{-\alpha} \mathcal{C}[f] \|_{ L^\infty} \lesssim \|\omega^{-\alpha} f \|_{L^\infty}^3 \int_0^{2\pi} \max(1,\omega_2^{\frac{3}{2} + 2 \alpha}) \dd \xi_2 \lesssim \|\omega^{-\alpha} f \|_{L^\infty}^3, 
$$
since this integral is finite if and only if $\alpha > -\frac{5}{4}$. 

Note that these calculations indicate that the  threshold $\alpha > -\frac{5}{4}$, is in fact the optimal one for the local well-posedness in $\omega^{\alpha} L^\infty(\mathbb{T})$.
\\
\\
\\


\end{proof}

The above lemma is rather elementary, but together with the theorem it implies, it gives a complete picture of local well-posedness questions on spaces modeled on $L^\infty$. It is surprisingly difficult to go beyond $L^\infty$ for the construction of solutions. It is proved in \cite{GermainLaMenegaki} that $\mathcal{C}$ is unbounded in $L^p$, $p \leq 3$. One expects local well-posedness to hold in $L^p$, $p \geq 3$, but this seems a suprisingly difficult question.

It is instuctive to compare with the question of local well-posedness for the kinetic NLS equation in dimension 3. For this model, optimal local well-posedness in $L^p$ spaces was proved in \cite{GermainIonescuTran} and global weak solutions could be built in the radially symmetric setting. Since \eqref{KFPUT} is a one-dimensional model, it is in a sense more singular (the resonant manifold, viewed as a subset of $\mathbb{T}^4$, has dimension $2$ while for NLS the resonant manifold has dimension $8$ in $\mathbb{T}^{12}$, the singularity being even weaker in the radially symmetric case).

\subsection{First properties of the linearized operator}

\label{sectionfirstproperties}

Recall that the RJ equilibria are given by
$$ \mathfrak{f}(\xi) = \mathfrak{f}_{\beta, \gamma} (\xi) = \frac{1}{\beta \omega(\xi) + \gamma},$$
where the allowed range of $\beta,\gamma \in \mathbb{R}$ is given in \eqref{rangemunu} (it simply corresponds to the condition $\beta \omega(\xi) + \gamma \geq 0$ for any $\xi$)

We linearise around $\mathfrak{f}(\xi)$: If $f(t,\xi) = \mathfrak{f}(\xi) (1+ g(t,\xi))$, then recall that $g(t,\xi)$ satisfies the following equation
\begin{equation*} 
\begin{split} 
\partial_tg(t,\xi) & = \mathfrak{f}_0^{-1} \int_{\mathbb{T}^3} 
 \delta(\Sigma) \delta(\Omega)
  \left[ \prod_{\ell=0}^3 \omega_\ell \mathfrak{f}_\ell \right]\times \\
  &\qquad \qquad \prod_{\ell=0}^3 (1+ g_\ell ) \left(\frac{1- g_0}{\mathfrak{f}_0}+\frac{1- g_1}{\mathfrak{f}_1}-\frac{1- g_2}{\mathfrak{f}_2}-\frac{1- g_3}{\mathfrak{f}_3}\right)\,\dd \xi_{1,2,3} + \mathcal{O}(g^2)\\ 
& = [L g](\xi) + \mathcal{O}(g^2). 
 \end{split}
\end{equation*} 
where the linarized operator is given by the formula
$$
[L g](\xi) = 
 \frac{1}{\mathfrak{f}_0} \int_{\mathbb{T}^3} 
 \delta(\Sigma) \delta(\Omega)
  \left[ \prod_{\ell=0}^3 \omega_\ell \mathfrak{f}_\ell \right]
  \left[ -\frac{g}{\mathfrak{f}}-\frac{g_1}{\mathfrak{f}_1}+\frac{g_2}{\mathfrak{f}_2}+\frac{g_3}{\mathfrak{f}_3}\right]\,\dd \xi_{1,2,3}.
$$

This operator is symmetric in $L^2$ and the Dirichlet form is positive:
$$ \langle -L g,g \rangle = \frac{1}{4} \int \delta({\Sigma})\delta({\Omega}) \left[ \prod_{\ell=0}^3 \omega_\ell \mathfrak{f}_\ell \right] \left[ \frac{g_3}{\mathfrak{f}_3}+ \frac{g_2}{\mathfrak{f}_2} - \frac{g}{\mathfrak{f}}  - \frac{g_1}{\mathfrak{f}_1}\right]^2 \dd \xi_{1,2,3} \geq 0 .$$
The next step is to split  $L$ into $L = -A+ K$, where $A$ is the multiplication operator
$$
[Ag](\xi)= a(\xi) g(\xi), \qquad a(\xi) = \frac{\omega_0}{\mathfrak{f}_0}\int \delta(\Sigma) \delta(\Omega)
\left[ \prod_{\ell=1}^3 \omega_\ell \mathfrak{f}_\ell \right]
\dd \xi_{1,2,3}
$$
and $K$ is the integral operator
$$
K = - K_1 + 2 K_2, \qquad 
\begin{cases}
& \displaystyle [K_1 g](\xi) = \frac{1}{\mathfrak{f}_0} \int \delta(\Sigma) \delta(\Omega) \left[ \prod_{\ell=0}^3 \omega_\ell \mathfrak{f}_\ell \right] \frac{g_1}{\mathfrak{f}_1} \, \dd \xi_{1,2,3} \\
& \displaystyle [K_2 g](\xi) = \frac{1}{\mathfrak{f}_0} \int \delta(\Sigma) \delta(\Omega) \left[ \prod_{\ell=0}^3 \omega_\ell \mathfrak{f}_\ell \right] \frac{g_2}{\mathfrak{f}_2} \, \dd \xi_{1,2,3}.
\end{cases}
$$
One checks that $K_1$ and $K_2$ are symmetric operators. In order to manipulate more easily these operators, one has to parameterize the resonant manifold $\Omega=\Sigma=0$. After some delicate calculus \cite{LukkarinenSpohn2008}, and identifying an operator with its kernel, it follows that
\begin{equation}
\label{chardonneret}
\begin{split}
& a(\xi) = \frac{\omega(\xi)}{\mathfrak{f}(\xi)} 
\int_0^{2\pi}
\omega_1 \omega_2 \omega_3 \mathfrak{f}_1 \mathfrak{f}_2 \mathfrak{f}_3 \frac{\dd \xi_2}{\sqrt{ F^+_{0,2}}}, \qquad \xi_1 = h(\xi,\xi_2),\; \xi_3 = \xi + \xi_1 - \xi_2 \\
& K_1(\xi,\xi_1) = \frac{\mathbf{1}_{F^-_{0,1}>0}}{\sqrt{F^-_{0,1}}}\omega  \omega_1 \omega_2 \omega_3 \mathfrak{f}_2 \mathfrak{f}_3, \qquad \xi_1 = h(\xi,\xi_2),\; \xi_3 = \xi + \xi_1 - \xi_2 \\
& K_2(\xi,\xi_2) = \frac{1}{\sqrt{F^+_{0,2}}} \omega  \omega_1 \omega_2 \omega_3 \mathfrak{f}_1 \mathfrak{f}_3, \qquad \xi_1 = h(\xi,\xi_2),\; \xi_3 = \xi + \xi_1 - \xi_2.
\end{split}
\end{equation}
where
\begin{align*}
& F^+(\xi_0,\xi_2) = {\left[ \cos \left( \frac{\xi_0}{2}\right) +\cos \left( \frac{\xi_2}{2}\right) \right]^2 +4\sin\left( \frac{\xi_0}{2}\right)\sin\left( \frac{\xi_2}{2}\right)} \\
& F^-(\xi_0,\xi_1) = {\left[ \cos \left( \frac{\xi_0}{2}\right) -\cos \left( \frac{\xi_1}{2}\right) \right]^2 -4\sin\left( \frac{\xi_0}{2}\right)\sin\left( \frac{\xi_1}{2}\right)}
\end{align*}

The next four lemmas lay out the basic properties of the linear operators $L,A,K_1,K_2$. These results all appear in \cite{LukkarinenSpohn2008}, except for the case $\gamma \neq 0$ of Lemma \ref{asymptoticsmultiplier}.

\begin{lemma}[Kernel of the linearized operator]
The kernel of $L$ is spanned by $\mathfrak{f}$ and $\omega \mathfrak{f}$.
\end{lemma}

\begin{proof}
This is a rephrasing of Theorem \ref{theo:FPU collisional_invar}.
\end{proof}

\begin{lemma}[Asymptotics of the multiplier function]
\label{asymptoticsmultiplier}
For any $\beta,\gamma>0$,
$$
a(\xi) \sim_{\beta,\gamma} \left| \sin \left( \frac \xi 2 \right) \right|^{\frac 53}.
$$
\end{lemma}

\begin{proof} This follows from a precise examination of the integral defining $a$ in \eqref{chardonneret}. \end{proof}

\begin{lemma}[Compactness of the weighted $K_2$ operator]
For any $\beta$, $\gamma$, the operator $a^{-\frac 12} K_2 a^{-\frac 12}$ is compact on $L^2$. 
\end{lemma}

\begin{proof}
Using Lemma \ref{asymptoticsmultiplier} and the fact that $\sqrt{F^+_{0,2}} \geq \sqrt{\omega_0 \omega_2}$, we find that
$$
a(\xi)^{-\frac 12} K_2(\xi,\xi_2) a(\xi_2)^{-\frac 12} \lesssim \left| \sin \left( \frac{\xi}{2} \right) \right|^{-\frac 13} \left| \sin \left( \frac{\xi_2}{2} \right) \right|^{- \frac 13}, 
$$
so that the corresponding operator is Hilbert-Schmidt, and therefore compact on $L^2$.
\end{proof}

\begin{lemma}[Compactness of the weighted $K_1$ operator]
For any $\beta$, $\gamma$, the operator $a^{-\frac 12} K_1 a^{- \frac 12}$ is compact on $L^2$.
\end{lemma}

\begin{proof}
This operator is not Hilbert-Schmidt, but it can be approximated in the operator norm by bounded operators. By closedness of the set of compact operators in that topology, the desired result follows.
\end{proof}

Combining these four lemmas, we obtain a crucial lower bound on the linearized operator. In the following we denote by $\langle \cdot, \cdot \rangle$ the usual $L^2$ inner product. The space $L^2_a$ is the $L^2$ space on $\mathbb{T}$ endowed with the measure $a(x) \dd x$, with the inner product
$$
\langle f,g \rangle_a := \int a(\xi) f(\xi) g(\xi)\, \dd \xi.
$$
The notation $g \perp_{L^2_a} \operatorname{Ker}(L)$ means the function $g$ is orthogonal to the subspace $\operatorname{Ker} L$ with respect to the weighted inner product.

\begin{prop} \label{propositiondissipation} For any $\beta,\gamma >0$, we have for all $g \in L^2$ such that $g \perp_{L^2_a} \operatorname{Ker} L$, 
$$ \langle -L g,g\rangle \gtrsim \int a(\xi) |g(\xi)|^2 \dd \xi. $$ 
\end{prop} 

\begin{proof}
Let $\widetilde{K} = a^{-\frac 12} K a^{- \frac 12}$. Then
\begin{align*}
   0 \leq - \langle L g,g \rangle = \langle A g,g \rangle - \langle K g,g \rangle &= \langle a^{\frac 12} g, a^{\frac 12} g \rangle - \langle \widetilde{K} (a^{\frac 12} g), (a^{\frac 12} g) \rangle \\ & = 
   \langle (\operatorname{Id} - \widetilde{K}) a^{\frac 12} g ,a^{\frac 12} g \rangle. 
\end{align*}
Since $\widetilde{K}$ is compact and self-adjoint, its spectrum is discrete away from $0$. It follows from the above formula that $\widetilde{K}$ cannot have eigenvalues $>1$, due to the positive sign of the Dirichlet form, and that the eigenspace associated to the eigenvalue $1$ of $\widetilde{K}$ coincides with the space $a^{\frac 12}\operatorname{Ker} L$, i.e. that $\operatorname{Ker} (I - \widetilde{K}) = a^{\frac 12}\operatorname{Ker} L$. Indeed, 
$$\widetilde{K}f = f \iff a^{-\frac{1}{2}} K  a^{-\frac{1}{2}} f=f \iff Kg=Ag \iff L g=0, $$
with $f =a^{\frac 12} g$, where we stress that $L$ is acting on the original variable $g$, while $\widetilde{K}$ acts on the conjugated variable $a^{\frac{1}{2}}g$, in flat $L^2$.

Finally, if $f =a^{\frac 12} g  \in \operatorname{Ker}(\operatorname{Id} - \widetilde{K})^\perp$, then due to spectral theorem 
$\langle (\operatorname{Id} - \widetilde{K})f,f \rangle \geq \delta \| f \|_{L^2}^2$  for some $\delta>0$, 
and thus
$$
 - \langle L g,g \rangle = \langle A g,g \rangle - \langle \widetilde{K} a^{\frac 12} g, a^{\frac 12} g \rangle \gtrsim \int a(\xi) |g(\xi)|^2 \dd \xi.
$$
\end{proof}

Turning to the evolution problem, we note first that it admits a unique solution for $L^2$ data, by boundedness on $L^2$.
As a consequence of the lower bound proved in the previous proposition, we get the following corollary.

\begin{corrolary}[Dissipation inequality]
\label{corodissi}
For $g_0 \in  L^2 \cap \big( \operatorname{Ker} L\big)^{\perp_a}$,
$$
\int_0^\infty \int a(\xi) |e^{tL} g_0(\xi)|^2 \dd \xi \dd t \lesssim \| g_0 \|_{L^2}^2.
$$
\end{corrolary}

This corollary quantifies time decay for the solution, but it will be insufficient for our purposes. Indeed, the equation is ill-posed in weighted $L^2$ spaces, so that this topology cannot be used to control nonlinear terms. In the next section, we aim at obtaining pointwise decay, which will correct this shortcoming.

\subsection{Pointwise decay for the linearized operator}

\label{sectionpointwise}

In all that follows we assume that the chemical potential $\gamma>0$, that is we study the linearised operator around non-singular RJ equilibria. 
In this section, we follow \cite{GermainLaMenegaki} and investigate how the energy dissipation leads to a polynomially fast relaxation for the linearized semigroup, pointwise and away from the edges. To this end, we explore how the edges of the domain, where the weight in the Poincaré Inequality of the previous section vanishes, interact with the bulk of the domain, where the weight is lower-bounded. 

\subsubsection{Energy decay in the bulk}
We define
$$
\langle t \rangle = 10 + |t|
$$
and the following subintervals of $[0,2\pi]$, corresponding to the edges and the bulk of the domain respectively
\begin{equation}
\begin{split}
 & \mathcal{E}_{t,\alpha} = \{ \xi \in [0,2\pi]: \xi< \langle t \rangle^{-\alpha},\ \xi> 2\pi-\langle t \rangle^{-\alpha}  \} \text{ and}\\ 
 &  \mathcal{B}_{t,\alpha} = \{ \xi \in [0,2\pi]: \langle t \rangle^{-\alpha} \leq \xi \leq 2\pi- \langle t \rangle^{-\alpha}\}.
\end{split}
\end{equation}

We now define the following functionals: 
\begin{equation}
    \begin{split}
    m(t) = 
    \int_{\mathcal{B}_{t,\alpha}} |g(t,\xi)|^2 d \xi,  \quad
    n(t)  = \int_{\mathcal{E}_{t,\alpha}  }|g(t,\xi)|^2 d \xi, \quad
    q(t) = \sup_{\xi\ \in\ \mathcal{E}_{t,\alpha} } |g(t,\xi)|.
 \end{split}
\end{equation}

\begin{lemma} \label{lemm:ODEf for m}
Assume $g_0 \in L^2 \cap \big( \operatorname{Ker}L\big)^\perp$, let  $\alpha < \frac 35$, and suppose that 
$$
m(0)+n(0) \leq 1 \qquad \mbox{and} \qquad q(t) \leq \langle t \rangle^{e}\quad \text{ for all } t \geq 0, 
$$
where $e \in \mathbb{R}$. Then
$$
m(t) \lesssim_\alpha \langle t \rangle^{2e - \alpha}.
$$
\end{lemma}

\begin{proof}
We consider solutions with initial data $g_0 \in (\operatorname{Ker} L)^\perp$, which is propagated by the linear flow, so that $g_t \in (\operatorname{Ker} L)^\perp$ for $t>0$. 

We decompose $g_t$ into its projection on $\operatorname{Ker}  L$ with respect to the weighted inner product $\langle \cdot , \cdot \rangle_a$ and the orthogonal component: 
$$g_t = g^{\parallel}_t + g_t^{\perp_a}, $$ where 
$$
g_t^{\parallel} \in \operatorname{Ker} L \quad \mbox{and} \quad \int g_t^{\perp_a} a \phi_i \dd p=0, \ \quad \text{ for } \ \phi_i \in \operatorname{Ker} L, \quad \ i=1,2.$$ The elements of the two-dimensional $\operatorname{Ker} L$ are given in Theorem \ref{theo:FPU collisional_invar}: $\phi_1=\mathfrak{f}$,  $\phi_2=\mathfrak{f}\omega$. 
By symmetry of $L$ and the fact that $L g_t^{\parallel} =0$, we have for all $t$
$$ -\langle Lg_t , g_t \rangle =  - \langle L g_t^{\perp_a} , g_t^{\perp_a} \rangle.$$

This implies, together with Lemma \ref{asymptoticsmultiplier} and Proposition \ref{propositiondissipation} for $g_t^{\perp_\alpha}$,
\begin{equation} \label{theODE}
\begin{split}
-\frac{d}{dt} (m(t) + n(t))  = - 2\int g(t,p) L g(t,p) \dd p &= - 2\int g^{\perp_a}(t,p) L g^{\perp_a}(t,p) \dd p 
\\ &  \gtrsim \int a(p) |g^{\perp_a}(t,p)|^2 \dd p \\ 
&\gtrsim \int_{\mathcal{B}_{t,\alpha}} \omega(p)^{\frac 53} |g^{\perp_a}(t,p)|^2  \dd p.
    \end{split}
\end{equation}
We will now lower-bound the right-hand side by a quantity depending on $m$ and $n$. 

We start with the decomposition
$$
g_t^{\parallel}  = c_1(t)\phi_1 + c_2(t)\phi_2, \qquad c_1(t), c_2(t) \in \mathbb{R}
$$
and the fact that, since $g_t \in (\operatorname{Ker} L)^\perp$, 
$$0=\langle g_t, \phi_i \rangle  =  \int_{\mathcal{B}_{t,\alpha}}(g_t^{\perp_a}  + g_t^\parallel) \phi_i \dd p + \int_{\mathcal{E}_{t,\alpha}} g_t \phi_i \dd p , \ \quad  i=1,2, $$
which after rearranging gives 
\begin{align} \label{eq:eq gram bulk matrix}
- \int_{\mathcal{B}_{t,\alpha}} g_t^{\perp_a} \phi_i  - \int_{\mathcal{E}_{t,\alpha}} g_t \phi_i = \int_{\mathcal{B}_{t,\alpha}} [c_1(t) \phi_1 + c_2(t) \phi_2]\  \phi_i, \ \quad i =1,2. 
\end{align} 
This can be written in a matrix equation form where the right-hand side is $G_{\mathcal{B}_{t,\alpha}}(c_1(t), c_2(t))^T$, with $G_{\mathcal{B}_{t,\alpha}}$ being the $2 \times 2$ matrix with entries 
$$(G_{\mathcal{B}_{t,\alpha}})_{i,j=1}^2 = \int_{\mathcal{B}_{t,\alpha}}\phi_i \phi_j. $$
 Now observe that since $\phi_1, \phi_2$ are linearly independent, the matrix with entries $(G_{\mathbb{T}})_{i,j=1}^2 = \int_{\mathbb{T}}\phi_i \phi_j$ is invertible. Since 
$\mathcal{B}_{t,\alpha} \uparrow \mathbb{T}$ as $t \to \infty$, for times $t$ large enough, the Gram matrix $G_{\mathcal{B}_{t,\alpha}}$ approximates the Gram matrix $G_{\mathbb{T}}$ and thus it is also invertible. We leave the details of this claim for the end of the proof. 

Thus, we may solve \eqref{eq:eq gram bulk matrix} for the coefficients $c_1(t), c_2(t)$ and bound them: 
$$c(t): = [c_1(t), c_2(t)]^T =  - G_{\mathcal{B}_{t,\alpha}}^{-1} \begin{bmatrix}
    \int_{\mathcal{B}_{t,\alpha}} g_t^{\perp_a} \phi_1  + \int_{\mathcal{E}_{t,\alpha}} g_t \phi_1\\
    \int_{\mathcal{B}_{t,\alpha}} g_t^{\perp_a} \phi_2  + \int_{\mathcal{E}_{t,\alpha}} g_t \phi_2, 
\end{bmatrix}
$$
from which we get, with the help of the Cauchy-Schwarz inequality,
\begin{align} \label{eq:bound on c_i(t)}
|c(t)| \lesssim \| G_{\mathcal{B}_{t,\alpha}}^{-1} \| \left( \left\vert \begin{bmatrix}
    \int_{\mathcal{B}_{t,\alpha}} g_t^{\perp_a} \phi_1\\
    \int_{\mathcal{B}_{t,\alpha}} g_t^{\perp_a} \phi_2
\end{bmatrix}\right\vert  + \left\vert 
\begin{bmatrix}
    \int_{\mathcal{E}_{t,\alpha}} g_t \phi_1 \\
    \int_{\mathcal{E}_{t,\alpha}} g_t \phi_2
\end{bmatrix}\right\vert \right) \lesssim \|g_t^{\perp_a}\|_{L^2(\mathcal{B}_{t,\alpha})} + \sqrt{n(t)\langle t \rangle^{-\alpha}},
\end{align} 
since $\| G_{\mathcal{B}_{t,\alpha}}^{-1} \|, \phi_1, \phi_2$ are bounded (recall $\gamma>0$).   
By boundedness of $\phi_1, \phi_2$, this allows us to estimate
$\|g_t^\parallel\|_{L^2(\mathcal{B}_{t,\alpha})}$: 
$$ \|g_t^\parallel\|_{L^2(\mathcal{B}_{t,\alpha})}\lesssim |c(t)| 
\lesssim \|g_t^{\perp_a}\|_{L^2(\mathcal{B}_{t,\alpha})} + \sqrt{n(t)\langle t \rangle^{-\alpha}}, 
$$
so that eventually we control the whole $\|g_t\|_{L^2(\mathcal{B}_{t,\alpha})} = \sqrt{m(t)}$ norm: 
$$
\sqrt{m(t)} =
\|g_t\|_{L^2(\mathcal{B}_{t,\alpha})} \lesssim \|g_t^{\perp_a}\|_{L^2(\mathcal{B}_{t,\alpha})}  + \|g_t^{\parallel}\|_{L^2(\mathcal{B}_{t,\alpha})} \lesssim \|g_t^{\perp_a}\|_{L^2(\mathcal{B}_{t,\alpha})} + \sqrt{n(t)\langle t \rangle^{-\alpha}}. $$
Squaring both sides then 
 $$ m(t) \lesssim  \|g_t^{\perp_a}\|_{L^2(\mathcal{B}_{t,\alpha})}^2 + n(t)\langle t \rangle^{-\alpha} \quad \Rightarrow \quad \|g_t^{\perp_a}\|_{L^2(\mathcal{B}_{t,\alpha})}^2 \geq C_0 m(t) - n(t)\langle t \rangle^{-\alpha},
$$
for a constant $C_0$. 

We are now ready to go back to our ODE \eqref{theODE}. We use on the one hand the above inequality and on the other the fact that
$$
n(t) \lesssim q(t)^2 \langle t \rangle^{-\alpha} \lesssim \langle t \rangle^{2e-\alpha};
$$
this gives
\begin{equation*}
\begin{split} 
-\frac{d}{dt} (m(t) + n(t))  
\gtrsim \int_{\mathcal{B}_{t,\alpha}} \omega(p)^{\frac 53} |g^{\perp_a}(t,p)|^2  \dd p &\gtrsim C_0 \langle t \rangle^{- \frac{5\alpha}{3}} \Big( m(t) - n(t) \langle t \rangle^{-\alpha} \Big)
\\
& \gtrsim C_0 \langle t \rangle^{- \frac{5\alpha}{3}} m(t)  -  \langle t \rangle^{- \frac{8 \alpha}{3}+ 2e}, 
\end{split}
\end{equation*}
since $\langle t \rangle^{- 11 \alpha /3} \lesssim \langle t \rangle^{- 8 \alpha /3}$. 
Rearranging the ODE we need to solve: 
\begin{equation} \label{eq:theODE2}
\begin{split} 
\frac{d}{dt} \big(m(t) + n(t)\big) + c \langle t \rangle^{- \frac{5\alpha}{3}} \big(m(t)+n(t)\big) & \lesssim \langle t \rangle^{- \frac{8 \alpha}{3}+ 2e} + c n(t) \langle t \rangle^{- \frac{5 \alpha}{3}} 
\\
& \lesssim  \langle t \rangle^{- \frac{8 \alpha}{3}+ 2e}.
\end{split}
\end{equation} 

Multiply the differential inequality \eqref{eq:theODE2} by $e^{cI(t)}$ for 
$$I(t) = \int_0^t \langle s \rangle^{-\frac {5\alpha} 3} \dd s \sim \langle t \rangle^{1 - \frac {5\alpha}{3}}, $$
and 
integrate in time to get  
$$
m(t) + n(t) \lesssim e^{-cI(t)}\big(m(0) + n(0)\big) + e^{-cI(t)} \int_0^t e^{cI(s)}\langle s \rangle^{- \frac{8 \alpha}{3}+ 2e} \dd s.
$$

The first term in the right-hand side is decaying exponentially fast, so we turn immediately to the second term which behaves as 
$$ e^{- c \langle t \rangle^{1 - \frac{5\alpha} 3}} \int_0^t \langle s \rangle^{2e - \frac{8 \alpha}{3}} e^{c \langle s \rangle^{1 - \frac{5\alpha} 3}} \dd s \sim \langle t \rangle^{2e -\alpha}, $$
where we used the identity
\begin{equation} \label{identity: integral order}
\int_1^T e^{t^a} t^b \dd t \sim \int_1^{T^a} e^s s^{\frac b a + \frac 1 a -1} \dd s \sim e^{T^a} T^{b-a+1}, \qquad \mbox{if $T>2$, $a>0$, $b \in \mathbb{R}$}
\end{equation}
(which is itself a consequence of $\int_1^t e^s s^a \dd s \sim e^t t^a$, valid for any $a \in \mathbb{R}$).  
Thus we conclude that indeed $m(t) \lesssim \langle t \rangle^{2e -\alpha}.$

Before we finish the proof we briefly verify the remaining claim that the matrix $G_{\mathcal{B}_{t,\alpha}}$ is invertible for $t$ sufficiently large:  
This claim follows immediately using the fact that the set of invertible matrices is open subset of the space of $2\times 2$ matrices. 
Indeed we know that the full Gram matrix $G_{\mathbb{T}}$ is invertible since  $\phi_1, \phi_2$ are linearly independent.  
Also, we write $$G_{\mathcal{B}_{t,\alpha}} = G_{\mathbb{T}} - G_{\mathcal{E}_{t,\alpha}}, \quad \text{ with }\ \left( G_{\mathcal{E}_{t,\alpha}}\right)_{i,j=1}^2:= \int_{\mathcal{E}_{t,\alpha}} \phi_i \phi_j $$
and we have $$\left\vert \left( G_{\mathcal{E}_{t,\alpha}} \right) _{i,j}\right\vert \lesssim |\mathcal{E}_{t,\alpha}|\lesssim \langle t \rangle^{-\alpha}, $$
since $\phi_i, \phi_j$ are bounded. So, 
$$ G_{\mathcal{E}_{t,\alpha}} \to 0, \quad t\to \infty \quad \Rightarrow \quad G_{\mathcal{B}_{t,\alpha}} \to G_{\mathbb{T}}.$$ 
Since $G_{\mathbb{T}}$ is invertible and the set of invertible matrices is open, it follows that $G_{\mathcal{B}_{t,\alpha}}$ is invertible for all sufficiently large $t$.
\end{proof}

\subsubsection{A weak pointwise bound}
We will need a very weak bound, which will be a stepping stone to start the proof of pointwise decay.

\begin{lemma} \label{lem:weakLinfty}
Let $g$ solve $\partial_t g - L g =0$ with $g(t=0) =g_0 \in L^\infty$. Then
$$
\| g(t) \|_{L^\infty(I)} \lesssim \langle t \rangle^{1+} \| g_0 \|_{L^\infty(I)}. 
$$
 \end{lemma}
The proof is not hard and can be found in \cite{GermainLaMenegaki}.

\subsubsection{Pointwise bounds}
 \begin{lemma} \label{lem: Step 1 of the iteration}Let $\alpha < \frac{3}{5}$, and assume that
$$
\| g_0 \|_{L^\infty} \leq 1 \quad \mbox{and} \quad \| g(t) \|_{L^2} \leq \langle t \rangle^b.
$$
\begin{itemize}
\item[(i)] If $b>-1$, for any $p \in [0,2\pi]$,
$$|g(t,\xi)| \lesssim |g_0(\xi)| + \omega^{\frac{3}{2}} \langle t \rangle^{b+1+}  . $$
\item[(ii)] In the bulk: if $\omega > \langle t \rangle^{-\alpha}$,
$$|g(t,\xi)| \lesssim \omega^{-\frac 16} \langle t \rangle^{b+}$$
\end{itemize}
 \end{lemma}
\begin{proof}
$(i)$ By formula \eqref{chardonneret} for $K_1$ and $K_2$,
\begin{align*}
\frac{d}{dt} |g(\xi)| & \lesssim  \left| \int \frac{\omega \omega_1 \omega_2 \omega_3}{\sqrt{F^+_{0,2}}} g(\xi_2) \dd \xi_2 \right|  + \left| \int \frac{\omega \omega_1 \omega_2 \omega_3 1_{F^-_{0,1} >0}}{\sqrt{F^-_{0,1}}} g(\xi_1) \dd \xi_1 \right| = I_+ + I_-.
\end{align*}

We start with $I_+$, for which we prove that when $\xi$ is close to $0$, 
$$I_+ \lesssim \omega^{\frac{3}{2}}\left( 1+ \sqrt{ \log \frac{2}{\omega}}\right) \| g\|_{L^2(\mathbb{T})}.$$
The case $\xi$ close to $2 \pi$ is symmetric, and if $\xi$ is away from the edges, the integral $I_+$ is easy since then $F^+_{0,2}$ is lower bounded. Indeed one checks that the only zeros of $F^+_{0,2}$ are the points $(0,2\pi)$ and $(2\pi,0)$.  
Thus we consider $0<\theta \ll 1$ so that $\xi \in[0,\theta]$ and we set $\varepsilon: = \sin \frac{\xi}{2}$ and,  $s:=\sin \frac{2\pi - \xi_2}{2}$.

We split the $\xi_2$ integral into two domains: $[0, 2\pi - c_0]$ and $[2\pi - c_0, 2\pi]$ for $c_0>0$ is a small constant. Then: 
$$\left| \int \frac{\omega \omega_1 \omega_2 \omega_3}{\sqrt{F^+_{0,2}}} g(\xi_2) \dd \xi_2 \right|\lesssim \left| \int_{0}^{2\pi - c_0} \frac{\omega \omega_1 \omega_2 \omega_3}{\sqrt{F^+_{0,2}}} g(\xi_2) \dd \xi_2\right| +  \left| \int_{2\pi - c_0}^{2\pi} \frac{\omega \omega_1 \omega_2 \omega_3}{\sqrt{F^+_{0,2}}} g(\xi_2) \dd \xi_2 \right|:= I_+^{b}(\xi) +  I_+^{e}(\xi).$$

For  $I_+^b(\xi)$: We notice that $F^+_{0,2}$ is lower bounded since $\xi_2$ stays away from $2\pi$ in that region, and $\xi \in [0,\theta]$ with $\theta$ small. We use the bound $\omega_1\omega_2\omega_3 \lesssim \omega = \varepsilon$ (cf. \cite[Lemma 19]{GermainLaMenegaki}). Thus the integral kernel is bounded by $\lesssim \varepsilon^2$. Cauchy-Schwarz then yields
$$I_+^b \lesssim \varepsilon^2 \|g\|_{L^2}.$$

For $I_+^e(\xi)$: we notice that when $\xi_2 \in (2\pi - c_0, 2\pi)$,  $s\leq s_0$ for some $s_0 \sim c_0$. Then we know that (cf. \cite[Lemma 18]{GermainLaMenegaki}) $\omega_1\omega_2\omega_3 \sim \frac{\varepsilon s^2}{\varepsilon^2 + s^2}$ and also that 
$F^+_{0,2} (\xi,\xi_2) \sim (\varepsilon^2 - s^2)^2 +4 \varepsilon s$. Thus, 
$$ I_+^e(\xi) \lesssim  \left| \int_{2\pi - c_0}^{2\pi} \frac{\varepsilon^2 s^2}{\varepsilon^2 + s^2} \frac{1}{\sqrt{(\varepsilon^2 - s^2)^2 +4 \varepsilon s}} g(\xi_2) \dd \xi_2 \right|. $$
Now if $ 0 < s \leq \varepsilon $, we have  $ \varepsilon^2 + s^2 \sim \varepsilon^2$ and the denominator is lower bounded by $\sqrt{\varepsilon s}$. The integral kernel then is $\lesssim \frac{s^2}{\sqrt{\varepsilon s}} = s^{\frac{3}{2}}\varepsilon^{-\frac{1}{2}} \lesssim \varepsilon$, implying that 
$$\int_{0<s\leq \varepsilon}|g(\xi_2)| \varepsilon \dd \xi_2 \leq \varepsilon^{\frac{3}{2}} \|g\|_{L^2}.$$
Finally if  $\varepsilon < s \leq s_0$, we have $ \varepsilon^2 + s^2 \sim s^2$ and the integral kernel then is $\lesssim \frac{\varepsilon^2}{\sqrt{\varepsilon s}} = \varepsilon^{\frac{3}{2}}s^{-\frac{1}{2}}.$ Then for any $\kappa_0>0$, 
$$\int_{\varepsilon \leq s \leq s_0} |g(\xi_2)|\varepsilon^{\frac{3}{2}}s^{-\frac{1}{2}} \dd \xi_2 \lesssim \varepsilon^{\frac{3}{2}} \left( \int_{\varepsilon}^{s_0} \frac{1}{s} \right)^{1/2} \|g\|_{L^2} \lesssim \varepsilon^{\frac{3}{2}}  \sqrt{\log\left( \frac{s_0}{\varepsilon} \right)}\|g\|_{L^2} \lesssim_{\kappa_0} \omega^{\frac{3}{2} - \kappa_0} \langle t \rangle^b.$$
Altogether $$I_+ = I_+^b+ I_+^e \lesssim_{\kappa_0}  \omega^2\langle t \rangle^b  + \omega^{\frac{3}{2} - \kappa_0}\langle t \rangle^b \lesssim_{\kappa_0} \omega^{\frac{3}{2} - \kappa_0 }\langle t \rangle^b .$$


For the singular term $I_-$, we first use the already mentioned inequality $\omega_1 \omega_2 \omega_3 \lesssim \omega$ to obtain that
$$
I_- \leq \left| \int \frac{\omega \omega_1 \omega_2 \omega_3 \mathbf{1}_{F_->0}}{\sqrt{F_-}} g(\xi_1) \dd \xi_1 \right| \lesssim \omega \left| \int \frac{\omega \mathbf{1}_{F^- > 0}}{\sqrt{F^-}} g(\xi_1) \dd \xi_1 \right|.
$$
At this point, we need more refined information on $F^-$. It is shown in \cite{LukkarinenSpohn2008} that $\xi_1 \mapsto F^-(\xi_0,\xi_1)$ has two zeros on $[0,2\pi]$, which will be denoted $\xi_1'$ and $\xi_1''$, and that $F^-$ is negative on $(\xi_1',\xi_1'')$. Furthermore, there holds $F^-(\xi_0,\xi_1) \gtrsim \omega_0 (\xi_1'-\xi)$ on $[0,\xi_1']$, and $F^-(\xi_0,\xi_1) \gtrsim \omega_0 (\xi-\xi_1'')$ on $[\xi_1'',2\pi]$. This enables us to bound
$$
\omega \left| \int_0^{\xi_1'} \frac{\omega \mathbf{1}_{F_->0}}{\sqrt{F_-}} g(\xi_2) \dd \xi_2 \right| \lesssim \omega^{\frac 32} \left| \int_0^{\xi_1'} \frac{1} {\sqrt{\xi_1'-\xi_2}} g(\xi_2) \dd \xi_2 \right| 
$$

Finally, splitting the integral between $|\xi-\xi_1'| < \langle t \rangle^{-N}$ and $|\xi-\xi_1'| \geq \langle t \rangle^{-N}$ for $N$ big enough, and using Lemma \ref{lem:weakLinfty}, we get
$$
I_- \lesssim \omega^{\frac 32} \left[ \log \langle t \rangle \| g \|_{L^2} + \langle t \rangle^{-\frac N 2} \| g \|_{L^\infty} \right] \lesssim \omega^{\frac 32} \langle t \rangle^{b+} +  \omega^{\frac 32} \langle t \rangle^{2-\frac N 2} .
$$

We now combine the estimates on $I_+$ and $I_-$ to integrate the previous ODE in time, which gives 
$$
|g(t,\xi)| \lesssim | g_{0}(\xi) |+ \omega^{\frac 32} \langle t \rangle^{ b + 1+}.
$$

\medskip

\noindent $(ii)$ We proceed as in $(i)$, except that we do not neglect the term $a(\xi) f(t,\xi)$. This gives the differential inequality
$$
\left| \frac{d}{dt} g(t,\xi) + a(\xi) g(t,\xi) \right| \lesssim \omega^{\frac 32} \langle t \rangle^{b+}
$$
or
$$
\left| \frac{d}{dt} \left[ e^{a(\xi) t} g(t,\xi) \right] \right| \lesssim e^{a(\xi) t} \omega^{\frac 32} \langle t \rangle^{b+},
$$
which can be integrated to give
$$
|g(t,\xi)| \lesssim e^{- a(\xi) t} g_0(\xi) + \omega^{\frac 32} e^{-a(\xi) t} \int_0^t e^{a(\xi)s} \langle s \rangle^{b+} \dd s.
$$
If $\omega > \langle t \rangle^{-\alpha}$, the first term on the right-hand side decays faster than any polynomial. Turning to the second term, we bound it in a straightforward way and use that $a(\xi) \sim \omega^{5/3}$ to get
$$
\left| \omega^{\frac 32} e^{-a(\xi) t} \int_0^t e^{a(\xi)s} \langle s \rangle^{b+} \dd s \right| \lesssim \omega^{\frac 32} \langle t \rangle^{b+} \int_0^t e^{a(\xi)s} \dd s \lesssim \omega^{- \frac 16} \langle t \rangle^{b+}.
$$
\end{proof}

\subsubsection{Iterative improvement}

\begin{theorem}[Pointwise decay] \label{theoremdecay}
Assume that $g_0 \in L^\infty$, with a zero projection (in $L^2$) on $\operatorname{Ker} L$. Then for any $\mu ,\nu \in (\frac 16, \frac 12)$ and any $\delta>0$
$$
\| \omega^\mu e^{tL} g \|_{L^\infty} \lesssim_\delta \langle t \rangle^{- \frac{3}{5}(\mu + \nu)+ \delta} \| \omega^{-\nu} g_0(\xi) \|_{L^\infty}.
$$
\end{theorem}

\begin{proof}
We apply iteratively lemmas \ref{lemm:ODEf for m} and \ref{lem: Step 1 of the iteration}, having set $\alpha = \frac 35 - \varepsilon_0$, with $0 <\varepsilon_0 \ll 1$. In the following, we denote $t^{a +}$ for $t^{a + C\varepsilon_0}$, for a constant $C$, instead of our usual notation $t^{a+}$ meaning $t^{a + \delta}$ for any $\delta>0$.
\begin{itemize}
\item Applying first Lemma \ref{lem: Step 1 of the iteration} $(i)$ with $b=0$ gives $|g(t,\xi)| \lesssim |g_0(\xi)| + \omega^{\frac 32} \langle t \rangle^{1+} .
$ provided $f_0 \in L^2$.
Thus, assuming $|g_0(\xi)| \leq 1$, we get
$$
q(t) \lesssim 1 + \langle t \rangle^{1 - \frac 32 \cdot \frac 35 +} \lesssim \langle t \rangle^{ \frac 1 {10} +}.
$$

\item Applying Lemma \ref{lemm:ODEf for m} with $e = \frac 1 {10} +$ gives $m(t) \lesssim \langle t \rangle^{\frac 15 - \frac 35 +} = \langle t \rangle^{- \frac 25 +}$, which gives in turn
$$
\| g(t) \|_{L^2} \lesssim m(t)^{\frac 12} + q(t) \langle t \rangle^{-\frac 3 {10} +} \sim \langle t \rangle^{-\frac 15 +}.
$$

\item Applying Lemma \ref{lem: Step 1 of the iteration} $(i)$ with $b = - \frac 15 +$ gives $|g(t,\xi)| \lesssim |g_0(\xi)| + \omega^{\frac 32} \langle t \rangle^{\frac 45+}$. Assuming $|g_0(\xi)| \lesssim \omega^{\frac 16}$, this yields
$$
q(t) \lesssim \langle t \rangle^{- \frac 16 \cdot \frac 35 +} + \langle t \rangle^{\frac 45 - \frac 32 \cdot \frac 35 +} \sim \langle t \rangle^{- \frac 1 {10} +}.
$$

\item Applying Lemma \ref{lemm:ODEf for m} with $e = - \frac 1 {10}$ gives $m(t) \lesssim \langle t \rangle^{-\frac 15 - \frac 35 +} = \langle t \rangle^{-\frac 45 +}$ hence
$$
\| g(t) \|_{L^2} \lesssim \langle t \rangle^{-\frac 25 +} + \langle t \rangle^{-\frac 1 {10} - \frac 3 {10} +} \sim \langle t \rangle^{-\frac 25 +}.
$$

\item Applying Lemma \ref{lem: Step 1 of the iteration} $(i)$ with $b = - \frac 25$, we get $|g(t,\xi)| \lesssim |g_0(\xi)| + \omega^{\frac 32 } \langle t \rangle^{\frac 35+}$. If $|g_0(\xi)| \leq \omega^\nu$ with $\nu \in [\frac 16,\frac 12]$, this gives
\begin{equation}
\label{boundq}
q(t) \lesssim \langle t \rangle^{- \frac 35 \nu +} + \langle t \rangle^{- \frac 3 {10} +} \sim  \langle t \rangle^{- \frac 35 \nu +}.
\end{equation}

\item Applying Lemma \ref{lemm:ODEf for m} with $e =- \frac 35 \nu +$ gives finally
$m(t) \lesssim \langle t \rangle^{- \frac {6 \nu}{5} - \frac 35 +}$ hence
\begin{equation}
\label{boundm}
\| g(t) \|_{L^2} \lesssim \langle t \rangle^{- \frac {3 \nu}{5} - \frac 3{10} +}. 
\end{equation}
\end{itemize}

We can now prove our final pointwise bound, under the assumption made above that $|g_0(\xi)| \leq \omega^\nu$. In the edges, we use \eqref{boundq} to get that
$$
\omega^\mu | g(t,\xi) | \lesssim \langle t \rangle^{-\frac 35 (\mu + \nu) +} \qquad \mbox{if $\omega < \langle t \rangle^{-\alpha}$}.
$$
In the bulk, we use Lemma \ref{lem: Step 1 of the iteration} $(ii)$ and \eqref{boundm} to get that
$$
\omega^{\frac 16} | g(t,\xi) | \lesssim \langle t \rangle^{- \frac {6 \nu + 3}{10}} \qquad \mbox{if $\omega > \langle t \rangle^{-\alpha}$}.
$$
\end{proof}

\subsection{Nonlinear stability} We now come to our main theorem, which gives nonlinear stability of the RJ states. We give the main idea of the proof and refer to \cite{GermainLaMenegaki} for the full details.
\label{sectionnonlinear}
\begin{theorem}
\label{thmnonlinstab}
For any $\beta,\gamma >0$, there exists $\varepsilon_0$ such that the following holds. If
$$
f(t=0) = \mathfrak{f}_{\beta,\gamma} [ 1 + g_0]
$$
where 
$$
\int \mathfrak{f}_{\beta,\gamma} g_0 \dd \xi = \int \omega \mathfrak{f}_{\beta,\gamma} g_0 \dd \xi = 0
$$
and
$$
\| \omega^{-\frac 12} g_0 \|_{L^\infty} = \varepsilon < \varepsilon_0,
$$
then there exists a global solution which can be written
$$
f(t) = \mathfrak{f}_{\beta,\gamma} [ 1 + g] \qquad \mbox{where} \quad \| \omega^{\frac{1}{2}} g(t,\xi) \|_{L^\infty_\xi} \lesssim \varepsilon \langle t \rangle^{-\frac 35 + \frac{1}{1000}} \quad \mbox{for all $t \geq 0$}.
$$
\end{theorem}

\begin{proof} \underline{The equation on the perturbation.} The full equation satisfied by the perturbation $g$ is
\begin{equation}
\label{equationg}
\partial_t g - L g = \mathcal{Q}(g) + \mathcal{C}(g),
\end{equation}
where
\begin{align*} 
& \mathcal{Q}(g)(t,\xi_0)= \frac{2}{\mathfrak{f}_0} \int \frac{ \omega_0 \omega_1 \omega_2 \omega_3}{\sqrt{F^+_{0,2}}}  \left[ \mathfrak{f}_2 \mathfrak{f}_3 g_2 g_3 - \mathfrak{f}_0 \mathfrak{f}_1 g_0 g_1 \right] \dd \xi_2 \\
&  \mathcal{C}(g)(t,\xi_0)= \frac{1}{\mathfrak{f}_0} \int \frac{ \omega_0 \omega_1 \omega_2 \omega_3 }{\sqrt{F^+_{0,2}}}  \prod_{\ell=0}^3 \mathfrak{f}_\ell g_\ell \left(\frac{1}{\mathfrak{f} g}+\frac{1}{\mathfrak{f}_1 g_1}-\frac{1}{\mathfrak{f}_2 g_2}-\frac{1}{\mathfrak{f}_3 g_3}\right)
 \dd \xi_2.
\end{align*} 

Duhamel's formula gives the equivalent formulation
\begin{equation}
\label{Duhamelform}
g(t) = S_t g_0 + \int_0^t S_{t-s} \left( \mathcal{Q}(g)(s,\xi) +  \mathcal{C}(g)(s,\xi)\right) \dd s.
\end{equation}

\medskip

\noindent \underline{The bootstrap argument} The norm that will be used to analyze this problem is, for $T>0$,
\begin{equation} \label{def:B_t norm}
\| g \|_{\mathcal{B}_T} := 
 \left\| \langle t \rangle^{\frac 2 5 - \delta} \omega^{\frac 1 6} g \right\|_{L^\infty_{t,\xi}([0,T] \times \mathbb{T})} + \left\| \langle t \rangle^{\frac 3 5 - \delta} \omega^{\frac 1 2} g \right\|_{L^\infty_{t,\xi}([0,T] \times \mathbb{T})}.
\end{equation}
We claim that there exists a constant $C_0 >0$ such that: for any $T>0$, if $g$ is a solution in $\mathcal{C}([0,T],\omega^{- \frac 16} L^\infty)$, then
\begin{equation}
\label{bootstrap}
\| g \|_{\mathcal{B}_T} \leq C_0 \left[ \varepsilon + \| g \|_{\mathcal{B}_T}^2 + \| g \|_{\mathcal{B}_T}^3 \right]
\end{equation}
where $\varepsilon := \| \omega^{-\frac 12} g_0 \|_{L^\infty}$ and where $\|\cdot\|_{\mathcal{B}_T}$ is defined in \eqref{def:B_t norm}.

If this inequality holds, then the desired result follows by a continuous induction argument, with the help of a local well-posedness result. We omit these standard steps.

\medskip

\noindent \underline{Proof of the bootstrap inequality \eqref{bootstrap}}
The linear term on the right-hand side of \eqref{Duhamelform} can be bounded by Theorem \ref{theoremdecay}:
$$
\left\| S_t g_0 \right\|_{\mathcal{B}_\infty} \lesssim \| \omega^{-\frac 12} g \|_{L^\infty} < \varepsilon 
$$

To deal with the quadratic and cubic terms, we will use the conservation laws of mass and energy. They imply that 
\begin{align*}
& \int \mathfrak{f}_{\beta,\gamma}(\xi)(1 + g(t,\xi)) \dd \xi = \int \mathfrak{f}_{\beta,\gamma}(\xi) \dd \xi \\
& \int \omega(\xi) \mathfrak{f}_{\beta,\gamma}(\xi)(1 + g(t,\xi)) \dd \xi = \int \omega(\xi) \mathfrak{f}_{\beta,\gamma}(\xi) \dd \xi 
\end{align*}
or in other words
$$
\int \mathfrak{f}_{\beta,\gamma}(\xi) g(t,\xi) \dd \xi = \int \omega(\xi) \mathfrak{f}_{\beta,\gamma}(\xi) g(t,\xi) \dd \xi = 0.
$$

As a consequence, the orthogonal projection (in $L^2$) of $g$ and $Lg$ on $\operatorname{Ker} L$ is zero. For the equation \eqref{Duhamelform} satisfied by $g$, this implies that the projection of the quadratic and cubic terms on $\operatorname{Ker} L$ is also zero. Therefore, we can apply Theorem \ref{theoremdecay} with $\mu = \nu = \frac 12$: if $t \in [0,T]$,
\begin{align*}
& \left\| \omega^{\frac 12} \int_0^t S_{t-s} \left( \mathcal{Q}(g)(s,\xi) +  \mathcal{C}(g)(s,\xi)\right) \dd s\right\|_{L^\infty_p} \\
& \qquad \lesssim \int_0^t  \langle t-s \rangle^{-\frac 35 + \delta} \left\| \omega^{-\frac 12}  \mathcal{Q}(g)(s,\xi)\right\|_{L^\infty} \dd s + \int_0^t  \langle t-s \rangle^{-\frac 35 + \delta} \left\| \omega^{-\frac 12}  \mathcal{C}(g)(s,\xi) \right\|_{L^\infty} \dd s\\
& \qquad = I + II.
\end{align*}
Using that $F_+ \gtrsim \omega_0 \omega_2$ and $\omega_1 \omega_2 \omega_3 \lesssim \omega_0$,
\begin{align*}
I &\lesssim \int_0^t \langle t-s \rangle^{-\frac 35 + \delta}\left\| \int \sqrt{\omega_2} \omega_1 \omega_3 [ |g_2 g_3| + |g_0 g_1|] \dd \xi_2 \right\|_{L^\infty} \dd s  \\
&\lesssim  \int_0^t \langle t-s \rangle^{-\frac 35 + \delta}  \| \omega^{\frac{1}{2}}g \|_{L^\infty}^2 \dd s  \lesssim \| g\|_{\mathcal{B}}^2 \int_0^t \langle t-s \rangle^{-\frac 35 + \delta}  \langle s \rangle^{\left(- \frac 35 + \delta \right)\cdot 2} \dd s \\
& \lesssim \| g \|_{\mathcal{B}}^2 \langle t \rangle^{- \frac 35 + \delta}.
\end{align*}
Similarly,
\begin{align*}
II & \lesssim  \int_0^t \langle t-s \rangle^{-\frac 35 + \delta}\left\| \int \sqrt{\omega_2} \omega_1 \omega_3 [ |g_1 g_2 g_3| +  |g_0 g_2 g_3| + |g_0 g_1 g_2| + |g_0 g_1 g_3| ] \dd \xi_2 \right\|_{L^\infty} \dd s  \\
& \lesssim  \int_0^t \langle t-s \rangle^{-\frac 35 + \delta} \| \omega^{\frac{1}{6}} g \|_{L^\infty} \| \omega^{\frac{1}{3}} g \|_{L^\infty} \| \omega^{\frac{1}{2}} g \|_{L^\infty} \dd s \\
& \lesssim  \| g \|_{\mathcal{B}}^3 \int_0^t \langle t-s \rangle^{-\frac 35 + \delta} \langle s \rangle^{ -\frac 25 - \frac 12 - \frac 35 + 4\delta} \dd s \\
& \lesssim   \| g\|_{\mathcal{B}}^3  \langle t \rangle^{-\frac 35 + \delta}.
\end{align*}

Using now Theorem \ref{theoremdecay} with $\mu = \frac 16$ and $\nu = \frac 12$,
\begin{align*}
& \left\| \omega^{\frac 16} \int_0^t S_{t-s} \left( \mathcal{Q}(g)(s,\xi) +  \mathcal{C}(g)(s,\xi)\right) \dd s\right\|_{L^\infty_p} \\
& \qquad \lesssim \int_0^t  \langle t-s \rangle^{-\frac 25 + \delta} \left\| \omega^{-\frac 12}  \mathcal{Q}(g)(s,\xi)\right\|_{L^\infty} \dd s + \int_0^t  \langle t-s \rangle^{-\frac 25 + \delta} \left\| \omega^{-\frac 12}  \mathcal{C}(g)(s,\xi) \right\|_{L^\infty} \dd s\\
& \qquad = III + IV.
\end{align*}

Taking once again advantage the inequalities $F^+ \gtrsim \omega_0 \omega_2$ and $\omega_1 \omega_2 \omega_3 \lesssim \omega_0$, 
\begin{align*}
III & \lesssim  \int_0^t \langle t-s \rangle^{-\frac 25 + \delta}\left\| \int \sqrt{\omega_2} \omega_1 \omega_3 [ |g_2 g_3| + |g_0 g_1|] \dd \xi_2 \right\|_{L^\infty} \dd s \\
& \lesssim \int_0^t \langle t-s \rangle^{-\frac 25 + \delta} \| \omega^{\frac 12} g \|_{L^\infty_p} \dd s \\
& \lesssim \| g \|_{\mathcal{B}_t}^2 \int_0^t  \langle t-s \rangle^{-\frac 25 + \delta} \langle s \rangle^{\left(\frac 12 + \delta \right) \cdot 2} \\
& \lesssim \| g \|_{\mathcal{B}_t}^2 \langle t \rangle^{- \frac 25 + \delta}.
\end{align*}

Similarly,
\begin{align*}
IV & \lesssim  \int_0^t \langle t-s \rangle^{-\frac 25 + \delta}\left\| \int \sqrt{\omega_2} \omega_1 \omega_3 [ |g_1 g_2 g_3| +  |g_0 g_2 g_3| + |g_0 g_1 g_2| + |g_0 g_1 g_3| ] \dd \xi_2 \right\|_{L^\infty} \dd s  \\
& \lesssim  \int_0^t \langle t-s \rangle^{-\frac 25 + \delta} \| \omega^{\frac{1}{6}} g \|_{L^\infty} \| \omega^{\frac{1}{3}} g \|_{L^\infty} \| \omega^{\frac{1}{2}} g \|_{L^\infty} \dd s \\
& \lesssim  \| g \|_{\mathcal{B}}^3 \int_0^t \langle t-s \rangle^{-\frac 25 + \delta} \langle s \rangle^{ -\frac 25 - \frac 12 - \frac 35 + 4\delta} \dd s \\
& \lesssim   \| g \|_{\mathcal{B}}^3  \langle t \rangle^{-\frac 25 + \delta}.
\end{align*}
This gives the desired estimate.
\end{proof}

\section{Hydrodynamic limits in the non-degenerate case}
\label{sec:Hydro_limit}

In this section, we derive formally the hydrodynamic limit of the inhomogeneous \eqref{KWE} in the nondegenerate case where the linearized operator $L$ of the homogeneous \eqref{KWE} around RJ states has a spectral gap.

We will consider two possible ansatz: either a perturbation of a global RJ equilibrium, or RJ equilibria with slowly varying (in space and time) parameters.

\subsection{A detour through the Boltzmann case}
Before starting with the discussion of the hydrodynamic limit of \eqref{KWE}, it is interesting to draw a comparison with hydrodynamic limits in the case of the Boltzmann equation. Among various excellent references for the hydrodynamic limits of the Boltzmann equations, we will follow the exposition from \cite{SaintRaymond}. Starting with the Boltzmann equation, which we write
\begin{equation*}
     \partial_t f + v \cdot \nabla_x f = \mathcal{Q} (f,f),
\end{equation*}
we first perform non-dimensionalization by introducing the characteristic length, time, and thermal speed $\ell_0, t_0, c$, with $c = \sqrt{5k T_0 / 3m }$ (here $T_0$ is the characteristic temperature and $m$ is the molecular mass of a particle): then the non-dimensionalized variables are 
\begin{equation*}
    \tilde t = \frac{t}{t_0}, \tilde x = \frac{x}{\ell_0}, \tilde v = \frac{v}{c}, 
\end{equation*}
and non-dimensionalized density variable becomes
\begin{equation*}
    \tilde f (\tilde t, \tilde x, \tilde v) = \frac{\ell_0^3 c^3}{N} f(t,x,v) = \frac{c^3}{R_0 } f(t,x,v).
\end{equation*}
Finally, the collision kernel should be non-dimensionalized: the collision kernel introduces a time-scale by $\mathcal{Q} (M_{R_0, 0, T_0} , M_{R_0, 0, T_0} ) = N / (\ell_0^3 \tau)$, and it can be interpreted as the average time scale of particles in the equilibrium density distribution $M_{R_0, 0, T_0}$ with density $R_0$ and temperature $T_0$ can travel freely without collision. Using this, we define mean free path $\lambda = c \tau$. 

Finally, we arrive at the following non-dimensional form of the Boltzmann equation: 
\begin{equation*}
    \text{St} \partial_t f + v \cdot \nabla_x f = \frac{1}{\text{Kn} } \mathcal{Q} (f,f),
\end{equation*}
where the kinetic Strouhal number reads $\text{St} = \frac{\ell_0}{c t_0}$, and the Knudsen number reads $\text{Kn} = \frac{\lambda}{\ell_0}.$

Now other non-dimensional numbers can be represented using $\text{St}$ and $\text{Kn}$. First, we consider the Mach number $\text{Ma} = u_0 / c$, where $u_0$ is the characteristic velocity of the fluid. In the case $u_0 = \ell_0 / t_0$, $\text{St} = \text{Ma}$. On the other hand, if we are interested in the flow which is much slower than the characteristic velocity of the geometry ($\ell_0/t_0$), for example, small fluctuations around some reference flow, then it can happen that $\text{Ma} << \text{St}$. Here, we only consider $\text{St} = \text{Ma}$ case. Also, the Reynolds number $\text{Re}$ is given by von Karman relation: $\text{Re} = \text{Ma}/\text{Kn}$.

Then, the hydrodynamic limit corresponds to the asymptotic limit $\text{Kn} \rightarrow 0$: collision becomes dominant, and therefore the system equilibrates and hydrodynamic variables describe the system. There are several cases:

\begin{enumerate}
    \item $\text{St} = \text{Ma} = O(1), \text{Kn} \rightarrow 0$: compressible Euler limit,
    \item $\text{St} = \text{Ma} \rightarrow 0, \text{Kn} \rightarrow 0, \text{Re}  \rightarrow \infty$: incompressible Euler limit,
    \item $\text{St} = \text{Ma} \rightarrow 0, \text{Kn} \rightarrow 0, \text{Re}  = O(1)$: incompressible Navier-Stokes limit.
\end{enumerate}

\subsection{The general setup} We start by rescaling the inhomogeneous \eqref{KWE} by setting $f(t,x,\xi) = \widetilde{f}(\varepsilon^A t, \varepsilon^B x, \xi)$ and then removing the tildes to obtain
\begin{equation}
\label{eq_f}
\varepsilon^A \partial_t f + \varepsilon^B \omega' \partial_x f = \mathcal{C} (f),
\end{equation}
which can also be written
$$
\varepsilon^{A-B} \partial_t f + \omega' \partial_x f = \varepsilon^{-B} \mathcal{C} (f),
$$
where it becomes apparent that the numbers $\varepsilon^{A-B}$ and $\varepsilon^{B}$ correspond to the Strouhals and Knudsen numbers respectively in the terminology of the Boltzmann equation.

Next, we consider a perturbation  of size $\varepsilon$ of a RJ equilibrium given by the ansatz
$$
f(t,x,\xi) = \mathfrak{f}_{\beta, \gamma} (\xi) (1+\varepsilon g (t,x,\xi) ).
$$
We now introduce some useful notation: since $\mathcal{C}(\mathfrak{f})=0$, the collision operator can be expanded as
$$
\mathcal{C}(\mathfrak{f}(1+\varepsilon g)) = \varepsilon \mathcal{L}_{\beta,\gamma} (g) + \varepsilon^2 \mathcal{Q}_{\beta,\gamma} (g) + \varepsilon^3 \mathcal{C}(\mathfrak{f} g).
$$
Here, we write $\mathcal{Q}_{\beta,\gamma}(g)$ or $\mathcal{Q}_{\beta,\gamma}(g,g)$ for the quadratic term, with the understanding that it is symmetric in its entries when we use the notation $\mathcal{Q}_{\beta,\gamma}(g,g)$. The same applies to $\mathcal{C}$. We also set
$$
L_{\beta,\gamma} g = \mathfrak{f_{\beta,\gamma}}^{-1} \mathcal{L}_{\beta,\gamma} (g), \qquad Q_{\beta,\gamma}(g) = \mathfrak{f}_{\beta,\gamma}^{-1} \mathcal{Q}_{\beta,\gamma} (g), \qquad C(g) = \mathfrak{f}_{\beta,\gamma}^{-1} \mathcal{C}(\mathfrak{f}_{\beta,\gamma} g).
$$

The \textit{non-degeneracy} of $L_{\beta,\gamma}$ corresponds to the two following conditions, which are assumed throughout this whole section. 
\begin{itemize} 
\item[(C1)] The kernel of the linearized operator $L_{\beta,\gamma}$ reduces to $\operatorname{Span}\{ \mathfrak{f}_{\beta,\gamma},\omega \mathfrak{f}_{\beta,\gamma} \}$
\medskip

\item[(C2)] The linearized operator $L_{\beta,\gamma}$ has a bounded inverse on the orthogonal of its kernel.
\end{itemize}

Since $L_{\beta,\gamma}$ is self-adjoint, we let $\Pi_0$ and $\Pi_1$ denote the orthogonal projections (in $L^2$) on $\operatorname{Ker} L_{\beta,\gamma}$ and $(\operatorname{Ker} L_{\beta,\gamma})^\perp$ respectively.

Any function of the form $F(\omega) \omega'$ belongs to $(\operatorname{Ker} L_{\beta,\gamma})^\perp$ since
$$
\int \frac{F(\omega)}{\beta \omega + \gamma} \omega' \dd \xi = \int \frac{F(\omega) \omega}{\beta \omega + \gamma} \omega' \dd \xi = 0.
$$
Thus
\begin{equation}
\label{Pi0Fomega}
\Pi_0 (F(\omega) \omega') = 0.
\end{equation}

\subsection{Perturbation around slowly varying Rayleigh-Jeans profiles}
Here we want to justify the nonlinear diffusion formula (6.4) in Aoki-Lukkarinen-Spohn \cite{AokiLukkSpohn}). We do not claim any rigor; in particular, the stability estimate justifying order-of-magnitude analysis seems difficult.

We will assume the ansatz
\begin{equation*}
    f (t,x,\xi) = \mathfrak{f}_{\beta, \gamma}(\xi) (1 + \varepsilon g (t,x,\xi)), \qquad \beta = \beta(t,x), \gamma = \gamma(t,x).
\end{equation*}
Besides $(\beta,\gamma)$, another natural set of coordinates is
$$
n(\beta,\gamma) = \int \mathfrak{f}_{\beta,\gamma} \dd \xi, \qquad e(\beta,\gamma) = \int \omega \mathfrak{f}_{\beta,\gamma} \dd \xi.
$$
Both sets of coordinates are equivalent; we refer to Section \ref{section_rigorous_stationary} for the allowed values of $(\beta,\gamma)$ and $(e,n)$.

The equation we start from is
$$
\varepsilon^A \partial_t \left[ \mathfrak{f}_{\beta,\gamma} (1+\varepsilon g) \right] + \varepsilon^B \omega' \partial_x \left[ \mathfrak{f}_{\beta,\gamma} (1+\varepsilon g) \right] = \mathcal{C}(\mathfrak{f}_{\beta,\gamma} (1+\varepsilon g))
$$
Using that $\omega \omega' \mathfrak{f}$ and $\omega' \mathfrak{f}$ have zero average and are orthogonal to $\operatorname{Ker} L_{\beta,\gamma}$, the usual mass and energy balance equation becomes
\begin{equation*}
    \varepsilon^A \partial_t \begin{pmatrix}
        e \\ n
    \end{pmatrix} + \varepsilon^{B+1} \partial_x \begin{pmatrix}
\langle \omega \omega' \mathfrak{f}_{\beta, \gamma}  , \Pi_1 g \rangle \\
\langle \omega'\mathfrak{f}_{\beta, \gamma} , \Pi_1 g \rangle
    \end{pmatrix} =  O(\varepsilon^{A+1}). 
\end{equation*}

Now we investigate the propagation equation for the remainder $g$: we have
\begin{equation*}
    \begin{split}
        -\mathfrak{f}_{\beta, \gamma}^2 &( \varepsilon^A \omega \partial_t \beta + \varepsilon^A \partial_t \gamma + \varepsilon^B \omega' (\omega \partial_x \beta + \partial_x \gamma) ) (1+ \varepsilon g ) + \varepsilon \mathfrak{f}_{\beta, \gamma} (\varepsilon^A \partial_t g + \varepsilon^B \omega' \partial_x g ) \\
        &= \mathcal{C} (\mathfrak{f}_{\beta, \gamma} (1+\varepsilon g)) = \varepsilon \mathfrak{f}_{\beta, \gamma} (L_{\beta, \gamma} g + \varepsilon Q_{\beta, \gamma} (g) + \varepsilon^2 C_{\beta, \gamma} (g) ),
    \end{split}
\end{equation*}
where $L_{\beta, \gamma}, Q_{\beta, \gamma}, C_{\beta, \gamma}$ are linearization, quadratic perturbation, and cubic perturbation around $\mathfrak{f}_{\beta, \gamma}$. With the aforementioned properties of $L_{\beta, \gamma}$, we have 
\begin{equation*}
    \Pi_1 g = L_{\beta, \gamma}^{-1} (- \varepsilon^{B-1} \omega' ( \partial_x\beta \omega + \partial_x \gamma ) \mathfrak{f}_{\beta, \gamma}  + O(\varepsilon^{\min(A, B, 1)}),
\end{equation*}
since $ \mathfrak{f}_{\beta, \gamma} (\omega \partial_t \beta + \partial_t \gamma) \in \operatorname{Ker} L_{\beta, \gamma}.$
Now plugging in this to the mass and energy balance equation, we obtain
\begin{equation*}
    \begin{split}
    \varepsilon^A \partial_t \begin{pmatrix}
        e \\ n 
    \end{pmatrix} &- \varepsilon^{2B} \partial_x \begin{pmatrix}
 \langle \omega' \omega \mathfrak{f}_{\beta, \gamma} , L_{\beta, \gamma}^{-1} (\omega' \omega \mathfrak{f}_{\beta, \gamma} ) \partial_x \beta  + L_{\beta, \gamma}^{-1} (\omega'  \mathfrak{f}_{\beta, \gamma} ) \partial_x \gamma \rangle \\
 \langle \omega' \mathfrak{f}_{\beta, \gamma} , L_{\beta, \gamma}^{-1} (\omega' \omega \mathfrak{f}_{\beta, \gamma} ) \partial_x \beta  + L_{\beta, \gamma}^{-1} (\omega'  \mathfrak{f}_{\beta, \gamma} ) \partial_x \gamma \rangle
    \end{pmatrix} \\
    &+ O(\varepsilon^{B +1+ \min (A, B, 1) }  + \varepsilon^{A+1}) = 0,
    \end{split}
\end{equation*}
or
\begin{equation*}
    \begin{split}
    \varepsilon^A \partial_t \begin{pmatrix}
        e \\ n
    \end{pmatrix} &+ \varepsilon^{2B} \partial_x \left ( D \begin{pmatrix}
        \partial_x \beta \\ \partial_x \gamma
    \end{pmatrix}  \right ) = O(\varepsilon^{B+1+\min(A,B,1) } + \varepsilon^{A+1} ), \\
\end{split}
\end{equation*}
where
$$
\boxed{D(\beta,\gamma) = - \begin{pmatrix}
\langle \omega \omega' \mathfrak{f}_{\beta, \gamma}, L_{\beta, \gamma}^{-1} (\omega' \omega \mathfrak{f}_{\beta, \gamma} ) \rangle && \langle \omega \omega' \mathfrak{f}_{\beta, \gamma}, L_{\beta, \gamma}^{-1} (\omega'  \mathfrak{f}_{\beta, \gamma} ) \rangle \\
\langle \omega' \mathfrak{f}_{\beta, \gamma}, L_{\beta, \gamma}^{-1} (\omega' \omega \mathfrak{f}_{\beta, \gamma} ) \rangle  && \langle \omega' \mathfrak{f}_{\beta, \gamma}, L_{\beta, \gamma}^{-1} (\omega'  \mathfrak{f}_{\beta, \gamma} ) \rangle
\end{pmatrix};
}
$$
note that this matrix is positive-definite by the Cauchy-Schwarz inequality.

Therefore, we have the following sub-cases: hydrodynamic equations arise when $\min(A, 2B) < B+1+\min(A,B,1)$, and we have the following subcases:
\begin{enumerate}
   \item $A> 2B$: we have $\partial_x \left ( D \begin{pmatrix}
        \partial_x \beta \\ \partial_x \gamma
    \end{pmatrix}  \right ) = 0 $,
    \item $A < 2B$: we have $\partial_t \begin{pmatrix}
        n \\ e
    \end{pmatrix}=0,$ and
    \item $A = 2B$, we have \begin{equation} \label{nonlindiff} \boxed{\partial_t \begin{pmatrix}
        e \\ n
    \end{pmatrix}  + \partial_x \left ( D(\beta,\gamma) \, \partial_x \begin{pmatrix} \beta \\ \gamma
    \end{pmatrix}  \right ) = 0 .}\end{equation}
\end{enumerate}

This equation has an interesting structure. First, it obviously conserves the total mass and the total energy
$$
\frac{\dd}{\dd t} \int n(t,x) \dd x = 0, \qquad \frac{\dd}{\dd t} \int e(t,x) \dd x = 0.
$$
Second, it enjoys a monotonic quantity, namely the total entropy
$$
\int S(\beta,\gamma) \dd x, \qquad \mbox{where} \qquad S(\beta,\gamma) = \int \log \mathfrak{f}_{\beta,\gamma} \dd x. 
$$
The entropy $S(\beta,\gamma)$ can equivalently be written $S(e,n)$, and we saw in Section \ref{section_first_look} that $\partial_e S(e,n) = \beta$ and $\partial_n S(e,n) = \gamma$. Therefore,
\begin{align*}
\frac{\dd}{\dd t} \int S(n,e) \dd x 
& = \int \left[ \partial_t e \cdot \partial_e S + \partial_t n \cdot \partial_n S \right] \dd x\\
& = \int \left[ \partial_t e \cdot \beta + \partial_t n \cdot \gamma \right] \dd x \\
& = - \int \left \langle \begin{pmatrix} \beta \\ \gamma \end{pmatrix}\,,\, \partial_x \left(  D \begin{pmatrix} \partial_x \beta \\ \partial_x \gamma \end{pmatrix} \right) \right \rangle \dd x \\
& = \int \left \langle \begin{pmatrix} \partial_x \beta \\ \partial_x \gamma \end{pmatrix}\,,\, D \begin{pmatrix} \partial_x \beta \\ \partial_x \gamma \end{pmatrix} \right \rangle \dd x \geq 0.
\end{align*}

Finally, this equation is of diffusion type and locally well-posed, since its leading order symbol is elliptic. Indeed, keeping only top order derivatives, the equation becomes
\begin{equation} 
\label{equationAD}
A \partial_t \begin{pmatrix} \beta \\ \gamma \end{pmatrix} =  D \partial_x^2 \begin{pmatrix} \beta \\ \gamma \end{pmatrix}.
\end{equation}
where
$$
A(\beta,\gamma) = - \frac{\partial (e,n)}{\partial (\beta,\gamma)} = 
\begin{pmatrix}
\int \omega^2 \mathfrak{f}_{\beta,\gamma}^2 \dd \xi & \int \omega \mathfrak{f}_{\beta,\gamma}^2 \dd \xi \\
\int \omega \mathfrak{f}_{\beta,\gamma}^2 \dd \xi & \int \mathfrak{f}_{\beta,\gamma}^2 \dd \xi
\end{pmatrix}
$$
(which is self-adjoint by the Cauchy-Schwarz inequality).
Since $A$ and $D$ are positive self-adjoint matrices, then $A^{-1}D$ is diagonalizable with positive eigenvaluess\footnote{For completeness, we will give a proof of this linear algebra statement: if $M$ and $N$ are two positive symmetric matrices, then $MN$ is diagonalizable with positive eigenvalues. Indeed, let $P = \sqrt M N \sqrt M$; it is a positive symmetric matrix. Furthermore, if $u \neq 0$ is an eigenvector of $P$ associated to the eigenvalue $\lambda$, or in other words $Pu = \lambda u$, then $MN \sqrt{M} u = \lambda \sqrt{M} u$. In other words, $\sqrt M$ maps the eigenvectors of $P$ to eigenvectors of $MN$ with the same eigenvalue, which gives the desired result.}, and \eqref{equationAD} is a locally well-posed diffusion equation.

\medskip

Finally, following \cite{AokiLukkSpohn}, we want to consider the case where $\gamma=0$. There are two reasons for this: first, we already saw that the case $\gamma=0$ was of particular physical relevance; and second, this leads to a simple scalar equation which helps understand the more complicated vector-valued problem above. The derivation is predicated on the reasonable assumption that the case $\gamma=0$ is dynamically stable.

Following the earlier computations, we find that
$$
\left( \int \omega^2 \mathfrak{f}_{\beta,0} \dd \xi \right) \partial_t \beta = - \partial_x \left( \langle \omega \omega' \mathfrak{f}_{\beta,0}, L_{\beta,0}^{-1} ( \omega' \omega  \mathfrak{f}_{\beta,0}) \rangle \partial_x \beta \right),
$$
which can be simplified to give
$$
\boxed{
\partial_t \beta = c_0 \beta^2 \partial_x (\beta \partial_x \beta), \qquad c_0 = - \langle \omega' , L_{1,0}^{-1} \omega' \rangle.
}
$$
This equation enjoys the conserved quantity (energy) $\int \beta^{-1} \dd x$ and the monotone quantity (entropy) $-\int \log \beta \dd x$.

Since $\gamma = 0$, we have the simple relation $e = \beta^{-1}$. As we saw in Section \ref{section_first_look}, the parameter $\beta$ can be interpreted as the inverse temperature, so that $e=T$ in this case. The equation on the temperature becomes
$$
\boxed{\partial_t T = c_0 \partial_x (T^{-3} \partial_x T),}
$$
with energy $\int T \dd x$ and entropy $\int \log T \dd x$.

\subsection{Perturbations of a global RJ equilibrium} 
\subsubsection{Summary of the results}
We start with the ansatz
\begin{equation*}
    f (t,x,\xi) = \mathfrak{f}_{\beta_0, \gamma_0}(\xi) (1 + \varepsilon g (t,x,\xi)), \qquad \beta = \beta_0, \gamma = \gamma_0, 
\end{equation*}
where $\beta_0$ and $\gamma_0$ are fixed and will be omitted from now on.
The equation satisfied by $g$ is
\begin{equation}
\label{eq_g}
\varepsilon^{A+1} \partial_t g + \varepsilon^{B+1} \omega' \partial_x g = \varepsilon L_{\beta,\gamma} g + \varepsilon^2 Q_{\beta,\gamma} (g) + \varepsilon^3 C_{\beta,\gamma}(g).
\end{equation}

and we set
\begin{equation}
\label{defVv}
\Pi_0 g(t,x,\xi) = b(t,x) \omega(\xi) \mathfrak{f}(\xi) + c(t,x) \mathfrak{f}(\xi) = V(t,x), \qquad 
v(t,x) = \begin{pmatrix} b(t,x) \\ c(t,x) \end{pmatrix}.
\end{equation}

Besides $b$ and $c$, another natural set of coordinates is provided by the mass and energy of the perturbation:
$$
e(t,x,\xi) = \int \mathfrak{f}(\xi) \omega(\xi) \Pi_0 g(t,x,\xi) \dd \xi, \qquad n(t,x,\xi) = \int \mathfrak{f}(\xi) \Pi_0 g(t,x,\xi) \dd \xi.
$$
The change of coordinates from $(b,c)$ to $(e,n)$ is given by
$$
\begin{pmatrix} e \\ n \end{pmatrix} = A \begin{pmatrix} b \\ c \end{pmatrix} \qquad \mbox{with} \qquad A = \begin{pmatrix}
\int \omega^2 \mathfrak{f}^2 \dd \xi & \int \omega \mathfrak{f}^2 \dd \xi \\
\int \omega \mathfrak{f}^2 \dd \xi & \int \mathfrak{f}^2 \dd \xi
\end{pmatrix}.
$$


In the following subsection, we will prove formally that, as $\varepsilon \to 0$, the orthogonal component $\Pi_1 g$ becomes negligible with respect to $\Pi_0 g$; and furthermore, the dynamics of $g$ is given by a hydrodynamic equation. The relevant scalings are as follows
\begin{itemize}
\item \underline{Linear diffusion.} if $0 <B < 2$ and $A = 2B$, then the limiting equation is
$$
\boxed{\partial_t \begin{pmatrix} e \\ n \end{pmatrix} = D \partial_x^2 \begin{pmatrix} b \\ c \end{pmatrix}, \qquad  D = - \begin{pmatrix}
            \langle \omega ' \omega \mathfrak{f}, L^{-1} \omega' \omega \mathfrak{f} \rangle & \langle \omega' \omega \mathfrak{f}, L^{-1} \omega' \mathfrak{f} \rangle \\ \langle \omega' \mathfrak{f}, L^{-1} \omega' \omega \mathfrak{f} \rangle & \langle \omega' \mathfrak{f}, L^{-1} \omega' \mathfrak{f} \rangle
        \end{pmatrix}.}
$$
This equation is the linearization around $(b,c)=(0,0)$ of \eqref{nonlindiff}. Thus it inherits its properties: it is a linear diffusion equation, which conserves the total mass and the total energy.

\medskip

\item \underline{Conservation law.} If $B>2$ and $A = B+2$, the limiting equation is
$$
\boxed{\partial_t \begin{pmatrix} e \\ n \end{pmatrix} + \partial_x P\begin{pmatrix} b \\ c \end{pmatrix} = 0, \qquad v \in \mathbb{R}^2}
$$
where 
\begin{equation}
\label{defP}
\boxed{
P \begin{pmatrix} b \\ c \end{pmatrix} = \begin{pmatrix}
\langle L^{-1}(\omega' \omega \mathfrak{f}), 2 Q(V,L^{-1} Q( V))+ C (V)\rangle \\ \langle L^{-1}(\omega' \mathfrak{f}), 2 Q(V,L^{-1} Q( V))+ C (V)\rangle \end{pmatrix} , \qquad V = b \omega \mathfrak{f} + c \mathfrak{f}.}
\end{equation}
Each of the two coordinates of $P\begin{pmatrix} b \\ c \end{pmatrix}$ is a homogeneous polynomial of degree 3 in the 
two variables $\beta,\gamma$.


\medskip

\item \underline{Nonlinear diffusion.} If $B=2$ and $A=4$, the limiting equation is
$$
\boxed{\partial_t v \begin{pmatrix} e \\ n \end{pmatrix} - D \partial_x^2 \begin{pmatrix} b \\ c \end{pmatrix} + \partial_x P \begin{pmatrix} b \\ c \end{pmatrix} = 0.}
$$

\medskip

\item \underline{Degenerate scalings.} All other ranges of $A$ and $B$ lead to degenerate limiting problems which are not evolution equations.
\end{itemize}

\subsubsection{Some useful identities}
We start by recording some identities which will be needed in the derivation.

\begin{itemize}
\item The quadratic term $Q(g)$ and the cubic term $C(g)$ belong to $(\operatorname{Ker} L)^\perp$ for any function $g$:
$$
\Pi_0 Q(g) = \Pi_0 C(g) = 0.
$$
For the cubic term, this is just a rewriting of the conservation of mass and energy
$$
\left \langle C(g) , \begin{pmatrix} \mathfrak{f} \\ \mathfrak{f} \omega \end{pmatrix} \right \rangle = \left \langle \mathfrak{f}^{-1} \mathcal{C}(\mathfrak{f} g) , \begin{pmatrix} \mathfrak{f} \\ \mathfrak{f} \omega \end{pmatrix} \right \rangle = \left \langle \mathcal{C}(\mathfrak{f} g) , \begin{pmatrix} 1 \\ \omega \end{pmatrix} \right\rangle = 0.
$$
For the quadratic term, this follows since
$$
Q(g) = \mathfrak{f}^{-1} \mathcal{C}(\mathfrak{f}(1+g)) - L g - C(g)
$$

\item The last identity (see \cite{bardos1991fluid}) we need is
\begin{equation}
\label{Q0L0}
Q(\Pi_0 g) = - L((\Pi_0 g)^2).
\end{equation}
To see why this is true, we start from 
$$
\mathcal{C} ( \mathfrak{f}') = 0, \quad \mbox{where} \quad \mathfrak{f}' = \frac{1}{(\beta_0 + \beta') \omega + (\gamma_0+\gamma')}.
$$
The expansion of $\mathfrak{f}'$ in $\beta',\gamma' \ll 1$ is
\begin{equation*}
\mathfrak{f'} = \mathfrak{f} \left (1 - \mathfrak{f} (\beta' \omega + \gamma' ) + \mathfrak{f}^2 (\beta' \omega + \gamma')^2 - \cdots  \right )
\end{equation*}
Expanding $\mathcal{C} ( \mathfrak{f}')$ to order 2, this means that
\begin{equation*}
\mathcal{Q}((\beta' \omega + \gamma')\mathfrak{f}) + \mathcal{L} ((\beta'\omega+ \gamma')^2 \mathfrak{f}^2 )= 0,
\end{equation*}
which is the desired identity.
\end{itemize}

\subsubsection{The heuristic derivation} \underline{Step 1: mass and energy balance.}  Multiplying the equation \eqref{eq_f} by $(1, \omega)$ and integrating with respect to $\xi$, we have 
    \begin{equation*}
        \varepsilon^{A} \partial_t \int \left ( \begin{matrix}
            1 \\ \omega 
        \end{matrix}\right ) \mathfrak{f} (1+\varepsilon g) \dd\xi + \varepsilon^{B} \partial_x \int \omega' \left ( \begin{matrix}
            1 \\ \omega 
        \end{matrix}\right ) \mathfrak{f} (1+\varepsilon g) \dd\xi = 0,
    \end{equation*}
    where the right-hand side vanishes since $1, \omega$ are collision invaraints. Next, as $\mathfrak{f}, \omega \mathfrak{f} \in \text{Ker} (L)$ and $\omega' \mathfrak{f}, \omega' \omega \mathfrak{f} \in \text{Ker} (L)^\perp$, we have the following:
    \begin{equation*}
        \varepsilon^{A} \partial_t  \left \langle  \left ( \begin{matrix}
            1 \\ \omega 
        \end{matrix}\right ) \mathfrak{f}, \Pi_0 g \right \rangle + \varepsilon^{B} \partial_x  \left \langle  \left ( \begin{matrix}
            1 \\ \omega 
        \end{matrix}\right ) \omega' \mathfrak{f}, \Pi_1 g \right \rangle = 0.
    \end{equation*}

\medskip

\noindent 
\underline{Step 2: recovering $\Pi_1 g$ through the equation.}
Next, it follows from the equation \eqref{eq_g} for $g$ that
    \begin{equation*}
    \begin{split}
 \Pi_1 g &= \varepsilon^A \partial_t L^{-1} \Pi_1 g + \varepsilon^B \partial_x L^{-1} ( \omega' \Pi_0 g ) + \varepsilon^B \partial_x L^{-1} \Pi_1 ( \omega' \Pi_1 g ) - \varepsilon L^{-1} (Q (g) + \varepsilon C (g) ).
    \end{split}
    \end{equation*}
Recalling that $\Pi_0 g =(b \omega + c) \mathfrak f$ and plugging the formula for $\Pi_1 g$ above in the mass and energy balance equation, we have
\begin{equation*}
\begin{split}
        &\varepsilon^{A} \partial_t   \left \langle  \left ( \begin{matrix}
            1 \\ \omega 
        \end{matrix}\right ) \mathfrak{f}, (b \omega + c) \mathfrak{f}  \right \rangle + \varepsilon^{2B} \partial_x  \left \langle  \left ( \begin{matrix}
            1 \\ \omega 
        \end{matrix}\right ) \omega' \mathfrak{f}, \partial_x L^{-1} (\omega' b \omega + c) \mathfrak{f} ) \right \rangle \\
        &- \varepsilon^{B+1} \left \langle  \left ( \begin{matrix}
            1 \\ \omega 
        \end{matrix}\right ) \omega' \mathfrak{f}, \partial_x L^{-1} (Q(g) + \varepsilon  C (g) ) \right \rangle  \\
        &+ \varepsilon^{A+B}   \partial_x  \left \langle  \left ( \begin{matrix}
            1 \\ \omega 
        \end{matrix}\right ) \omega' \mathfrak{f}, \partial_t L^{-1} \Pi_1 g \right \rangle + \varepsilon^{2B} \partial_x  \left \langle  \left ( \begin{matrix}
            1 \\ \omega 
        \end{matrix}\right ) \omega' \mathfrak{f}, \partial_x L^{-1} \Pi_1 (\omega' \Pi_1 g ) \right \rangle 
          =0.
    \end{split}
    \end{equation*}
    Then we use the following heuristic: from the equation for $\Pi_1 g$, we may write
\begin{equation}
\label{Pi1g}
\Pi_1 g = \varepsilon^{m} h, \qquad m = \min (A, B, 1),
\end{equation}
for a function $h$ of size $O(1)$. Thus we see that 
    \begin{equation}
\label{grosseequation}
        \begin{split}
& \varepsilon^A \partial_t \begin{pmatrix} b \langle \mathfrak{f}, \omega \mathfrak{f} \rangle + c \langle \mathfrak{f}, \mathfrak{f} \rangle \\ 
b \langle \omega \mathfrak{f}, \omega \mathfrak{f} \rangle + c \langle \omega \mathfrak{f}, \mathfrak{f} \rangle \end{pmatrix}
+ \varepsilon^{2B} \partial_x^2 \begin{pmatrix} ( b \langle \omega' \mathfrak{f}, L^{-1} (\omega' \omega \mathfrak{f} ) \rangle + c \langle \omega' \mathfrak{f}, L^{-1} (\omega' \mathfrak{f} ) \rangle \\
 b \langle \omega  \omega' \mathfrak{f}, L^{-1} (\omega' \omega \mathfrak{f} ) \rangle + c \langle \omega \omega' \mathfrak{f}, L^{-1} (\omega' \mathfrak{f} ) \rangle \end{pmatrix} \\
&- \varepsilon^{B+1}\partial_x \begin{pmatrix}   \langle \omega' \mathfrak{f}, L^{-1} Q (g)  \rangle \\ \langle \omega' \omega  \mathfrak{f}, L^{-1} Q (g) \rangle \end{pmatrix} -  \varepsilon^{B+2} \partial_x \begin{pmatrix} \langle \omega' \mathfrak{f},L^{-1} C (g) \rangle   \\  \langle \omega' \omega \mathfrak{f}, L^{-1}  C (g) \rangle \end{pmatrix} \\
&+ O(\varepsilon^{\min( A+B+m, 2B+m)} ) = 0.
\end{split}
\end{equation}

\medskip
\noindent \underline{Step 3: expanding the quadratic and cubic terms}
We can now expand with the help of \eqref{Pi1g} and \eqref{Q0L0}
\begin{align*}
Q(g) = Q(g,g) & = Q(\Pi_0 g) + 2 Q(\Pi_0 g, \Pi_1 g) + Q(\Pi_1 g, \Pi_1 g) \\
& = -L (\Pi_0 g)^2 + 2 Q(\Pi_0 g, \Pi_1 g) + O(\varepsilon^{2m}),
\end{align*}
so that, using now that $\int \omega' F(\omega) \dd \xi = 0$ for a smooth $F$, and denoting 
$$
V = \Pi_0 g = b\omega \mathfrak{f} + c \mathfrak{f},
$$
we have
\begin{equation}
\label{eqforQ}
\left \langle \begin{pmatrix} \omega' \mathfrak{f} \\ \omega' \omega \mathfrak{f} \end{pmatrix} , L^{-1} Q(g) \right \rangle = \varepsilon^m \left \langle \begin{pmatrix} \omega' \mathfrak{f} \\ \omega' \omega \mathfrak{f} \end{pmatrix} , L^{-1} Q(V,h) \right \rangle + O(\varepsilon^{2m}).
\end{equation}
while
\begin{equation}
\label{eqforC}
C(g) = C(V) + O(\varepsilon^m).
\end{equation}
Plugging \eqref{eqforQ} and \eqref{eqforC} into \eqref{grosseequation}, and recalling the definition of the matrix $D$, we find
\begin{equation} \label{asymptoticeqn}
\begin{split}
\varepsilon^A \partial_t \begin{pmatrix} e \\ n \end{pmatrix} &- \varepsilon^{2B} D \,\partial_x^2 \begin{pmatrix} b \\ c \end{pmatrix} \\ 
&- \varepsilon^{B+1+m} \partial_x  \left ( \begin{matrix}  \langle \omega\omega ' \mathfrak{f}, L^{-1} Q (V,h) \rangle \\ \langle \omega '   \mathfrak{f}, L^{-1} Q (V,h)  ) \rangle
\end{matrix} \right ) - \varepsilon^{B+2} \partial_x  \left ( \begin{matrix}  \langle \omega\omega ' \mathfrak{f}, L^{-1}  C (V) ) \rangle \\ \langle \omega '   \mathfrak{f}, L^{-1} C (V)  ) \rangle
\end{matrix} \right )\\
& \qquad \qquad = O(\varepsilon^{\min(A+B+m, 2B+m, {B+1+2m,B+2+m} )} )
\end{split}
\end{equation}

\medskip
\noindent \underline{Step 4: distinguishing between the different regimes.}
Since 
$$
\min (A, 2B, B+1+m, B+2) < \min (A+B+m, 2B+m, B+1+2m,B+2+m),
$$ 
all possible leading order terms are on the left-hand side.
Now we need to distinguish several cases depending on the relative size of the powers of $\varepsilon$ on the left-hand side. We take an extra care for the $O(\varepsilon^{B+1+m})$-size term, as we did not fully represent it with $\beta$ and $\gamma$ yet. Recalling $m = \min (A, B, 1)$, we have the following cases.

\begin{enumerate}
\item $A<B$. In this case, $\min (A, 2B, B+1+m, B+2) = A$, and the leading order equation is $\partial_t A v = 0.$

\medskip

\item $A \ge B$. In this case, $m = \min (B, 1)$. We have several sub-cases.

\medskip

\begin{enumerate}
\item $B \le 1$. In this case, $B+1+m = 2B+1$, and also $2B < B+2$. Therefore, $\min (A, 2B, B+1+m, B+2) = \min (A, 2B). $ Thus, if $A=2B$, we have $\partial_t A v - \partial_x^2 D v = 0,$ if $A>2B$, we have $\partial_x^2 D v = 0$, and if $A < 2B$, we have $\partial_t A v = 0.$

\medskip

\item $B > 1$. In this case, $B+1+m = B+2$. Moreover, the leading order term for $\Pi_1 g$ is $-\varepsilon L^{-1} \mathcal{Q} (V) . $ Plugging this into the second line of \eqref{asymptoticeqn}, we get
                \begin{equation*}
        \begin{split}
\varepsilon^A \partial_t \begin{pmatrix} e \\ n \end{pmatrix} - \varepsilon^{2B} D \, \partial_x^2 \begin{pmatrix} b \\ c \end{pmatrix} + \varepsilon^{B+2} \partial_x P \begin{pmatrix} b \\ c \end{pmatrix} = O(\varepsilon^{\min (A+B+1, 2B+1 )} ),
        \end{split}
    \end{equation*}
where $P$ was defined in \eqref{defP}.

    Thus, we have the following subcases:
    \begin{enumerate}
        \item $B>2$. We have $2B > B+2$. Thus if $A > B+2$, we have $\partial_x P^3 (v) = 0$, if $A = B+2$, $\partial_t A v + \partial_x P^3 (v) = 0,$ and if $A < B+2$, again we have $\partial_t A v = 0.$
        \item $B=2$. In this case, $2B = B+2$, and thus if $A > 2B$, we have $-\partial_x^2 D v + \partial_x P^3 (v) = 0$, if $A = 2B$, we have $\partial_t A v - \partial_x^2 D v + \partial_x P^3 (v) = 0, $ and if $A<2B$, we have $\partial_t A v = 0.$
        \item $B<2$. We have $2B < B+2$, and thus if $A>2B$, we have $-\partial_x^2 D v = 0$, if $A=2B$, we have $\partial_t A v - \partial_x^2 D v = 0,$ and if $A<2B$, we have $\partial_t A v = 0.$
    \end{enumerate}
        \end{enumerate}
    \end{enumerate}

\subsection{Hilbert expansion}

Yet a question remains: can we justify this scaling analysis: for this purpose, we introduce one candidate, Hilbert expansion approach. Once we have enough machinery, in particular, $L^p, p>2$ control of the solution, this can be put rigorously. At this point, we keep our discussion formal.

We first consider the perturbation around slowly varying R-J profiles. We keep $A=2, B=1$, and consider an ansatz with further correctors:
\begin{equation*}
    f = \mathfrak{f}_{\beta, \gamma}  + \varepsilon \mathfrak{f}_{\beta, \gamma} g_1 + \varepsilon^2 \mathfrak{f}_{\beta, \gamma} g_2 + \varepsilon^\gamma \mathfrak{f}_{\beta, \gamma} r,
\end{equation*}
and plug in to the equation
\begin{equation*}
    \varepsilon \partial_t f + \omega' \partial_x f - \frac{1}{\varepsilon} \mathcal{C} (f) = 0
\end{equation*}
to obtain
\begin{equation*}
    \begin{split}
        &\varepsilon^{\gamma+1} \partial_t (\mathfrak{f}_{\beta, \gamma} r ) + \varepsilon^\gamma \omega' \partial_x (\mathfrak{f}_{\beta,\gamma} r) - \varepsilon^{\gamma-1} \mathcal{C} (\mathfrak{f}_{\beta,\gamma}, \mathfrak{f}_{\beta,\gamma}, \mathfrak{f}_{\beta,\gamma}r ) - \varepsilon^\gamma \mathcal{C} ( \mathfrak{f}_{\beta,\gamma}, \mathfrak{f}_{\beta,\gamma}g_1, \mathfrak{f}_{\beta,\gamma}r) \\
        &-\varepsilon^{\gamma+1} (\text{linear in } \mathfrak{f}_{\beta,\gamma} r) - \varepsilon^{2\gamma -1} (\text{ quadratic in } \mathfrak{f}_{\beta,\gamma} r ) - \varepsilon^{3\gamma -1} (\text{ cubic in } \mathfrak{f}_{\beta,\gamma} r ) \\ 
        &- \varepsilon^{-1} \mathcal{C} (\mathfrak{f}_{\beta,\gamma}) + (\omega' \partial_x \mathfrak{f}_{\beta,\gamma} - \mathcal{C} (\mathfrak{f}_{\beta,\gamma}, \mathfrak{f}_{\beta,\gamma}, \mathfrak{f}_{\beta,\gamma} g_1) ) \\
        &+\varepsilon (\partial_t \mathfrak{f}_{\beta,\gamma} + \omega' \partial_x (\mathfrak{f}_{\beta,\gamma} g_1)  - \mathcal{C} (\mathfrak{f}_{\beta,\gamma},\mathfrak{f}_{\beta,\gamma}, \mathfrak{f}_{\beta,\gamma}g_2) - \mathcal{C} (\mathfrak{f}_{\beta,\gamma}, \mathfrak{f}_{\beta,\gamma}g_1, \mathfrak{f}_{\beta,\gamma}g_1) ) \\
        &+\varepsilon^{2} (\partial_t (\mathfrak{f}_{\beta,\gamma} g_1) + \omega' \partial_x (\mathfrak{f}_{\beta,\gamma} g_2) - \mathcal{C} (\mathfrak{f}_{\beta,\gamma}, \mathfrak{f}_{\beta,\gamma}g_1, \mathfrak{f}_{\beta,\gamma}g_2) - \mathcal{C} (\mathfrak{f}_{\beta,\gamma} g_1, \mathfrak{f}_{\beta,\gamma} g_1, \mathfrak{f}_{\beta,\gamma}g_1 ) ) \\
        &+ O(\varepsilon^3).
    \end{split}
\end{equation*}
Here we slightly abused notation, and whenever we write $\mathcal{C}(f, g, h)$, then it actually represents summation over all terms obtained by permutations of the arguments, i.e. $\mathcal{C}(g,f,h), \mathcal{C}(h,g,f)$, and so on.

Now we see that if source terms are too large, the estimate for $r$ will not remain stable. Thus, we determine $g_1, g_2$ so that source terms cancel each other. 
\begin{enumerate}
    \item $O(\varepsilon^{-1}$: $\mathfrak{f}_{\beta,\gamma}^{-1} \mathcal{C} (\mathfrak{f}_{\beta,\gamma}) = 0$ immediately. 
    \item $O(1)$: this amounts to 
    \begin{equation*}
        \frac{\omega' \partial_x \mathfrak{f}_{\beta,\gamma}}{\mathfrak{f}_{\beta,\gamma}} - L_{\beta, \gamma} g_1 = 0,
    \end{equation*}
    which is satisfied by
    \begin{equation*}
        \Pi_1 g_1 = L_{\beta, \gamma} ^{-1} \left ( \frac{ \omega' \partial_x \mathfrak{f}_{\beta,\gamma}}{\mathfrak{f}_{\beta,\gamma}} \right ).
    \end{equation*}
    \item $O(\varepsilon^{1} )$: this amounts to
    \begin{equation*}
        \begin{split}
            &\Pi_0 \left ( \frac{ \partial_t \mathfrak{f}_{\beta,\gamma}}{\mathfrak{f}_{\beta,\gamma}} + \omega' \frac{\partial_x ( \mathfrak{f}_{\beta,\gamma} \Pi_0 g_1)  } {\mathfrak{f}_{\beta,\gamma}} + \omega' \frac{\partial_x ( \mathfrak{f}_{\beta,\gamma} \Pi_1 g_1) }{\mathfrak{f}_{\beta,\gamma}} \right )  = 0, \\
            &\Pi_1 \left ( \frac{ \partial_t \mathfrak{f}_{\beta,\gamma}}{\mathfrak{f}_{\beta,\gamma}} + \omega' \frac{\partial_x ( \mathfrak{f}_{\beta,\gamma} \Pi_0 g_1)  } {\mathfrak{f}_{\beta,\gamma}} + \omega' \frac{\partial_x ( \mathfrak{f}_{\beta,\gamma} \Pi_1 g_1) }{\mathfrak{f}_{\beta,\gamma}}  \right ) - L_{\beta, \gamma} g_2 - Q(g_1) = 0,
        \end{split}
    \end{equation*}
    however, since $\Pi_0 \omega' \frac{\partial_x ( \mathfrak{f}_{\beta,\gamma} \Pi_0 g_1)  } {\mathfrak{f}_{\beta,\gamma}} = 0, \Pi_1 \frac{ \partial_t \mathfrak{f}_{\beta,\gamma}}{\mathfrak{f}_{\beta,\gamma}} = 0, $
    we end up with
    \begin{equation*}
        \begin{split}
            &\Pi_0 \left (\frac{ \partial_t \mathfrak{f}_{\beta,\gamma}}{\mathfrak{f}_{\beta,\gamma}} - \omega' \frac{1}{\mathfrak{f}_{\beta, \gamma}} \partial_x \left ( \mathfrak{f}_{\beta, \gamma} (-L_{\beta, \gamma} )^{-1} \left ( \frac{\omega' \partial_x \mathfrak{f}_{\beta, \gamma}}{\mathfrak{f}_{\beta, \gamma}} \right ) \right ) \right ) = 0, \\
            &\Pi_1 g_2 = L_{\beta, \gamma} ^{-1} Q (g_1) + L_{\beta, \gamma} ^{-1} \Pi_1 \left ( \frac{ \partial_t \mathfrak{f}_{\beta,\gamma}}{\mathfrak{f}_{\beta,\gamma}} + \omega' \frac{\partial_x ( \mathfrak{f}_{\beta,\gamma} \Pi_0 g_1)  } {\mathfrak{f}_{\beta,\gamma}} + \omega' \frac{\partial_x ( \mathfrak{f}_{\beta,\gamma} \Pi_1 g_1) }{\mathfrak{f}_{\beta,\gamma}}  \right )
        \end{split}
    \end{equation*}
    The first equation is exactly the equation we had before.
    \item $O(\varepsilon^2)$: again, by dividing by $\mathfrak{f}_{\beta,\gamma}$ and taking $\Pi_0$, we have an evolution equation for $\Pi_0 g_1$; on the other hand, $\Pi_1$ part cannot be canceled out, as we already have determined $g_1$ and $\Pi_1 g_2$ - we do not have a term like $L g_3$.  
\end{enumerate}
To summarize, the remainder equation becomes
\begin{equation}
\begin{split}
     \frac{1}{\mathfrak{f}_{\beta, \gamma}} \partial_t (\mathfrak{f}_{\beta, \gamma} r) &+ \frac{1}{\varepsilon} \frac{1}{\mathfrak{f}_{\beta, \gamma}} \omega' \partial_x (\mathfrak{f}_{\beta, \gamma} r ) - \frac{1}{\varepsilon^2 }  L_{\beta, \gamma} r  - \frac{1}{\varepsilon} \frac{1}{\mathfrak{f}_{\beta, \gamma}} \mathcal{C} (\mathfrak{f}_{\beta, \gamma}, \mathfrak{f}_{\beta, \gamma}g_1, \mathfrak{f}_{\beta, \gamma}r ) \\
     &+ Lin(\mathfrak{f}_{\beta, \gamma}r) + \varepsilon^{\gamma-2 } Quad(\mathfrak{f}_{\beta, \gamma}r ) + \varepsilon^{2\gamma -2} Cub(\mathfrak{f}_{\beta, \gamma}r) \\
     &+ O(\varepsilon^2) \operatorname{Ker} L_{\beta,\gamma}^\perp  + O(\varepsilon^3) = 0.
\end{split}
\end{equation}
Now since every term originated from collision operators are orthogonal to $\operatorname{Ker} L_{\beta,\gamma}$, we might be able to control these terms using the coercivity of $L$ on $\operatorname{Ker} \operatorname{L}_{\beta,\gamma} ^\perp$ and appropriate higher $L^p$ norm control: on the other hand, the issue with the convection term $\frac{1}{\varepsilon} \frac{1}{\mathfrak{f}_{\beta, \gamma} } \omega' \partial_x (\mathfrak{f}_{\beta, \gamma}  r )$ seems more serious: we  would need a coercivity of the modified operator $\mathfrak{f}_{\beta, \gamma} L_{\beta,\gamma} \mathfrak{f}_{\beta, \gamma}^{-1}$ to control this term. 

\section{Hydrodynamic limits in the degenerate case}
\label{sec:Hydro_limit_degenerate}

This section is dedicated to hydrodynamic limits in the degenerate case where the linearized operator for \eqref{KWE} around RJ is not invertible. This case is considerably more delicate than the non-degenerate case, and indeed we will be content with treating the linearized problem. As we shall see, the interesting consequence of this degeneracy is a fractional (also called anomalous) diffusion equation.

The following exposition follows \cite{mellet2015anomalous}, and we only provide a  rough sketch, refering to this article for a meticulous treatment. We also refer to \cite{MelletMouhotMischler11} and references therein, for the derivation of fractional diffusion equations from kinetic equations such as the \eqref{KWE} studied here. Also the more recent works  \cite{MouhotBouin22, BKM24} for a quantitative abstract method towards fractional diffusion equations from linear kinetic models, where cases with one and several collision invariants respectively, are treated.  

\subsection{The case of FPU-$\beta$, $\mathfrak{f} = \omega^{-1}$} The general method which was followed in the non-degenerate case ceases to work in the degenerate case. We will rely on the Fourier-Laplace transform to go around this difficulty.

We start from \eqref{eq_g} with the usual ansatz $f = \mathfrak{f} (1+\varepsilon g).$ Instead of using the mass and energy balance equation (since the approach requires that the leading order term of $\Pi_1 g$ can be written in terms of $\Pi_0 g$, which is difficult to even say if $L$ is not invertible) one directly starts from the equation for $g$ and takes Fourier-Laplace transform (that is, Fourier in $x$ and Laplace in $t$ variable). We have
\begin{equation*}
    \varepsilon^A \partial_t g + \varepsilon^B \omega' \partial_x g = Lg + \varepsilon R,
\end{equation*}
where $R$ consists of quadratic and cubic nonlinearities. We will be working on the linearized problem and set therefore $R=0$.

Two key properties will be needed: first, as we saw in Section \ref{section_rigorous_FPU}, the linearized operator can be written as
$$
L = - A + K,
$$
where $K$ is a compact operator and $A$ is the multiplication operator by the collision frequency $a$:
$$
[A g](\xi) = a(\xi) g(\xi), \qquad \mbox{with} \qquad a(\xi) \simeq |\xi|^{\frac 5 3} \quad \mbox{as $\xi \to 0$.}
$$
The above $\simeq$ notation means that $\lim_{\xi \to 0} \frac{a(\xi)}{|\xi|^{\frac 5 3}}=1$; this is proved in \cite{LukkarinenSpohn2008} with a multiplicative constant, which can be set to $1$ by rescaling time.

Second, we will admit the bound
$$
\left\| \varepsilon^{A/2} \Pi_1 g \right\|_{L^2 (a\dd \xi dx dt )} = O(1)
$$
uniformly in $\varepsilon$, which is proved in \cite{mellet2015anomalous}.

\medskip

\noindent
\underline{Step 1: The Fourier-Laplace transform.} Setting
\begin{equation*}
    \widehat g (p,k,\xi) := \int_0 ^\infty \int_{\mathbb{R}} e^{-pt}  e^{-ikx} g(t,x,\xi) dxdt,
\end{equation*}
and we treat $a$ as the main part of $L$ and $K$ as a perturbation:
\begin{equation*}
    \varepsilon^{A} (p \widehat g - \widehat g_0 ) + \varepsilon^B ik \omega' \widehat g + a \widehat g = K \widehat g.
\end{equation*}
Therefore, we have
\begin{equation} \label{FLtransform}
    \widehat g (p, k, \xi) = \frac{\varepsilon^A}{\varepsilon^A p + \varepsilon^B ik\omega' + a(\xi)} \widehat g_0 + \frac{1}{\varepsilon^A p + \varepsilon^B ik\omega' + a(\xi)} \int K (\xi, \xi') \widehat g (p, k, \xi')\dd \xi'.
\end{equation}

\medskip

\noindent
\underline{Step 2: modified mass and energy balance.} Let us first convince ourselves that the standard approach fails in the case $\mathfrak{f} = \omega^{-1}$. Since we want an equation for the hydrodynamic part of $g$, we may project $\widehat g$ to $\text{Ker} L = \operatorname{Span} \{ \mathfrak{f}, \omega \mathfrak{f} \}$: we have
\begin{equation*}
\begin{split}
    \int & \left ( \begin{matrix}
        \omega \\ 1 
    \end{matrix} \right ) \mathfrak{f}(\xi) \widehat g (p, k , \xi)\dd \xi = \int  \left ( \begin{matrix}
        \omega \\ 1 
    \end{matrix} \right ) \mathfrak{f}(\xi) \frac{\varepsilon^A}{\varepsilon^A p + \varepsilon^B ik\omega' + a(\xi)} \widehat g_0(\xi)\dd \xi \\
    &+ \int \left ( \begin{matrix}
        \omega \\ 1 
    \end{matrix} \right ) \mathfrak{f}(\xi) \frac{1}{\varepsilon^A p + \varepsilon^B ik\omega' + a(\xi)} \int K (\xi, \xi') \widehat g (p, k, \xi')\dd \xi'\dd \xi 
\end{split}
\end{equation*}
and writing $\widehat g = (\widehat \beta \omega + \widehat \gamma) \mathfrak{f} $, formally we have an equation for
\begin{equation*}
    \begin{split}
        A \widehat {v} = \left ( \begin{matrix}
            \widehat \beta \langle \omega \mathfrak{f}, \omega \mathfrak{f} \rangle + \widehat \gamma \langle \omega \mathfrak{f}, \mathfrak{f} \rangle \\\widehat\beta \langle  \mathfrak{f}, \omega \mathfrak{f} \rangle + \widehat \gamma \langle  \mathfrak{f}, \mathfrak f \rangle 
        \end{matrix} \right ) .
    \end{split}
\end{equation*}
However, this approach fails when $\mathfrak{f} = \omega^{-1}$, as $\langle \omega \mathfrak{f}, \mathfrak{f} \rangle , \langle \mathfrak{f}, \mathfrak{f} \rangle = \infty$.  

Instead, we will project the equation on $a \mathfrak{f}, a \omega \mathfrak{f}$, or in other words integrate \eqref{FLtransform} against $a(\xi) \mathfrak{f}(\xi) \begin{pmatrix} \omega(\xi) \\ 1 \end{pmatrix}$, keeping in mind that, since $L \mathfrak f = L (\omega \mathfrak f) = 0$, 
$$
a(\xi) \mathfrak{f}(\xi) \begin{pmatrix} \omega(\xi) \\ 1 \end{pmatrix} = \int K(\xi,\xi') \mathfrak{f}(\xi') \begin{pmatrix} \omega(\xi') \\ 1 \end{pmatrix} \dd \xi'.
$$
Multiplying \eqref{FLtransform} by $a(\xi) \mathfrak{f}(\xi) \begin{pmatrix} \omega(\xi) \\ 1 \end{pmatrix}$ and using the above identity together and the self-adjointness of $K$, we find

\begin{equation*}
\begin{split}
&\int \mathfrak{f}(\xi) \left ( \begin{matrix} \omega(\xi) \\ 1 \end{matrix} \right ) \left (  \frac{a(\xi)}{\varepsilon^A p + \varepsilon^B ik \omega'(\xi) + a(\xi)} - 1 \right ) K(\widehat g ) (\xi)\dd \xi \\
&= - \int \mathfrak{f}(\xi) \left ( \begin{matrix} \omega(\xi) \\ 1 \end{matrix} \right )  \frac{\varepsilon^A a (\xi) }{\varepsilon^A p + \varepsilon^B i k \omega'(\xi) + a(\xi) } \widehat g_0 (\xi)\dd \xi. 
\end{split}
\end{equation*}

\medskip

\noindent
\underline{Step 3: getting rid of $\Pi_1 g$.}
The next key idea is the following: since 
$$
K (\widehat g) = K (\Pi_0\widehat g) + K (\Pi_1 \widehat g) = a \Pi_0\widehat g + K ( \Pi_1 \widehat g),
$$
and $\Pi_1 \widehat g = O(\varepsilon^{A/2})$, we can obtain an equation for 
$$
\Pi_0\widehat g = \widehat \beta \omega \mathfrak{f} + \widehat \gamma \mathfrak{f}
$$
by ignoring the contribution from $\Pi_1 g$ (in fact, one needs to be more careful since $\Pi_1 \widehat g$ lives in $L^2 (a\dd \xi dx)$, where $a$ is degenerate at $\xi=0$, and therefore $\Pi_1g$ may be singular near $\xi=0$. For the details, see \cite{mellet2015anomalous}.) Moreover, we obtain an additional factor $a$, which mediates the singularity near $\xi=0$ (in particular for $\mathfrak{f}^2 = 1/\omega^2$ term). 
More precisely, we have 
\begin{equation} \label{FL_result}
    \begin{split}
        \varepsilon^{-A} &\left ( \begin{matrix} \widehat \beta \int \omega ^2 \mathfrak{f}^2 a \left ( \frac{a}{\varepsilon^A p + \varepsilon^B i k \omega' + a} - 1\right )\dd \xi + \widehat \gamma \int \omega \mathfrak{f}^2 a \left ( \frac{a}{\varepsilon^A p + \varepsilon^B i k \omega' + a} - 1\right )\dd \xi  \\
        \widehat \beta \int \omega  \mathfrak{f}^2 a \left ( \frac{a}{\varepsilon^A p + \varepsilon^B i k \omega' + a} - 1\right )\dd \xi + \widehat \gamma \int  \mathfrak{f}^2 a \left ( \frac{a}{\varepsilon^A p + \varepsilon^B i k \omega' + a} - 1\right )\dd \xi
        \end{matrix} \right ) \\
        &= -  \left ( \begin{matrix} \widehat \beta_0 \int \omega ^2 \mathfrak{f}^2 a \left ( \frac{a}{\varepsilon^A p + \varepsilon^B i k \omega' + a} \right )\dd \xi + \widehat \gamma_0 \int \omega  \mathfrak{f}^2 a \left ( \frac{a}{\varepsilon^A p + \varepsilon^B i k \omega' + a} \right )\dd \xi  \\
        \widehat \beta_0 \int \omega  \mathfrak{f}^2 a \left ( \frac{a}{\varepsilon^A p + \varepsilon^B i k \omega' + a} \right )\dd \xi + \widehat \gamma_0 \int  \mathfrak{f}^2 a \left ( \frac{a}{\varepsilon^A p + \varepsilon^B i k \omega' + a} \right )\dd \xi
        \end{matrix} \right ).
    \end{split}
\end{equation}

\medskip

\noindent \underline{Step 4: limits of the integrals} The problem boils down to finding the $\varepsilon \to 0$ limit for the six integrals above. For the integrals on the right-hand side, the Lebesgue dominated convergence theorem immediately gives convergence to constants.

For the integrals on the left-hand side, we carefully investigate their asymptotic limits. First, we see that
\begin{equation*}
    \begin{split}
        \varepsilon^{-A} &\int \omega ^2 \mathfrak{f}^2 a \left ( \frac{a}{\varepsilon^A p + \varepsilon^B i k \omega' + a} - 1\right )\dd \xi \\
        &= -p \int \omega^2 \mathfrak{f}^2 \frac{a}{\varepsilon^A p} \frac{ (\varepsilon^A p (\varepsilon^A p + a) + \varepsilon^{2B} k^2 (\omega')^2  ) + i \varepsilon^B k \omega' a}{(\varepsilon^A p + a)^2 + (\varepsilon^B k \omega')^2 }\dd \xi,
\end{split}
\end{equation*}
which can be decomposed as
\begin{align}
& \label{oiecendree1} = -p \int \omega^2 \mathfrak{f}^2 \frac{a(\varepsilon^A p + a)}{(\varepsilon^A p + a)^2 + (\varepsilon^B k \omega')^2}\dd \xi  \\
& \label{oiecendree2} \qquad - \int \omega^2 \mathfrak{f}^2 \frac{\varepsilon^{2B-A}a k^2 (\omega')^2}{(\varepsilon^A p + a)^2 + (\varepsilon^B k \omega')^2}\dd \xi  \\
& \label{oiecendree3} \qquad - \int \omega^2 \mathfrak{f}^2 \frac{i \varepsilon^{B-A} k \omega' a^2}{(\varepsilon^A p + a)^2 + (\varepsilon^B k \omega')^2}\dd \xi.
\end{align}
It is clear that the inner integral of \eqref{oiecendree1} converges to $1$ as $\varepsilon \to 0$, hence this term goes to $-p$. Turning to \eqref{oiecendree3}, it is zero by oddness of $\omega'$ and evenness of $a$. There remains \eqref{oiecendree2}, for which  we recall that $|a| \simeq |\xi|^{5/3}$ and $\omega \simeq |\xi|$ as $\xi \to 0$. Now for $\xi$ away from 0 (say $|\xi| > \delta$), the leading order term reads
\begin{equation*}
    -\int \omega^2 \mathfrak{f}^2 \frac{\varepsilon^{2B- A} k^2 (\omega')^2 }{a}\dd \xi
\end{equation*}
and for $2B - A > 0$, this vanishes as $\varepsilon \to 0$. On the other hand, for $\xi \sim 0$, we have $|\omega'(\xi)| \eqsim 1$, and thus \eqref{oiecendree2} is approximately (assuming that $A-B >0$)
\begin{equation*}
\begin{split}
& -\varepsilon^{-A} \int_{|\xi| \ll 1} \omega^2 \mathfrak{f}^2 \frac{k^2 }{(\varepsilon^{-B}|\xi|^{5/3} )^2 + k^2} |\xi|^{5/3}\dd \xi \\
    & \qquad \qquad = - \varepsilon^{-A} \int_{|\xi| \ll 1} \omega^2 \mathfrak{f}^2 \frac{1}{(\varepsilon^{-B} k^{-1} |\xi|^{5/3})^2 + 1 } (\varepsilon^{-B} k^{-1} |\xi|^{5/3} ) \varepsilon^B k \dd \xi \\
    & \qquad \qquad  \simeq -\varepsilon^{-A} \int_{|w| \ll \varepsilon^{-B}k^{-1}}\frac{w}{w^2+1} \varepsilon^{8B/5} |k|^{8/5} \frac{3}{5} w^{- 2/5} \dd w .
\end{split}
\end{equation*}
(after setting $w=\varepsilon^{-B} |k|^{-1} |\xi|^{\frac 53} \operatorname{sign} \xi$).
Thus, when $A = 8B/5$, we find that
$$
\eqref{oiecendree2} \overset{\varepsilon \to 0}{\longrightarrow} - \kappa_1 |k|^{8/5} \qquad \mbox{with} \qquad \kappa_1 = \frac{6}{5} \int_0 ^\infty \frac{w^{3/5}}{w^2+1} \dd w.
$$

For the other terms, it turns out that the integral becomes more singular: we have
\begin{equation*}
    \begin{split}
        \varepsilon^{-A+3B/5} &\int \omega  \mathfrak{f}^2 a \left ( \frac{a}{\varepsilon^A p + \varepsilon^B i k \omega' + a} - 1\right )\dd \xi \\
        &= -\varepsilon^{3B/5} p \int \omega \mathfrak{f}^2 \frac{a}{\varepsilon^A p} \frac{ (\varepsilon^A p (\varepsilon^A p + a) + \varepsilon^{2B} k^2 (\omega')^2  ) + i \varepsilon^B k \omega' a}{(\varepsilon^A p + a)^2 + (\varepsilon^B k \omega')^2 }\dd \xi,
\end{split}
\end{equation*}
which can be split into
\begin{align}
& \label{oieegyptienne1}= -\varepsilon^{3B/5} p \int \omega \mathfrak{f}^2 \frac{ a (\varepsilon^A p + a)}{(\varepsilon^A p + a)^2 + (\varepsilon^B k \omega')^2 }\dd \xi\\
&\qquad \qquad \label{oieegyptienne2} -\varepsilon^{3B/5} p \int \omega \mathfrak{f}^2 \frac{a}{p} \frac{ \varepsilon^{2B-A} k^2 (\omega')^2  }{(\varepsilon^A p + a)^2 + (\varepsilon^B k \omega')^2 }\dd \xi\\
&\qquad \qquad \label{oieegyptienne3} -\varepsilon^{3B/5} p \int \omega \mathfrak{f}^2 \frac{a}{\varepsilon^A p} \frac{  i \varepsilon^{B-A} k \omega' a}{( p + a)^2 + (\varepsilon^B k \omega')^2 }\dd \xi.
\end{align}
Just like earlier, \eqref{oieegyptienne3} vanishes by oddness of $\omega'$ and evenness and $a$; furthermore, \eqref{oieegyptienne1} is of order $O(\varepsilon^{\frac{3B}{5}})$. There remains 
\begin{equation*}
\begin{split}
\eqref{oieegyptienne2}  &\simeq -\varepsilon^{-A + 3B/5} \int_{|\xi| <<1} \frac{1}{|\xi|} \frac{k^2 }{(\varepsilon^{-B}|\xi|^{5/3} )^2 + k^2} |\xi|^{5/3}\dd \xi \\
&\simeq -\varepsilon^{-A+3B} \int_{|w| << \varepsilon^{-B}k^{-1}}\frac{w}{w^2+1} \varepsilon^{8B/5} |k|^{8/5} \frac{3}{5} w^{- (2+3)/5} \varepsilon^{-3B/5}|k|^{-3/5}  \dd w \\
&= -\varepsilon^{8B/5 - A}  \frac{3}{5} |k|^1 \int_{|w|<<\varepsilon^{-B}k^{-1} } \frac{1}{w^2+1} \dd w.
\end{split}
\end{equation*}.
Thus,
$$
\eqref{oieegyptienne2}  \overset{\varepsilon \to 0}{\longrightarrow} \kappa_2 |k| \qquad \mbox{with} \qquad \kappa_2 = \frac{6}{5} \int_0 ^\infty \frac{1}{w^2+1} \dd w.
$$
Similarly
\begin{equation*}
    \begin{split}
        \varepsilon^{-A+6B/5} &\int   \mathfrak{f}^2 a \left ( \frac{a}{\varepsilon^A p + \varepsilon^B i k \omega' + a} - 1\right )\dd \xi \\
        &\simeq \varepsilon^{8B/5 - A} \frac{3}{5} |k|^{2/5} \int_{|w| < < \varepsilon^{-B}k^{-1}} \frac{1}{w^2+1} w^{-3/5} \dd w
\end{split}
\end{equation*}
so that the limit as $\varepsilon \to 0$ of the above term is
$$
\kappa_3 |k|^{2/5}  \qquad \mbox{with} \qquad \kappa_3 = \frac{6}{5} \int_0 ^\infty \frac{1}{w^2+1} w^{-3/5} \dd w.
$$

\medskip

\noindent \underline{Step 5: conclusion.}
Overall, by renaming $\Gamma = \varepsilon^{-3B/5} \gamma$, multiplying $\varepsilon^{3B/5}$ to the second row of \eqref{FL_result}, and taking $\varepsilon \rightarrow 0$, we obtain
\begin{equation*}
\left (\begin{matrix}
\widehat \beta (-p - \kappa_1 |k|^{8/5} ) + \widehat \Gamma (-\kappa_2 |k|^1 ) \\ \widehat \beta (-\kappa_2 |k|^1) + \widehat \Gamma ( - \kappa_3 |k|^{2/5} )
\end{matrix} \right ) 
= \left (\begin{matrix} \widehat \beta_0 + \widehat \Gamma_0 \cdot 0 \\ \widehat \beta_0 \cdot 0 + \widehat \Gamma_0 \cdot 0 \end{matrix} \right )
\end{equation*}
Thus, we end up with
\begin{equation*}
\boxed{\partial_t \beta + \left( \kappa_1 - \frac{\kappa_2^2}{\kappa_3} \right)(-\Delta)^{8/5} \beta = 0, \qquad\gamma = -\frac{\kappa_2}{\kappa_3 } \varepsilon^{\frac {3B}5} (-\Delta)^{3/5} \beta.}
\end{equation*}
Note that $\kappa_1 - \frac{\kappa_2^2}{\kappa_3} > 0$ as follows by the Cauchy-Schwarz inequality.

\subsection{The case of FPU-$\beta$, $\mathfrak{f} = \frac{1}{\beta_0 \omega + \gamma_0}, \gamma_0 > 0$}

In this case, the argument goes mostly same as before, yet as $\mathfrak{f}(0) > 0$ all integrals become more regular: therefore, under the assumption $B< A < 2B$, we have
\begin{equation*}
    \begin{split}
                \varepsilon^{-A} &\int \omega ^j \mathfrak{f}^2 a \left ( \frac{a}{\varepsilon^A p + \varepsilon^B i k \omega' + a} - 1\right )\dd \xi \\
                &= -p \int \omega^j \mathfrak{f}^2 \frac{a(\varepsilon^A p + a)}{(\varepsilon^A p + a)^2 + (\varepsilon^B k \omega')^2}\dd \xi  \\
        &- \int \omega^j \mathfrak{f}^2 \frac{\varepsilon^{2B-A}a k^2 (\omega')^2}{(\varepsilon^A p + a)^2 + (\varepsilon^B k \omega')^2}\dd \xi  
    \end{split}
\end{equation*}
and the first term converges to $ -p \int \omega^j \mathfrak{f}^2\dd \xi =: - \nu_j p$. For $j \ge 1$, the leading order term for the second term is
\begin{equation*}
    \varepsilon^{2B-A} \int \frac{\omega^j \mathfrak{f}^2}{a} k^2 (\omega')^2\dd \xi,
\end{equation*}
which is integrable for $j \ge 1$, and thus for $2B-A>0$ this term converges to 0. For $j=0$, the same calculation as in $\mathfrak{f} = \omega^{-1}$ case (with $\omega^2 \mathfrak{f}^2 = 1$) we have $- \kappa_1 \gamma_0^{-2} |k|^{8/5}$. Thus, we have
\begin{equation*}
    \partial_t (\nu_2 \beta+\nu_1 \gamma) = 0, \quad \partial_t (\nu_1 \beta + \nu_0 \gamma)+\kappa_ 1 \gamma_0^{-2}  (-\Delta)^{8/5} \gamma = 0.
\end{equation*}
Again, using $\nu_1 ^2 < \nu_0 \nu_2 $ by Cauchy-Schwarz inequality, we end up with a diffusive system
\begin{equation*}
    \boxed{\left ( \nu_0 - \frac{\nu_1^2}{\nu_2} \right )\partial_t \gamma + \kappa_1 \gamma_0^{-2} (-\Delta)^{8/5} \gamma = 0,  \qquad   \partial_t \beta + \frac{\nu_1}{\nu_2} \partial_t \gamma = 0}.
\end{equation*}

\section{The Fermi-Pasta-Ulam-Tsingou paradox}
\label{sec:FPUT_paradox}
In this section, we present the numerical simulation of an oscillator chain performed by Fermi, Pasta, Ulam and Tsingou in 1955 and recount its paradoxical outcome.
This was one of the first numerical experiments, and certainly one of the most influential, since it spurred research in nonlinear dynamics for decades. A number of ideas were put forward to explain the result of the experiment, and we survey some of them.

We will then argue that wave turbulence and the associated kinetic theory are the right tools to understand the result of this experiment, at least if appropriate scaling conditions are met.

\subsection{The numerical experiment of 1955 and its surprising outcome}

In 1955, E. Fermi, who was working at the Los Alamos National Laboratory, had the idea to use automatic computing - which was of course in its infancy - to investigate problems from Physics which seemed out of reach of analytic computation or more physical scaling arguments.

Together with J. Pasta, S. Ulam and M. Tsingou, who was in charge of coding proper, they focused on what is now known as the (FPUT) problem, which was already presented in Section \ref{section_examples} and which we recall here: it is the periodic problem
\begin{equation}
\tag{FPUT} \label{FPUT}
\ddot{q_j} = V'(q_{j+1} - q_j) - V'(q_j - q_{j-1}), \qquad q_j \in \mathbb{R}, \;\; j \in \{1,\dots,N\}
\end{equation}
where
$$
V(x) = \frac{1}{2} x^2 + \frac{\alpha}{3} x^3 + \frac{\beta}{4} x^4,
$$
for some $\alpha, \beta \in \mathbb{R}$. The (FPUT-$\alpha$) problem corresponds to the case $\beta=0$, while the (FPUT-$\beta$) problem corresponds to the case $\alpha=0$. As we saw in Section \ref{section_examples}, the associated kinetic wave equation \eqref{KWE} is only depends on $\alpha$ and $\beta$ through a multiplicative coefficient!

The idea of Fermi was to investigate the question of ergodicity, which had interested him at the beginning of his career: his article \cite{Fermi} was titled "Proof that in general a normal mechanical system is quasi-ergodic".

The Hamiltonian of the system is $\mathscr{H}(\mathbf{p},\mathbf{q}) = \sum \frac{1}{2} p_j^2 + V(q_{j+1}-q_j)$ and is a conserved quantity. Due to the nonlinear interactions among the Fourier modes, which would allow energy exchange among the modes, 
Fermi expected that energy initially inserted in a few modes would eventually be redistributed among all modes, leading to  equipartition of Hamiltonian energy a process known as  \textit{thermalization} in Physics.
Indeed, in the weakly nonlinear regime (low energy regime) which we are considering here, the Hamiltonian energy is essentially equivalent to its quadratic part. The ergodic hypothesis predicts that the trajectory should be equidistributed on microcanonical state given by constant energy, which is essentially equivalent in our case to equidistribution of the energy among Fourier modes.

But to the surprise of Fermi and his collaborators, the system did not evolve irreversibly towards a statistical equilibrium where energy is equidistributed. Rather, as can be seen in Figure \ref{figfpu} borrowed from their paper, it displayed a nearly periodic behavior! 

\begin{figure}
  \includegraphics[width=\linewidth]{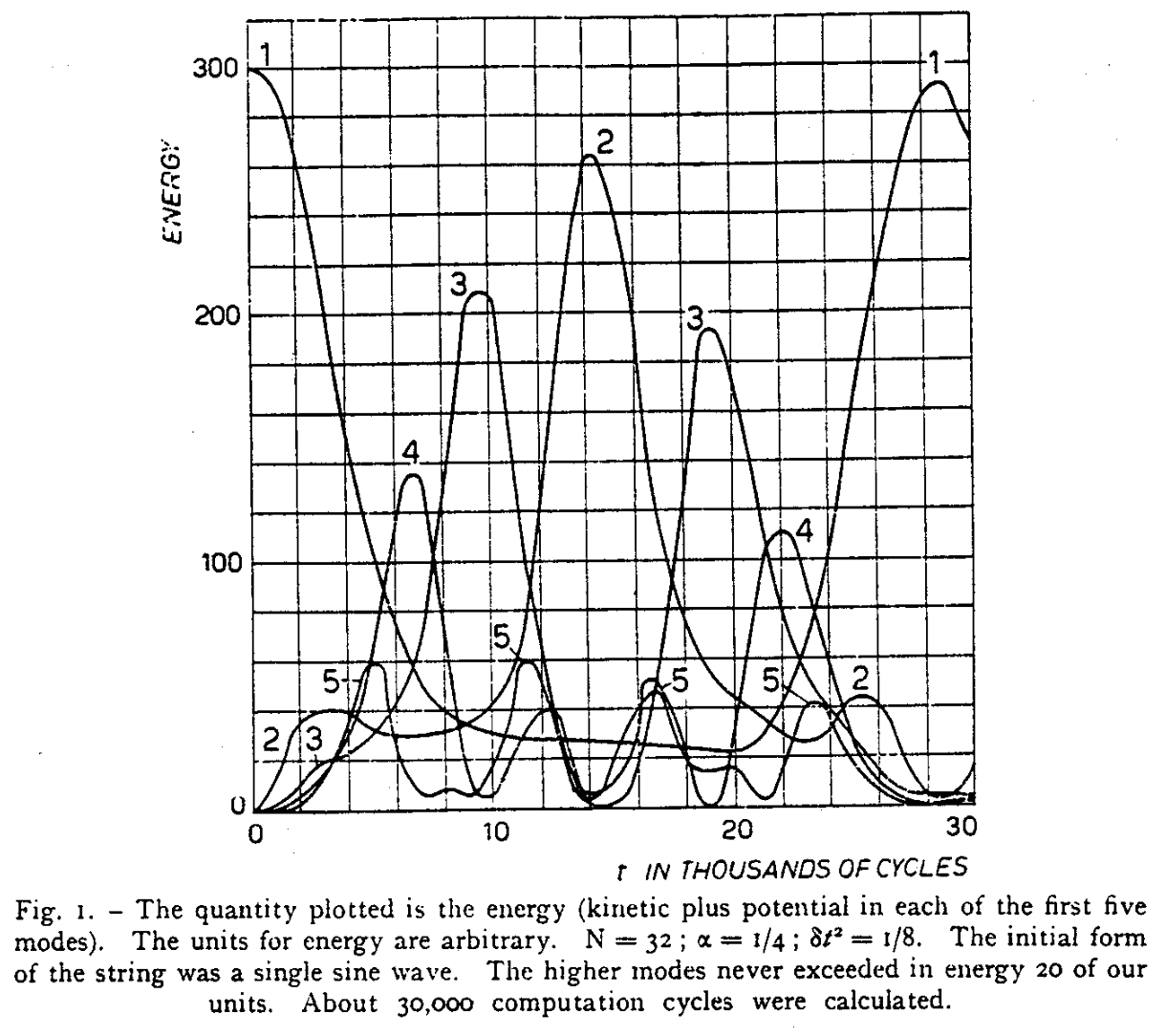}
  \caption{One of the pictures in the orginal paper by Fermi, Pasta, Ulam and Tsingou. The curves labeled 1 to 5 correspond to the energy of the first five Fourier modes as a function of time; the higher modes are not represented. After much oscillation, the system returns very close to its starting point.}
  \label{figfpu}
\end{figure}

What is the explanation for this paradoxical outcome? We will present in the next section the early development of the subject and the main explanations which were put forward. While they only seem partially convincing today, they opened up new and fundamental directions in nonlinear dynamics.

\subsection{Before the weak turbulence viewpoint} 
In this section, we hope to give an idea of the main ideas which were put forward to make sense of the (FPUT) numerical experiment which was described above; some are more rigorous and analytic, and some more heuristic and numerically validated. We will be brief and therefore superficial; the interested reader is referred to the reviews \cite{BermanIzrailev, Ford,Gallavotti, Weissert} for much more.

\medskip

\noindent
\underline{Poincar\'e recurrence.} This is the most natural explanation of the observed recurrence: the classical theorem of Poincar\'e states that, in a measure preserving dynamical system, a trajectory will return arbitrarily close to its starting point for almost all initial data. However, the time scales involved grow extremely fast with the dimension of the system; therefore, this cannot be the proper explanation.

\medskip

\noindent
\underline{KAM Theory} The theory of Kolmogorov-Arnold-Moser \cite{Kolmogorov,Arnold,Moser} was developed between the end of the fifties and the beginning of the seventies. It applies to perturbations of integrable dynamical systems - here, "integrable" can be understood in the naive sense that a formula for the solution can be derived. Since the linearized system is integrable (by viewing the evolution problem on the Fourier side), does the KAM theory apply to (FPUT) for small solutions, in which case the nonlinear problem is a perturbation of the linear problem? The first obstacle is that a nondegeneracy condition at the heart of KAM theory is not satisfied by the (FPUT) equation; this obstacle was recently overcome \cite{HenriciKappeler,Rink} by an ingenious construction. The second obstacle has to do with the proof of the KAM theorem: the constants it produces depend adversely on the dimension of the system considered. For this reason, it is hard to discern whether the observed recurrence can be attributed to KAM theory, or whether something else is at stake.

\medskip

\noindent
\underline{Integrable systems} In the years following the (FPUT) experiment, it was realized that integrable systems lurk very close. First, the Toda lattice \cite{Toda}, which was already mentioned, is an example of unpinned oscillator chain which is completely integrable; by fitting parameters, it can be made to agree with the ($\alpha$-FPUT) system to second order. Second, the KdV equation arises as the long-wave limit of the (FPUT) system; Zabusky and Kruskal \cite{ZabuskyKruskal} observed famously that it displayed the features of an integrable system, after which integrability was indeed proved \cite{GardnerGreeneKruskalMiura}. While at the time these striking discoveries looked like decisive steps, their applicability to the dynamics of the (FPUT) equation now seems questionable. Indeed, beyond the specific regimes where Toda and KdV can be shown to play a role, it is not clear why these systems should say anything about the long time (FPUT) dynamics.

\medskip

\noindent
\underline{Frequency overlap} It was observed by H\'enon-Heiles \cite{HenonHeiles} that certain low-dimensional systems can display very different behavior depending on the level of energy (or the size of the solutions): either periodic trajectories or chaotic behavior, with a rather sharp transition. This idea was then extended by Izrailev anc Chirikov \cite{IzrailevChirikov,Chirikov} to high-dimensional systems. Most importantly, they proposed a criterion to estimate whether a given Hamiltonian equation falls into the category of periodic trajectories or chaotic behavior. Let $\delta \omega$ denote the typical spacing between (temporal) frequencies and $NL$ denote the typical shift in frequency due to the nonlinearity. Then chaotic dynamics correspond to the regime $NL \gg \delta \omega$, and regular trajectories to $\delta \omega \ll NL$, which is a natural guess: if the nonlinearity is not sufficient to shift a temporal frequency to the neighboring one, then averaging will destroy nonlinear interactions between modes, and the system will behave like its linearized version.

What does this condition become in the scaling of Section \ref{Sec: Derivation_WKE}, and how does it compare with the condition of validity of the kinetic limit established in Section \ref{section_from_NOC_to_KWE}? The scaling retained in section \ref{Sec: Derivation_WKE} is such that $a(\xi)$ solves \eqref{eq_epsilon_N}, or equivalently that $b(\xi)$ solves
$$
-i \partial_t b(\underline{\xi}) = \omega(\underline{\xi}) b(\underline{\xi}) + \frac{\varepsilon^2}{N} \sum_{\sigma_i = \pm 1} \sum_{\underline{\xi_1},\underline{\xi_2},\underline{\xi_3}} C_{1,2,3}^\sigma b_1^{\sigma_1}b_2^{\sigma_2}b_3^{\sigma_3} \delta^\sigma_{0,1,2,3}.
$$
There remains to identify the scales $\delta \omega$ and $NL$ in this context. It is clear that the typical spacing of frequencies is $\delta \omega \sim \frac{1}{N}$. Turning to $NL$, we need to estimate the sum in the above right-hand side. On the one hand, there are $N^2$ summands; on the other hand, these summands have random phases and one expects a square root cancellation from the central limit theorem. Therefore, one finds for the nonlinear term on the right-hand side the typical size $NL \sim \varepsilon^2$. The Izrailev-Chirikov criterion for chaotic dynamics becomes
$$
\delta \omega \ll NL \qquad \Longleftrightarrow \qquad N^{-1} \ll \varepsilon^2.
$$
This agrees with condition (C3) of Section \ref{section_from_NOC_to_KWE}!

\medskip \noindent
\underline{Computational results.} 
Before 2000, computations seemed inconclusive and there was no consensus on the eventual thermalization, or the lack of it, for (FPUT) chains. THe literature is extremely abundant, we only mention some papers \cite{Kantz,KantzLiviRuffo,LiviPettiniRuffoSparpaglioneVulpiani} which advocate thermalization  and others \cite{BocchieriScottiBearziLoinger,CasettiCerrutiSolaPettiniCohen} which report the existence of an energy threshold below which thermalization does not occur, in the spirit of the Izrailev-Chirikov criterion.

\subsection{The triumph of the wave turbulence viewpoint}
The first application of weak turbulence and kinetic theory to the (FPUT) problem goes back to 2002, and it seems in retrospect astonishing that it took so long since the connection is so natural. The two papers \cite{BielloKramerLvov,KramerBielloLvov} were published that year and they are the first attempt to set up the kinetic wave equation and compare time scales for thermalization with the predictions of wave turbulence theory.

With the development of computing power, it became clear that thermalization was universally observed for the (FPUT) chain. The validity of the kinetic wave equation is by now well-established in the \textit{thermodynamic limit} where $N$ is sufficiently large depending on $\varepsilon$. The first signature of the wave kinetic equation is the kinetic time scale $T_{\operatorname{kin}} \sim \varepsilon^{-4}$ which is observed as characteristic scale over which thermalization occurs. This was reported in \cite{PistoneChibbaroBustamanteLvovOnorato} and confirmed in \cite{FuZhangZhao,FuZhangZao2} on a variety of examples. A second signature of the wave kinetic equation is the RJ spectrum, which is the attractor of \eqref{KWE} and has been reported in \cite{MendlLuLukkarinen}.

What happens much before the kinetic time scale? Over a long time, the dynamics of the (FPUT) chain seem to be integrable and in particular recurrent; this is not so surprising since the Toda chain can, for instance, be fitted to match the FPUT-$\alpha$ chain to order two. This integrable behavior then slowly gives way to diffusion amongst Fourier modes until the kinetic behavior kicks in and the system thermalizes. This transition is studied in \cite{BenettinChristodoulidiPonno,PonnoChristodoulidiSkokosFlach}.

Finally, what happens much beyond the kinetic time scale? Recall that the RJ states $\frac{1}{\beta \omega + \gamma}$ are the attractors of \eqref{KWE}, but they do not display energy equidistribution if $\gamma \neq 0$. However, we saw that the parameter $\gamma$ is linked to the conservation of mass $\mathcal{M}$, which holds at the kinetic level but not microscopically. It was observed in \cite{MendlLuLukkarinen,HuvennersLukkarinen} that the $RJ$ solution was reached around the kinetic time scale, and that for times much larger, $\gamma$ converges to $0$ until energy equidistribution is satisfied.

Overall, we can propose the following summary, valid in the weakly nonlinear ($\varepsilon$ small) and thermodynamic ($N$ sufficiently large depending on $\varepsilon$) limit
\begin{itemize}
\item Pre-kinetic regime $t \ll T_{\operatorname{kin}}$. The system looks integrable at first, displaying oscillations and recurrence, but it is slowly entering the weakly turbulent regime characterized by spread Fourier support and random phases.
\item Kinetic regime $t \sim T_{\operatorname{kin}}$. The system is in the weakly turbulent regime and its dynamics are governed by \eqref{KWE}. As a result, it converges to a RJ state $\frac{1}{\beta \omega + \gamma}$.
\item Post-kinetic regime $t \gg T_{\operatorname{kin}}$. The system remains in the weakly turbulent regime, but $\gamma$ goes to zero, with the spectrum eventually converging to $\frac{1}{\beta \omega}$.
\end{itemize}

So in the end, Fermi's intuition was right! The system does reach energy equidistribution, and the key mechanism to get there is \eqref{KWE}.

\subsection{Limitations of the kinetic viewpoint} It seems that energy equidistribution is always reached by \eqref{FPUT}, even though the kinetic description through \eqref{KWE} might not hold. Going back to the conditions given in Section \ref{section_from_NOC_to_KWE} for the validity of \eqref{KWE}, all are immediately satisfied for \eqref{FPUT} except the scaling condition
$$
N^{-\frac 12} \ll \varepsilon \ll 1
$$
(here, $N$ is the number of particles and $\varepsilon$ is a measure of the strength of nonlinear interactions whose definition is given in Section \ref{section_weaklyturbulent}). We will discuss what happens if this scaling condition is not satisfied

\medskip
\noindent
\underline{Strong turbulence} corresponds to the situation where $\varepsilon \ll 1$. Then the system cannot be regarded as weakly nonlinear and there is no reason that resonances should play any role. The dynamics are fully nonlinear, in a way close to hydrodynamic turbulence, and there is no governing kinetic equation. The regime of strong turbulence might explain the behavior found in \cite{LvovOnorato} for the equilibration time. In this regime, a description in terms of fluctuating hydrodynamics was proposed by Spohn \cite{Spohn2014}.

\medskip
\noindent
\underline{Finite size effects} correspond to the regime where the system is too small, in more quantitative terms $N \ll \varepsilon^{-2}$. In that case, quasiresonances are too rare: the average number of quasiresonances is zero under a natural equidistribution hypothesis.

To be more specific, going back to the derivation of \eqref{KWE} in Section \ref{Sec: Derivation_WKE}, it is natural to define quasi resonances at frequency $\underline{\xi_0}$ as $(\underline{\xi_1},\underline{\xi_2})$ such that
\begin{equation}
\label{quasiresonances}
\left| \omega(\underline{\xi_0}) + \omega(\underline{\xi_1}) - \omega(\underline{\xi_2}) - \omega(\underline{\xi_0} + \underline{\xi_0} - \underline{\xi_2}) \right| < \frac{1}{T_{\operatorname{kin}}}.
\end{equation}
Here, recall that the frequencies $\underline{\xi_i}$ range in $\frac{2\pi}{N} \mathbb{Z} \mod 2\pi$. Furthermore, the threshold $\frac{1}{T_{\operatorname{kin}}}$ is natural in view of the derivation of \eqref{KWE} given in Section \ref{Sec: Derivation_WKE}; more generally, it is natural to cutoff quasiresonances at a threshold which is the reciprocal $\frac{1}{T}$ of the time over which the system is observed. 

On the one hand, if the frequencies $\xi_1$ and $\xi_2$ are considered to be continuous variables, the set defined by \eqref{quasiresonances} has size $\sim \frac{1}{T_{\operatorname{kin}}}$. On the other hand, the frequencies $\underline{\xi_1},\underline{\xi_2}$ assume $N^2$ different values. Therefore, under a natural equidistribution hypothesis, the number of quasiresonances is $\frac{N^2}{T_{\operatorname{kin}}}$, which is $\ll 1$ if $N \ll \varepsilon^{-2}$.
Then the justification for the sum to integral limit in the derivation of \eqref{KWE} becomes shaky, and what to expect is uncertain...

It was argued in \cite{OnoratoVozellaPromentLvov,PistoneChibbaroBustamanteLvovOnorato} based on convinving numerics that exact 6-wave resonances become the key to the dynamics if $N$ is very small. These exact 6-wave resonances have the form $(\xi_1,\xi_2,\xi_3,\pm \xi_1,\pm \xi_2,\pm \xi_3)$. With this specific ansatz, these resonances are in a way trivial and not generic, and this anszatz invalidates the equidistribution argument which was mentioned above.

Would it be possible to derive a wave kinetic equation based on these exact 6-wave resonances? This seems like an interesting question, close in spirit to  \cite{FaouGermainHani}, where only exact resonances were considered, but with the difference that the random phase approximation is assumed.

\section{Fourier’s law and conductivity}
\label{section_Fourier}

In this section, we aim to connect the wave turbulence approach to thermalization,
developed in the previous sections, with the classical problem of deriving the inhomogeneous Fourier law of heat conduction from microscopic dynamics.
We shall review the main formulations of thermal conductivity used in the literature, explain how they are related, and clarify what is known rigorously and what remains open.

Understanding how macroscopic laws of
energy transport emerge from microscopic dynamics is a central problem in non-equilibrium statistical mechanics.  
Among these, Fourier’s law of heat conduction plays a distinguished role. Fourier’s law is an inhomogeneous macroscopic law which states the following: in a medium (at least when) close to equilibrium subjected to a small temperature gradient, the macroscopic heat flux $J = J(t,x)$
is proportional to the temperature gradient
\begin{align} \label{eq:Fourier law}
   J = - \kappa \nabla T,
\end{align}
where the proportionality constant $\kappa>0$ is the \textit{thermal conductivity}. When Fourier's law is combined with energy conservation, leads to the linear heat equation governing the
evolution of the temperature profile. So in other words establishing Fourier's law is equivalent to proving that macroscopic energy evolves diffusively.

Despite the apparent simplicity of Fourier's law and the overwhelming existing experimental evidence verifying it is true for many materials, 
deriving Fourier’s law from microscopic Hamiltonian dynamics and understanding the microscopic conditions under which it holds remains one of the major open challenges in mathematical physics \cite{BLR00, LLP03, ICMVillani2007}.

\subsection{A stationary microscopic version of Fourier’s law} The following microscopic formulation of Fourier's law does not describe the time evolution of temperature profiles, but rather whether a steady heat current
scales linearly with an imposed temperature gradient. 

\subsubsection{A microscopic definition of Fourier's law} 
Consider a finite microscopic system of length $N$, coupled at its boundaries to heat reservoirs at temperatures $T_1$ and $T_2$. Assuming the existence of a nonequilibrium stationary  state (NESS), for instance described by a phase-space
probability measure (existence of such measure we discuss in a minute), we will define the heat flux average  $\langle J_N\rangle_{\mathrm{NESS}}$. 

At the microscopic level, the heat current $J_N$ is defined through the local conservation of
energy. For oscillator chains with nearest-neighbor interactions, one can associate to each bond
$(i,i+1)$ on the lattice $\mathbb{Z}$, an energy flux $j_{i,i+1}$, obtained from the continuity equation for the local energy. Locally around $i$, this will be like a discrete gradient for the flux: 
\begin{align} \label{eq:local_continuity_eq}
  \tfrac{d}{dt} e_i(t)=j_{i,i+1}(t)-j_{i-1,i}(t).
\end{align}

The explicit expression of $j_{i,i+1}$ depends on the specific model under consideration
(could be deterministic or stochastic dynamics, and with boundary or with bulk noise). 
Now in a stationary non-equilibrium state and in the absence of bulk energy sources (so that there are no terms that inject or remove energy in the bulk\footnote{In works of Basile-Bernandin-Olla where additional bulk noise is usually introduced, it is to ensure that the energy (or also momentum) are the only conserved quantities. We refer to \cite{Bernandin_notes} for an exposition.}), 
its expectation can be directly computed, by taking the expectation in the above relation \eqref{eq:local_continuity_eq}, which is independent of $i$. It is this common value that we denote by $\langle J_N\rangle_{\mathrm{NESS}}$.

The \textit{finite-size conductivity} (microscopic conductivity) is then defined by
\begin{align*} 
\kappa_N (T_1,T_2) = \frac{\langle J_N\rangle_{\mathrm{NESS}}}{(T_1-T_2)/N},
\end{align*} 
and Fourier’s law corresponds to the existence of the thermodynamic limit
$$
\kappa(T) := \lim_{N\to\infty}\,\lim_{T_1,T_2\to T}\,\kappa_N(T_1,T_2),
$$
with $\kappa(T)\in(0,\infty)$. 
Establishing this limit rigorously for realistic microscopic models is a major open problem.

From a mathematical perspective, this difficulty is already visible in one-dimensional oscillator
chains. The harmonic case, i.e. when the particle interactions are assumed to be quadratic potentials corresponding to linear microscopic dynamics, is one of the very few cases where the NESS is known explicitly: it is a Gaussian state. The conductivity can then be computed and it turns out that it diverges linearly with the system size $N$:  $\langle J_N \rangle_{\mathrm{NESS}} = O(1)$ \cite{RLL67}, reflecting ballistic transport and the failure of Fourier’s law. 
For (FPUT-$\beta$) chains introduced in Section \ref{section_examples}, numerical
simulations and heuristic arguments \cite{LLP03, Spohnreview06, LukkarinenSpohn2008} predict anomalous transport with
$$
\kappa_N \sim N^\alpha,\qquad \alpha\in(0,1),
$$ 
connected to the fractional diffusion in the hydrodynamic limit, see \cite{mellet2015anomalous} and references therein.

We are very far from a complete rigorous mathematical understanding. In  general, the conjecture is that some anharmonicity in the microscopic level is crucial (maybe enough for higher dimensional models) in order to derive the linear heat equation \cite{ICMVillani2007}. But, as the (FPUT-$\beta$) example shows, it is not enough!  

\subsubsection{Rigorous results.} Mathematically, a lot of progress has been achieved on microscopic systems of $N$ oscillators in contact at the boundaries with two heat reservoirs at different temperatures. The reservoirs are modelled by stochastic diﬀerential equations where the stochastic noise aﬀects only the velocities of the two boundary oscillators, while all the other variables follow classical Hamiltonian dynamics. 
Particles interact typically with their nearest neighbors. This leads to a system of SDEs, where $2N-2$ equations are deterministic and $2$ are stochastic. In other words, the associated differential operator, obtained by Itô's formula, is highly degenerate, giving an extreme example of a  hypoelliptic operator\footnote{The associated operator is similar to the kinetic Fokker-Planck, arising from Langevin dynamics, cf e.g.  \cite{Greg_book} but it is much more degenerate, since the diffusion only acts at the boundaries. A more relevant but simpler case is the so-called generalised Langevin dynamics, see for example \cite{GregGabrielVaes}, which corresponds to a chain of just $N=2$ oscillators.} as the noise enters through two out of the $2N$ degrees of freedom. 


Under suitable assumptions on the interaction and pinning potential (if present), e.g. smoothness, non-degeneracy, polynomial growth and pinning weaker than interactions, existence and uniqueness of a non-equilibrium steady state (NESS) has been proven for such systems and the system is also known to relax to a  \cite{EPR99a, EPR99b, EH00, RBT02, Car07} exponentially fast. Quantitative estimates on the speed of  convergence in terms of the system size $N$ have been obtained in certain cases by hypocoercivity techniques: in \cite{Men2020} for potentials being perturbations around harmonic (and the perturbations to decay polynomially with $N$) and on the whole spectrum (optimal scalings in $N$) of such operators in harmonic cases \cite{BMY25,BM22}. 

The main mathematical difficulties are the absence of explicit formulas for such steady states in principle (except in the linear case). This makes it extremely hard to compute the scaling of $\langle J_N\rangle_{\mathrm{NESS}}$ in terms of $N$. Moreover, studying the long time behavior of such systems and how fast they approach a NESS (if they do..) is a challenge due to the extremely hypoelliptic operator describing it. There are cases where the operator is not even hypoelliptic if one adds next-to-nearest-neighbor interactions \cite{BMY25}, and also cases where no stationary state is proven to even exist (or no spectral gap) if the pinning potential is stronger than interactions \cite{HM09}. 

These questions are extremely interesting and thus the mathematical community in Statistical Mechanics, PDEs and Stochastic Analysis has mostly focused on these. In terms of Fourier's law, of course having a steady state which is physically relevant (i.e. the system converging towards it fast enough) is essential in order to define the microscopic conductivity as above. However, exhibiting or not exponential convergence towards it, does not imply by itself anomalous or not energy transport. For instance, one can prove exponential convergence for the harmonic chain (even though note that the sharp spectral gap in that case is not independent of the number of particles $N$ and in fact closes as $N \to \infty$ (\cite{BM22})!) 


\subsection{Microscopic linear response and the Green–Kubo formula.}
An alternative approach to a microscopic expression for the thermal conductivity is provided by linear response theory. Rather than imposing boundary noise as above, one studies the response to infinitesimal perturbations, of an infinite homogeneous system on the lattice at equilibrium. 
We briefly recall here the standard argument after \cite{Spohn_book, Lukkarinen2016, AokiLukkSpohn}. 

We assume that energy is the only relevant conserved quantity and that the equilibrium dynamics is ergodic.
We work in infinite volume at equilibrium temperature $T$ (canonical Gibbs state), and denote by 
$\mathbb{E}_{\mu_T}(\cdot)$ the corresponding expectation with respect to the stationary measure. As previously, let $e_i(t)$ denote the energy density at the site $i \in \mathbb{Z}$. 


\noindent \underline{The basis of this argument} is the assumption that, on macroscopic space-time scales, the temperature profile $T(t,x)$ evolves according to a nonlinear diffusion equation with  diffusion coefficient $D(\cdot)$: 
 \begin{align} \label{eq:NL diffusion}
 \partial_t T(t,x) = \partial_x( D(T(t,x) \partial_x T(t,x) ).
 \end{align} 
We consider a perturbation of the global equilibrium of the form $$T(t,x) = T + \theta(t,x),$$ where $T$ is the global equilibrium temperature and $\theta \in L^1(\mathbb{R})$ has finite second moment: $\int_{\mathbb{R}}x^2 |\theta| dx< \infty$ and $\int \theta \dd x \neq 0$.

Macroscopically, this means injecting locally a finite amount of energy; 
microscopically, this corresponds to formally perturbing the equilibrium Gibbs measure by a small localized perturbation of the temperature.


Linear response theory corresponds to linearising the (macroscopic) equation \eqref{eq:NL diffusion} around a homogeneous equilibrium state $T$, which yields, to leading order, the linear heat equation for the perturbation 
$$\partial_t \theta  = D(T)\partial_x^2 \theta.$$

The solution to this linear diffusion equation is explicit, it is given by convolution with the heat kernel, and this enables to recover $D(T)$ as follows. The second moment of the perturbation $\theta(t,x)$ is
\begin{align*}
    M_2(t)&:= \int_{\mathbb{R}} x^2 \theta(t,x) dx = \frac{1}{\sqrt{4 \pi t D(T)}}\int_{\mathbb{R}} x^2 \int_{\mathbb{R}} e^{-\tfrac{(x-y)^2}{4t D(T)} } \theta(0,y)\, dy\, dx 
    \\ &= 
    \int_{\mathbb{R}} \theta(0,y) [ y^2 + 2tD(T) ] dy = M_2(0) +  2tD(T) M_0(0)
\end{align*}
    where $M_0(0)$ denotes the zeroth moment (the mass) of the perturbation at time $0$, since the mass is conserved (we have assumed that $\theta \in L^1$). 
    Therefore, 
$$ 
\lim_{t \to \infty} \frac{ M_2(t)}{t} = 2D(T) M_0(0). 
$$
The above formula makes sense, as we have assumed the moments appearing to be finite. 

\noindent
\underline{Microscopic analogue.} What defines the diffusion coefficient for the microscopic system is a direct discretization of the above identity: 

\begin{equation}\label{eq:diffusiveS}
\begin{split} 
\lim_{t\to\infty}\frac1t \mathcal{M}_2(t) &:= 
\lim_{t\to\infty}\frac1t \sum_{i \in \mathbb{Z}} i^2 \mathrm{Cov}_{\mu_{T}}\big(e_i(t),e_0(0)\big)
\\ &=
2 D(T)\,\sum_{i\in\mathbb Z} \left[ \mathbb{E}_{\mu_T} ( e_i(t) e_0(0)) - \mathbb{E}_{\mu_T}( e_i(t)) \mathbb{E}_{\mu_T}( e_0(0)) \right] \\ & =: 2 D(T)\mathcal{M}_0(t) = 2 D(T)\mathcal{M}_0(0).
\end{split}
\end{equation}
This formula defines the diffusion coefficient $D(T)$.

In the above discretization, note that the analogy is
\begin{align*}
& M_2(t) \sim \mathcal{M}_2(t) = \sum_{i \in \mathbb{Z}} i^2 \mathrm{Cov}_{\mu_{T}}\big(e_i(t),e_0(0)\big) \\
& M_0(t) \sim \mathcal{M}_0(t) = \sum_{i\in\mathbb Z} \left[ \mathbb{E}_{\mu_T} ( e_i(t) e_0(0)) - \mathbb{E}_{\mu_T}( e_i(t)) \mathbb{E}_{\mu_T}( e_0(0)) \right], 
\end{align*}
because the macroscopic perturbation $\theta$ is identified with the $\operatorname{Cov}_{\mu_T}(e_j(t), e_0(0))$. Here the covariance $\mathrm{Cov}_{\mu_{T}}\big(e_i(t),e_0(0)\big)$ is taken with respect to the equilibrium measure $\mu_T$: we take initial data distributed according to $\mu_T$, we let the system evolve until time $t$, and then we compute the covariance of $e_i(t)$ and $e_i(0)$.

The above discretisation is the correct one in the linear response regime, since, as the following computation showsm the covariance is the linear response of the expected energy of site $j$ at time $t$.  

Indeed, the covariance appears as the derivative at zero perturbation of the expectation with respect to the perturbed Gibbs measure: For the (normalised) perturbed measure $$\mu_{T, \tau} = Z_{T,\tau}^{-1} e^{ - T^{-1}\sum_{i \in \mathbb{Z} } e_i + \tau \sum_{i \in \mathbb{Z}}e_i \varphi(i) } =Z_{T,\tau}^{-1} e^{- T^{-1}\sum_{i \in \mathbb{Z} } e_i + \tau e_0 }$$ for the choice $\varphi(i)=\delta_{0}$. We calculate for any observable $F$
\begin{equation} \label{eq:computation deriv_at_0_perturb = cov}
\begin{split}
    \frac{d}{d\tau} \left. \left( 
  \frac{\int F e^{\tau e_0} d \mu_T}{\int e^{\tau e_0} d\mu_T} \right)\right|_{\tau=0}  &= \left. \left[\frac{ \int F e_0 e^{\tau e_0}d\mu_T }{\int e^{\tau e_0} d\mu_T} - \frac{\int F e^{\tau e_0} d \mu_T}{\int e^{\tau e_0} d\mu_T}\frac{\int e_0 e^{\tau e_0}d\mu_T }{\int e^{\tau e_0} d\mu_T} \right] \right|_{\tau=0}\\
  &= \left[ \mathbb{E}_{\mu_{T, \tau}}(Fe_0) - \mathbb{E}_{\mu_{T, \tau}}(F)\mathbb{E}_{\mu_{T, \tau}}(e_0)\right]\vert_{\tau=0} = \operatorname{Cov}_{\mu_T} (F, e_0)
\end{split}
\end{equation}
    which, in the particular case of choosing the observable $F$ to be the energy of the $j$-th particle at time $t$, yields the linear response of the averaged  energy: 
    $$\frac{d}{d\tau} \mathbb{E}_{\mu_{T,\tau}}(e_j(t))\vert_{\tau=0} =\operatorname{Cov}_{\mu_T}(e_j(t), e_0(0)).$$

Upon dividing the eq. \eqref{eq:diffusiveS} by $\mathcal{M}_0(0)$, we have $$
D(T) = \frac{1}{2} \lim_{t \to \infty} \frac{\mathcal{M}_2(t)}{t \mathcal{M}_0(0)}.$$
From Einstein's relation \cite[Ch. 2.5, Eq. (2.72)]{Spohn_book},  the conductivity is defined by $\kappa(T) = D(T) \frac{d}{dT}\mathbb{E}_{\mu_T}(e_0)$, where $\frac{d}{dT}\mathbb{E}_{\mu_T}(e_0)$ is the specific heat,\footnote{Einstein's relation essentially says that the conductivity is given as the speed of the energy spreading, $D(T)$, times the expected energy per unit temperature, i.e. the specific heat.} and explicit computations give \footnote{One should first do the calculations rigorously on the truncated space $\Lambda_L$ with the Gibbs measure $\mu_T^L:= Z(T)^{-1} \operatorname{exp}(- \sum_{i \in \Lambda_L} e(p_i,q_i)) \prod_{i \in \Lambda_L} dp_i dq_i $, with periodic boundary conditions, and then pass to the limit $L \to \infty$, where $\Lambda_L \to \mathbb{Z},$ which would be justified under sufficient decay of correlations-hypothesis for $e_0$. This is how we get the stated formula}: 
$$\frac{d}{dT} \left[ \mathbb{E}_{\mu_T}(e_0) \right] = T^{-2} \mathcal{M}_0(0). $$ 
Altogether, this implies the formula 
$$ \kappa(T) = \frac{1}{2T^2} \lim_{t \to \infty} \frac{\mathcal{M}_2(t)}{t}.$$


\noindent 
\underline{Green-Kubo formula in terms of fluxes}. We obtain the standard Green-Kubo formula by using the local conservation of energy law, and the fact that the Gibbs state is invariant under both spatial translations and the time evolution, which implies $\mathcal{M}_2(t)$ is symmetric in $t$. 
 One can then rewrite the growth of $\mathcal{M}_2(t)$ in terms of current-current correlations.

Formally (and rigorously under summability/mixing assumptions),
$$
\frac{d^2}{dt^2}\mathcal{M}_2(t)=2\sum_{i\in\mathbb Z}\mathrm{Cov}_{\mu_{T}}\big(j_{i,i+1}(t),j_{0,1}(0)\big).
$$
Integrating twice in time yields
$$
\lim_{t\to\infty}\frac{\mathcal{M}_2(t)}{t}
=2\int_0^\infty \sum_{i\in\mathbb Z}\mathrm{Cov}_{\mu_{T}}\big(j_{i,i+1}(s),j_{0,1}(0)\big)\,ds,
$$
provided the time integral converges. Then again by Einstein's relation:
\begin{equation}\label{eq:GK}
\kappa(T)
=\frac{1}{T^2}\int_0^\infty \sum_{i\in\mathbb Z}\mathrm{Cov}_{\mu_{T}}\big(j_{i,i+1}(s),j_{0,1}(0)\big)\,ds.
\end{equation}
Thus to understand conductivity for large times, we need to understand time-decay of current-current correlations.
Integrability of the current-current correlation function $C_{T}(r): =\sum_{i\in\mathbb Z}\mathrm{Cov}_{\mu_{T}}\big(j_{i,i+1}(r),j_{0,1}(0)\big) $ is sufficient for diffusive energy transport, while slow decay of $C_T(\cdot)$ leads to divergence of the Green-Kubo integral, in which case we should expect superdiffusive behaviour.
\\

\noindent
\underline{Conductivity $\kappa_N$ and Green-Kubo}: What is generally agreed upon but not proven is that the stationary conductivity defined from boundary-driven non-equilibrium steady states, $\kappa_N$ in the thermodynamic limit is identified with $\kappa(T)$ from Green-Kubo: $\kappa = \lim_{N\to \infty}\kappa_N$ \cite[Sec. 7]{BLR00} and \cite[Sec. 7.2]{Spohn_book}.

The advantage of the Green-Kubo formulation is that it is given entirely in terms of equilibrium dynamics and so we do not need an explicit description of the NESS, which in general we do not have.
By contrast, estimating $\kappa_N$ directly requires quantitative information on the non-equilibirum steady state and its dependence on $N$. On the other hand computing the long time decay of energy correlations is also highly non-trivial. There are however specific cases where the conductivity has been computed via Green-Kubo (and in certain regimes normal transport has been established) in the presence of bulk noise to ensure conservation of momentum or/and energy, see \cite{BernandinOlla11, BasileBernandinOlla09,BernandinOlla05, Bernandin_notes} and references therein.

\subsection{Conductivity via Wave Turbulence approach} We are now ready to make the connection with the main topic of this review, namely the kinetic wave equation.  
As discussed above, starting from weakly anharmonic chains, one derives, at least formally, a phonon Boltzmann equation for a phase-space density $f(t,x,\xi)$.
 Its spatially homogeneous version (when $f$ is independent of $x$), describes how energy is redistributed in frequency (momentum) space. As we  will se, besides the fact that in this scale we see the relaxation towards the thermodynamic equilibrium RJ, where the energy is equidistributed which was E. Fermi's expectation in the (FPUT) experiment as explained in Section \ref{sec:FPUT_paradox},
wave turbulence/kinetic theory also allows us 
to derive a formula for the conductivity by exploiting the Green-Kubo formula.

For the sake of the presentation let us fix the microscopic system to be (DNLKG) as introduced in Section \ref{subsec:pinnedsystems}, as was done in \cite{AokiLukkSpohn}. We scale the Hamiltonian as follows
\begin{align} \label{eq:Hamilt_delta}
    \mathscr{H}(\textbf{p}, \textbf{q}) & = \sum_{j\in \mathbb{Z}} \left( \frac{p_j^2}{2} +  \frac{1}{2} \left[ \delta (q_{j+1} - q_j)^2 + (1-2\delta) q_j^2  \right] \right) + \frac{\varepsilon}{4} \sum_{j\in \mathbb{Z}} q_j^4 \\ 
& = \mathscr{H}_{\text{quad}}(\textbf{p}, \textbf{q}) + \frac{\varepsilon}{4} \sum_{j\in \mathbb{Z}} q_j^4,
\end{align}  
(with $\delta \in (0,\frac 12)$) and the dispersion relation which is determined by the quadratic part of the Hamiltonian is
$$
\omega(\xi) = \sqrt{1-2 \delta \cos(\xi)}.
$$
The energy conservation law locally yields, if $e_j(t)$ is the summand for each $j$ in the Hamiltonian \eqref{eq:Hamilt_delta}
$$
\frac{d}{dt} e_j(t) = J_{j-1,j} - J_{j,j+1},
$$
from which we can deduce the energy current from $j$ to $j+1$
\begin{align} \label{eq:current explicit}
    J_{j,j+1} = -\frac{\delta}{2}(p_jq_{j+1} - p_{j+1}q_j).
\end{align}

Then the total current is computed by summing over all local currents, i.e. 
$$
J_{\text{total}} = \sum_{i \in \mathbb{Z}} J_{i,i+1}.
$$

We look at the problem through the evolution of the complex valued field $b(\xi)$ that takes into account both $\textbf{p}$ and $\textbf{q}$: 
$$
b(\xi ) = \frac{1}{\sqrt{\omega(\xi)}} [i \widehat{p}(\xi) + \omega(\xi) \widehat{q}(\xi)]
$$
whose time evolution is (as computed explicitly in \eqref{eqb}):
\begin{align*}
&  \partial_t b(\xi) + i\omega(\xi) b(\xi) = \varepsilon \times [\text{Nonlinearities}], 
\end{align*}
We will be working in the weakly nonlinear regime $\varepsilon \to 0$. We see that the energy of the dominant linear (harmonic) part is 
$$ 
\mathscr{H}_{\text{quad}} = \int_{\mathbb{T}} \omega(\xi) \overline{b (\xi)} b(\xi) \dd \xi. $$
We can then compute the current by inserting,  for each $j \in \mathbb{Z}$, 
$p_j = \int e^{i j \xi} \widehat{p}(\xi)\dd \xi$ and the same for $q_j$, in \eqref{eq:current explicit}:  
$$
\sum_j \left\{ \iint_{\mathbb{T}^2} e^{i \xi} \widehat{p}(\xi')\widehat{q}(\xi) e^{ij(\xi+\xi')}\dd \xi\dd \xi' - \iint_{\mathbb{T}^2} e^{i \xi} \widehat{p}(\xi)\widehat{q}(\xi') e^{ij(\xi+\xi')}\dd \xi\dd \xi' \right\}.
$$
Re-arranging and using that $\sum_{j \in \mathbb{Z}} e^{ij(\xi+\xi')} = 2\pi \delta_{\mathbb{T}}(\xi+\xi') $ we get 
$$
J_{\text{total}} = -  \delta \pi \int_{\mathbb{T}}e^{i \xi}[ \widehat{p}(-\xi)\widehat{q}(\xi)  - \widehat{p}(\xi)\widehat{q}(-\xi) ]\dd \xi.  $$
 By reality of $p,q$ and changing the variable in the second term to $-\xi$: 
 \begin{align*}
J_{\text{total}} & = -\pi \delta \int_{\mathbb{T}}[e^{i \xi} \overline{\widehat{p}(\xi)}\widehat{q}(\xi)  - e^{-i\xi}\overline{\widehat{p}(\xi)}\widehat{q}(\xi) ]\dd \xi 
 = -2 i \delta \pi  \int \overline{\widehat{p}(\xi)} \widehat{q}(\xi) \sin(\xi) \dd \xi 
 \\ & = 
 2\pi \int [-i\overline{\widehat{p}(\xi)} \widehat{q}(\xi)] \omega(\xi)\omega'(\xi)\dd \xi =
 2 \pi \int \operatorname{Im}(\overline{\widehat{p}(\xi)} \widehat{q}(\xi)) \omega(\xi)\omega'(\xi)\dd \xi
 \end{align*}
 since the total current is real. Now use that $$\overline{b (\xi) } b(\xi) = 2 \operatorname{Im}(\overline{\widehat{p}(\xi)} \widehat{q}(\xi)) + \omega(\xi) |\widehat{q}(\xi)|^2+
 \frac{|\widehat{p}(\xi)|^2}{\omega(\xi)}. $$ to get 
$$
J_{\text{total}} = \pi \int  \overline{b (\xi) } b(\xi) \omega(\xi)\omega'(\xi)\dd \xi. $$
Note that the term $\omega(\xi) |\widehat{q}(\xi)|^2+
 |\widehat{p}(\xi)|^2/ \omega(\xi)$ in the integrand  vanishes when integrated against $\omega \omega'$ over the torus, since the latter is an odd function and the former an even function. 

 Rewriting now the harmonic Gibbs measure at temperature $T$, in terms of the complex field $b$, we have a centered complex Gaussian which is characterised by the covariance 
 $$\operatorname{Cov}_T (\overline{b(\xi)}, b(\xi') ) = \langle \overline{b(\xi)} b(\xi') \rangle_T = \delta(\xi-\xi') f_{T}(\xi),\quad \text{ with }  f_{T}(\xi) =T/\omega(\xi).$$

\noindent
 \textbf{Perturbing initial data:} As we aim to study conductivity through linear response, as in the previous subsection, we consider the perturbed Gaussian measure: $$\mu_{T, \tau} = Z_{T, \tau}^{-1}\exp \Big(- \frac 1T \mathscr{H}_{\text{quad}}(\textbf{p},\textbf{q})  + \tau J_{\text{total}} \Big)$$ with $\tau$ small, where $Z_{T, \tau}$ is the normalization constant. This state, given the previous calculations on the current with respect to the field $b$, is still Gaussian with covariance 
 $$\operatorname{Cov}_{T,\tau} (\overline{b(\xi)}, b(\xi') ) = \delta(\xi-\xi') f_{T,\tau}(\xi),\quad \text{ with }  f_{T, \tau}(\xi) =\frac{1}{\pi( \beta \omega(\xi) - \tau (\omega'\omega)(\xi)}.
 $$

We assume that the derivation of the Kinetic Wave Equation \eqref{KWE} holds as is explained heuristically in Section \ref{section_heuristic_proof}. Thus formally we set 
$$
\lim_{\varepsilon \to 0}\mathbb{E}(|b(\varepsilon^{-2} t,\xi)|^2) = f_{T, \tau}(t,\xi)
$$
where $f_{T, \tau}(t,\xi)$ solves the spatial-homogeneous version of \eqref{KWE}: 
$$
\begin{cases} 
\partial_tf_{T, \tau}(t,\xi) = \mathcal{C}(f_{T, \tau}) (t,\xi) \\
f(t=0,\xi) = f_{T, \tau}(\xi) =(\beta \omega(\xi) -  \tau \pi (\omega'\omega)(\xi))^{-1}.
\end{cases}
$$
(recall also the setup in \eqref{data_Gaussian} in the heuristic derivation of the KWE).

\noindent
\textbf{Linearization Set up:} We linearise the wave kinetic operator around the RJ state $f_{\beta}(\xi)=1/\beta \omega$, i.e. we look at solutions of the form 
$$f_{T,\tau}(t,\xi) = f_{\beta}(\xi) + f_{\beta}^2(\xi) g_{\beta, \tau}(t,\xi).$$ 
The choice of this perturbation makes the linearized operator $L$, written below, self-adjoint in $L^2(\mathbb{T})$. 

Now for initial data as above we expand for small $\tau$ and we have $f_{T,\tau}(0,\xi) = (\beta \omega)^{-1}+ \tau \pi \frac{\omega' \omega}{(\beta \omega)^2}$. Inserting then our ansatz, we get that in this linear response regime the perturbation at initial time $g_{\beta, \tau}(0,\cdot)$ is  
$$ g_{\beta, \tau}(0,\xi) = \tau \pi  \omega' \omega = \tau \pi \delta \sin(\xi).$$

The perturbation $g_{\beta, \tau}(t,\xi)$ evolves according to the linear phonon Boltzmann semigroup generated by the operator 
\begin{align} \label{eq:lin Oper_L_pert f+f^2g}
L g_{\beta, \tau} = \int_{\mathbb{T}^3}  \delta(\Sigma) \delta(\Omega) \frac{\beta^{-4}}{[\omega_0\omega_1\omega_2\omega_3]^2}\left[g_{\beta, \tau}(\xi_0) + g_{\beta, \tau}(\xi_1) - g_{\beta, \tau}(\xi_2) - g_{\beta, \tau}(\xi_3) \right] d\xi_1 d\xi_2 d\xi_3, 
\end{align}
which is a self-adjoint operator in $L^2(\mathbb{T})$. The linear evolution is given by 
$$
g_{\beta, \tau}(t,\xi) 
= e^{-t f_{\beta}^{-2} L} g_{\beta, \tau}(0,\xi) 
= e^{-t f_{\beta}^{-2} L}
\left( \tau \pi \delta \sin(\xi) \right)
$$ 
and the evolution of the actual perturbation considered is given by 
 $$ \partial_t(f_{T,\tau}(t,\xi)  - f_\beta) =\partial_t (f_\beta^2 g_{\beta,\tau}(t,\xi)) = - L (f_\beta^{-2}(f_{T,\tau}(t,\xi)  - f_\beta))), $$
 meaning that 
 \begin{align} \label{eq:evol_linear perturb}
      f_{T,\tau}(t,\xi) = f_\beta(\xi) + [e^{-t L f_{\beta}^{-2}}(f_\beta^2 g_{\beta,\tau}(0))](\xi) 
 \end{align}

Now since linear response is the derivative at zero perturbation we set 
$$
\partial_\tau g_{\beta, \tau}(0,\xi)\vert_{\tau=0} = \pi \delta \sin(\xi).
$$ 

\noindent
\textbf{Linear Response and Kinetic Limit:}
Next we use that the derivative at zero perturbation of the (normalised) expectation with respect to the perturbed measure $\mu_{T,\tau}$, $\mathbb{E}_{T,\tau}(F)$, for an observable $F$,  gives the equilibrium covariance (recall also the computation in \eqref{eq:computation deriv_at_0_perturb = cov}): 
\begin{align*}
    \partial_\tau \big\vert_{\tau=0} (\mathbb{E}_{T,\tau}( F ) ) &= 
    \big[ \mathbb{E}_{T,\tau}( F(t) J_{\text{total}}(0)  ) -  \mathbb{E}_{T,\tau}( F(t) ) \mathbb{E}_{T,\tau}(J_{\text{total}}(0)) \big]\big\vert_{\tau=0} \\ & =  \mathbb{E}_{T}( F J_{\text{total}})  = \operatorname{Cov}_{T}(F, J_{\text{total}} )
\end{align*}
since $\mathbb{E}_T(J_{\text{total}})=0$. Apply this now for $F=J_{0,1}(t)$: 
\begin{align} \label{eq:linear response for J_01}
\partial_\tau \big\vert_{\tau=0} (\mathbb{E}_{T,\tau}( J_{0,1}(t) ) ) = \operatorname{Cov}_{T}(J_{0,1}(t), J_{\text{total}} ).
\end{align}
Having the Green-Kubo formula in mind, \eqref{eq:GK}, we aim to understand the conductivity in the kinetic limit. We thus compute in the kinetic limit formally the current density, i.e. the expectation of the local current:
$$ \lim_{\varepsilon \to 0}\mathbb{E}_{T, \tau}(J_{0,1}(\varepsilon^{-2}t)) = \pi \int_{\mathbb{T}} f_{T,\tau}(t,\xi)\omega(\xi)\omega'(\xi)\dd\xi.$$

Now if we are allowed to exchange the limit and integral, on the one hand we have (from linear response in \eqref{eq:linear response for J_01}) 
$$ \lim_{\varepsilon \to 0} \partial_\tau \mathbb{E}_{T,\tau}[ J_{0,1}(\varepsilon^{-2}t) ] \big\vert_{\tau=0} = \lim_{\varepsilon \to 0} \mathbb{E}_{T} [ J_{0,1}(\varepsilon^{-2} t) J_{\text{total}}(0) ]
$$
and on the other hand, in the kinetic limit, the left-hand side is 
\begin{align*} 
\lim_{\varepsilon \to 0} \partial_\tau \mathbb{E}_{T,\tau}[ J_{0,1}(\varepsilon^{-2}t) ] \big\vert_{\tau=0} 
& = \partial_\tau \Big[ \pi \delta \int \sin(\xi) f_{T,\tau}(t,\xi) \dd \xi   \Big]\big\vert_{\tau=0}\\ 
& = 
\partial_\tau 
\Big[ \pi \delta \int \sin(\xi)\left\{ f_{\beta}(\xi) + [e^{-tLf_{\beta}^{-2}}f_{\beta}^2 g_{\beta, \tau}(0)](\xi) \right\} \dd \xi \Big]\big\vert_{\tau=0} \\ 
& = \pi^2 \delta^2  \left\langle  \sin , e^{-t L f_{\beta}^{-2}} f_{\beta}^2 \sin \right\rangle_{L^2(\mathbb{T})}, 
\end{align*}
where we used that \eqref{eq:evol_linear perturb} for the evolution of the initial data $f_{T,\tau}(0,\xi)$.  
Now applying the conjugation of semigroup property (after writing $L f_\beta^{-2} =f_\beta f_\beta^{-1} (L f_\beta^{-1} )f_\beta^{-1} $) 
we eventually have 
\begin{align} \label{eq:Kin Conj identity}
\lim_{\varepsilon \to 0}  \operatorname{Cov}_{T}\left( J_{0,1}(\varepsilon^{-2} t), J_{\text{total}}(0) \right) =
\pi^2 \delta^2   \left\langle  \sin(\xi) , f_{\beta} e^{-t f_{\beta}^{-1} L f_{\beta}^{-1}} f_{\beta} \sin(\xi)\right\rangle_{L^2(\mathbb{T})}.
\end{align} 
Integrate over time, and given that we are allowed to exchange the time integral and the kinetic limit on the LHS, we recover the conductivity due to Green-Kubo formula  \eqref{eq:GK}
\footnote{From the definition of $J_{\text{total}}$:  $\operatorname{Cov}_{T}\left( J_{0,1}(t), J_{\text{total}}(0) \right) =\sum_{i \in \mathbb{Z}}\operatorname{Cov}_{T}\left( J_{0,1}(t), J_{i,i+1}(0) \right)$. Using then stationarity of the invariant equilibrium measure: $\operatorname{Cov}_{T}\left( J_{0,1}(t), J_{i,i+1}(0) \right) = \operatorname{Cov}_{T}\left( J_{0,1}(0), J_{i,i+1}(-t) \right)$ which equals to $\operatorname{Cov}_{T}\left( J_{0,1}(0), J_{i,i+1}(t) \right)$ due to time invariance. This recovers precisely the integrand in the Green-Kubo formula.}. 

In particular this gives us an indication of whether the conductivity of our system is finite or diverges in an anomalous manner, by determining the decay rate of the semigroup $e^{-t f_\beta^{-1} L f_\beta^{-1}} \omega'$, see also \cite[Kinetic Conjecture, Eq. (1.18)]{LukkarinenSpohn2008}. 
\\

\noindent 
$\bullet$ Provided that the semigroup exhibits sufficiently fast decay in $t$ so that the time integral $\int_0^\infty \langle \sin  , f_{\beta} e^{-t f_{\beta}^{-1} L f_{\beta}^{-1}} f_{\beta} \sin \rangle_{L^2(\mathbb{T})}  \dd t$ converges, which is the case if the operator has a spectral gap for example, we may rewrite:  
\begin{align}
    \lim_{\varepsilon \to 0} \varepsilon^2 T^2 \kappa(T) = \pi^2 \delta^2 \int_{\mathbb{T}} \sin(\xi) f_\beta^2 L^{-1} f_\beta^2\sin(\xi)
    \dd \xi = \pi^2 \delta^2 \beta^{-4} \int_{\mathbb{T}} \sin(\xi) \omega^{-2} L^{-1} \sin(\xi) \omega^{-2} \dd \xi 
\end{align}
and expect finite conductivity. 
Note that here since $\sin$ is odd and $1,\omega$ are even functions and $1,\omega$ are expected to span the $\operatorname{Ker}(L)$ (see discussion in Section \ref{sec:collisional invariants}), we are already working on the complement of the Kernel, thus $L^{-1} \sin$ is well-defined.

Note that the RHS does not depend on $\beta$, since remember that in the definition of the operator $L$, in \eqref{eq:lin Oper_L_pert f+f^2g}, we included the factor $\beta^4$, but it depends on $\delta$.
\\

\noindent 
$\bullet$ In the unpinned FPUT case ($\delta = 1/2$) as was discussed in \cite{LukkarinenSpohn2008}, the linear decay is of order $t^{-3/5}$, \footnote{See also Theorem \ref{theo:decay linear semigr FPUT} where the same slow time decay is established in weighted $L^\infty$ spaces - even though the initial data required in Theorem \ref{theo:decay linear semigr FPUT} do not cover this setting here.} and so not sufficiently fast.
In particular Lukkarinen \& Spohn by a resolvent expansion they rigorously establish, \cite[Theorem  2.5 \& Corollary 2.6]{LukkarinenSpohn2008} that the Laplace transform of the current-current correlation function $C_T(t) = \operatorname{Cov}_T(J_{0,1}(t),J_{\text{total}(0)})$, which is given by (for $\lambda>0$ and given the heuristically derived formula \eqref{eq:Kin Conj identity})
\begin{align*} 
R(\lambda)  =  \langle \omega',(\widetilde{L} + \lambda)^{-1} \omega' \rangle = \langle\omega',\int_0^\infty e^{-\lambda t} e^{-t\widetilde{L}} \omega' \dd t \rangle  &= \int_0^\infty e^{-\lambda t} \langle\omega', e^{-t\widetilde{L}} \omega' \rangle \dd t  \\ & = c(T) \int_0^\infty e^{-\lambda t} C_T(t) \dd t
\end{align*}
for some constant $c(T)$, 
satisfies 
$$ R(\lambda) = \mathcal{O}(\lambda^{-2/5}), \quad \lambda \ll 1, $$
where $\widetilde{L} := f_\beta^{-1} L f_\beta^{-1}$ \footnote{This is up to constants depending on $\beta$, since in \cite{LukkarinenSpohn2009}, they analyse the operator $\omega L \omega$}. 

From this, by Tauberian theory which translates the small $\lambda$ behaviour of $R(\lambda)$ to large-time asymptotics of $C_T(t)$, they obtain 
$$C_T(t) = \mathcal{O}(t^{-3/5}), \quad t \gg 1 \ \text{ and so } \int_0^t C_T(\tau) \dd \tau = \mathcal{O}(t^{2/5}). $$ 
Thus the Green-Kubo time integral diverges in the weak turbulent/kinetic regime, which given the above heuristic discussion, this indicates anomalous conductivity in that system. In fact the same exponent $2/5$ appears in numerical simulations directly in the microscopic $N$-sized model, i.e. that $\kappa(N) = \mathcal{O}(N^{2/5})$ for $N \to \infty$, \cite{Lep16}. This also suggests that the Phonon Boltzmann approximation predicts well the microscopic conductivity. 

Before we close, let us note that the failure of Fourier's law here, is not due to a complete lack of thermalisation. In the physics literature, anomalous conductivity has been investigated through both microscopic simulations and kinetic descriptions, see for instance \cite{Onoratoetal1, Onoratoetal2, pereverzev2003fermi}. These works suggest that the anomalous scaling of the conductivity can be interpreted as the result of a coexistence of two transport regimes: the energy carried by low wavenumber modes is transported in a ballistic manner (in the sense that this part  contributes as $\kappa_N \sim N$), whereas in higher wavenumbers it is transported in a regular/diffusive manner, as expected in Fourier law. 




\section{Related equations}

\label{section_related}

\subsection{The discrete nonlinear Schr\"odinger equation} In its most basic version, it is the following evolution problem
\begin{equation}
\label{DNLS} \tag{DNLS}
i \dot{u_j} + \Delta_d u_j = \sigma |u_j|^2 u_j, \qquad u_j \in \mathbb{C}, \;\; j \in \mathbb{Z}
\end{equation}
(here, $\sigma = \pm 1$ is a fixed sign and the discrete Laplacian is defined by $\Delta_d u_j = -2 u_{j}+ u_{j-1} + u_{j+1}$). We refer to \cite{HennigTsironis,Kevrekidis} as entry points to the literature on this model. We notice immediately that \eqref{DNLS} 
enjoys two conserved quantities, the Hamiltonian $\mathcal{H}$ and mass $\mathcal{M}$
\begin{equation}
\label{hamiltonianNLS}
\mathscr{H}(\mathbf{u}) = \frac{1}{2} \sum_j |u_{j+1}-u_j|^2 + \frac{\sigma}{4} \sum_j |u_j|^4 \quad \mbox{and} \quad \mathcal{M} = \sum |u_j|^2.
\end{equation}
This should of course be contrasted with \eqref{NOC}, for which only the Hamiltonian is conserved. For \eqref{DNLS}, the conservation of mass and energy at the kinetic level is mirrored by the conservation of corresponding quantities at the microscopic level.

Variants of \eqref{DNLS} include choosing a higher-order nonlinearity, a non-local nonlinearity, or replacing the linear part of the equation $\Delta_d$  by a more general translation-invariant linear operator $L$
$$
L u_j = \sum_{k \in \mathbb{Z}} \alpha_k u_{j-k} \quad \Leftrightarrow \quad \widehat{L u}(\xi) = \widehat{\alpha}(\xi) \widehat{u} (\xi).
$$
(the linear stability condition is always satisfied since $\widehat{\alpha}$ is even and real-valued like $(\alpha_j)$). With a more general $L$, the equation is sometimes refered to as the Discrete Self-Trapping equation \cite{Eilbeck}.

A last variant consists in replacing the nonlinearity by $|u_j|^2(u_{j-1} + u_{j+1})$, in which case the equation is known as Ablowitz-Ladik and integrable \cite{AblowitzLadik}. Just like the Toda lattice was the integrable form of the \eqref{NOC} equation, the Ablowitz-Ladik equation is the integrable version of \eqref{DNLS}.

We now come to the kinetic theory, which can be found in \cite{ODPPBR,LukkarinenPirnesVuoksenmaa}. The kinetic wave equations which can be obtained is of the same kind as in the case of \eqref{NOC}; for instance, for the Hamiltonian \eqref{hamiltonianNLS}, one finds \eqref{KWE} with $K=1$ and $\omega = 2 \sin^2(\frac \xi 2)$. There is an important difference though: the dispersion relation is now $\omega = \widehat{\alpha}$ instead of $\omega = \sqrt{\widehat{\alpha}}$. In particular, it can take negative values!

\subsection{More nonlinear oscillator chains} There is a great variety of models of oscillator chains which seek to model physical effects occurring in different materials. We present a selection of such models, which would all yield a kinetic theory in the weakly turbulent regime, though this does not appear to have been explored.

\medskip

\noindent \underline{Diatomic chain.} To account for a crystal made up of atoms of different weights, the Hamiltonian of \eqref{NOC} is modified as follows
$$
\mathscr{H}(\mathbf{p},\mathbf{q}) = \frac{m_0}{2} \sum_j p_{2j}^2 + \frac{m_1}{2} \sum_j p_{2j+1}^2 + \sum_j V(q_{j+1}-q_j);
$$
see \cite{Askar,Maugin} for a physical discussion and \cite{GalganiGiorgilliMartinoliVanzini,PezziDengLvovLorenzoOnorato} for numerical experiments. This model results in two distinct dispersion relations, called \textit{acoustic} and \textit{optical}, the former vanishing at the origin but not the latter.

\medskip
\noindent \underline{Constant magnetic field.} In the model proposed and analyzed in \cite{SaitoSasada, SaitoSasadaSuda}, the displacement $q_j$ takes values in the plane $\mathbb{R}^2$ and is submitted to the Lorentz force from a constant magnetic field in the vertical direction. The Hamiltonian becomes
$$
\mathscr{H}(\mathbf{p},\mathbf{q}) = \frac 12 \sum_j \left| p_{j} - \frac{1}{2} B R q_j \right| ^2 + \sum_j V(q_{j+1}-q_j)
$$
which leads to the dynamics
$$
\begin{cases} 
\dot q_j = p_j \\
\dot p_j = \Delta_d q_j - B R p_j + V'(q_{j+1} - q_j) - V'(q_{j} - q_{j-1}) .
\end{cases}
$$

\medskip

\noindent \underline{Electrostatic effects.} To model an ionic crystal \cite{Askar}, one can imagine that the particles carry alternatively positive and negative charges. Submitting the chain to a constant electric field gives the Hamiltonian
$$
\mathscr{H}(\mathbf{p},\mathbf{q}) = \frac{1}{2} \sum_j p_{2j}^2 + \sum_j V(q_{j+1}-q_j) + E \left[ \sum_j  q_{2j}^2 - \sum_j q_{2j+1}^2 \right].
$$
For $E=0$, one can also take into account the interaction of the charges with the electrostatic field they generate collectively \cite{Askar}.

\medskip

\noindent \underline{Micropolar chain.} Micropolar materials exhibit additional degrees of freedom besides the displacement of the particles; this is typically the case for molecules which can rotate around different axes. A basic model is to consider dumbbells aligned along an axis, which can rotate (angle $\psi$) and be translated (displacement $q$). The following Hamiltonian is derived in \cite{Askar}
\begin{align*}
\mathscr{H}(\mathbf{p},\mathbf{q},\mathbf{\varphi},\mathbf{\psi}) & = \frac{m}{2} \sum p_j^2 + \frac{1}{2} \sum_j \left[I \varphi_j^2 + mp_j^2\right] + \sum_j \left[ (q_{j+1}-q_j)^2 + (\psi_j - (q_j-q_{j+1}))^2 \right.\\
& \qquad \qquad \qquad \qquad \left. + (\psi_{j+1} - (q_j-q_{j+1}))^2 + (\psi_j + \psi_{j+1} -2 (q_j-q_{j+1}))^2,
\right]
\end{align*}
where $(p,\varphi)$ are the conjugate variables to $(q,\psi)$. See \cite{Maugin} for many more references.

\medskip

\noindent \underline{Ferroelectric and piezoelectric chains.} We refer to \cite{Askar,Maugin} and references therein for these more elaborate models.

\subsection{Quantized equation} The classical Boltzmann equation is a model for gases of classical particles, but it has counterparts for bosons and fermions. These are known as Quantum Boltzmann equations, or, in the bosonic case, as Nordheim-Boltzmann and Uehling–Uhlenbeck, after the authors who wrote it down first \cite{Nordheim,UehlingUhlenbeck}. The theory of these equations for bosons and fermions has been developed by a number of mathematicians, see for instance \cite{Dolbeault,EscobedoVelazquez2,Lu1,Lu2} and the introductory texts \cite{PomeauTran,Villani}.

In the case of crystals, such a quantum kinetic equation was already derived in Peierls' foundational paper \cite{Peierls} of 1929 which marked the beginning of wave turbulence. This same year saw the publication of Nordheim's article \cite{Nordheim} which was already mentioned. This shows how intertwined the classical and quantum cases are from a physical perspective.

Turning to the nonlinear oscillator chain \eqref{NOC}, a kinetic equation similar to \eqref{KWE} can be derived from its quantized version \cite{Spohn2005,LukkarinenSpohn2009,LukkarinenPirnesVuoksenmaa}. It takes the following form
\begin{equation}
\label{QKE} \tag{$QKE_\pm$}
\partial_t f (t,\xi) = \mathcal{C}_{\pm}[f(t)](\xi)
\end{equation}
where
\begin{align*}
\mathcal{C}_{\pm}[f](\xi) & = \int K_{0,1,2,3} \left[\pm \left( f_1 f_2 f_3 + f_0 f_2 f_3 - f_0 f_1 f_2 - f_0 f_1 f_3 \right) + f_2 f_3 - f_0 f_1 \right] \\ 
& \qquad \qquad\qquad \qquad \qquad \qquad \qquad \qquad \qquad \qquad \delta(\Omega_{0,1,2,3}) \delta (\Sigma_{0,1,2,3}) \dd \xi_{1,2,3},
\end{align*}
the bosonic and fermionic cases corresponding to $+$ and $-$ respectively. 

For bosons, the only restriction is $f \geq 0$ as is the case for \eqref{KWE}; for fermions, the unknown $f$ takes values in $[0,1]$. One checks that this is propagated by the equation since, as long as $f_1,f_2,f_3 \in [0,1]$
\begin{align*}
&- \left( f_1 f_2 f_3 + f_0 f_2 f_3 - f_0 f_1 f_2 - f_0 f_1 f_3 \right) + f_2 f_3 - f_0 f_1 \\
& \qquad \qquad \qquad \qquad \qquad \qquad  \qquad \qquad \qquad \qquad  = \begin{cases} 
(1- f_1) f_2 f_3 \geq 0 & \mbox{if $f_0=0$} \\
 f_1 (f_2-1)(1-f_3) \leq 0 & \mbox{if $f_0=1$}.
\end{cases}
\end{align*}

The mass and energy as still conserved with identical definitions to \eqref{KWE}, but the entropy becomes
$$
S_{\pm} = \int \left[ (1 \pm f(\xi)) \log (1 \pm f(\xi)) \mp f(\xi) \log f(\xi) \right] \dd \xi
$$

The stationary (Rayleigh-Jeans) solutions are modified to
$$
\mathfrak{f}_{\beta,\gamma}^{\pm}(\xi) = \frac{1}{e^{-\beta \omega(\xi) + \gamma} \mp 1}
$$

We are not aware of any rigorous result for the quantum kinetic equation \eqref{QKE}, but because of the similarities in the formulas, it is to be expected that the theories for \eqref{KWE} and \eqref{QKE} are closely related. To be more specific, \eqref{KWE} can be viewed as the large field limit of \eqref{QKE}: indeed, $\mathcal{C}_+$ and $S_+$ become equivalent to $\mathcal{C}$ and $S$ as $f \to \infty$ (pointwise), while $\mathfrak{f}_{\beta,\gamma}^{+}(\xi) \sim  \mathfrak{f}_{\beta,\gamma}(\xi)$ as $\beta,\gamma \to 0$.

\subsection{Higher-dimensional case} The obvious generalization of the one-dimensional case is given by unknown functions $(p_j,q_j)$ indexed by $j \in \mathbb{Z}^d$ and a Hamiltonian
$$
\mathscr{H}(\mathbf{p},\mathbf{q}) = \frac{1}{2} \sum_j p_j^2 + F(\mathbf q),
$$
where the potential function $F$ is invariant by translation in $j$: $F((q_j)) = F((q_{j+a}))$ for any $a \in \mathbb{Z}^d$, it is for instance considered in \cite{Prigogine}.

For the model
$$
\mathscr{H}(\mathbf{p},\mathbf{q}) = \frac{1}{2} \sum_j p_j^2 + \frac{1}{2} \sum_j \alpha_k q_{j} q_{j+k} + \frac{\lambda}3 \sum_j q_j^3
$$
set on $\mathbb{Z}^3$, Spohn proceeded in the foundational paper \cite{Spohn06} to develop the kinetic theory and spell out its consequences; this was not done in full rigor but gives a great panorama of the subject.

The rigorous theory remains to be written! An additional difficulty compared to the one-dimensional case lies in the more involved geometry of the resonant set. However, the derivation of the kinetic model might be more accessible, since the one-dimensional case presents some specific obstacles pointed out in \cite{Vassilev,BoyangWu}.

\subsection{More realistic crystals?}
The work of Peierls \cite{Peierls} was the starting point of the kinetic theory of phonons and of the theory of wave turbulence. Is it possible to give a rigorous - or even a precise - treatment of this theory in general?

To answer this question, we must start with the microscopic structure of a crystal. 
A crystal has a periodic structure, which can be described by the basis $(e_j)_{j=1,\dots,d}$ of a lattice in $\mathbb{R}^d$. Each cell of this lattice can be labeled by $\mathbf{j} = (j_1,\dots,j_d) \in \mathbb{Z}$) and atoms within each cell are indexed by $a \in \{1, \dots ,A\}$. The displacement of the $a$-th atom from its rest position is $q^a_{j_1,\dots,j_d}$, and its momentum $p^a_{j_1,\dots,j_d}$. The potential energy of the crystal can then be written in full generality \cite{Srivastava} as a function $F$ of all the displacements $q^a_{\mathbf{j}}$, and its Hamiltonian as
$$
\mathscr{H}(\mathbf{p},\mathbf{q}) = \frac{1}{2} \sum_{\mathbf{j},a} \frac{|p^a_{\mathbf{j}}|^2}{m_a} + F(\mathbf{q}),
$$
where $m_a$ is the mass of the corresponding atom and the potential function $F$ is given by the Physics of the crystal, and must in particular respect its symmetries.

The potential  can be expanded to third order as
\begin{align*}
F(\mathbf{q}) & = F_2 (\mathbf{q}) + F_3(\mathbf{q}) + \{ \mbox{higher order} \} \\
& = \sum_{\substack{a,a',\mathbf{j},\mathbf{j}'}} c_2(\mathbf{j},\mathbf{j}',a,a') q^a_{\mathbf{j}} q^{a'}_{\mathbf{j}'} + \sum_{\substack{a,a',\mathbf{j},\mathbf{j}'}} c_3(\mathbf{j},\mathbf{j}',\mathbf{j}'',a,a',a'') q^a_{\mathbf{j}} q^{a'}_{\mathbf{j}'} + q^{a''}_{\mathbf{j}''} + \{ \mbox{higher order} \} 
\end{align*}

The quadratic terms in the potential energy make up the Born-von Karman model and give the Hamiltonian
$$
\mathscr{H}_{\operatorname{quad}}(p,q) = \frac{1}{2} \sum_{\mathbf{j},a} \frac{|p^a_{\mathbf{j}}|^2}{m_a} + F(\mathbf{q}).
$$
This so-called harmonic Hamiltonian can be simplified through additional assumptions \cite{Callaway,Srivastava,Ziman}; it determines the dispersion relations of the model.

The next step taken in all textbooks \cite{Callaway,Prigogine,Srivastava,Ziman} is to quantize the model; indeed, the classical case receives little attention. Finally, the formula for the kinetic equation can be established following the same procedure as was demonstrated in Section \ref{Sec: Derivation_WKE}.

It it possible and meaningful to give a rigorous treatment of this derivation? The technical complications and the variety of possible configurations seem daunting in full generality. However, it might be possible to treat the case of a 'generic' crystal by leaving aside exceptional configurations.

\section{Open Mathematical Questions on the kinetic limit of \eqref{NOC}}

\label{section_open}

\begin{itemize}
\item A rigorous derivation of \eqref{KWE} from \eqref{NOC} is still missing. This is true for all models of \eqref{NOC}, and in the inhomogeneous as well as homogeneous cases. Reaching the kinetic time scale for such a derivation would be very impressive, and longer time scales even more impressive!

\medskip

\item Such a derivation should be simpler - but still very interesting - if the microscopic model involves random forcing.

\medskip

\item It is natural to ask whether the homogeneous \eqref{KWE} admits global smooth solutions for smooth data, or whether singularity formation is possible. Assuming that a solution is global, it is natural to expect that the RJ states are global attractors, but this has not been proved.

\medskip

\item Turning to the inhomogeneous case, set on the line, the torus, or the interval with appropriate boundary conditions, the question of global regularity and asymptotic behavior is also very important. In the case of a compact domain, the attractor should be a global RJ state.

\medskip

\item A question which is less physically relevant, but mathematically very natural, is that of the lowest regularity for local well-posedness of \eqref{KWE} in the homogeneous as well as inhomogeneous cases.

\medskip

\item For the inhomogeneous \eqref{KWE} set on the interval with non-trivial and non-matching boundary conditions, is there a stationary solution?

\medskip

\item Is it possible to justify mathematically the hydrodynamic limits of Section \ref{sec:Hydro_limit} in the non-degenerate case?

\medskip

\item In the degenerate case, only the linearized case seems to have been investigated, even formally. Is it possible to identify the correct limiting nonlinear problem? Can these hydrodynamic limits be made rigorous?

\end{itemize}


\bibliographystyle{alpha}
\bibliography{references}

\end{document}